\documentclass{llncs}
\usepackage{hyperref}
\usepackage{graphicx}
\usepackage{amssymb} 
\usepackage{url}

\usepackage{xspace}
\usepackage{graphicx}
\usepackage{url}
\usepackage{amsmath}
\usepackage{xcolor}
\usepackage{algorithm2e}
\usepackage{mathptmx}
\usepackage{tabularx}
\usepackage{subcaption}
\usepackage{array}

\usepackage{placeins}

\usepackage{wasysym}
\usepackage{wrapfig}

\usepackage{pgf}
\usepackage{tikz}
\usepackage{pgfplots}
\usetikzlibrary{matrix,trees,arrows,arrows.meta,fit,shapes,positioning}
\usetikzlibrary{decorations.pathmorphing,decorations.pathreplacing,decorations.markings,decorations} 

\urldef{\mailsa}\path|{alfred.hofmann, ursula.barth, ingrid.haas, frank.holzwarth,|
\urldef{\mailsb}\path|anna.kramer, leonie.kunz, christine.reiss, nicole.sator,|
\urldef{\mailsc}\path|erika.siebert-cole, peter.strasser, lncs}@springer.com|   

\newcommand{\kwd}[1]{\ensuremath{\mathsf{#1}}}

\newcommand{\Lit}{\ensuremath{\mathcal{L}}\xspace}
\newcommand{\ltrl}{\ensuremath{l}\xspace}

\newcommand{\update}{\oplus}

\newcommand{\Obl}{\ensuremath{\mathcal{O}}\xspace}

\newcommand{\Obls}{\circledcirc}
\newcommand{\Ob}{\ensuremath{\odot}}

\newcommand{\ltrls}{\ensuremath{L}\xspace}
\newcommand{\set}[1]{\{ #1 \}}
\newcommand{\seq}[1]{( #1 )}

\newcommand{\exe}{\ensuremath{\epsilon}}
\newcommand{\Exe}[1]{\Sigma(#1)}

\newcommand{\trace}{\ensuremath{\theta}}
\newcommand{\traces}[1][]{%
  \@ifmtarg{#1}%
  {\ensuremath{\Theta}}%
  {\ensuremath{\Theta}(#1)}%
  }
  
\newcommand{\PCompliance}[3][]{#2%
  \;\vdash^{\hspace{-0.1cm}\raisebox{0.5mm}{\tiny #1}}\; #3}
\newcommand{\NPCompliance}[3][]{#2%
  \;\not\vdash^{\hspace{-0.1cm}\raisebox{0.5mm}{\tiny #1}}\; #3}

\newcommand{\red}[1]{{\color{red} #1}}

\newcolumntype{?}{!{\vrule width 1pt}}

\newcolumntype{L}[1]{>{\raggedright\let\newline\\\arraybackslash\hspace{0pt}}m{#1}}
\newcolumntype{C}[1]{>{\centering\let\newline\\\arraybackslash\hspace{0pt}}m{#1}}
\newcolumntype{R}[1]{>{\raggedleft\let\newline\\\arraybackslash\hspace{0pt}}m{#1}}

\newtheorem{assumption}{Assumption}

\begin{document}

\pgfdeclarelayer{background}
\pgfdeclarelayer{foreground}
\pgfsetlayers{background,main,foreground}

\markboth{\LaTeXe{} Class for Lecture Notes in Computer
Science}{\LaTeXe{} Class for Lecture Notes in Computer Science}
\title{Business Process Full Compliance with Respect to a Set of Conditional Obligation in Polynomial Time}


%
%
\author{Silvano Colombo Tosatto\and Guido Governatori \and Nick Van Beest}


\authorrunning{Silvano Colombo Tosatto \and Guido Governatori \and Nick Van Beest}

\institute{Data61, CSIRO, Dutton Park, Australia\\
\email{\{silvano.colombotosatto;guido.governatori;nick.vanbeest\}@data61.csiro.au}}

%
%

\toctitle{Lecture Notes in Computer Science}
\tocauthor{Authors' Instructions}
\maketitle

\begin{abstract}
In this paper, we present a new methodology to evaluate whether a business process model is fully
compliant with a regulatory framework composed of a set of conditional obligations. The methodology
is based failure $\Delta$-constraints that are evaluated on bottom-up aggregations of a tree-like representation of business process
models.
While the generic problem of proving full compliance is in co{\bf NP}-complete, we show that
verifying full compliance can be done in polynomial time using our methodology, for an acyclic
structured process model given a regulatory framework composed by a set of conditional obligations,
whose elements are restricted to be represented by propositional literals.

\end{abstract}

\section{Introduction}

Business process models are often used by companies and organisations to compactly represent how their business goals and objectives are achieved. It is, therefore, vital for the stakeholders to show that such business process models comply with the relevant regulations in place.

Various approaches have been proposed in the past, proving whether business process model are compliant with a regulatory framework~\cite{knuplesch2010enabling,elgammal2016formalizing,hoffmann11compliance,Governatori2013RegorousAB}. However, as shown by Colombo Tosatto et al.~\cite{10.1109/TSC.2014.2341236}, the generic problem of proving regulatory compliance is complex enough to be in \textbf{NP}-complete. Colombo Tosatto et al.~\cite{BPM2019_complexity} also show that when considering the problem of proving \emph{partial compliance} of business process models (i.e. checking whether the model contains at least one compliant execution), the majority of the variants of the problem are indeed in \textbf{NP}-complete. 

In this paper, we consider the problem of proving \emph{full compliance} of a business process, to ensure that every execution of the model is compliant with the given regulatory framework.
While the generic problem of proving full compliance is in co\textbf{NP}-complete~\cite{10.1109/TSC.2014.2341236}, we show that verifying full compliance can be done in polynomial time for an acyclic structured process model given a regulatory framework composed by a set of conditional obligations, whose elements are restricted to be represented by propositional literals.

The proposed approach is based on verifying an alternative representation of the constraints defined by the regulatory frameworks, referred to as failure $\Delta$-constraints by Colombo Tosatto et al.~\cite{Mirel2019_delta_c}. These constraints are evaluated on a tree-like structure representing the business process model, adopting a \emph{divide \& conquer} strategy by independently evaluating the leaves and cumulatively aggregating the results towards the root of the process tree.

The remainder of the paper is structured as follows: Section \ref{sec:problem} introduces the studied variant of the compliance problem. Section \ref{sec:deltaconstraints} introduces the failure $\Delta$-constraints. Subsequently, the basic principles of the evaluation procedure is described in Section~\ref{sec:stem}. Next, the aggregation of the results towards the root of the process tree is presented in Section~\ref{sec:aggregations}, followed by a conclusion of the work in Section~\ref{sec:conclusion}.

\section{Proving Regulatory Compliance}\label{sec:problem}

In this preliminary section, we introduce the problem of proving regulatory compliance of business process models. As the goal of this paper is to provide a solution, which is in polynomial time with respect to the size of the problem, we focus on a simplified variant rather than the general problem. Similar to every regulatory compliance problem, this variant is composed by two distinct components, which we introduce separately. First, we introduce the business process model, describing the possible executions capable of achieving an objective. Second, we introduce the regulatory framework, describing the compliance requirements that the executions of a process model must follow.

\subsection{Acyclic Structured Business Processes}

We focus our approach on solving a variant of the problem involving acyclic structured process models. We limit our approach to structured process models, as the soundness\footnote{A process is sound, as defined by van der Aalst \cite{Aalst:1997:VWN:647744.733919,Aalst1998workflow}, if it avoids livelocks and deadlocks.} of such models can be verified in polynomial time with respect to their size. These processes are similar to structured workflows defined by Kiepuszewski et al. \cite{Kiepuszewski:2000:SWM:646088.679917}. An additional advantage of limiting ourselves to deal with these kind of process models is that their structure is composed by properly nested components\footnote{Considering a structured process model as a set of nested components, by properly nested we mean that every component, except the outer process model itself, always has exactly one parent compoent.}, can be exploited to navigate the structure using an alternative \emph{tree-like} representation.

Finally, the advantage of considering acyclic process models, is that the number of possible executions of a model is always finite. When considering process models including cycles, on the other hand, unless an upper limit of iteration is enforced on a cycle, a cycle can be run any number of times and to each number of possible iterations corresponds a potential execution of the model. While an infinite execution does not make sense, as to be considered a possible execution of the model, it must be allowed to terminate. It is still easy to see how such cycles can lead to models containing an infinite amount of executions. In this paper, we focus on models foregoing such structures. 


\begin{definition}[Process Block]\label{def:pb}
A process block $B$ is a directed graph: the nodes are called \emph{elements} and the directed edges
are called \emph{arcs}. The set of elements of a process block are identified by the function $V(B)$
and the set of arcs by the function $E(B)$. The set of elements is composed of tasks and
coordinators. There are 4 types of coordinators: \kwd{and\_split}, \kwd{and\_join}, \kwd{xor\_split}
and \kwd{xor\_join}. Each process block $B$ has two distinguished nodes called the \emph{initial}
and \emph{final} element. The initial element has no incoming arc from other elements in $B$ and is
denoted by $b(B)$. Similarly the final element has no outgoing arcs to other elements in $B$ and is
denoted by $f(B)$.

A directed graph composing a process block is defined inductively as follows:
\begin{itemize}
\item A single task constitutes a process block. The task is both initial and final element of the block.
\item Let $B_1, \dots , B_n$ be distinct process blocks with $n > 1$:
\begin{itemize} 
\item $\kwd{SEQ}(B_1, \dots, B_n)$ denotes the process block with node set $\bigcup_{i=0}^{n}V(B_i)$ and edge set $\bigcup_{i=0}^{n}(E(B_i) \cup  \{(f(B_i), b(B_{i+1})): 1 \leq i < n\})$. \\The initial element of $\kwd{SEQ}(B_1, \dots, B_n)$ is $b(B_1)$ and its final element is $f(B_n)$.
\item $\kwd{XOR}(B_1, \dots, B_n)$ denotes the block with vertex set $\bigcup_{i=0}^{n} V(B_i) \cup \{\kwd{xsplit}, \kwd{xjoin}\}$ and edge set $\bigcup_{i=0}^{n}( E(B_i) \cup  \{(\kwd{xsplit}, b(B_i)), (f(B_i), \kwd{xjoin}): 1 \leq i \leq n\})$ where $\kwd{xsplit}$ and $\kwd{xjoin}$ respectively denote an \kwd{xor\_split} coordinator and an \kwd{xor\_join} coordinator, respectively. The initial element of $\kwd{XOR}(B_1, \dots, B_n)$ is $\kwd{xsplit}$ and its final element is $\kwd{xjoin}$.
\item $\kwd{AND}(B_1, \dots, B_n)$ denotes the block with vertex set $\bigcup_{i=0}^{n} V(B_i) \cup \{\kwd{asplit}, \kwd{ajoin}\}$ and edge set $\bigcup_{i=0}^{n}( E(B_i) \cup  \{(\kwd{asplit}, b(B_i)), (f(B_i), \kwd{ajoin}): 1 \leq i \leq n\})$ where $\kwd{asplit}$ and $\kwd{ajoin}$ denote an \kwd{and\_split} and an \kwd{and\_join} coordinator, respectively. The initial element of $\kwd{AND}(B_1, \dots, B_n)$ is $\kwd{asplit}$ and its final element is $\kwd{ajoin}$.
\end{itemize}
\end{itemize}

\end{definition}

By enclosing a process block as defined in Definition~\ref{def:pb} along with a \kwd{start} and \kwd{end} task in a sequence block, we obtain a \emph{structured process model}.
The effects of executing the tasks contained in a business process model are described using annotations as shown in Definition \ref{def:annpro}.

\begin{definition}[Annotated process]\label{def:annpro}
Let $P$ be a structured process and let $T$ be the set of tasks contained in $P$. An annotated process is a pair $(P, \kwd{ann})$, where \kwd{ann} is a function associating a consistent set of literals to each task in $T$: $\kwd{ann}:T\mapsto 2^{\Lit}$. The function $\kwd{ann}$ is constrained to the consistent literals sets in $2^\Lit$.
\end{definition}

The update between the states of a process during its execution is inspired by the AGM belief revision operator~\cite{Alchourron1985-ALCOTL-2} and is used in the context of business processes to define the transitions between states, which in turn are used to define the \emph{traces}.

\begin{definition}[State update]\label{def:litupd}
Given two consistent sets of literals $\ltrls_1$ and $\ltrls_2$, representing the process state and the annotation of a task being executed, the
  update of $\ltrls_1$ with $\ltrls_2$, denoted by $\ltrls_1\update \ltrls_2$ is a set of literals defined as
  follows:
$$ \ltrls_1\update \ltrls_2 = \ltrls_1\setminus\set{\overline{\ltrl}\mid
  \ltrl\in \ltrls_2}\cup \ltrls_2
$$
\end{definition}

\begin{definition}[Executions and Traces]\label{def:ser}
Let $\mathbb{P}_1 = (\mathcal{S}_1, \prec_{s_1})$ and $\mathbb{P}_2 = (\mathcal{S}_2, \prec_{s_2})$ be partial ordered sets, we define the following four operations:
\begin{itemize}
\item Union: $\mathbb{P}_1 \cup_\mathbb{P} \mathbb{P}_2 = (\mathcal{S}_1 \cup \mathcal{S}_2, \prec_{s_1} \cup \prec_{s_2})$, where $\cup$ is the disjoint union.
\item Intersection:  $\mathbb{P}_1 \cap_\mathbb{P} \mathbb{P}_2 = (\mathcal{S}_1 \cap \mathcal{S}_2, \prec_{s_1} \cap \prec_{s_2})$
\item Concatenation: $\mathbb{P}_1 +_\mathbb{P} \mathbb{P}_2 = (\mathcal{S}_1 \cup \mathcal{S}_2, \prec_{s_1} \cup \prec_{s_2} \cup \{s_1 \prec s_2 | s_1 \in \mathcal{S}_1 \mbox{ and } s_2 \in \mathcal{S}_2\})$.
\item Linear Extensions: $\mathcal{E}(\mathbb{P}_1) = \{(\mathcal{S}, \prec_s)| \mathcal{S} = \mathcal{S}_1, (\mathcal{S}, \prec_s) \mbox{ is a sequence and }$ $\prec_{s_1}$ $\subseteq$ $\prec_s\}$.
\end{itemize}
The \emph{associative} property holds for Union, Intersection and Concatenation.

Given a structured process model identified by a process block $B$, the set of its executions, written $\Exe{B} = \{\exe | \exe \mbox{ is a sequence and is an execution of $B$}\}$. The function $\Exe{B}$ is defined as follows:

\begin{enumerate}
\item If $B$ is a task  $t$, then $\Exe{B} = \{(\{t\}, \emptyset)\}$
\item if $B$ is a composite block with subblocks $B_1,\dots,B_n$ let $\exe_i$ be the projection of $\exe$ on block $B_i$ (obtained by ignoring all tasks which do not belong to $B_i$)
\begin{enumerate}
\item If $B = \kwd{SEQ}(B_1, \dots, B_n)$, then $\Exe{B} = \{\exe_1 +_\mathbb{P} \dots +_\mathbb{P}  \exe_n | \exe_i \in \Exe{B_i}\}$ 
\item If $B = \kwd{XOR}(B_1, \dots, B_n)$, then $\Exe{B} = \Exe{B_1} \cup \dots \cup \Exe{B_n}$
\item If $B = \kwd{AND}(B_1, \dots, B_n)$, then $\bigcup_{\exe_1,\dots,\exe_n  } \mathcal{E}(\exe_1\cup_\mathbb{P} \dots \cup_\mathbb{P} \exe_n | \forall\exe_i \in \Exe{B_i})$
\end{enumerate}
\end{enumerate}

Given an annotated process $(B, \kwd{ann})$ and an execution $\exe = \seq{t_1, \ldots, t_n}$ such that $\exe \in \Exe{B}$. A trace $\trace$ is a finite sequence of states: $\seq{ \sigma_1, \ldots ,\sigma_n}$. Each state of $\sigma_i \in \trace$ contains a set of literals $\ltrls_i$ capturing what holds after the execution of a task $t_i$.  Each $\ltrls_i$ is a set of literals such that:

  \begin{enumerate}
  \item $\ltrls_0 = \emptyset$
  \item $\ltrls_{i+1} = \ltrls_i\update\kwd{ann}(t_{i+1})$, for $1\le i < n$.

  \end{enumerate}

To denote the set of possible traces resulting from an process model $(B,\kwd{ann})$, we use $\Theta(B,\kwd{ann})$.
\end{definition}


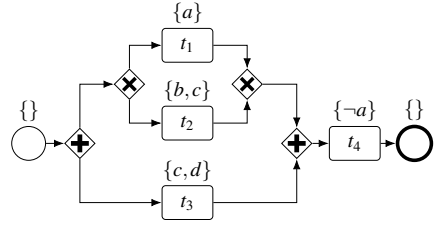
\begin{wrapfigure}{r}{6cm}
\centering
\scalebox{0.45}{\begin{tikzpicture}[thick,font=\LARGE]

\newcommand{\gatesize}{0.9cm}

\node[draw,align=center,circle, minimum size=1cm] (start) at (0,0) {};

\node[xshift=1cm,draw,diamond,align=center,minimum height=\gatesize,minimum width=\gatesize] (and1) at (start.east) {};
\draw[line width=4pt] ([yshift=-2mm]and1.north) -- ([yshift=2mm]and1.south);
\draw[line width=4pt] ([xshift=-2mm]and1.east) -- ([xshift=2mm]and1.west);

\node[xshift=1cm,yshift=1.75cm,draw,diamond,align=center,minimum height=\gatesize,minimum width=\gatesize] (xor1) at (and1.east) {};
\draw[line width=3pt] ([xshift=0.5mm,yshift=-0.5mm]xor1.north west) -- ([xshift=-0.5mm,yshift=0.5mm]xor1.south east);
\draw[line width=3pt] ([xshift=-0.5mm,yshift=-0.5mm]xor1.north east) -- ([xshift=0.5mm,yshift=0.5mm]xor1.south west);

\node[xshift=1.25cm,yshift=1.15cm,draw,rounded corners=3pt,align=center,minimum height=1cm,minimum width=1.5cm] (t1) at (xor1.east) {$t_1$};
\node[xshift=1.25cm,yshift=-1.15cm,draw,rounded corners=3pt,align=center,minimum height=1cm,minimum width=1.5cm] (t2) at (xor1.east) {$t_2$};
\node[yshift=-1.75cm,draw,rounded corners=3pt,align=center,minimum height=1cm,minimum width=1.5cm] (t3) at (t2 |- and1) {$t_3$};

\node[xshift=1cm,draw,diamond,align=center,minimum height=\gatesize,minimum width=\gatesize] (xor2) at (t1.east |- xor1) {};
\draw[line width=3pt] ([xshift=0.5mm,yshift=-0.5mm]xor2.north west) -- ([xshift=-0.5mm,yshift=0.5mm]xor2.south east);
\draw[line width=3pt] ([xshift=-0.5mm,yshift=-0.5mm]xor2.north east) -- ([xshift=0.5mm,yshift=0.5mm]xor2.south west);

\node[xshift=1cm,draw,diamond,align=center,minimum height=\gatesize,minimum width=\gatesize] (and2) at (xor2.east |- and1) {};
\draw[line width=4pt] ([yshift=-2mm]and2.north) -- ([yshift=2mm]and2.south);
\draw[line width=4pt] ([xshift=-2mm]and2.east) -- ([xshift=2mm]and2.west);

\node[xshift=1.25cm,draw,rounded corners=3pt,align=center,minimum height=1cm,minimum width=1.5cm] (t4) at (and2.east) {$t_4$};

\node[xshift=1cm,draw,line width=3pt,align=center,circle, minimum size=1cm] (end) at (t4.east) {};

\node[anchor=south] at (start.north) {\{\}};
\node[anchor=south] at (t1.north) {\{$a$\}};
\node[anchor=south] at (t2.north) {\{$b,c$\}};
\node[anchor=south] at (t3.north) {\{$c,d$\}};
\node[anchor=south] at (t4.north) {\{$\neg a$\}};
\node[anchor=south] at (end.north) {\{\}};

\draw[->,-{Latex[length=3mm]}] (start) -- (and1);
\draw[->,-{Latex[length=3mm]}] (and1) -- (and1 |- xor1) -- (xor1);
\draw[->,-{Latex[length=3mm]}] (and1) -- (and1 |- t3) -- (t3);
\draw[->,-{Latex[length=3mm]}] (xor1) -- (xor1 |- t1) -- (t1);
\draw[->,-{Latex[length=3mm]}] (xor1) -- (xor1 |- t2) -- (t2);
\draw[->,-{Latex[length=3mm]}] (t1) -- (xor2 |- t1) -- (xor2);
\draw[->,-{Latex[length=3mm]}] (t2) -- (xor2 |- t2) -- (xor2);
\draw[->,-{Latex[length=3mm]}] (xor2) -- (and2 |- xor2) -- (and2);
\draw[->,-{Latex[length=3mm]}] (t3) -- (and2 |- t3) -- (and2);
\draw[->,-{Latex[length=3mm]}] (and2) -- (t4);
\draw[->,-{Latex[length=3mm]}] (t4) -- (end);

\end{tikzpicture}}
\caption{An annotated process}\label{ex:processexample}
\label{f:p03}
\vspace{-18pt}
\end{wrapfigure}
\vspace{-12pt}
\textit{\example Annotated Process Model}. Fig. \ref{f:p03} shows a structured process containing four tasks labelled $t_1, t_2, t_3$ and $ t_4$ and their annotations. The process contains an \kwd{AND} block followed by a task and an \kwd{XOR} block nested within the \kwd{AND} block. The annotations indicate what has to hold after a task is executed. If $t_1$ is executed, then the literal $a$ has to hold in that state of the process.


\begin{table*}
\center
\begin{tabular}{ccc}
$\Exe{B}$  &\vline& $\Theta(B, \kwd{ann})$ \\\hline
\small$(\kwd{start},t_1, t_3, t_4,\kwd{end})$&\vline  & \small$((\kwd{start}, \emptyset), (t_1, \{a\}), (t_3, \{a,c,d\}), (t_4, \{\neg a,c,d\}), (\kwd{end}, \{\neg a,c,d\}))$\\
\small$(\kwd{start},t_2, t_3, t_4\kwd{end})$&\vline &  \small$((\kwd{start}, \emptyset), (t_2, \{b,c\}), (t_3,\{b,c,d\}), (t_4,\{\neg a,b,c,d\}), (\kwd{end},\{\neg a,b,c,d\}))$ \\
\small$(\kwd{start},t_3, t_1, t_4\kwd{end})$&\vline & \small$((\kwd{start}, \emptyset), (t_3, \{c,d\}), (t_1, \{a,c,d\}), (t_4,\{\neg a,c,d\}), (\kwd{end},\{\neg a,c,d\}))$ \\
\small$(\kwd{start},t_3, t_2, t_4\kwd{end})$&\vline & \small$((\kwd{start}, \emptyset), (t_3, \{c,d\}), (t_2, \{b,c,d\}), (t_4.\{\neg a,b,c,d\}), (\kwd{end},\{\neg a,b,c,d\}))$
\vspace{0.3cm}
\end{tabular}
\caption{Executions and Traces of the annotated process in Fig. \ref{f:p03}.}\label{tab:01}
\vspace{-18pt}
\end{table*}


\subsection{Regulatory Framework}\label{s:basic}

We now introduce now the regulatory framework, the requirements that a business process model needs to follow to be considered compliant. We aim to solve a variant of the problem, hence we introduce in this section conditional obligations. Considering that a regulatory framework is composed by a set of them, then an execution is required to comply with every obligation belonging to the set. One of the limitations we are applying to the regulatory framework is that the components defining the obligations are limited to propositional literals. We use a subset of Process Compliance Logic, introduced by Governatori and Rotolo~\cite{PCL}, to describe the conditional obligations composing the regulatory framework.

\begin{definition}[Conditional Obligation]\label{def:obligation} A local obligation $\Ob$ is a tuple $\langle o,r,t,d\rangle$, where $o\in\{a,m\}$ and represents the type of the obligation. The elements $c, t$ and $d$ are propositional literals in $\mathcal{L}$. $r$ is the requirement of the obligation, $t$ is the trigger of the obligation and $d$ is the deadline of the obligation. 

We use the notation $\Ob = \Obl^o \langle r, t, d\rangle$ to represent a conditional obligation.
\end{definition}

Recalling one of the limitations of the variant of the problem being solved in this paper, the elements of a conditional obligation are limited to be represented by propositional literals. Thus, the \emph{requirement}, \emph{trigger}, and \emph{deadline} of a conditional obligation are represented using propositional literals. These elements are evaluated against the states of a trace, as illustrated for instance in Table \ref{tab:01}, as the elements are propositional literals, they are considered to be positively satisfied if they are contained in a state.

Given a trace of a business process model and a conditional obligation, if a state of the trace satisfies the the obligation's triggering element, then the obligation is set in force. Additionally, when an obligation is in force over a trace, and the deadline element is satisfied by a state of a trace, then conditional obligation ceases to be in force. Figure \ref{f:obl} graphically illustrates the in force interval of a conditional obligation over a trace of a business process model. 

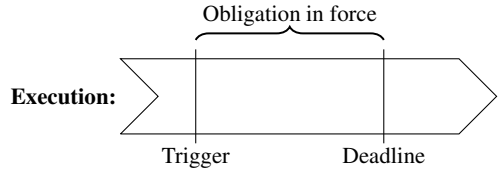
\begin{figure}[h]
\centering
\scalebox{0.5}{
\begin{tikzpicture}[thick, font=\LARGE]

\linespread{0.75}
\newcommand{\minsize}{1.5cm}
\newcommand{\wi}{10cm}
\newcommand{\hi}{2cm}
\newcommand{\incline}{1cm}
\newcommand{\trigger}{0.2*\wi}
\newcommand{\deadline}{0.7*\wi}

\draw (0,0) -- ([xshift=-\incline]\wi,0) -- (\wi, 0.5*\hi) -- ([xshift=-\incline]\wi,\hi) -- (0,\hi) -- (\incline,0.5*\hi) -- (0,0);
\node[anchor=east] at (0,0.5*\hi) {\bf Execution:};

\draw ([yshift=0.25cm]\trigger,\hi) -- (\trigger,-0.25);
\node[anchor=north] at (\trigger,-0.25) {Trigger};

\draw ([yshift=0.25cm]\deadline,\hi) -- (\deadline,-0.25);
\node[anchor=north] at (\deadline,-0.25) {Deadline};

\draw [decorate,decoration={brace,amplitude=10pt},line width=1.5pt] ([yshift=0.5cm]\trigger,\hi) -- ([yshift=0.5cm]\deadline,\hi);

\node[yshift=0.75cm,anchor=south] at (0.5*\deadline+0.5*\trigger,\hi) {Obligation in force};

\end{tikzpicture}
}
\caption{Obligation in force}\label{f:obl}
\end{figure}


Notice that when a conditional obligation is in force, then the requirement element of the obligation is evaluated within the in force interval to determine whether the obligation is satisfied or violated in that particular in force interval. How the requirement element is evaluated depends on the type of the obligation. We consider two types of obligations, \emph{achievement} and \emph{maintenance}:

\begin{description}
\item[Achievement] When this type of obligation is in force, the requirement specified by the regulation must be satisfied by at least one state within the in force interval, in other words, before the deadline is satisfied. When this is the case, the obligation in force is considered to be satisfied, otherwise it is violated.
\item[Maintenance] When this type of obligation is in force, the requirement must be satisfied \emph{continuously} in every state of the in force interval, until the deadline is satisfied. Again, if this is the case, the obligation in force is then satisfied, otherwise it is violated.
\end{description}

Before proceeding with the formal introduction of the different types of obligations, we introduce two syntactical shorthands to keep the subsequent definitions more compact.

\begin{definition}[Syntax Shorthands]\label{def:abbr}
To avoid cluttering, we adopt the following shorthands:
\begin{itemize}
\item $\sigma \in \trace \mbox{ such that } \sigma \models l$ is abbreviated as: $\sigma_l$
\item A task-state pair appearing in a trace: $(t, \sigma)$ such that $l \in \kwd{ann}(t)$, is abbreviated as: $\kwd{contain}(l, \sigma)$
\end{itemize}
\end{definition}

Notice that an \emph{in force} interval instance of an obligation, having $l$ as its trigger, is always started from a state $\sigma$, where $\kwd{contain}(l, \sigma)$ is true. Therefore, potentially multiple in force intervals of an obligation can co-exist at the same time. However, multiple in force intervals can be fulfilled by a single event happening in a trace, as shown in the following definitions.


An achievement obligation is fulfilled by a trace if the requirement holds at least in one of the trace states when the obligation is in force.

 \begin{definition}[Comply with Achievement]\label{def:laof}
Given an achievement obligation $\Obl^a$ $\langle r,$ $t,$ $d\rangle$ and a trace $\trace$, $\trace$ is compliant with $\Obl^a \langle r, t, d\rangle$ if and only if:
$\forall \sigma_t,$ $\exists \sigma_r | \kwd{contain}(t, \sigma_t)$ $and$ $\sigma_t \preceq \sigma_r$ $and$ $\neg \exists \sigma_d | \sigma_t \preceq \sigma_d \prec \sigma_r$.

A trace not complying with an achievement obligations is considered to violate it.
\end{definition}

A maintenance obligation, on the other hand, is fulfilled if the requirement holds in each of the states where the obligation is in force.

\begin{definition}[Comply with Maintenance]\label{def:lmof}
Given a maintenance obligation $\Obl^m$ $\langle r,$ $t,$ $d\rangle$ and a trace $\trace$, $\trace$ is compliant with $\Obl^m \langle r, t, d\rangle$ if and only if:
$\forall \sigma_t,$ $\exists \sigma_d | \kwd{contain}(t, \sigma_t)$ $\mbox{ and }$ $\sigma_t \preceq \sigma_d \mbox{ and } \forall \sigma | \sigma_t \preceq \sigma \preceq \sigma_d, r \in \sigma$.

A trace not complying with a maintenance obligations is considered to violate it.
\end{definition}

Notice that according to Definition \ref{def:laof} and Definition \ref{def:lmof}, an obligation can be still in force when the end of a trace is reached, in other words, none of the tasks before and including the last one in a trace contained the \emph{deadline condition} of the obligation being evaluated. While this can be a viable scenario in the real world, in this paper we are interested on verifying whether a trace is compliant with an obligation, hence we force the semantics of the tasks belonging to a trace to always allow that. More precisely, we assume that the last task of a trace always satisfies the \emph{deadline} of an obligation in force, in this way we can always decide at the end of a trace whether an obligation is fulfilled or violated. Finally, notice that according to how we defined \emph{structured process models}, the last task in a trace is always the \emph{end task}.

\begin{assumption}[Final Deadline]\label{ass:final_deadline}
Given a structured process model $(P, \kwd{ann})$ and an obligation $\Obl^o$ $\langle r,$ $t,$ $d\rangle$. For every $\trace \in \Theta(P, \kwd{ann})$:
\begin{itemize}
\item if $e$ is the last task in $\trace$, then $d \in \kwd{ann}(e)$. 
\end{itemize}
\end{assumption}

\subsubsection{Process Compliance}

Considering a trace to be compliant with a regulatory framework, if it complies with every obligation belonging to the set composing the framework. A business process model can be compliant with a regulatory framework to various degrees, depending on how many of the traces of the model comply with the regulatory framework. A process model is said to be fully compliant if every trace of the process is compliant with the regulatory framework. A process model is partially compliant if there exists at least one trace compliant with the regulatory framework, and not compliant if there is no trace complying with the framework.

\begin{definition}[Process Compliance]$\:$\label{def:pcompliance}
Given a process $(P, \kwd{ann})$ and a regulatory framework composed by a set of obligations $\Obls$, the compliance of $(P, \kwd{ann})$ with respect to $\Obls$ is determined as follows:
  \begin{itemize}
 \item \textbf{Full compliance} $\PCompliance[F]{(P, \kwd{ann})}{\Obls}$ if and only if\\ $\forall \trace \in \Theta[(P, \kwd{ann})], \PCompliance{\trace}{\Obls}$.
 \item\textbf{Partial compliance} $\PCompliance[P]{(P, \kwd{ann})}{\Obls}$ if and only if\\ $\exists \trace \in \Theta[(P, \kwd{ann})], \PCompliance{\trace}{\Obls}$.
  \item\textbf{Not compliant} $\NPCompliance{(P, \kwd{ann})}{\Obls}$ if and only if\\ ${\neg \exists} \trace \in \Theta[(P, \kwd{ann})], \PCompliance{\trace}{\Obls}$.
  \end{itemize}
\end{definition}

Notice that in this paper we focus on verifying whether a given business process model is \emph{fully} compliant with a regulatory framework. In particular, the approach relies on verifying whether the complementary condition of full compliance is true: in other words, whether a business process model contains a trace not compliant with the regulatory framework. This is equivalent to finding a trace in the process model violating at least an obligation composing the regulatory framework.


\section{Failure $\Delta$-constraints}\label{sec:deltaconstraints}

As full compliance of a business process model can be verified by checking whether it contains a trace violating at least one of the obligations composing the regulatory framework, we formally introduce the definitions for violating the different types of obligations.
For each of the two types of obligations, achievement and maintenance, we define a set of \emph{Failure $\Delta$-constraints}, an alternative representation of the requirements required to be satisfied by a trace to violate a given obligation. These sets of constraints are the ones going to be verified by the approach proposed in the paper.

Before introducing the \emph{Failure $\Delta$-constraints} we need to consider an assumption, as such constraints are verifying the complement of the fulfilment condition of the obligation being checked. In order for the complementary condition of an obligation to be correctly verified, it is necessary that the complement of the requirement condition belongs to the state in order to satisfy the \emph{Failure $\Delta$-constraints}. As it is not always the case that the complementary condition is introduced by the tasks executed in a trace, we assume that the initial task contain the complementary required condition. While we are aware that introducing a syntactical representation for missing literals in a state would have provided the same result, we opted for this solution as it does not introduce additional notation.

\begin{assumption}[Complementary Start]\label{ass:starting_complementary}
Given a structured process model $(P, \kwd{ann})$ and an obligation $\Obl^o$ $\langle r,$ $t,$ $d\rangle$. For every $\trace \in \Theta(P, \kwd{ann})$:
\begin{itemize}
\item if $s$ is the first task in $\trace$, then $\{\neg r\} \in \kwd{ann}(s)$.
\end{itemize}
\end{assumption}

\subsection{Achievement}

We describe in the following, Definition \ref{def:laof_fail}, the complement of Definition \ref{def:laof}, the failure condition for an achievement conditional obligation.

\begin{definition}[Achievement Failure]\label{def:laof_fail}
Given an achievement obligation $\Obl^a \langle r, t, d\rangle$ and a trace $\trace$, $\trace$ is not compliant with $\Obl^a \langle r, t, d\rangle$ if and only if:
$\exists \sigma_t,$ $\sigma_d | \kwd{contain}(t, \sigma_t)$ $\mbox{ and }$ $\sigma_t \preceq \sigma_d \mbox{ and } \neg \exists \sigma_r | \sigma_t \preceq \sigma_r \preceq \sigma_d$.
\end{definition}

In the following definition we translate the requirements described in Definition \ref{def:laof_fail} into its Failure $\Delta$-constraints representation.

\begin{definition}[Achievement Failure $\Delta$-constraints]\label{def:af_delta}

Given a conditional achievement obligation $\Obl^a \langle r, t, d\rangle$ and an execution $\exe$, $\exe$ is compliant with $\Obl^a \langle r, t, d\rangle$, if and only if one of the following conditions is satisfied:

\bigskip
\noindent\begin{tabular}{@{}ll}
$\overline{A}\Delta1$ & $\exists t_t$ such that: \\
& $\exists t_{\neg r}, t_d | t_{\neg r} \preceq t_t \preceq t_d$ and \\
& $\neg \exists t_{r} | t_{\neg r} \preceq t_{r} \preceq t_d$ and \\
& $\exists t_{\neg d}, \neg \exists t_d | t_{\neg d} \preceq t_d \preceq t_t$ \\
\end{tabular}

\bigskip
\noindent\begin{tabular}{@{}ll}
$\overline{A}\Delta2$ & $\exists t_t$ such that: \\
& $\exists t_{\neg r}, \neg \exists t_{r}  | t_{\neg r} \preceq t_{r} \preceq t_t$ and \\
& $\exists t_d, \neg \exists t_{\neg d} | t_d \preceq t_{\neg d} \preceq t_t$ \\
\end{tabular}


\end{definition}

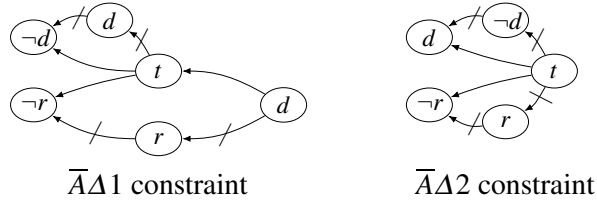
\begin{figure}[h!]
\centering
\scalebox{0.45}{
\begin{tikzpicture}[thick, font=\huge]

\node[ellipse,draw,align=center,minimum size=1cm,text width=0.75cm] (t1) at (0,0) {$t$};
\node[xshift=-3cm,yshift=1cm,ellipse,draw,align=center,minimum size=1cm,text width=0.75cm] (notd1) at (t1.west) {$\neg d$};
\node[xshift=-3cm,yshift=-1cm,ellipse,draw,align=center,minimum size=1cm,text width=0.75cm] (notc1) at (t1.west) {$\neg r$};
\node[xshift=3cm,ellipse,draw,align=center,minimum size=1cm,text width=0.75cm] (d1a) at (t1.east |- notc1) {$d$};
\node[yshift=-2cm,ellipse,draw,align=center,minimum size=1cm,text width=0.75cm] (c1) at (t1) {$r$};
\node[xshift=-0.75cm,yshift=1.5cm,ellipse,draw,align=center,minimum size=1cm,text width=0.75cm] (d1b) at (t1.west) {$d$};

\draw[->,-{Latex[length=2.5mm]}] (t1) to[out=180,in=330] (notd1);
\draw[->,-{Latex[length=2.5mm]}] (t1) to[out=120,in=330] (d1b) node[xshift=0.125cm,yshift=-0.25cm] {$\not$};
\draw[->,-{Latex[length=2.5mm]}] (d1b) to[out=170,in=30] (notd1) node[xshift=0.5cm,yshift=0.25cm] {$\not$};

\draw[->,-{Latex[length=2.5mm]}] (t1) to[out=190,in=20] (notc1);
\draw[->,-{Latex[length=2.5mm]}] (d1a) to[out=150,in=0] (t1);
\draw[->,-{Latex[length=2.5mm]}] (d1a) to[out=210,in=0] (c1) node[xshift=1cm,yshift=0.125cm] {$\not$};
\draw[->,-{Latex[length=2.5mm]}] (c1) to[out=180,in=330] (notc1) node[xshift=1cm,yshift=-0.5cm] {$\not$};

\node[xshift=11cm,ellipse,draw,align=center,minimum size=1cm,text width=0.75cm] (t2) at (t1.east) {$t$};
\node[xshift=-3cm,yshift=1cm,ellipse,draw,align=center,minimum size=1cm,text width=0.75cm] (d2) at (t2.west) {$d$};
\node[xshift=-3cm,yshift=-1cm,ellipse,draw,align=center,minimum size=1cm,text width=0.75cm] (notc2) at (t2.west) {$\neg r$};
\node[xshift=-0.75cm,yshift=1.5cm,ellipse,draw,align=center,minimum size=1cm,text width=0.75cm] (notd2) at (t2.west) {$\neg d$};
\node[yshift=-1.5cm,ellipse,draw,align=center,minimum size=1cm,text width=0.75cm] (c2) at (notd2 |- t2) {$r$};

\draw[->,-{Latex[length=2.5mm]}] (t2) to[out=170,in=340] (d2);
\draw[->,-{Latex[length=2.5mm]}] (t2) to[out=120,in=330] (notd2) node[xshift=0.125cm,yshift=-0.25cm] {$\not$};
\draw[->,-{Latex[length=2.5mm]}] (notd2) to[out=170,in=30] (d2) node[xshift=0.5cm,yshift=0.25cm] {$\not$};

\draw[->,-{Latex[length=2.5mm]}] (t2) to[out=190,in=20] (notc2);
\draw[->,-{Latex[length=2.5mm]}] (t2) to[out=240,in=30] (c2) node[xshift=0.5cm,yshift=0.125cm,rotate=90] {$\not$};
\draw[->,-{Latex[length=2.5mm]}] (c2) to[out=190,in=330] (notc2) node[xshift=0.5cm,yshift=-0.25cm] {$\not$};

\node[yshift=-0.75cm] (l1) at (c1.south) {\Huge $\overline{A}\Delta1$ constraint};
\node (l2) at (c2.south |- l1) {\Huge $\overline{A}\Delta2$ constraint};

\end{tikzpicture}
}
\caption{Achievement constraints}\label{f:aDC}
\end{figure}

Figure \ref{f:aDC} illustrates a graphical representation of the Failure $\Delta$-constraints for achievement obligations. The nodes represent annotations in the tasks and the arrows represent the ordering relations between the tasks that must be fulfilled by an execution of the process model to fulfil the Failure $\Delta$-constraint. A slashed arrow denotes the required absence of the respective element in the interval identified by the surrounding elements.

\begin{theorem}
Given an execution $\exe$, represented as a sequence of tasks, if $\exe$ satisfies either of the Achievement Failure $\Delta$-Constraints (Definition \ref{def:af_delta}), then $\exe$ satisfies the obligation related to the Achievement Failure $\Delta$-Constraints.
\end{theorem}

\subsection{Maintenance}

The translation to Failure $\Delta$-constraints is also applied to local maintenance obligations, in a similar was as for local achievement obligations.

Again, we first introduce the complement of local maintenance obligation fulfilment, Definition \ref{def:lmof}, in Definition \ref{def:lmof_fail}.

\begin{definition}[Maintenance Failure]\label{def:lmof_fail}
Given a conditional maintenance obligation $\Obl^m \langle r, t, d\rangle$ and a trace $\trace$, $\trace$ is not compliant with $\Obl^m \langle r, t, d\rangle$ if and only if:
$$\exists \sigma_t \forall \sigma_d | \kwd{contain}(t, \sigma_t) \mbox{ and } \sigma_t \preceq \sigma_d \mbox{ and } \exists \sigma_{\neg r} | \sigma_t \preceq \sigma_{\neg r} \preceq \sigma_d$$
\end{definition}

Definition \ref{def:mf_delta} illustrates how the conditions of Definition \ref{def:lmof_fail} can be represented as a set of Failure $\Delta$-constraints.

\begin{definition}[Maintenance Failure $\Delta$-constraints]\label{def:mf_delta}
Given a maintenance obligation\\ $\Obl^m \langle r, t, d\rangle$ and an execution $\exe$, $\exe$ is not compliant with $\Obl^m \langle r, t, d\rangle$ if and only if one of the following conditions is satisfied:

\bigskip
\noindent\begin{tabular}{@{}ll}
$\overline{M}\Delta1$ & $\exists t_t$ such that: \\
& $\exists t_{\neg r}, \neg \exists t_{r} | t_{\neg r} \preceq t_{r} \preceq t_t$ \\
\end{tabular}

\bigskip

\noindent\begin{tabular}{@{}ll}
$\overline{M}\Delta2$ & $\exists t_t$ such that: \\
& $\exists t_{r}, \neg \exists t_{\neg r} | t_{r} \preceq t_{\neg r} \preceq t_t$ and \\
& $\forall t_d (\exists t_{\neg r} | t_t \preceq t_{\neg r} \prec t_d)$ \\
\end{tabular}


\end{definition}

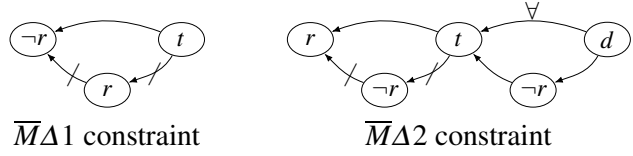
\begin{figure}[h!]
\centering
\scalebox{0.45}{
\begin{tikzpicture}[thick, font=\huge]

\node[ellipse,draw,align=center,minimum size=1cm,text width=0.75cm] (t1) at (0,0) {$t$};
\node[xshift=-1.5cm,yshift=-1.5cm,ellipse,draw,align=center,minimum size=1cm,text width=0.75cm] (c1) at (t1.west) {$r$};
\node[xshift=-1.5cm,ellipse,draw,align=center,minimum size=1cm,text width=0.75cm] (notc1) at (c1.west |- t1) {$\neg r$};

\draw[->,-{Latex[length=2.5mm]}] (t1) to[out=240,in=20] (c1) node[xshift=0.5cm,yshift=0.25cm] {$\not$};
\draw[->,-{Latex[length=2.5mm]}] (t1) to[out=160,in=20] (notc1);
\draw[->,-{Latex[length=2.5mm]}] (c1) to[out=160,in=315] (notc1) node[xshift=0.5cm,yshift=-0.625cm] {$\not$};

\node[xshift=7.5cm,ellipse,draw,align=center,minimum size=1cm,text width=0.75cm] (t2) at (t1.east) {$t$};
\node[xshift=-1.5cm,yshift=-1.5cm,ellipse,draw,align=center,minimum size=1cm,text width=0.75cm] (notc2) at (t2.west) {$\neg r$};
\node[xshift=-1.5cm,ellipse,draw,align=center,minimum size=1cm,text width=0.75cm] (c2) at (notc2.west |- t2) {$r$};
\node[xshift=1.5cm,ellipse,draw,align=center,minimum size=1cm,text width=0.75cm] (notc3) at (t2.east |- notc2) {$\neg r$};
\node[xshift=1.5cm,ellipse,draw,align=center,minimum size=1cm,text width=0.75cm] (d2) at (notc3.east |- t2) {$d$};

\draw[->,-{Latex[length=2.5mm]}] (t2) to[out=240,in=20] (notc2) node[xshift=0.5cm,yshift=0.25cm] {$\not$};
\draw[->,-{Latex[length=2.5mm]}] (t2) to[out=160,in=20] (c2);
\draw[->,-{Latex[length=2.5mm]}] (notc2) to[out=160,in=315] (c2) node[xshift=0.5cm,yshift=-0.625cm] {$\not$};

\draw[->,-{Latex[length=2.5mm]}] (notc3) to[out=160,in=315] (t2); 
\draw[->,-{Latex[length=2.5mm]}] (d2) to[out=240,in=20] (notc3); 
\draw[->,-{Latex[length=2.5mm]}] (d2) to[out=160,in=20] (t2) node[xshift=1.55cm,yshift=0.625cm] {$\forall$};

\node[yshift=-0.75cm] (l1) at (c1.south) {\Huge $\overline{M}\Delta1$ constraint};
\node (l2) at (t2.south |- l1) {\Huge $\overline{M}\Delta2$ constraint};

\end{tikzpicture}
}
\caption{Maintenance constraints}\label{f:mDC}
\end{figure}

\begin{theorem}
Given an execution $\exe$, represented as a sequence of tasks, if $\exe$ satisfies either of the Maintenance Failure $\Delta$-constraints (Definition \ref{def:mf_delta}), then $\exe$ violates the obligation related to the Maintenance Failure $\Delta$-constraints.
\end{theorem}

\subsubsection{Translation Complexity}

Translating a given conditional obligation in the corresponding set of failure $\Delta$-constraints can be done in constant time, since depending on the type of obligations, the associated failure $\Delta$-constraints just need to be instantiated with the parameters of the respective obligation.
\section{Approach: \emph{Stem Evaluation}}\label{sec:stem}

In this section, we introduce the approach verifying whether a business process model is fully compliant with a regulatory framework. The proposed approach evaluates the \emph{tree-like} representation according to the failure $\Delta$-constraints associated to each obligation of the framework.

\subsection{Process Tree Model}

We refer to the tree-like representation of a structured business process model as the \emph{Process Tree Model}, and it is defined as follows.

\begin{definition}[Process Tree Model]\label{def:process_tree}
Let $P$ be a structured process model. A \emph{Process Tree} $PT$ is an abstract hierarchical representation of $P$, where:
\begin{itemize}
\item Each process block $B$ in $P$ corresponds to a node $N$ in $PT$.
\item Given a process block $B(B_1, \dots, B_n)$, where $B_1, \dots, B_n$ are the process blocks directly nested in $B$, the nodes $N_1, \dots, N_n$ in $PT$, corresponding to $B_1, \dots, B_n$ in $P$, are children of $N$, corresponding to $B$ in $P$. Mind that the order between the sub-blocks of a process block is preserved between the children of the same node.
\end{itemize}
\end{definition}

In accordance to Definition \ref{def:process_tree}, the \emph{main block} of a structured process model always represents the root of its corresponding \emph{process tree} model, as it is not contained as a sub-process block in any other process blocks. Similarly, as the tasks in a process model are also process blocks, they also represent nodes in the tree. However, since a block representing a task does not have any sub-process block, such type of process block always corresponds to the leaves of a process tree. Figure \ref{fig:process_tree_model_example} illustrates the process tree model of the structured process model used in Example \ref{ex:processexample}.


\begin{figure}[h]
\centering
\scalebox{0.35}{
\begin{tikzpicture}[thick,font=\LARGE]

\newcommand{\gatesize}{0.9cm}


\node[draw,align=center,circle, minimum size=1cm] (start) at (0,0) {};

\node[xshift=1.5cm,draw,diamond,align=center,minimum height=\gatesize,minimum width=\gatesize] (and1) at (start.east) {};
\draw[line width=4pt] ([yshift=-2mm]and1.north) -- ([yshift=2mm]and1.south);
\draw[line width=4pt] ([xshift=-2mm]and1.east) -- ([xshift=2mm]and1.west);

\node[xshift=1.5cm,yshift=1.75cm,draw,diamond,align=center,minimum height=\gatesize,minimum width=\gatesize] (xor1) at (and1.east) {};
\draw[line width=3pt] ([xshift=0.5mm,yshift=-0.5mm]xor1.north west) -- ([xshift=-0.5mm,yshift=0.5mm]xor1.south east);
\draw[line width=3pt] ([xshift=-0.5mm,yshift=-0.5mm]xor1.north east) -- ([xshift=0.5mm,yshift=0.5mm]xor1.south west);

\node[xshift=1.25cm,yshift=1.15cm,draw,rounded corners=3pt,align=center,minimum height=1cm,minimum width=1.5cm] (t1) at (xor1.east) {$t_1$};
\node[xshift=1.25cm,yshift=-1.15cm,draw,rounded corners=3pt,align=center,minimum height=1cm,minimum width=1.5cm] (t2) at (xor1.east) {$t_2$};
\node[yshift=-1.75cm,draw,rounded corners=3pt,align=center,minimum height=1cm,minimum width=1.5cm] (t3) at (t2 |- and1) {$t_3$};

\node[xshift=1cm,draw,diamond,align=center,minimum height=\gatesize,minimum width=\gatesize] (xor2) at (t1.east |- xor1) {};
\draw[line width=3pt] ([xshift=0.5mm,yshift=-0.5mm]xor2.north west) -- ([xshift=-0.5mm,yshift=0.5mm]xor2.south east);
\draw[line width=3pt] ([xshift=-0.5mm,yshift=-0.5mm]xor2.north east) -- ([xshift=0.5mm,yshift=0.5mm]xor2.south west);

\node[xshift=1.75cm,draw,diamond,align=center,minimum height=\gatesize,minimum width=\gatesize] (and2) at (xor2.east |- and1) {};
\draw[line width=4pt] ([yshift=-2mm]and2.north) -- ([yshift=2mm]and2.south);
\draw[line width=4pt] ([xshift=-2mm]and2.east) -- ([xshift=2mm]and2.west);

\node[xshift=1.5cm,draw,rounded corners=3pt,align=center,minimum height=1cm,minimum width=1.5cm] (t4) at (and2.east) {$t_4$};

\node[xshift=1cm,draw,line width=3pt,align=center,circle, minimum size=1cm] (end) at (t4.east) {};

\node[anchor=south] at (start.north) {\{\}};
\node[anchor=south] at (t1.north) {\{$a$\}};
\node[anchor=south] at (t2.north) {\{$b,c$\}};
\node[anchor=south] at (t3.north) {\{$c,d$\}};
\node[anchor=south] at (t4.north) {\{$\neg a$\}};
\node[anchor=south] at (end.north) {\{\}};

\draw[->,-{Latex[length=3mm]}] (start) -- (and1);
\draw[->,-{Latex[length=3mm]}] (and1) -- (and1 |- xor1) -- (xor1);
\draw[->,-{Latex[length=3mm]}] (and1) -- (and1 |- t3) -- (t3);
\draw[->,-{Latex[length=3mm]}] (xor1) -- (xor1 |- t1) -- (t1);
\draw[->,-{Latex[length=3mm]}] (xor1) -- (xor1 |- t2) -- (t2);
\draw[->,-{Latex[length=3mm]}] (t1) -- (xor2 |- t1) -- (xor2);
\draw[->,-{Latex[length=3mm]}] (t2) -- (xor2 |- t2) -- (xor2);
\draw[->,-{Latex[length=3mm]}] (xor2) -- (and2 |- xor2) -- (and2);
\draw[->,-{Latex[length=3mm]}] (t3) -- (and2 |- t3) -- (and2);
\draw[->,-{Latex[length=3mm]}] (and2) -- (t4);
\draw[->,-{Latex[length=3mm]}] (t4) -- (end);

\node[draw,align=center,minimum height=1cm,minimum width=1.5cm] (troot) at (20,4) {Root};

\node[xshift=-3cm,yshift=-1cm,draw,ellipse,align=center,minimum size=1cm,anchor=north] (tc1) at (troot.south) {};
\node[xshift=-1cm,yshift=-1cm,draw,align=center,minimum height=1cm,minimum width=1.5cm,anchor=north] (tc2) at (troot.south) {$+$};
\node[xshift=1cm,yshift=-1cm,draw,align=center,minimum height=1cm,minimum width=1.5cm,anchor=north] (tc3) at (troot.south) {$t_4$};
\node[xshift=3cm,yshift=-1cm,draw,ellipse,align=center,minimum size=1cm,anchor=north,line width=3pt] (tc4) at (troot.south) {};

\node[xshift=-1cm,yshift=-1cm,draw,align=center,minimum height=1cm,minimum width=1.5cm,anchor=north] (tc5) at (tc2.south) {$\times$};
\node[xshift=1cm,yshift=-1cm,draw,align=center,minimum height=1cm,minimum width=1.5cm,anchor=north] (tc6) at (tc2.south) {$t_3$};

\node[xshift=-1cm,yshift=-1cm,draw,align=center,minimum height=1cm,minimum width=1.5cm,anchor=north] (tc7) at (tc5.south) {$t_1$};
\node[xshift=1cm,yshift=-1cm,draw,align=center,minimum height=1cm,minimum width=1.5cm,anchor=north] (tc8) at (tc5.south) {$t_2$};

\draw[->,-{Latex[length=3mm]}] (troot) -- (tc1);
\draw[->,-{Latex[length=3mm]}] (troot) -- (tc2);
\draw[->,-{Latex[length=3mm]}] (troot) -- (tc3);
\draw[->,-{Latex[length=3mm]}] (troot) -- (tc4);

\draw[->,-{Latex[length=3mm]}] (tc2) -- (tc5);
\draw[->,-{Latex[length=3mm]}] (tc2) -- (tc6);

\draw[->,-{Latex[length=3mm]}] (tc5) -- (tc7);
\draw[->,-{Latex[length=3mm]}] (tc5) -- (tc8);

\draw[->,-{Latex[length=3mm]}] ([xshift=0.5cm,yshift=2cm]end.east |- t1.north) to[in=145,out=25] (troot.north west);

\draw ([xshift=-0.5cm,yshift=2cm]start.west |- t1.north) -- 
	  ([xshift=0.5cm,yshift=2cm]end.east |- t1.north) -- 
	  ([xshift=0.5cm,yshift=-1cm]end.east |- t3.south) -- 
	  ([xshift=-0.5cm,yshift=-1cm]start.west |- t3.south) -- 
	  ([xshift=-0.5cm,yshift=2cm]start.west |- t1.north);

\draw ([xshift=-0.5cm,yshift=1cm]xor1.west |- t1.north) -- 
	  ([xshift=0.5cm,yshift=1cm]xor2.east |- t1.north) -- 
	  ([xshift=0.5cm,yshift=-0.25cm]xor2.east |- t2.south) -- 
	  ([xshift=-0.5cm,yshift=-0.25cm]xor1.west |- t2.south) -- 
	  ([xshift=-0.5cm,yshift=1cm]xor1.west |- t1.north);
\node[anchor=north east] (xor) at ([xshift=0.5cm,yshift=1cm]xor2.east |- t1.north) {$\times$};
\draw (xor.north west) -- (xor.south west) -- (xor.south east);

\draw ([xshift=-0.5cm,yshift=1.5cm]and1.west |- t1.north) -- 
	  ([xshift=0.25cm,yshift=1.5cm]and2.east |- t1.north) -- 
	  ([xshift=0.25cm,yshift=-0.5cm]and2.east |- t3.south) -- 
	  ([xshift=-0.5cm,yshift=-0.5cm]and1.west |- t3.south) -- 
	  ([xshift=-0.5cm,yshift=1.5cm]and1.west |- t1.north);
\node[anchor=north east] (and) at ([xshift=0.25cm,yshift=1.5cm]and2.east |- t1.north) {$+$};
\draw (and.north west) -- (and.south west) -- (and.south east);
	  
\end{tikzpicture}
}
\caption{Process Tree Example}\label{fig:process_tree_model_example}
\end{figure}
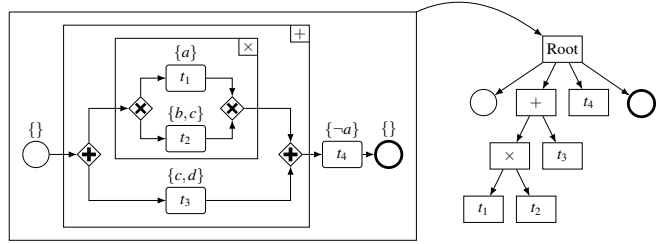

Notice that a process block can only be the sub-block of a single other process-block, hence Definition \ref{def:process_tree} is sufficient, as the process blocks nesting ensures that a process tree is a proper tree. Moreover, notice that, as the only process blocks that do not contain nested blocks are tasks, then each leaf in a process tree corresponds to one of the tasks\footnote{A minor exception to this, is that the \kwd{start} and \kwd{end} elements of the process models are also represented as leaves in the process tree. Despite these two elements not being proper tasks, they can be treated as such for the sake of the evaluation procedure.} of the original business process model.

Intuitively, the tasks having the same parent in a process tree representation are related by an ordering relationship, or by a mutual exclusion relationship, depending on the sub-process block type associated to the parent node. By navigating the structure bottom-up, it is possible to recursively infer the relations between the tasks, and we exploit this in our approach evaluating the compliance state of a process model.

\subsubsection{Process Tree Stem}

The proposed approach uses the sufficient conditions, described in Definition \ref{def:af_delta} and Definition \ref{def:mf_delta}, expressed as failure $\Delta$-constraints, to identify whether a process model contains an execution violating an obligation being evaluated. By checking whether a business process model contains an execution violating one or more obligations, it is possible to verify whether such model is fully compliant with a regulatory framework.

Recalling some details from Definition \ref{def:af_delta} and Definition \ref{def:mf_delta}, every sufficient condition is in relation to the existence of a \emph{trigger task}, a task having annotated the trigger of the conditional obligation being evaluated. With this in mind, given a \emph{process tree model} and a \emph{trigger task}, we refer to the path between the root and the \emph{trigger task} as the \emph{Stem}, as illustrated in Figure \ref{fig:stem_tree}.


\begin{figure}[h]
\centering
\scalebox{0.35}{
\begin{tikzpicture}[thick,font=\LARGE]

\node[draw,align=center,minimum height=1cm,minimum width=1.5cm] (root) at (0,0) {Root};

\node[draw,ellipse,align=center,minimum size=1cm,anchor=north] (trigger) at (0,-5) {Trigger leaf};

\node (stem) at (-1.5,-4) {Stem};

\draw (root.south) -- 
      ([xshift=-0.5cm,yshift=-1cm]root.center) --
      ([xshift=0.5cm,yshift=-2cm]root.center) --
      ([xshift=-0.5cm,yshift=-3cm]root.center) --
      ([xshift=0.5cm,yshift=-4cm]root.center) --
      (trigger.north);
                               
\draw (root.west) -- ([xshift=-3cm]trigger.west);
\draw (root.east) -- ([xshift=3cm]trigger.east);

\draw (stem.north) to[out=75,in=190] ([xshift=-0.7cm,yshift=-3.2cm]root.center);

\end{tikzpicture}
}
\caption{Process Tree and Stem}\label{fig:stem_tree}
\end{figure}
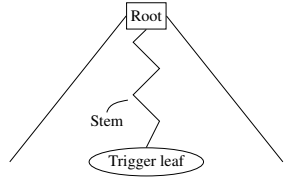

Given a process tree and a conditional obligation, notice that the tree can contain multiple trigger tasks, hence its tree representation can contain multiple stems. The proposed approach investigates each stem in a process tree model independently.

Additionally, the nodes in a process tree can be considered to represent sets of sub-executions of the related business process model\footnote{As each node in a process tree is associated to a process block, even simple tasks, every execution obtainable by applying Definition \ref{def:ser} to the process block is always a sub execution of a possible execution of the associated business process model.}, and that these partial executions involve the tasks belonging to the sub-tree of the node. Moreover, the nodes on the \emph{Stem} represent sub-executions involving the \emph{trigger task}, which we can use to verify the conditions expressed in the failure $\Delta$-constraints. Finally, verifying a failure $\Delta$-constraint condition in a set of partial executions of a model corresponds to verifying it in the set of executions of the model itself, as it is stated by Lemma \ref{lem:sub_to_super}.

\subsection{Stem Evaluation Procedure}

The inputs of the \emph{Stem Evaluation} procedure are a business process model $P$ and a $\Delta$-constraint $\delta$. The procedure is detailed in the pseudo-code shown in Algorithm \ref{a:stem}.

\begin{algorithm}
{\fontsize{8}{7.5}\selectfont
\SetKwData{Left}{left}\SetKwData{This}{this}\SetKwData{Up}{up}
\SetKwFunction{Union}{Union}\SetKwFunction{FindCompress}{FindCompress}
\SetKwInOut{Input}{input}\SetKwInOut{Output}{output}

\Input{A business process model: $P$, and a failure $\Delta$-constraint: $\delta$}
\Output{Whether $P$ satisfies $\delta$}

\BlankLine

\nl Let $\mathcal{T}$ be the set of tasks in $P$ containing $t$ of $\delta$ in their annotation\;
\nl Let $\mathcal{P}$ be the tree representation of $P$\;
\nl\For{ $\tau \in \mathcal{T}$}{
\nl \If{!\kwd{checkLeafFail}($\tau$, $\delta$)}
{
\nl Let the path between the root in $\mathcal{P}$ to $\tau$ be the \kwd{stem}\;
\nl Let $\mathcal{P}'$ be the tree where each \kwd{XOR} node on the \kwd{stem} is replaced with its child on the \kwd{stem}\;
\nl \If{\kwd{Evaluate}($\mathcal{P}'$, $\tau$, $\delta$)}
{\nl return \kwd{true}\;}
}
}
\nl return \kwd{false}\;
}
\caption{Stem Evaluation}\label{a:stem}
\end{algorithm}\DecMargin{1em}

Algorithm \ref{a:stem} identifies the \emph{trigger leaves} ($\mathcal{T}$) of the tree representation of the business process in line 1. 
Notice that in line 4 \kwd{checkLeafFail}($\tau$, $\delta$) verifies whether the trigger leaf contains certain annotations that would immediately falsify the $\Delta$-constraint being evaluated. \kwd{checkLeafFail} is defined as follows depending on the value of $\delta$:

\begin{description}
\item[$\overline{A}\Delta1$] if $c \in \kwd{ann}(\tau)$ then return true, false otherwise
\item[$\overline{A}\Delta2$] if $c \in \kwd{ann}(\tau)$ then return true, false otherwise
\item[$\overline{M}\Delta1$] if $c \in \kwd{ann}(\tau)$ then return true, false otherwise
\item[$\overline{M}\Delta2$] if $d \in \kwd{ann}(\tau)$ then return true, false otherwise
\end{description}

For each of the trigger leaves, the procedure prunes the tree in line 5, removing each \kwd{XOR} branch on the \kwd{stem} that would not lead to the considered trigger leaf, this is graphically illustrated in Figure \ref{f:stem_pruning}. Intuitively, as the tree is a representation of a business process model, the pruning procedure \kwd{XOR} branches on the \kwd{stem}, removes from the pool of the possible executions all the execution not containing the trigger leaf ($\tau$).


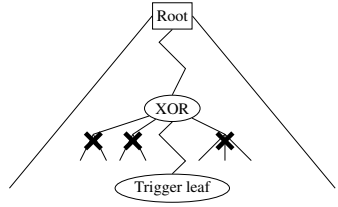
\begin{figure}[h]
\centering
\scalebox{0.35}{
\begin{tikzpicture}[thick,font=\LARGE]

\newcommand{\prunesize}{0.3cm}
\newcommand{\prune}[1]{%
\draw[line width=5pt] ([xshift=-\prunesize,yshift=-\prunesize]#1.center) --
					  ([xshift=\prunesize,yshift=\prunesize]#1.center);
\draw[line width=5pt] ([xshift=-\prunesize,yshift=\prunesize]#1.center) --
					  ([xshift=\prunesize,yshift=-\prunesize]#1.center);
}

\node[draw,align=center,minimum height=1cm,minimum width=1.5cm] (root) at (0,0) {Root};

\node[draw,ellipse,align=center,minimum size=1cm,anchor=north] (xor) at (0,-3) {XOR};

\node[draw,ellipse,align=center,minimum size=1cm,anchor=north] (trigger) at (0,-6) {Trigger leaf};

\draw (root.south) -- 
      ([xshift=-0.5cm,yshift=-1cm]root.center) --
      ([xshift=0.5cm,yshift=-2cm]root.center) --
      (xor.north);

\draw (xor.south) -- 
      ([xshift=-0.5cm,yshift=-1cm]xor.center) --
      ([xshift=0.5cm,yshift=-2cm]xor.center) --
      (trigger.north);

\draw (xor) --
	  (-1.5,-4.5);
\draw (-1.5,-4.5) -- (-1,-5.5);
\draw (-1.5,-4.5) -- (-2,-5.5);

\draw (xor) --
	  (-3,-4.5);
\draw (-3,-4.5) -- (-2.5,-5.5);
\draw (-3,-4.5) -- (-3.5,-5.5);

\draw (xor) --
	  (2,-4.5);
\draw (2,-4.5) -- (2,-5.5);
\draw (2,-4.5) -- (1,-5.5);
\draw (2,-4.5) -- (3,-5.5);

\node (p1) at (-1.5,-4.75) {};
\node (p2) at (2,-4.75) {};
\node (p3) at (-3,-4.75) {};

\prune{p1};
\prune{p2};
\prune{p3};
                                                      
\draw (root.west) -- ([xshift=-4cm]trigger.west);
\draw (root.east) -- ([xshift=4cm]trigger.east);

\end{tikzpicture}
}
\caption{Stem Pruning}\label{f:stem_pruning}
\end{figure}

As shown later in this section, the \kwd{Evaluation} procedure verifies whether a tree representation of a model satisfies a failure $\Delta$-constraint given a specific task having the trigger of the constraint annotated, the trigger leaf $\tau$. Thus, Lemma \ref{l:stem_prune} formally states that the pruning procedure removes exactly every execution in the process not containing $\tau$.

\begin{theorem}[Algorithm \ref{a:stem} Correctness]
Given a business process model $P$ and a failure $\Delta$-constrain $\delta$, if there exists an execution of $P$ such that the execution satisfies $\delta$, then Algorithm \ref{a:stem} returns \kwd{true}, otherwise it returns \kwd{false}.
\end{theorem}

\subsection{Function \kwd{Evaluate}}

Before proceeding with describing \kwd{Evaluate}, notice that we use the term \emph{overnodes} to refer to the nodes on the stem, and the term \emph{undernodes} to refer to the others. The \emph{trigger leaf} is always an \emph{overnode}.

As reported in Lemma \ref{lem:sub_to_super}, if the evaluation of an overnode satisfies a failure $\Delta$-constraint, then it is the case that the root overnode satisfies the same failure $\Delta$-constraint, and the same for the corresponding process model. The consequence of this property is that the evaluation of whether a model satisfies a failure $\Delta$-constraint, can be sometimes assessed by evaluating only some of the associated overnodes.

To determine the extent to which a node in a process tree satisfies a failure $\Delta$-constraint being evaluated, we adopt a set of classifications. While we introduce the classes used to evaluate the failure $\Delta$-constraints and their sub-patterns in Section~\ref{sec:aggregations} in detail, for the sake of simplicity we introduce now an abstract function, \kwd{classification} as described in Definition \ref{def:class_fun}, which assigns a set of classifications to a node of a process tree according to the failure $\Delta$-constraint being evaluated.

\begin{definition}[Classification]\label{def:class_fun}
Let $B$ be a process block. The following classification function (\kwd{cl}) returns a set of classifications $\mathcal{C}$, where the requirements of each classification in $\mathcal{C}$ are satisfied by $B$.
$$\kwd{cl}(B) = \mathcal{C}$$
\end{definition}

The result of classifying a node in a process tree is generally a set of classes instead of a single class. This is the result of the properties of a node to match the requirements of multiple classes. Note, however, that even if the cartesian product suggests an explosion of the number of aggregations required to be performed, such number is limited by the number of classifications kept in the resulting set, which is again is going to be limited by preference orders defined for each type of classes, allowing to discard classifications from the result if another strictly better classification exists in the set.

The classification of the leaves in a process tree depends uniquely on which failure $\Delta$-constraint is being evaluated, and on the properties of the leaf. Differently, the classification of a non-leaf node in a process tree can be derived from the classifications of its children as shown in Definition \ref{def:aggregation_fun}. 

\begin{definition}[Aggregation]\label{def:aggregation_fun}

Let $B$ be a process block and let $B_1$, $B_2$, ..., $B_n$ be the sub-blocks composing $B$. Let \emph{aggregate} (\kwd{agg}) be a binary function taking two sets of classifications and returning the set of resulting aggregated classifications.

The classification of $B$ is equivalent to the cumulative aggregations of the classifications of its sub-blocks as follows:

$$\kwd{cl}(B) \equiv \kwd{agg}(\kwd{agg}(\kwd{agg}(\kwd{cl}(B_1), \kwd{cl}(B_2)), ...),\kwd{cl}(B_n))$$

Notice that the result of \kwd{agg} also depends on the type of the process block whose sub-blocks classifications are being aggregated. 
Moreover, the aggregation order follows the left to right order of appearance of the sub-blocks in the process tree, which maintains the order in the process description using the \emph{prefix representation}, shown in Definition~\ref{def:pb}. While the aggregation order is important only when the node being classified is of type \kwd{SEQ}, for the sake of simplicity we force a strict left-to-right aggregation for every type of node.
\end{definition}

\begin{example}[Classifications and Aggregations]

Considering the process tree representation in Figure \ref{fig:ad1s_example_tree_3}, where $\overline{A}\Delta1$ is evaluated through the classification\footnote{Some of the classes to evaluate $\overline{A}\Delta1$ are first introduced in Table \ref{tb:ispclass}.} of the nodes of the tree. 

The nodes in the process tree are classified to evaluate the constraint. For instance, it can be noticed that $\kwd{cl}(t_1) = \mathbf{\underline{\overset{0}{x}\overset{0}{z}}}$, and that $\kwd{cl}(t_2) = \mathbf{\underline{\overset{-}{xz}}}$.

Considering now the process block composed by $t_1$ and $t_2$, labelled $\times$ in Figure \ref{fig:ad1s_example_tree_3}, its classification is the following: $\kwd{cl}(\times) = \mathbf{\underline{\overset{0}{x}\overset{0}{z}}}$.  

The classification of $\times$ can also represented as the aggregations\footnote{Notice that in this particular case, as the node being evaluated is of type \kwd{XOR}, the aggregation consists of the union of the classifications of the blocks composing the \kwd{XOR}, and by removing the classes which are strictly worst than the ones included in the resulting set. This is further explained in a following subsection.} of the classifications of the process blocks composing it as follows: $\kwd{agg}(\mathbf{\underline{\overset{0}{x}\overset{0}{z}}}, \mathbf{\underline{\overset{-}{xz}}}) = \mathbf{\underline{\overset{0}{x}\overset{0}{z}}}$

\end{example}

\subsubsection{Aggregations for \kwd{SEQ} and \kwd{AND} Undernodes}

When evaluating a failure $\Delta$-constraint over a process-tree representation of a model, we are interested in the classifications of the overnodes. Keeping this in mind, and despite using a comparable left-to-right cumulative aggregation for both nodes of type \kwd{SEQ} and \kwd{AND}, the following difference must be adopted when such nodes are overnodes.

\begin{description}
\item[\kwd{SEQ}]
Classifying an overnode of type \kwd{SEQ} requires to evaluate and aggregate its left and right children independently with respect to the stem. This is graphically illustrated in Figure~\ref{fig:seq_ov_ev}. As the properties to verify a $\Delta$-constraints are always relative to a given \emph{trigger task}, which identifies the stem, and because of the execution semantics of a \kwd{SEQ} process block, classifications on the left side of the stem can only be used to verify the part of the $\Delta$-constraint on the left of the \emph{trigger task}. The same holds for classifications on the right side of the stem.

\item[\kwd{AND}]
When classifying an overnode of type \kwd{AND}, the left and right of the stem division is not required, and the classification of the overnode is obtained by aggregating together the evaluations of its direct undernode children and then aggregating the result with its overnode child. This is graphically illustrated in Figure \ref{fig:and_ov_ev}.
\end{description}

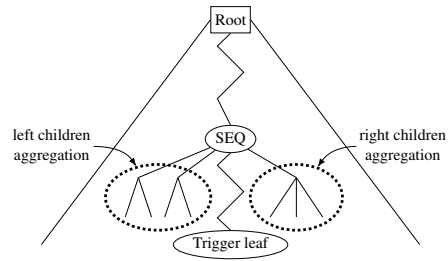
\begin{figure}[h]
\centering
\scalebox{0.35}{
\begin{tikzpicture}[thick,font=\LARGE]

\node[draw,align=center,minimum height=1cm,minimum width=1.5cm] (root) at (0,0) {Root};

\node[draw,ellipse,align=center,minimum size=1cm,anchor=north] (seq) at (0,-4) {SEQ};

\node[draw,ellipse,align=center,minimum size=1cm,anchor=north] (trigger) at (0,-8) {Trigger leaf};

\draw (root.south) -- 
      ([xshift=-0.5cm,yshift=-1cm]root.center) --
      ([xshift=0.5cm,yshift=-2cm]root.center) --
      ([xshift=-0.5cm,yshift=-3cm]root.center) --
      (seq.north);

\draw (seq.south) -- 
      ([xshift=-0.5cm,yshift=-1cm]seq.center) --
      ([xshift=0.5cm,yshift=-2cm]seq.center) --
      ([xshift=-0.5cm,yshift=-3cm]seq.center) --
      (trigger.north);

\draw (seq) --
	  (-2,-6);
\draw (-2,-6) -- (-1.5,-7.5);
\draw (-2,-6) -- (-2.5,-7.5);

\draw (seq) --
	  (-3.5,-6);
\draw (-3.5,-6) -- (-3,-7.5);
\draw (-3.5,-6) -- (-4,-7.5);

\draw (seq) --
	  (2.5,-6);
\draw (2.5,-6) -- (2.5,-7.5);
\draw (2.5,-6) -- (1.5,-7.5);
\draw (2.5,-6) -- (3.5,-7.5);

\draw (root.west) -- ([xshift=-5cm]trigger.west);
\draw (root.east) -- ([xshift=5cm]trigger.east);

\node[draw,ellipse,dashed,line width=3pt,minimum width=4cm,minimum height=2.5cm] (lte) at (-2.75,-6.75) {};
\node[xshift=-1cm,anchor=south east,minimum size=2cm,align=center] (ltelbl) at (lte.north west) {left children\\ aggregation};
\draw[->,-{Latex[length=3mm]}] (ltelbl.east) to[out=0,in=125] (lte);

\node[draw,ellipse,dashed,line width=3pt,minimum width=3.5cm,minimum height=2.5cm] (rte) at (2.5,-6.75) {};
\node[xshift=1cm,anchor=south west,minimum size=2cm,align=center] (rtelbl) at (rte.north east) {right children\\ aggregation};
\draw[->,-{Latex[length=3mm]}] (rtelbl.west) to[out=180,in=55] (rte);

\end{tikzpicture}
}
\caption{\kwd{SEQ} Overnode Evaluation}\label{fig:seq_ov_ev}
\end{figure}


\begin{figure}[h]
\centering
\scalebox{0.35}{
\begin{tikzpicture}[thick,font=\LARGE]

\node[draw,align=center,minimum height=1cm,minimum width=1.5cm] (root) at (0,0) {Root};

\node[draw,ellipse,align=center,minimum size=1cm,anchor=north] (and) at (0,-4) {AND};

\node[draw,ellipse,align=center,minimum size=1cm,anchor=north] (trigger) at (0,-10) {Trigger leaf};

\draw (root.south) -- 
      ([xshift=-0.5cm,yshift=-1cm]root.center) --
      ([xshift=0.5cm,yshift=-2cm]root.center) --
      ([xshift=-0.5cm,yshift=-3cm]root.center) --
      (and.north);

\draw (and.south) -- 
      ([xshift=-0.5cm,yshift=-1cm]and.center) --
      ([xshift=0.5cm,yshift=-2cm]and.center) --
      ([xshift=-0.5cm,yshift=-3cm]and.center) --
      ([xshift=0.5cm,yshift=-4cm]and.center) --
      ([xshift=-0.5cm,yshift=-5cm]and.center) --
      (trigger.north);

\draw (and) --
	  (-2,-6);
\draw (-2,-6) -- (-1.5,-7.5);
\draw (-2,-6) -- (-2.5,-7.5);

\draw (and) --
	  (-3.5,-6);
\draw (-3.5,-6) -- (-3,-7.5);
\draw (-3.5,-6) -- (-4,-7.5);

\draw (and) --
	  (2.5,-6);
\draw (2.5,-6) -- (2.5,-7.5);
\draw (2.5,-6) -- (1.5,-7.5);
\draw (2.5,-6) -- (3.5,-7.5);

\draw (root.west) -- ([xshift=-6cm]trigger.west);
\draw (root.east) -- ([xshift=6cm]trigger.east);

\node[draw,ellipse,dashed,line width=3pt,minimum width=4cm,minimum height=2.5cm] (lte) at (-2.75,-6.75) {};
\node[draw,ellipse,dashed,line width=3pt,minimum width=3.5cm,minimum height=2.5cm] (rte) at (2.5,-6.75) {};

\draw[->,-{Latex[length=5mm]},line width=3pt] ([xshift=0.65cm]lte.south west) -- ([xshift=-0.65cm]rte.south east);

\node[xshift=0.5cm,align=center,anchor=north] at (lte.south) {undernode children \\ aggregation};

\end{tikzpicture}
}
\caption{\kwd{AND} Overnode Evaluation}\label{fig:and_ov_ev}
\end{figure}
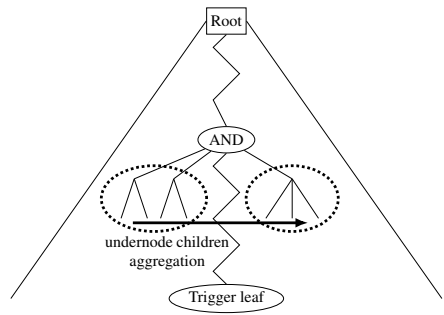

\subsubsection{Aggregations for \kwd{XOR} Undernodes}

The aggregation of undernodes of type \kwd{XOR} is performed by considering the union of the sets of classifications of the children of the \kwd{XOR} undernode. The resulting set of classification is then simplified in accordance to the preference orders of the classification being evaluated.

When the preference order of the classifications is linear, after simplifying the resulting set of an \kwd{XOR} aggregation, it always consists of a single classification. Differently, when the preference order of the classifications is not linear (i.e., the preference orders between the possible classifications is a lattice), then it is possible that the simplification of an aggregation contains more than a single classification.

Removing possible classifications from a result set does not interfere with verifying whether the failure $\Delta$-constraint being evaluated can be satisfied, as only the classifications being strictly worse than existing ones in the result set are eliminated. Hence, the remaining classifications allow to verify the fulfilment of the failure $\Delta$-constraint. Intuitively, when a classification is strictly better than another, the former is capable of achieving at least the same classification through aggregations than the latter, and possibly a better one.

\begin{example}[\kwd{XOR} Aggregation]
Let $B_1$ and $B_2$ be process blocks composing an \kwd{XOR} block, and let their classifications be the following:

\begin{itemize}
\item $\kwd{cl}(B_1) = \mathcal{C}_1$
\item $\kwd{cl}(B_2) = \mathcal{C}_2$
\end{itemize}

The result of the aggregation (and, as such, the classification of the \kwd{XOR} block) is the following:
$$\kwd{cl}( \kwd{XOR} (B_1, B_2)) \equiv \mathcal{C}_1 \cup \mathcal{C}_2$$
\end{example}


\subsubsection{Aggregations for the Overnode}

The evaluation of the overnode on the stem is performed in a similar way as for the undernodes of type \kwd{SEQ} and \kwd{AND}. The aggregation is defined in some aggregation tables, depending on the failure $\Delta$-constraint being evaluated, and the type of the overnode. This aggregation is performed over the classification of the overnode's child of type overnode, with the classifications of its children of type undernode. Again, the classification is always performed from left to right in accordance to where the children appear in the process tree.

When an overnode is of type \kwd{XOR}, then its evaluation simply corresponds to the classification of its child of type overnode, as its children of type undernode can be safely discarded by the \emph{Stem Pruning} procedure described in Lemma~\ref{l:stem_prune}.

\subsubsection{\kwd{Evaluate} Termination}

Finally, the evaluation of the stem (i.e. the overnodes of a process-tree) can be terminated either when \kwd{Evaluate} reaches the root of the process-tree and fails to verify the satisfaction of the failure $\Delta$-constraint, or when the failure $\Delta$-constraint is satisfied by any of the overnodes.

While the latter case ensures that there exists an execution in the model violating the obligation associated to the failure $\Delta$-constraint, the former only ensures that there is no execution in the process model violating, given that particular trigger leaf, if and only if the sister failure $\Delta$-constraint\footnote{Remember that each obligation produces a pair of failure $\Delta$-constraints, as shown in Definition \ref{def:af_delta} and Definition \ref{def:mf_delta}.} associated to the same obligation is also not satisfied for the same trigger leaf.

\subsection{Computational Complexity of the Evaluation}

The number of times the \kwd{Evaluate} function needs to be applied is in the worst case scenario equal to the number of overnodes times the number of trigger leaf for both failure $\Delta$-constraints representing a given obligation that is required to be checked. Thus this procedure must be considered for each obligation in the regulatory framework. However, this evaluation procedure is \emph{polynomial} with respect to the size of the process model and the size of the regulatory framework.

Furthermore, we need to look into the \kwd{Evaluate} function to actually show that the whole evaluation procedure is actually in time polynomial with respect to the size of the problem. From the provided description, we know that in order to classify the overnode being evaluated, \kwd{Evaluate} needs to aggregate the classification of each of its children classifications. Resulting in $n-1$ aggregations, where $n$ is the number of children of the overnode.

While the classification of the child of type overnode is provided by the previous application of \kwd{Evaluate}, the evaluations of the other children must be constructed by exploring the sub trees, hence classifying the leaves and aggregating their results up to the desired level where the children of the overnode reside. Notice that the number of aggregations required is polynomial with respect to the size of the process model, as it is at most equal to the number of sub-blocks of the block corresponding to the overnode.

Finally, we have to consider the complexity of each aggregation, which potentially involves aggregating sets of classifications together. As mentioned previously, the combinatorial explosion of aggregation is taken care of by providing a preference partial order between the possible classifications, hence allowing to discards classifications for which a strictly better one exists. As a consequence, for each aggregation between nodes in the process-tree, the amount of required aggregations between the classifications is reduced to a constant number of aggregations in the worst case scenario, depending on the preferences between the classifications.

The whole procedure, from aggregating classifications to the recursive application of \kwd{Evaluate} from the trigger leaf towards the root of the process-tree, involves steps that are polynomial or constant with respect to the size of the problem. Therefore, we can conclude that the computational complexity of the whole procedure is in time polynomial with respect to the size of the problem.
\section{An Overview of the Aggregations}\label{sec:aggregations}

In this section, we introduce the classifications used to classify the tasks while evaluating the stem of a process tree against a failure $\Delta$-constraints. It is not surprising that different classifications are required depending on the failure $\Delta$-constraint being evaluated. However, we also introduce sub-patterns that can be reused through the classification and aggregation procedures, limiting in this way the number of possible combinations required to be analysed.

Note that, for the sake of readability, we do not explicitly show the aggregation tables for the classifications introduced. The aggregation tables, along with the correctness proof for the classifications considered, and the aggregations can be found in Appendix~\ref{app:agg}.

\subsection{$\overline{A}\Delta1$}\label{sec:ad1}

We hereby show the classifications used by \kwd{Evaluate}($\mathcal{P}'$, $\tau$, $\delta$) when $\delta$ is $\overline{A}\Delta1$. Before proceeding with the classifications, we introduce a simplified version of $\overline{A}\Delta1$, which is used during the evaluation procedure, and is equivalent to evaluating $\overline{A}\Delta1$.

The simplified failure $\Delta$-constraint is introduced in Definition \ref{def:ad1s} and Figure \ref{fig:ad1s} illustrates side by side the two versions of the failure $\Delta$-constraint.

\begin{definition}[$\overline{A}\Delta1$ Simplified]\label{def:ad1s}
The Achievement Failure $\Delta$-constraint from Definition \ref{def:af_delta}, recalled below:


\bigskip
\noindent\begin{tabular}{@{}ll}
$\overline{A}\Delta1$ & $\exists t_t$ such that: \\
& $\exists t_{r}, t_d | t_{r} \preceq t_t \preceq t_d$ and \\
& $\neg \exists t_{\neg r} | t_{r} \preceq t_{\neg r} \preceq t_d$ and \\
& $\exists t_{\neg d}, \neg \exists t_d | t_{\neg d} \preceq t_d \preceq t_t$ \\
\end{tabular}

\bigskip

\noindent is simplified as follows:

\bigskip
\noindent\begin{tabular}{@{}ll}
$\overline{A}\Delta1S$ & $\exists t_t$ such that: \\
& $\exists t_{\neg r}, t_d | t_{\neg r} \preceq t_t \preceq t_d$ and \\
& $\neg \exists t_r | t_{\neg r} \preceq t_r \preceq t_d$ \\
\end{tabular}

\bigskip


\end{definition}

\begin{figure}[h!]
\centering
\scalebox{0.4}{
\begin{tikzpicture}[thick, font=\LARGE]

\node[ellipse,draw,align=center,minimum size=1cm,text width=0.75cm] (t1) at (0,0) {$t$};
\node[xshift=-3cm,yshift=1cm,ellipse,draw,align=center,minimum size=1cm,text width=0.75cm] (notd1) at (t1.west) {$\neg d$};
\node[xshift=-3cm,yshift=-1cm,ellipse,draw,align=center,minimum size=1cm,text width=0.75cm] (notc1) at (t1.west) {$\neg r$};
\node[xshift=3cm,ellipse,draw,align=center,minimum size=1cm,text width=0.75cm] (d1a) at (t1.east |- notc1) {$d$};
\node[yshift=-2cm,ellipse,draw,align=center,minimum size=1cm,text width=0.75cm] (c1) at (t1) {$r$};
\node[xshift=-0.75cm,yshift=1.5cm,ellipse,draw,align=center,minimum size=1cm,text width=0.75cm] (d1b) at (t1.west) {$d$};

\draw[->,-{Latex[length=2.5mm]}] (t1) to[out=180,in=330] (notd1);
\draw[->,-{Latex[length=2.5mm]}] (t1) to[out=120,in=330] (d1b) node[xshift=0.125cm,yshift=-0.25cm] {$\not$};
\draw[->,-{Latex[length=2.5mm]}] (d1b) to[out=170,in=30] (notd1) node[xshift=0.5cm,yshift=0.25cm] {$\not$};

\draw[->,-{Latex[length=2.5mm]}] (t1) to[out=190,in=20] (notc1);
\draw[->,-{Latex[length=2.5mm]}] (d1a) to[out=150,in=0] (t1);
\draw[->,-{Latex[length=2.5mm]}] (d1a) to[out=210,in=0] (c1) node[xshift=1cm,yshift=0.125cm] {$\not$};
\draw[->,-{Latex[length=2.5mm]}] (c1) to[out=180,in=330] (notc1) node[xshift=1cm,yshift=-0.5cm] {$\not$};

\node[xshift=10cm,ellipse,draw,align=center,minimum size=1cm,text width=0.75cm] (t2) at (t1.east) {$t$};
\node[xshift=-2cm,yshift=-1cm,ellipse,draw,align=center,minimum size=1cm,text width=0.75cm] (notc2) at (t2.west) {$\neg r$};
\node[xshift=2cm,ellipse,draw,align=center,minimum size=1cm,text width=0.75cm] (d2a) at (t2.east |- notc2) {$d$};
\node[yshift=-2cm,ellipse,draw,align=center,minimum size=1cm,text width=0.75cm] (c2) at (t2) {$r$};

\draw[->,-{Latex[length=2.5mm]}] (t2) to[out=180,in=45] (notc2);
\draw[->,-{Latex[length=2.5mm]}] (d2a) to[out=135,in=0] (t2);
\draw[->,-{Latex[length=2.5mm]}] (d2a) to[out=225,in=0] (c2) node[xshift=0.75cm,yshift=0.125cm] {$\not$};
\draw[->,-{Latex[length=2.5mm]}] (c2) to[out=180,in=315] (notc2) node[xshift=0.75cm,yshift=-0.5cm] {$\not$};

\node[yshift=-0.5cm] at (c1.south) {$\overline{A}\Delta1$};
\node[yshift=-0.5cm] at (c2.south) {$\overline{A}\Delta1S$};

\end{tikzpicture}
}
\caption{$\overline{A}\Delta1S$}\label{fig:ad1s}
\end{figure}

The idea behind using simplified versions of the failure $\Delta$-constraints is to have simpler patterns to evaluate in the nodes of the process tree. While the original formulation of some failure $\Delta$-constraints is easier to understand, it contains some superfluous sub-patterns that can be omitted, hence easing their overall evaluation. 

Moreover, Theorem \ref{the:adeltas} states that the simplification does not hinder the relation between the constraint and the model not being fully compliant.
In detail, the difference in the simplified constraint consists of removing the sub-pattern regarding the presence in a possible execution of a task having $\neg d$ annotated on the left of the trigger leaf, while not having a task with $d$ annotated between. The remainder of the pattern is sufficient to recognise executions violating the associated achievement obligation to $\overline{A}\Delta1$, as the removed sub-pattern purpose is mainly to differentiate it from $\overline{A}\Delta2$. The proof of Theorem \ref{the:adeltas} shows in details how the simplified version of $\overline{A}\Delta1$ is sufficient to recognise the same as the original version.

\begin{theorem}\label{the:adeltas}
If a violation in a business process model is identified by satisfying either $\overline{A}\Delta1$ or $\overline{A}\Delta2$, then the same violation is identified by satisfying either $\overline{A}\Delta1S$ or $\overline{A}\Delta2$.
\end{theorem}

\subsubsection{Interval Overnode Pattern}

We now introduce the Interval Overnode Pattern, a generalisation of $\overline{A}\Delta1S$.

\begin{definition}[Interval Overnode Pattern]

Let an Interval Overnode Pattern be the following:

$$\kwd{iop}(x, y, z)$$

Where $x$ represents the desired task on the left side of the aggregation, and $z$ the one desired on the right side. Moreover, $y$ represents the task that is required not to appear between the interval identified by $x$ and $y$, as shown in Figure \ref{fig:generic_iop}. 


\begin{figure}[h]
\centering
\scalebox{0.35}{
\begin{tikzpicture}[thick, font=\LARGE]

\node[xshift=10cm,ellipse,draw,align=center,minimum size=1cm,text width=0.75cm] (t2) at (t1.east) {$t$};
\node[xshift=-2cm,yshift=-1cm,ellipse,draw,align=center,minimum size=1cm,text width=0.75cm] (notc2) at (t2.west) {$x$};
\node[xshift=2cm,ellipse,draw,align=center,minimum size=1cm,text width=0.75cm] (d2a) at (t2.east |- notc2) {$z$};
\node[yshift=-2cm,ellipse,draw,align=center,minimum size=1cm,text width=0.75cm] (c2) at (t2) {$y$};

\draw[->,-{Latex[length=2.5mm]}] (t2) to[out=180,in=45] (notc2);
\draw[->,-{Latex[length=2.5mm]}] (d2a) to[out=135,in=0] (t2);
\draw[->,-{Latex[length=2.5mm]}] (d2a) to[out=225,in=0] (c2) node[xshift=0.75cm,yshift=0.125cm] {$\not$};
\draw[->,-{Latex[length=2.5mm]}] (c2) to[out=180,in=315] (notc2) node[xshift=0.75cm,yshift=-0.5cm] {$\not$};

\end{tikzpicture}
}
\caption{Interval Overnode Pattern}\label{fig:generic_iop}
\end{figure}
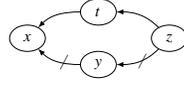

\end{definition}

The $\overline{A}\Delta1S$ pattern can be parametrised as an Interval Overnode Pattern as follows: 

$$\kwd{iop}(t_{\neg r}, t_{r}, t_{d})$$

As such, fulfilling an Interval Overnode Pattern with the above instantiation corresponds to fulfilling the corresponding $\overline{A}\Delta1S$.
An overview of the possible classifications of an overnode is provided in Table~\ref{tb:overnodeclass}.

\begin{table}
{\normalsize
 \begin{tabularx}{\textwidth}{m{0.75cm} X}
\hline
$\mathbf{\overset{+}{x}t\overset{+}{z}}$ & The associated process block contains an execution fulfilling the Interval Overnode Pattern. \\
$\mathbf{\overset{+}{x}t\overset{0}{z}}$ & The associated process block contains an execution fulfilling the left sub-pattern of the Interval Overnode Pattern. \\
$\mathbf{\overset{0}{x}t\overset{+}{z}}$ & The associated process block contains an execution fulfilling the right sub-pattern of the Interval Overnode Pattern. \\
$\mathbf{\overset{+}{x}t\overset{-}{z}}$ & The associated process block contains an execution fulfilling the left sub-pattern of the Interval Overnode Pattern, and fails the right sub-pattern. \\
$\mathbf{\overset{0}{x}t\overset{0}{z}}$ & The associated process block contains an execution not failing or fulfilling either sub-pattern of the Interval Overnode Pattern. \\
$\mathbf{\overset{-}{x}t\overset{+}{z}}$ & The associated process block contains an execution fulfilling the right sub-pattern of the Interval Overnode Pattern, and fails the left sub-pattern. \\
$\mathbf{\overset{0}{x}t\overset{-}{z}}$ & The associated process block contains an execution failing the right sub-pattern of the Interval Overnode Pattern. \\
$\mathbf{\overset{-}{x}t\overset{0}{z}}$ & The associated process block contains an execution failing the left sub-pattern of the Interval Overnode Pattern. \\
$\mathbf{\overset{-}{x}t\overset{-}{z}}$ & The associated process block contains only executions failing the Interval Overnode Pattern. \\
\hline
\end{tabularx}
}
\caption{Overnode classification classes.}\label{tb:overnodeclass}
\end{table}

Even though an overnode can potentially belong to multiple classes, only the classes for which no strictly better class is available, are considered as the classes of a node. The preference relations between the classes is shown in Figure \ref{fig:ad1s_stem_ev}.

\begin{figure}[h]
\centering
\scalebox{0.5}{
\begin{tikzpicture}[thick, font=\LARGE]

\node (a) at (0,0) {$\mathbf{\overset{+}{x}t\overset{+}{z}}$};

\node (b) at (-2,-2) {$\mathbf{\overset{+}{x}t\overset{0}{z}}$};
\node (c) at (2,-2) {$\mathbf{\overset{0}{x}t\overset{+}{z}}$};

\node (d) at (-4,-4) {$\mathbf{\overset{+}{x}t\overset{-}{z}}$};
\node (e) at (0,-4) {$\mathbf{\overset{0}{x}t\overset{0}{z}}$};
\node (f) at (4,-4) {$\mathbf{\overset{-}{x}t\overset{+}{z}}$};

\node (g) at (-2,-6) {$\mathbf{\overset{0}{x}t\overset{-}{z}}$};
\node (h) at (2,-6) {$\mathbf{\overset{-}{x}t\overset{0}{z}}$};

\node (i) at (0,-8) {$\mathbf{\overset{-}{x}t\overset{-}{z}}$};

\draw[->,-{Latex[length=2.5mm]}] (b) -- (a);
\draw[->,-{Latex[length=2.5mm]}] (c) -- (a);

\draw[->,-{Latex[length=2.5mm]}] (d) -- (b);
\draw[->,-{Latex[length=2.5mm]}] (e) -- (b);
\draw[->,-{Latex[length=2.5mm]}] (e) -- (c);
\draw[->,-{Latex[length=2.5mm]}] (f) -- (c);

\draw[->,-{Latex[length=2.5mm]}] (g) -- (d);
\draw[->,-{Latex[length=2.5mm]}] (g) -- (e);
\draw[->,-{Latex[length=2.5mm]}] (h) -- (e);
\draw[->,-{Latex[length=2.5mm]}] (h) -- (f);

\draw[->,-{Latex[length=2.5mm]}] (i) -- (g);
\draw[->,-{Latex[length=2.5mm]}] (i) -- (h);

\end{tikzpicture}
}
\caption{Interval Overnode Pattern Classes' Preference Lattice}\label{fig:ad1s_stem_ev}
\end{figure}
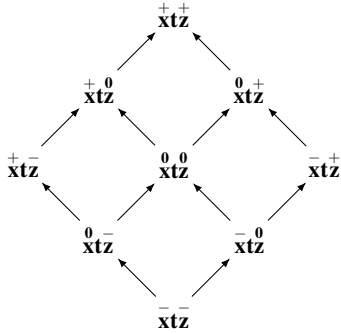

When classifying overnodes along the stem of a process tree, if one of the nodes is classified as $\mathbf{\overset{+}{x}t\overset{+}{z}}$ then the evaluation procedure can be terminated, as this is sufficient proof to conclude that $\overline{A}\Delta1S$ is satisfied. Consequently, the process model associated to the process tree being analysed violates the obligation associated to the failure $\Delta$-constraint.

\subsubsection{Interval Overnode Pattern: Trigger Leaf Classification}

The trigger leaf can be assigned to the following classes: $\mathbf{\overset{+}{x}t\overset{+}{z}}$, $\mathbf{\overset{+}{x}t\overset{0}{z}}$, $\mathbf{\overset{0}{x}t\overset{+}{z}}$, $\mathbf{\overset{0}{x}t\overset{0}{z}}$, depending on the literals annotated on the task associated to the trigger leaf.

Note that if the associated task has $y$ in its annotations, then the evaluation can safely be terminated, as the process-tree given this particular trigger leaf will never satisfy the Interval Overnode Pattern.

\subsubsection{Interval Overnode Pattern: Undernodes Classification for \kwd{SEQ} Overnode}

Considering the Interval Overnode Pattern, it can be divided into its left and right sides as illustrated in Figure~\ref{fig:iop_left_right_subs}. The undernodes can be then classified and aggregated using the simpler \emph{Left Sub-Pattern}, and \emph{Right Sub-Pattern}.

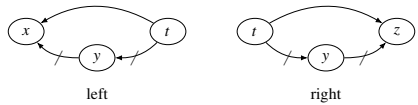
\begin{figure}[h]
\centering
\scalebox{0.35}{
\begin{tikzpicture}[thick, font=\LARGE]

\node[ellipse,draw,align=center,minimum size=1cm,text width=0.75cm] (m1) at (0,0) {$y$};
\node[xshift=-2cm,yshift=1cm,ellipse,draw,align=center,minimum size=1cm,text width=0.75cm] (l1) at (m1.west) {$x$};
\node[xshift=2cm,yshift=1cm,ellipse,draw,align=center,minimum size=1cm,text width=0.75cm] (r1) at (m1.east) {$t$};

\draw[->,-{Latex[length=2.5mm]}] (r1) to[out=150,in=30] (l1);
\draw[->,-{Latex[length=2.5mm]}] (r1) to[out=225,in=0] (m1); \node[xshift=-1cm,yshift=0cm] at (m1.west) {$\not$};
\draw[->,-{Latex[length=2.5mm]}] (m1) to[out=180,in=315] (l1); \node[xshift=0.5cm,yshift=0cm] at (m1.east) {$\not$};

\node[yshift=-0.5cm,anchor=north] at (m1.south) {left};

\node[xshift=8cm,ellipse,draw,align=center,minimum size=1cm,text width=0.75cm] (m2) at (m1.east) {$y$};
\node[xshift=-2cm,yshift=1cm,ellipse,draw,align=center,minimum size=1cm,text width=0.75cm] (l2) at (m2.west) {$t$};
\node[xshift=2cm,yshift=1cm,ellipse,draw,align=center,minimum size=1cm,text width=0.75cm] (r2) at (m2.east) {$z$};

\draw[->,-{Latex[length=2.5mm]}] (l2) to[in=150,out=30] (r2);
\draw[->,-{Latex[length=2.5mm]}] (m2) to[in=225,out=0] (r2); \node[xshift=-1cm,yshift=0cm] at (m2.west) {$\not$};
\draw[->,-{Latex[length=2.5mm]}] (l2) to[in=180,out=315] (m2); \node[xshift=0.5cm,yshift=0cm] at (m2.east) {$\not$};

\node[yshift=-0.5cm,anchor=north] at (m2.south) {right};

\end{tikzpicture}
}
\caption{Interval Overnode Pattern Left and Right Sub-Patterns}\label{fig:iop_left_right_subs}
\end{figure}

\subsubsection{Left Sub-Pattern}

\begin{definition}[Left Sub-Pattern]
Let a Left Sub-Pattern be the following:

$$\kwd{lsp}(x, y)$$

Where $x$ represents the desired task on the right side of the aggregation, and $y$ the undesired task on the right side of $x$, as shown in Figure \ref{fig:generic_lsp}.


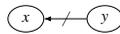
\begin{figure}[h]
\centering
\scalebox{0.35}{
\begin{tikzpicture}[thick, font=\LARGE]

\node[ellipse,draw,align=center,minimum size=1cm,text width=0.75cm] (x) at (0,0) {$x$};
\node[ellipse,draw,align=center,minimum size=1cm,text width=0.75cm] (y) at (3,0) {$y$};

\draw[->,-{Latex[length=2.5mm]}] (y) -- (x);
\node[xshift=-1cm] at (y.west) {$\not$};

\end{tikzpicture}
}
\caption{Left Sub-Pattern}\label{fig:generic_lsp}
\end{figure}
\end{definition}

Notice that for $\overline{A}\Delta1S$, the corresponding parametrisation of the left sub-pattern check is the following:

$$\kwd{lsp}(t_{\neg r}, t_{r})$$

Considering the left sub-pattern, an undernode in the process tree can belong to the following classes. The association to a class depends on the properties of the associated process block.

\begin{table}
{\normalsize
 \begin{tabularx}{\textwidth}{m{0.375cm} X}
$\mathbf{\overset{+}{x}}$ & The associated process block contains an execution capable of fulfilling the left sub-pattern. \\
$\mathbf{\overset{0}{x}}$ & The associated process block contains an execution that does not influence the left sub-pattern. \\
$\mathbf{\overset{-}{x}}$ & The associated process block contains only executions hindering the capability of fulfilling the left sub-pattern. \\
\end{tabularx}
}
\end{table}

Even though a process block can potentially belong to multiple classes, a process block belongs to the classes for which no better class is available, in accordance to the following linear preference order of the classes: $\mathbf{\overset{-}{x}} < \mathbf{\overset{0}{x}} < \mathbf{\overset{+}{x}}$.

\subsubsection{Right Sub-Pattern}

We discuss now the \emph{Right Sub-Pattern}, used to evaluate the right side of the stem parting from an overnode of type \kwd{SEQ}.

\begin{definition}[Right Sub-Pattern]
Let a Right Sub-Pattern be the following:

$$\kwd{rsp}(y, z)$$

Where $z$ represents the desired task on the left side of the aggregation, and $y$ the undesired task on the left side of $z$, as shown in Figure \ref{fig:generic_rsp}.


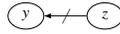
\begin{figure}[h]
\centering
\scalebox{0.35}{
\begin{tikzpicture}[thick, font=\LARGE]

\node[ellipse,draw,align=center,minimum size=1cm,text width=0.75cm] (y) at (0,0) {$y$};
\node[ellipse,draw,align=center,minimum size=1cm,text width=0.75cm] (z) at (3,0) {$z$};

\draw[->,-{Latex[length=2.5mm]}] (z) -- (y);
\node[xshift=-1cm] at (z.west) {$\not$};

\end{tikzpicture}
}
\caption{Right Sub-Pattern}\label{fig:generic_rsp}
\end{figure}
\end{definition}

Notice that for $\overline{A}\Delta1S$, the corresponding parametrisation of the right sub-pattern check is the following:

$$\kwd{rsp}(t_{r}, t_{d})$$

Considering the right sub-pattern, an undernode in the process tree can belong to the following classes. The association to a class depends on the properties of the associated process block.

\begin{table}
{\normalsize
\begin{tabularx}{\textwidth}{m{0.375cm} X}
$\mathbf{\overset{+}{z}}$ & The associated process block contains an execution capable of fulfilling the right sub-pattern of the right sub-pattern. \\
$\mathbf{\overset{0}{z}}$ & The associated process block contains an execution that does not influence the right sub-pattern of the right sub-pattern. \\
$\mathbf{\overset{-}{z}}$ & The associated process block contains only executions hindering the capability of the right sub-pattern. \\
\end{tabularx}
}
\end{table}

Even though a process block can potentially belong to multiple classes, a process block belongs to the classes for which no better class is available, in accordance to the following linear preference order of the classes: $\mathbf{\overset{-}{z}} < \mathbf{\overset{0}{z}} < \mathbf{\overset{+}{z}}$

\subsubsection{Interval Overnode Pattern: Undernodes Classification for \kwd{AND} Overnode}

We introduce a generalised sub-pattern to evaluate the undernodes: \emph{Interval Sub-Pattern}.

\begin{definition}[Interval Sub-Pattern]

Let an Interval Sub-Pattern be the following:

$$\kwd{isp}(x, y, z)$$

Where $x$ represents the desired task on the right side of the interval, and $z$ the left side within an aggregation. Additionally, $y$ represent a task that is undesired when appearing within the interval formed by $x$ and $z$, as shown in Figure \ref{fig:generic_isp}.


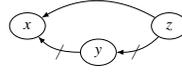
\begin{figure}[h]
\centering
\scalebox{0.35}{
\begin{tikzpicture}[thick, font=\LARGE]

\node[ellipse,draw,align=center,minimum size=1cm,text width=0.75cm] (m1) at (0,0) {$y$};
\node[xshift=-2cm,yshift=1cm,ellipse,draw,align=center,minimum size=1cm,text width=0.75cm] (l1) at (m1.west) {$x$};
\node[xshift=2cm,yshift=1cm,ellipse,draw,align=center,minimum size=1cm,text width=0.75cm] (r1) at (m1.east) {$z$};

\draw[->,-{Latex[length=2.5mm]}] (r1) to[out=150,in=30] (l1);
\draw[->,-{Latex[length=2.5mm]}] (r1) to[out=225,in=0] (m1); \node[xshift=-1cm,yshift=0cm] at (m1.west) {$\not$};
\draw[->,-{Latex[length=2.5mm]}] (m1) to[out=180,in=315] (l1); \node[xshift=0.5cm,yshift=0cm] at (m1.east) {$\not$};

\end{tikzpicture}
}
\caption{Interval Sub-Pattern}\label{fig:generic_isp}
\end{figure}
\end{definition}

For $\overline{A}\Delta1S$, the parametrisation of the interval sub-pattern is as follows:

$$\kwd{isp}(t_{\neg r}, t_{r}, t_{d})$$

\begin{table}
{\normalsize
\begin{tabularx}{\textwidth}{m{0.5cm} X}
\hline
$\mathbf{\underline{\overset{+}{x}\overset{+}{z}}}$ & The associated process block contains an execution containing an interval fulfilling the complete interval sub-pattern. \\

$\mathbf{\underline{\overset{+}{x}\overset{0}{z}}}$ & The associated process block contains an execution fulfilling the left side of the interval sub-pattern, and with no failure conditions for the right side. \\

$\mathbf{\underline{\overset{0}{x}\overset{+}{z}}}$ & The associated process block contains an execution fulfilling the right side of the interval sub-pattern, and with no failure conditions for the left side. \\

$\mathbf{\underline{\overset{+}{x}\overset{-}{z}}}$ & The associated process block contains an execution fulfilling the left side of the interval sub-pattern, while failing its right side. \\

$\mathbf{\underline{\overset{0}{x}\overset{0}{z}}}$ & The associated process block contains an execution not fulfilling the interval sub-pattern, and not containing any undesired element. \\

$\mathbf{\underline{\overset{-}{x}\overset{+}{z}}}$ & The associated process block contains an execution fulfilling the right side of the interval sub-pattern, while failing its left side. \\

$\mathbf{\underline{\overset{-}{xz}}}$ & The associated process block contains only executions failing either side of the interval sub-pattern without partially achieving any of them. \\
\hline
\end{tabularx}
}
\caption{Interval sub-pattern classification classes.}\label{tb:ispclass}
\end{table}

Note that in Table~\ref{tb:ispclass} the classes $\mathbf{\underline{\overset{0}{x}\overset{-}{z}}}$ and $\mathbf{\underline{\overset{-}{x}\overset{0}{z}}}$ are missing. It is not necessary to explicitly include these classes, as the class $\mathbf{\underline{\overset{-}{xz}}}$ covers them. The intuitive reason is the following: as the omitted classes represent scenarios where either the left side or the right side of the pattern contain elements leading to the failure of $\overline{A}\Delta1$, and no additional helpful element towards the satisfaction of the constraint are present. Such classifications can be reduced to the worst classification type, because as these sub-patterns are used to evaluate \kwd{AND} nodes, their semantics allows to disregard such sub-blocks when evaluating the failure $\Delta$-constraint. This means that we are not required to keep track of these cases, as they can all be aggregated to the worst class $\mathbf{\underline{\overset{-}{xz}}}$.

Even though a process block can potentially belong to multiple classes, a process block belongs to the classes for which no better class is available, in accordance to the preference lattice shown in Figure \ref{fig:ad1s_c_latt}.
%
Despite the fact that the evaluation of an overnode can increase due to a single evaluation aggregated with a classification producing multiple ones, the number of evaluations required to be kept is limited to at most 3 in the worst case scenario due to the preference order expressed in Figure \ref{fig:ad1s_c_latt}.


\begin{figure}[h]
\centering
\scalebox{0.5}{
\begin{tikzpicture}[thick, font=\LARGE]

\node (a) at (0,0) {$\mathbf{\overset{+}{x}\overset{+}{z}}$};

\node (b) at (-2,-2) {$\mathbf{\overset{+}{x}\overset{0}{z}}$};
\node (c) at (2,-2) {$\mathbf{\overset{0}{x}\overset{+}{z}}$};

\node (d) at (-2,-4) {$\mathbf{\overset{+}{x}\overset{-}{z}}$};
\node (e) at (0,-4) {$\mathbf{\overset{0}{x}\overset{0}{z}}$};
\node (f) at (2,-4) {$\mathbf{\overset{-}{x}\overset{+}{z}}$};

\node (g) at (0,-6) {$\mathbf{\overset{-}{xz}}$};

\draw[->,-{Latex[length=2.5mm]}] (b) -- (a);
\draw[->,-{Latex[length=2.5mm]}] (c) -- (a);

\draw[->,-{Latex[length=2.5mm]}] (d) -- (b);
\draw[->,-{Latex[length=2.5mm]}] (e) -- (b);
\draw[->,-{Latex[length=2.5mm]}] (e) -- (c);
\draw[->,-{Latex[length=2.5mm]}] (f) -- (c);

\draw[->,-{Latex[length=2.5mm]}] (g) -- (d);
\draw[->,-{Latex[length=2.5mm]}] (g) -- (e);
\draw[->,-{Latex[length=2.5mm]}] (g) -- (f);

\end{tikzpicture}
}
\caption{Interval Sub-Pattern Classes Preference Lattice}\label{fig:ad1s_c_latt}
\end{figure}
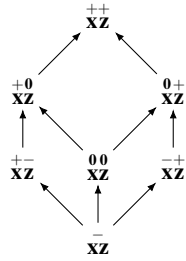

\subsubsection{$\overline{A}\Delta1S$: Computational Complexity of the Aggregations}

The evaluation of an overnode of type \kwd{SEQ} is linear with respect to the size of the undernode children sub-trees. When evaluating an overnode of type \kwd{AND}, we must consider the number of possible simultaneous classifications that can appear while evaluating an overnode sub-pattern and an interval sub-pattern. In both cases, the number is limited to 3, as it can be seen in the preference lattices shown in Figure \ref{fig:ad1s_stem_ev} and Figure \ref{fig:ad1s_c_latt}.
Therefore, we can conclude that the overall complexity of verifying whether a process model satisfies $\overline{A}\Delta1S$ is determined by the complexity of aggregating the evaluations of the overnodes belonging to its process tree. As evaluating an overnode of type \kwd{AND} is more complex than evaluating one of type \kwd{SEQ}, we can generalise the complexity of checking the satisfaction of $\overline{A}\Delta1S$ as:  is \textbf{O}($9n$) evaluations/aggregations, where $n$ is the number of nodes in the tree, representing the number of blocks in a business process. Therefore the complexity of evaluating is polynomial with respect to the size of the process model.

\subsubsection{Example: Evaluating $\overline{A}\Delta1S$ along a Process Tree}


In this sub-section, we illustrate an example of how the evaluation is performed on a process tree while checking whether the related process model contains an execution satisfying $\overline{A}\Delta1S$.

In Figure~\ref{fig:ad1s_example_tree}b, the process tree is shown derived from the process introduced in Figure~\ref{ex:processexample}. Notice that the \kwd{start} task contains the negation of the sought annotated requirement of the obligation, in accordance to Assumption \ref{ass:starting_complementary}, and the deadline is contained in the \kwd{end} task, in accordance to Assumption \ref{ass:final_deadline}.


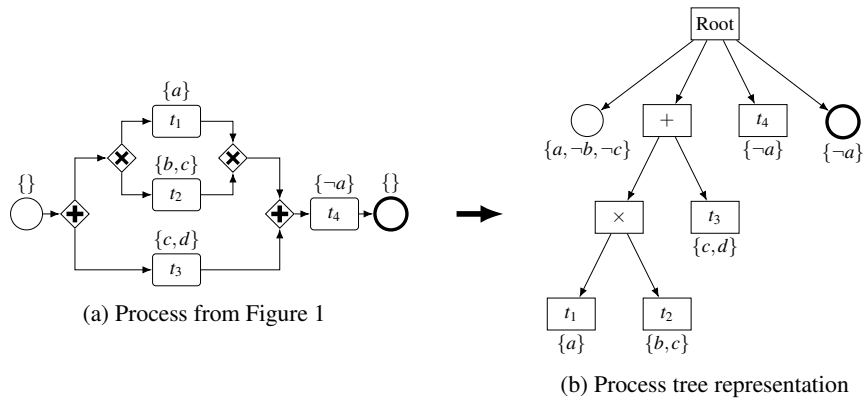
\begin{figure}[h!]
    \centering
    \begin{subfigure}[t!]{.5\textwidth}
        \centering
		\scalebox{0.425}{
		\begin{tikzpicture}[thick,font=\LARGE]

\newcommand{\gatesize}{0.9cm}

\node[draw,align=center,circle, minimum size=1cm] (start) at (0,0) {};

\node[xshift=1cm,draw,diamond,align=center,minimum height=\gatesize,minimum width=\gatesize] (and1) at (start.east) {};
\draw[line width=4pt] ([yshift=-2mm]and1.north) -- ([yshift=2mm]and1.south);
\draw[line width=4pt] ([xshift=-2mm]and1.east) -- ([xshift=2mm]and1.west);

\node[xshift=1cm,yshift=1.75cm,draw,diamond,align=center,minimum height=\gatesize,minimum width=\gatesize] (xor1) at (and1.east) {};
\draw[line width=3pt] ([xshift=0.5mm,yshift=-0.5mm]xor1.north west) -- ([xshift=-0.5mm,yshift=0.5mm]xor1.south east);
\draw[line width=3pt] ([xshift=-0.5mm,yshift=-0.5mm]xor1.north east) -- ([xshift=0.5mm,yshift=0.5mm]xor1.south west);

\node[xshift=1.25cm,yshift=1.15cm,draw,rounded corners=3pt,align=center,minimum height=1cm,minimum width=1.5cm] (t1) at (xor1.east) {$t_1$};
\node[xshift=1.25cm,yshift=-1.15cm,draw,rounded corners=3pt,align=center,minimum height=1cm,minimum width=1.5cm] (t2) at (xor1.east) {$t_2$};
\node[yshift=-1.75cm,draw,rounded corners=3pt,align=center,minimum height=1cm,minimum width=1.5cm] (t3) at (t2 |- and1) {$t_3$};

\node[xshift=1cm,draw,diamond,align=center,minimum height=\gatesize,minimum width=\gatesize] (xor2) at (t1.east |- xor1) {};
\draw[line width=3pt] ([xshift=0.5mm,yshift=-0.5mm]xor2.north west) -- ([xshift=-0.5mm,yshift=0.5mm]xor2.south east);
\draw[line width=3pt] ([xshift=-0.5mm,yshift=-0.5mm]xor2.north east) -- ([xshift=0.5mm,yshift=0.5mm]xor2.south west);

\node[xshift=1cm,draw,diamond,align=center,minimum height=\gatesize,minimum width=\gatesize] (and2) at (xor2.east |- and1) {};
\draw[line width=4pt] ([yshift=-2mm]and2.north) -- ([yshift=2mm]and2.south);
\draw[line width=4pt] ([xshift=-2mm]and2.east) -- ([xshift=2mm]and2.west);

\node[xshift=1.25cm,draw,rounded corners=3pt,align=center,minimum height=1cm,minimum width=1.5cm] (t4) at (and2.east) {$t_4$};

\node[xshift=1cm,draw,line width=3pt,align=center,circle, minimum size=1cm] (end) at (t4.east) {};

\node[anchor=south] at (start.north) {\{\}}; 
\node[anchor=south] at (t1.north) {\{$a$\}};
\node[anchor=south] at (t2.north) {\{$b,c$\}};
\node[anchor=south] at (t3.north) {\{$c,d$\}};
\node[anchor=south] at (t4.north) {\{$\neg a$\}};
\node[anchor=south] at (end.north) {\{\}}; 

\draw[->,-{Latex[length=7.5mm]},line width=5pt] ([xshift=1.5cm]end.east) -- ([xshift=3cm]end.east);

\draw[->,-{Latex[length=3mm]}] (start) -- (and1);
\draw[->,-{Latex[length=3mm]}] (and1) -- (and1 |- xor1) -- (xor1);
\draw[->,-{Latex[length=3mm]}] (and1) -- (and1 |- t3) -- (t3);
\draw[->,-{Latex[length=3mm]}] (xor1) -- (xor1 |- t1) -- (t1);
\draw[->,-{Latex[length=3mm]}] (xor1) -- (xor1 |- t2) -- (t2);
\draw[->,-{Latex[length=3mm]}] (t1) -- (xor2 |- t1) -- (xor2);
\draw[->,-{Latex[length=3mm]}] (t2) -- (xor2 |- t2) -- (xor2);
\draw[->,-{Latex[length=3mm]}] (xor2) -- (and2 |- xor2) -- (and2);
\draw[->,-{Latex[length=3mm]}] (t3) -- (and2 |- t3) -- (and2);
\draw[->,-{Latex[length=3mm]}] (and2) -- (t4);
\draw[->,-{Latex[length=3mm]}] (t4) -- (end);

\end{tikzpicture}
		}\caption{Process from Figure~\ref{ex:processexample} \;\;\;\;\;\;\;\;\;}
    \end{subfigure}%
    ~ 
    \begin{subfigure}[t!]{.5\textwidth}
        \centering
		\scalebox{0.425}{
		\begin{tikzpicture}[thick,font=\LARGE]

\node[draw,align=center,minimum height=1cm,minimum width=1.5cm] (troot) at (0,0) {Root};

\node[xshift=-4cm,yshift=-2cm,draw,ellipse,align=center,minimum size=1cm,anchor=north] (tc1) at (troot.south) {};
\node[xshift=-1.5cm,yshift=-2cm,draw,align=center,minimum height=1cm,minimum width=1.5cm,anchor=north] (tc2) at (troot.south) {$+$};
\node[xshift=1.5cm,yshift=-2cm,draw,align=center,minimum height=1cm,minimum width=1.5cm,anchor=north] (tc3) at (troot.south) {$t_4$};
\node[xshift=4cm,yshift=-2cm,draw,ellipse,align=center,minimum size=1cm,anchor=north,line width=3pt] (tc4) at (troot.south) {};

\node[xshift=-1.5cm,yshift=-2cm,draw,align=center,minimum height=1cm,minimum width=1.5cm,anchor=north] (tc5) at (tc2.south) {$\times$};
\node[xshift=1.5cm,yshift=-2cm,draw,align=center,minimum height=1cm,minimum width=1.5cm,anchor=north] (tc6) at (tc2.south) {$t_3$};

\node[xshift=-1.5cm,yshift=-2cm,draw,align=center,minimum height=1cm,minimum width=1.5cm,anchor=north] (tc7) at (tc5.south) {$t_1$};
\node[xshift=1.5cm,yshift=-2cm,draw,align=center,minimum height=1cm,minimum width=1.5cm,anchor=north] (tc8) at (tc5.south) {$t_2$};

\node[anchor=north] (an1) at (tc1.south) {$\{a, \neg b, \neg c\}$};
\node[anchor=north] (an3) at (tc3.south) {$\{\neg a\}$};
\node[anchor=north] (an4) at (tc4.south) {$\{\neg a\}$};
\node[anchor=north] (an6) at (tc6.south) {$\{c, d\}$};
\node[anchor=north] (an7) at (tc7.south) {$\{a\}$};
\node[anchor=north] (an8) at (tc8.south) {$\{b, c\}$};

\draw[->,-{Latex[length=3mm]}] (troot) -- (tc1);
\draw[->,-{Latex[length=3mm]}] (troot) -- (tc2);
\draw[->,-{Latex[length=3mm]}] (troot) -- (tc3);
\draw[->,-{Latex[length=3mm]}] (troot) -- (tc4);

\draw[->,-{Latex[length=3mm]}] (tc2) -- (tc5);
\draw[->,-{Latex[length=3mm]}] (tc2) -- (tc6);

\draw[->,-{Latex[length=3mm]}] (tc5) -- (tc7);
\draw[->,-{Latex[length=3mm]}] (tc5) -- (tc8);

\end{tikzpicture}
		}\caption{Process tree representation}
    \end{subfigure}
    \caption{Process tree representation.}\label{fig:ad1s_example_tree}
\end{figure}

Given the obligation $\Obl^a \langle b, c, \neg a\rangle$, there are two trigger leaves candidates: $t_2$ and $t_3$, as they both contain the annotation $c$.
Let us first consider the process tree using $t_2$ as trigger leaf and proceed with the classification of the leaves in the tree. However, since $t_2$ satisfies the termination condition while evaluating the trigger leaf, we can safely state that for this particular stem the process does not satisfy $\overline{A}\Delta1S$.
Let us now consider $t_3$ as the trigger leaf. We can classify the leaves of the tree as shown in Figure \ref{fig:ad1s_example_tree_1}, where:
\begin{itemize}
\item $t_3$, being the trigger leaf on the stem (and, therefore, an overnode) is classified according to the overnode sub-pattern.
\item $t_1$ and $t_2$ are classified according to the interval sub pattern, as their parent node on the stem (their parent overnode) is an \kwd{AND}.
\item \kwd{start} is classified according to the left sub-pattern as its parent overnode is a \kwd{SEQ} and it sits on the left of the stem.
\item $t_4$ and \kwd{end} are classified according to the righ sub-pattern as their parent overnode is a \kwd{SEQ} and it sits on the right of the stem.
\end{itemize}

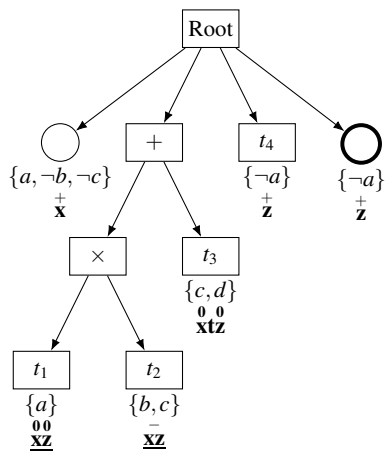
\begin{figure}[h!]
\centering
\scalebox{0.5}{
\begin{tikzpicture}[thick,font=\LARGE]

\node[draw,align=center,minimum height=1cm,minimum width=1.5cm] (troot) at (0,0) {Root};

\node[xshift=-4cm,yshift=-2cm,draw,ellipse,align=center,minimum size=1cm,anchor=north] (tc1) at (troot.south) {};
\node[xshift=-1.5cm,yshift=-2cm,draw,align=center,minimum height=1cm,minimum width=1.5cm,anchor=north] (tc2) at (troot.south) {$+$};
\node[xshift=1.5cm,yshift=-2cm,draw,align=center,minimum height=1cm,minimum width=1.5cm,anchor=north] (tc3) at (troot.south) {$t_4$};
\node[xshift=4cm,yshift=-2cm,draw,ellipse,align=center,minimum size=1cm,anchor=north,line width=3pt] (tc4) at (troot.south) {};

\node[xshift=-1.5cm,yshift=-2cm,draw,align=center,minimum height=1cm,minimum width=1.5cm,anchor=north] (tc5) at (tc2.south) {$\times$};
\node[xshift=1.5cm,yshift=-2cm,draw,align=center,minimum height=1cm,minimum width=1.5cm,anchor=north] (tc6) at (tc2.south) {$t_3$};

\node[xshift=-1.5cm,yshift=-2cm,draw,align=center,minimum height=1cm,minimum width=1.5cm,anchor=north] (tc7) at (tc5.south) {$t_1$};
\node[xshift=1.5cm,yshift=-2cm,draw,align=center,minimum height=1cm,minimum width=1.5cm,anchor=north] (tc8) at (tc5.south) {$t_2$};

\node[anchor=north] (an1) at (tc1.south) {$\{a, \neg b, \neg c\}$};
\node[anchor=north] (an3) at (tc3.south) {$\{\neg a\}$};
\node[anchor=north] (an4) at (tc4.south) {$\{\neg a\}$};
\node[anchor=north] (an6) at (tc6.south) {$\{c, d\}$};
\node[anchor=north] (an7) at (tc7.south) {$\{a\}$};
\node[anchor=north] (an8) at (tc8.south) {$\{b, c\}$};

\node[yshift=0.2cm,anchor=north] (lan1) at (an1.south) {$\mathbf{\overset{+}{x}}$}; 
\node[yshift=0.2cm,anchor=north] (lan3) at (an3.south) {$\mathbf{\overset{+}{z}}$}; 
\node[yshift=0.2cm,anchor=north] (lan4) at (an4.south) {$\mathbf{\overset{+}{z}}$}; 
\node[yshift=0.2cm,anchor=north] (lan6) at (an6.south) {$\mathbf{\overset{0}{x}t\overset{0}{z}}$}; 
\node[yshift=0.2cm,anchor=north] (lan7) at (an7.south) {$\mathbf{\underline{\overset{0}{x}\overset{0}{z}}}$}; 
\node[yshift=0.2cm,anchor=north] (lan8) at (an8.south) {$\mathbf{\underline{\overset{-}{xz}}}$};

\draw[->,-{Latex[length=3mm]}] (troot) -- (tc1);
\draw[->,-{Latex[length=3mm]}] (troot) -- (tc2);
\draw[->,-{Latex[length=3mm]}] (troot) -- (tc3);
\draw[->,-{Latex[length=3mm]}] (troot) -- (tc4);

\draw[->,-{Latex[length=3mm]}] (tc2) -- (tc5);
\draw[->,-{Latex[length=3mm]}] (tc2) -- (tc6);

\draw[->,-{Latex[length=3mm]}] (tc5) -- (tc7);
\draw[->,-{Latex[length=3mm]}] (tc5) -- (tc8);

\end{tikzpicture}
}
\caption{Evaluating $\overline{A}\Delta1S$ on a process tree (1).}\label{fig:ad1s_example_tree_1}
\end{figure}

In Figure \ref{fig:ad1s_example_tree_2}, we show the aggregations towards the classifications of the overnode \kwd{AND} on the stem. Note that the \kwd{XOR} undernode classification is the union of the classifications of the children, and then pruned according to the priority lattice.


\begin{figure}[h!]
\centering
\scalebox{0.5}{
\begin{tikzpicture}[thick,font=\LARGE]

\node[draw,align=center,minimum height=1cm,minimum width=1.5cm] (troot) at (0,0) {Root};

\node[xshift=-4cm,yshift=-2cm,draw,ellipse,align=center,minimum size=1cm,anchor=north] (tc1) at (troot.south) {};
\node[xshift=-1.5cm,yshift=-2cm,draw,align=center,minimum height=1cm,minimum width=1.5cm,anchor=north] (tc2) at (troot.south) {$+$};
\node[xshift=1.5cm,yshift=-2cm,draw,align=center,minimum height=1cm,minimum width=1.5cm,anchor=north] (tc3) at (troot.south) {$t_4$};
\node[xshift=4cm,yshift=-2cm,draw,ellipse,align=center,minimum size=1cm,anchor=north,line width=3pt] (tc4) at (troot.south) {};

\node[xshift=-1.5cm,yshift=-2cm,draw,align=center,minimum height=1cm,minimum width=1.5cm,anchor=north] (tc5) at (tc2.south) {$\times$};
\node[xshift=1.5cm,yshift=-2cm,draw,align=center,minimum height=1cm,minimum width=1.5cm,anchor=north] (tc6) at (tc2.south) {$t_3$};

\node[xshift=-1.5cm,yshift=-2cm,draw,align=center,minimum height=1cm,minimum width=1.5cm,anchor=north] (tc7) at (tc5.south) {$t_1$};
\node[xshift=1.5cm,yshift=-2cm,draw,align=center,minimum height=1cm,minimum width=1.5cm,anchor=north] (tc8) at (tc5.south) {$t_2$};

\node[anchor=north] (an1) at (tc1.south) {$\{a, \neg b, \neg c\}$};
\node[anchor=north] (an3) at (tc3.south) {$\{\neg a\}$};
\node[anchor=north] (an4) at (tc4.south) {$\{\neg a\}$};
\node[anchor=north] (an6) at (tc6.south) {$\{c, d\}$};
\node[anchor=north] (an7) at (tc7.south) {$\{a\}$};
\node[anchor=north] (an8) at (tc8.south) {$\{b, c\}$};

\node[yshift=0.2cm,anchor=north] (lan1) at (an1.south) {$\mathbf{\overset{+}{x}}$}; 
\node[yshift=0.2cm,anchor=north] (lan3) at (an3.south) {$\mathbf{\overset{+}{z}}$}; 
\node[yshift=0.2cm,anchor=north] (lan4) at (an4.south) {$\mathbf{\overset{+}{z}}$}; 
\node[yshift=0.2cm,anchor=north] (lan6) at (an6.south) {$\mathbf{\overset{0}{x}t\overset{0}{z}}$}; 
\node[yshift=0.2cm,anchor=north] (lan7) at (an7.south) {$\mathbf{\underline{\overset{0}{x}\overset{0}{z}}}$}; 
\node[yshift=0.2cm,anchor=north] (lan8) at (an8.south) {$\mathbf{\underline{\overset{-}{xz}}}$}; 

\node[anchor=west] (lan2) at (tc2.east) {$\mathbf{\overset{0}{x}t\overset{0}{z}}$}; 
\node[anchor=west] (lan5) at (tc5.east) {$\mathbf{\underline{\overset{0}{x}\overset{0}{z}}}$};

\draw[->,-{Latex[length=3mm]}] (troot) -- (tc1);
\draw[->,-{Latex[length=3mm]}] (troot) -- (tc2);
\draw[->,-{Latex[length=3mm]}] (troot) -- (tc3);
\draw[->,-{Latex[length=3mm]}] (troot) -- (tc4);

\draw[->,-{Latex[length=3mm]}] (tc2) -- (tc5);
\draw[->,-{Latex[length=3mm]}] (tc2) -- (tc6);

\draw[->,-{Latex[length=3mm]}] (tc5) -- (tc7);
\draw[->,-{Latex[length=3mm]}] (tc5) -- (tc8);

\end{tikzpicture}
}
\caption{Evaluating $\overline{A}\Delta1S$ on a process tree (2).}\label{fig:ad1s_example_tree_2}
\end{figure}
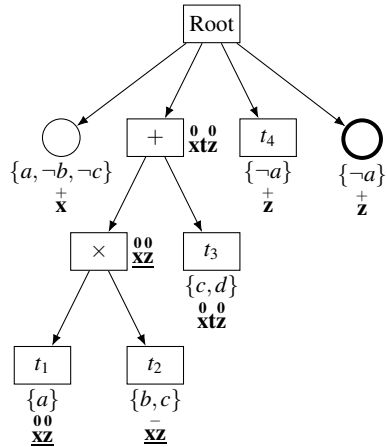

Figure \ref{fig:ad1s_example_tree_3} depicts the aggregation towards the classification of the overnode \kwd{SEQ}, the root. Notice that the resulting classification is $\mathbf{\overset{+}{x}t\overset{+}{z}}$\footnote{The aggregation is performed in two steps, the first step involves the overnode class with either one of the left or right classifications. Then the result is aggregated with the remainder, producing the final result.}, hence the process tree, using this stem, does satisfy $\overline{A}\Delta1S$. 


\begin{figure}[h]
\centering
\scalebox{0.5}{
\begin{tikzpicture}[thick,font=\LARGE]

\node[draw,align=center,minimum height=1cm,minimum width=1.5cm] (troot) at (0,0) {Root};

\node[xshift=-4cm,yshift=-2cm,draw,ellipse,align=center,minimum size=1cm,anchor=north] (tc1) at (troot.south) {};
\node[xshift=-1.5cm,yshift=-2cm,draw,align=center,minimum height=1cm,minimum width=1.5cm,anchor=north] (tc2) at (troot.south) {$+$};
\node[xshift=1.5cm,yshift=-2cm,draw,align=center,minimum height=1cm,minimum width=1.5cm,anchor=north] (tc3) at (troot.south) {$t_4$};
\node[xshift=4cm,yshift=-2cm,draw,ellipse,align=center,minimum size=1cm,anchor=north,line width=3pt] (tc4) at (troot.south) {};

\node[xshift=-1.5cm,yshift=-2cm,draw,align=center,minimum height=1cm,minimum width=1.5cm,anchor=north] (tc5) at (tc2.south) {$\times$};
\node[xshift=1.5cm,yshift=-2cm,draw,align=center,minimum height=1cm,minimum width=1.5cm,anchor=north] (tc6) at (tc2.south) {$t_3$};

\node[xshift=-1.5cm,yshift=-2cm,draw,align=center,minimum height=1cm,minimum width=1.5cm,anchor=north] (tc7) at (tc5.south) {$t_1$};
\node[xshift=1.5cm,yshift=-2cm,draw,align=center,minimum height=1cm,minimum width=1.5cm,anchor=north] (tc8) at (tc5.south) {$t_2$};

\node[anchor=north] (an1) at (tc1.south) {$\{a, \neg b, \neg c\}$};
\node[anchor=north] (an3) at (tc3.south) {$\{\neg a\}$};
\node[anchor=north] (an4) at (tc4.south) {$\{\neg a\}$};
\node[anchor=north] (an6) at (tc6.south) {$\{c, d\}$};
\node[anchor=north] (an7) at (tc7.south) {$\{a\}$};
\node[anchor=north] (an8) at (tc8.south) {$\{b, c\}$};

\node[yshift=0.2cm,anchor=north] (lan1) at (an1.south) {$\mathbf{\overset{+}{x}}$}; 
\node[yshift=0.2cm,anchor=north] (lan3) at (an3.south) {$\mathbf{\overset{+}{z}}$}; 
\node[yshift=0.2cm,anchor=north] (lan4) at (an4.south) {$\mathbf{\overset{+}{z}}$}; 
\node[yshift=0.2cm,anchor=north] (lan6) at (an6.south) {$\mathbf{\overset{0}{x}t\overset{0}{z}}$}; 
\node[yshift=0.2cm,anchor=north] (lan7) at (an7.south) {$\mathbf{\underline{\overset{0}{x}\overset{0}{z}}}$}; 
\node[yshift=0.2cm,anchor=north] (lan8) at (an8.south) {$\mathbf{\underline{\overset{-}{xz}}}$}; 

\node[anchor=west] (lanroot) at (troot.east) {$\mathbf{\overset{+}{x}t\overset{+}{z}}$}; 
\node[anchor=west] (lan2) at (tc2.east) {$\mathbf{\overset{0}{x}t\overset{0}{z}}$}; 
\node[anchor=west] (lan5) at (tc5.east) {$\mathbf{\underline{\overset{0}{x}\overset{0}{z}}}$};

\draw[->,-{Latex[length=3mm]}] (troot) -- (tc1);
\draw[->,-{Latex[length=3mm]}] (troot) -- (tc2);
\draw[->,-{Latex[length=3mm]}] (troot) -- (tc3);
\draw[->,-{Latex[length=3mm]}] (troot) -- (tc4);

\draw[->,-{Latex[length=3mm]}] (tc2) -- (tc5);
\draw[->,-{Latex[length=3mm]}] (tc2) -- (tc6);

\draw[->,-{Latex[length=3mm]}] (tc5) -- (tc7);
\draw[->,-{Latex[length=3mm]}] (tc5) -- (tc8);

\end{tikzpicture}
}
\caption{Evaluating $\overline{A}\Delta1S$ on a process tree (3).}\label{fig:ad1s_example_tree_3}
\end{figure}
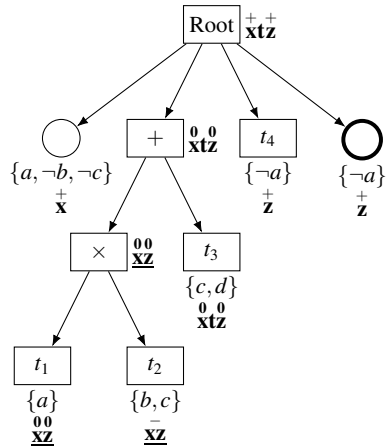
\FloatBarrier

\subsection{$\overline{A}\Delta2$}\label{sec:ad2}

We hereby show the classifications used by \kwd{Evaluate}($\mathcal{P}'$, $\tau$, $\delta$) when $\delta$ is $\overline{A}\Delta2$. Notice that the procedure is actually evaluated over two simplified constraints of $\overline{A}\Delta2$ obtained from its decomposition, which we first introduce.

\subsubsection{Decomposing $\overline{A}\Delta2$}

We now introduce the decomposed counterparts of $\overline{A}\Delta2$, which we refer to as $\overline{A}\Delta2.1$ and $\overline{A}\Delta2.2$. The two decomposed $\Delta$-constraints are derived from $\overline{A}\Delta2$ as illustrated in Figure \ref{f:ad2_pattern_reductions}.


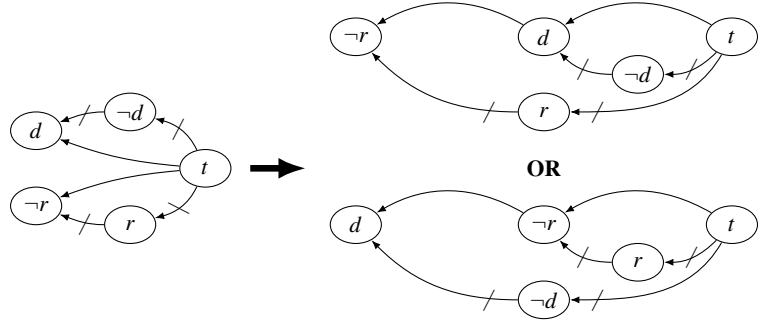
\begin{figure}[h!]
\centering
\scalebox{0.5}{
\begin{tikzpicture}[thick, font=\LARGE]

\node[ellipse,draw,align=center,minimum size=1cm,text width=0.75cm] (c) at (-4.5,-1) {$\neg r$};
\node[ellipse,draw,align=center,minimum size=1cm,text width=0.75cm] (b) at (-2,-1.5) {$r$};
\node[ellipse,draw,align=center,minimum size=1cm,text width=0.75cm] (a) at (0,0) {$t$};

\node[ellipse,draw,align=center,minimum size=1cm,text width=0.75cm] (e) at (-4.5,1) {$d$};
\node[ellipse,draw,align=center,minimum size=1cm,text width=0.75cm] (d) at (-2,1.5) {$\neg d$};

\draw[->,-{Latex[length=2.5mm]}] (a) to[out=185,in=20] (c);
\draw[->,-{Latex[length=2.5mm]}] (a) to[out=175,in=-20] (e);
\draw[->,-{Latex[length=2.5mm]}] (a) to[out=245,in=15] (b); \node[xshift=-0.75cm] at (b.west) {$\not$};
\draw[->,-{Latex[length=2.5mm]}] (b) to[out=180,in=-20] (c); \node[xshift=0.6cm,yshift=0.2cm,rotate=90] at (b.east) {$\not$};

\draw[->,-{Latex[length=2.5mm]}] (a) to[out=115,in=-15] (d); \node[xshift=-0.75cm,yshift=-0.1cm] at (d.west) {$\not$};
\draw[->,-{Latex[length=2.5mm]}] (d) to[out=180,in=20] (e); \node[xshift=0.35cm,yshift=-0.45cm] at (d.east) {$\not$};

\node[ellipse,draw,align=center,minimum size=1cm,text width=0.75cm] (p1) at (4,3.5) {$\neg r$};
\node[ellipse,draw,align=center,minimum size=1cm,text width=0.75cm] (q1) at (9,1.5) {$r$};
\node[ellipse,draw,align=center,minimum size=1cm,text width=0.75cm] (r1) at (9,3.5) {$d$};
\node[ellipse,draw,align=center,minimum size=1cm,text width=0.75cm] (s1) at (11.5,2.5) {$\neg d$};
\node[ellipse,draw,align=center,minimum size=1cm,text width=0.75cm] (t1) at (14,3.5) {$t$};

\draw[->,-{Latex[length=2.5mm]}] (r1) to[out=150,in=30] (p1);
\draw[->,-{Latex[length=2.5mm]}] (t1) to[out=240,in=0] (q1); \node[xshift=-1cm,yshift=0cm] at (q1.west) {$\not$};
\draw[->,-{Latex[length=2.5mm]}] (q1) to[out=180,in=315] (p1); \node[xshift=0.5cm,yshift=0cm] at (q1.east) {$\not$};

\draw[->,-{Latex[length=2.5mm]}] (t1) to[out=150,in=30] (r1);
\draw[->,-{Latex[length=2.5mm]}] (t1) to[out=225,in=0] (s1); \node[xshift=-1cm,yshift=0.125cm] at (s1.west) {$\not$};
\draw[->,-{Latex[length=2.5mm]}] (s1) to[out=180,in=315] (r1); \node[xshift=0.5cm,yshift=0.125cm] at (s1.east) {$\not$};

\node[ellipse,draw,align=center,minimum size=1cm,text width=0.75cm] (p2) at (4,-1.5) {$d$};
\node[ellipse,draw,align=center,minimum size=1cm,text width=0.75cm] (q2) at (9,-3.5) {$\neg d$};
\node[ellipse,draw,align=center,minimum size=1cm,text width=0.75cm] (r2) at (9,-1.5) {$\neg r$};
\node[ellipse,draw,align=center,minimum size=1cm,text width=0.75cm] (s2) at (11.5,-2.5) {$r$};
\node[ellipse,draw,align=center,minimum size=1cm,text width=0.75cm] (t2) at (14,-1.5) {$t$};

\draw[->,-{Latex[length=2.5mm]}] (r2) to[out=150,in=30] (p2);
\draw[->,-{Latex[length=2.5mm]}] (t2) to[out=240,in=0] (q2); \node[xshift=-1cm,yshift=0cm] at (q2.west) {$\not$};
\draw[->,-{Latex[length=2.5mm]}] (q2) to[out=180,in=315] (p2); \node[xshift=0.5cm,yshift=0cm] at (q2.east) {$\not$};

\draw[->,-{Latex[length=2.5mm]}] (t2) to[out=150,in=30] (r2);
\draw[->,-{Latex[length=2.5mm]}] (t2) to[out=225,in=0] (s2); \node[xshift=-1cm,yshift=0.125cm] at (s2.west) {$\not$};
\draw[->,-{Latex[length=2.5mm]}] (s2) to[out=180,in=315] (r2); \node[xshift=0.5cm,yshift=0.125cm] at (s2.east) {$\not$};

\draw[->,-{Latex[length=7.5mm]},line width=5pt] ([xshift=0.5cm]a.east) -- ([xshift=2cm]a.east);

\node at (r1 |- a) {\bf OR};

\end{tikzpicture}
}
\caption{$\overline{A}\Delta$2 Pattern Reductions}\label{f:ad2_pattern_reductions}
\end{figure}

\begin{definition}[$\overline{A}\Delta2$ Decomposed]\label{def:ad2d}
The Achievement Failure $\Delta$-constraint from Definition \ref{def:af_delta}, recalled below:

\bigskip
\noindent\begin{tabular}{@{}ll}
$\overline{A}\Delta2$ & $\exists t_t$ such that: \\
& $\exists t_{\neg r}, \neg \exists t_{r}  | t_{\neg r} \preceq t_{r} \preceq t_t$ and \\
& $\exists t_d, \neg \exists t_{\neg d} | t_d \preceq t_{\neg d} \preceq t_t$ \\
\end{tabular}

\bigskip


\noindent is simplified as follows:

\bigskip
\noindent\begin{tabular}{@{}ll}
$\overline{A}\Delta2.1$ & $\exists t_t$ such that: \\
& $\exists t_{\neg r}, t_d | t_{\neg r} \preceq t_d \preceq t_t$ and \\
& $\neg \exists t_{\neg d} | t_{d} \preceq t_{\neg d} \preceq t_t$ and \\
& $\neg \exists t_r | t_{\neg r} \preceq t_r \preceq t_d$ \\
$\overline{A}\Delta2.2$ & $\exists t_t$ such that: \\
& $\exists t_{\neg r}, t_d | t_d \preceq t_{\neg r} \preceq t_t$ and \\
& $\neg \exists t_{\neg d} | t_{d} \preceq t_{\neg d} \preceq t_{\neg r}$ and \\
& $\neg \exists t_r | t_{\neg r} \preceq t_r \preceq t_t$ \\
\end{tabular}

\bigskip

\end{definition}

$\overline{A}\Delta2$ contains two sub-patterns, and both of these sub-patterns sit on the left of the trigger leaf. The idea behind the simplification is to serialise the two possible interleavings of the sub-patterns to ease their verification. While the simplification introduces an additional constraint to verify, it allows to get rid of the complexity of verifying interleaved sub-patterns. Finally, Theorem \ref{the:ad2simple} shows how verifying the two decomposed constraints is equivalent to verifying the original.

\begin{theorem}\label{the:ad2simple}
If a violation in business process model is identified by satisfying $\overline{A}\Delta2$, then the same violation is identified by satisfying either $\overline{A}\Delta2.1$ or $\overline{A}\Delta2.2$.
\end{theorem}

\subsubsection{Generalised Sequence Pattern}

\begin{definition}[Generalised Sequence Pattern]

Let a Generalised Sequence Pattern be the following:

$$\kwd{gsp}(x, y, z, k)$$

Where $z$ represents the desired task on the right side of the aggregation, and $k$ the undesired task on the right side of $z$. Moreover $x$ represents a desired task on the left of $z$, and with no $y$ between $x$ and $z$, as shown in Figure \ref{f:ad2_gsp}. 


\begin{figure}[h!]
\centering
\scalebox{0.5}{
\begin{tikzpicture}[thick, font=\LARGE]

\node[ellipse,draw,align=center,minimum size=1cm,text width=0.75cm] (a) at (-3,0) {$x$};
\node[ellipse,draw,align=center,minimum size=1cm,text width=0.75cm] (b) at (0,-1.5) {$y$};
\node[ellipse,draw,align=center,minimum size=1cm,text width=0.75cm] (c) at (0,0) {$z$};

\node[ellipse,draw,align=center,minimum size=1cm,text width=0.75cm] (d) at (3,0) {$k$};

\draw[->,-{Latex[length=2.5mm]}] (c) to[out=150,in=30] (a);
\draw[->,-{Latex[length=2.5mm]}] (d) to[out=240,in=0] (b); \node[xshift=-1cm,yshift=0.125cm] at (b.west) {$\not$};
\draw[->,-{Latex[length=2.5mm]}] (b) to[out=180,in=315] (a); \node[xshift=0.5cm,yshift=0.125cm] at (b.east) {$\not$};

\draw[->,-{Latex[length=2.5mm]}] (d) to[out=150,in=30] (c); \node[xshift=-1cm,yshift=0.5cm] at (d.west) {$\not$};

\end{tikzpicture}
}
\caption{Generalised Sequence Constraint}\label{f:ad2_gsp}
\end{figure}
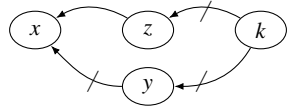

\end{definition}

Despite the \emph{generalised sequence pattern} looking different than $\overline{A}\Delta2.1$ and $\overline{A}\Delta2.2$, it captures the key properties that neither the unwanted elements $k$ and $y$ should appear on the right of the required element $z$, and maintaining the additional condition that there should be another required element $x$ on the left of $z$, and $y$ should also not be between them.

Notice that, given an obligation $\Obl^o$ $\langle r,$ $t,$ $d\rangle$, the parametrisations of its $\Delta$-constrains $\overline{A}\Delta2.1$ and $\overline{A}\Delta2.2$, in generalised sequence patterns, is the following:

\begin{description}
\item[$\overline{A}\Delta2.1$] $\kwd{gsp}(t_{\neg r}, t_{r}, t_{d}, t_{\neg d})$
\item[$\overline{A}\Delta2.2$] $\kwd{gsp}(t_{d}, t_{\neg d}, t_{\neg r}, t_{r})$
\end{description}

\subsubsection{Generalised Sequence Pattern Undernode's Classes}

The following classifications can be used to classify the undernodes while checking for a decomposed $\overline{A}\Delta2$ and evaluating an overnode of type \kwd{SEQ}. Notice also that due to the shape of the pattern, only the undernodes on the left of the stem are considered.
An overview of the classifications of the generalised sequence pattern is provided in Table~\ref{tb:gspclass}.

\begin{table}
{\normalsize
\begin{tabularx}{\textwidth}{m{0.5cm} X}
\hline
$\mathbf{\overset{+}{x}\overset{+}{z}}$ & Exists an execution such that: both $x$ and $z$ sub-patterns are fulfilled and in the correct order.\\
$\mathbf{\overset{+}{x}\overset{0}{z}}$ & Exists an execution such that: $x$ sub-pattern is satisfied, but $z$ sub-pattern is not-satisfied on the right of the former, neither invalidated.\\
$\mathbf{\overset{0}{x}\overset{+}{z}}$ & Exists an execution such that: $z$ sub-pattern is satisfied and $x$ is neither satisfied or invalidated on the left of the former.\\
$\mathbf{\overset{0}{x}\overset{0}{z}}$ & Exists an execution such that: neither sub-pattern invalidated or satisfied.\\
$\mathbf{\overset{-}{x}\overset{+}{z}}$ & Exists an execution such that: $z$ sub-pattern is satisfied and $x$ is invalidated on the left of the former.\\
$\mathbf{\overset{+}{x}\overset{-}{z}}$ & Exists an execution such that: $x$ sub-pattern is satisfied, but $z$ sub-pattern is invalidated on the right of the former.\\
$\mathbf{\overset{0}{x}\overset{-}{z}}$ & Exists an execution: $z$ is falsified and $x$ is neither satisfied or falsified on the left of $z$.\\
$\mathbf{\overset{-}{xz}}$ & For every execution: $z$ is falsified or $x$ is falsified without having a satisfied $z$ on its right. (Represents the remainder possibilities: $\mathbf{\overset{-}{x}\overset{-}{z}}$ and $\mathbf{\overset{-}{x}\overset{0}{z}}$). \\
\hline
\end{tabularx}
}
\caption{Classification classes of the generalised sequence pattern.}\label{tb:gspclass}
\end{table}

Even though a process block can potentially belong to multiple classes, a process block belongs to the classes for which no better class is available, in accordance to the preference lattice shown in Figure \ref{f:ad2_alt_sequndernodeclasses}.

\begin{figure}[h!]
\centering
\scalebox{0.5}{
\begin{tikzpicture}[thick, font=\LARGE]

\node (a) at (0,0) {$\mathbf{\overset{+}{x}\overset{+}{z}}$};

\node (b) at (-2,-2) {$\mathbf{\overset{+}{x}\overset{0}{z}}$};
\node (c) at (2,-2) {$\mathbf{\overset{0}{x}\overset{+}{z}}$};

\node (d) at (-2,-4) {$\mathbf{\overset{+}{x}\overset{-}{z}}$};
\node (e) at (0,-4) {$\mathbf{\overset{0}{x}\overset{0}{z}}$};
\node (f) at (2,-4) {$\mathbf{\overset{-}{x}\overset{+}{z}}$};

\node (g) at (0,-6) {$\mathbf{\overset{0}{x}\overset{-}{z}}$};
\node (h) at (2,-8) {$\mathbf{\overset{-}{xz}}$};

\draw[->,-{Latex[length=2.5mm]}] (b) -- (a);
\draw[->,-{Latex[length=2.5mm]}] (c) -- (a);

\draw[->,-{Latex[length=2.5mm]}] (d) -- (b);
\draw[->,-{Latex[length=2.5mm]}] (e) -- (b);
\draw[->,-{Latex[length=2.5mm]}] (e) -- (c);
\draw[->,-{Latex[length=2.5mm]}] (f) -- (c);

\draw[->,-{Latex[length=2.5mm]}] (g) -- (d);
\draw[->,-{Latex[length=2.5mm]}] (g) -- (e);

\draw[->,-{Latex[length=2.5mm]}] (h) -- (f);
\draw[->,-{Latex[length=2.5mm]}] (h) -- (g);

\end{tikzpicture}
}
\caption{Generalised Sequence Pattern Undernode Classes}\label{f:ad2_alt_sequndernodeclasses}
\end{figure}
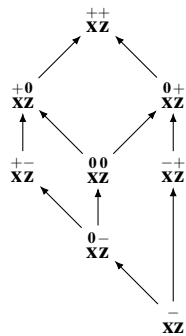

\subsubsection{Generalised Sequence Pattern Overnode's Classes}

The overnodes on the stem being evaluated while checking for $\overline{A}\Delta2.1$ or $\overline{A}\Delta2.2$, use a variant of the \emph{Generalised Sequence Pattern} classification.

In particular, the labels of the classes are the same, but with appended on the right \textbf{t}, signifying that the classification constraints of the definitions for the generalised sequence patterns are supposed to have always on the right the \emph{trigger leaf} task.\footnote{Notice that, since these classifications are always on the stem, and we are classifying the children nodes of the tree on the left of the stem, then this is always the case.} For instance, considering the class $\mathbf{\overset{+}{x}\overset{+}{z}}$ used to classify undernodes, its corresponding class used for overnodes would be $\mathbf{\overset{+}{x}\overset{+}{z}t}$, and it would have the same properties as $\mathbf{\overset{+}{x}\overset{+}{z}}$, with the addition of having a $t_t$ on the right of each considered element.The graphical representation of the pattern used to classify the overnodes is shown in Figure \ref{f:ad2_alt_yxsubpatterns}.


\begin{figure}[h!]
\centering
\scalebox{0.5}{
\begin{tikzpicture}[thick, font=\LARGE]

\node[ellipse,draw,align=center,minimum size=1cm,text width=0.75cm] (a) at (-5,0) {$y$};
\node[ellipse,draw,align=center,minimum size=1cm,text width=0.75cm] (b) at (-2.5,-1) {$\neg y$};
\node[ellipse,draw,align=center,minimum size=1cm,text width=0.75cm] (c) at (0,0) {$x$};

\node[ellipse,draw,align=center,minimum size=1cm,text width=0.75cm] (e) at (5,0) {$t$};
\node[ellipse,draw,align=center,minimum size=1cm,text width=0.75cm] (d) at (2.5,-1) {$\neg x$};

\draw[->,-{Latex[length=2.5mm]}] (c) to[out=150,in=30] (a);
\draw[->,-{Latex[length=2.5mm]}] (c) to[out=225,in=0] (b); \node[xshift=-1cm,yshift=0.125cm] at (b.west) {$\not$};
\draw[->,-{Latex[length=2.5mm]}] (b) to[out=180,in=315] (a); \node[xshift=0.5cm,yshift=0.125cm] at (b.east) {$\not$};

\draw[->,-{Latex[length=2.5mm]}] (e) to[out=150,in=30] (c);
\draw[->,-{Latex[length=2.5mm]}] (e) to[out=225,in=0] (d); \node[xshift=-1cm,yshift=0.125cm] at (d.west) {$\not$};
\draw[->,-{Latex[length=2.5mm]}] (d) to[out=180,in=315] (c); \node[xshift=0.5cm,yshift=0.125cm] at (d.east) {$\not$};

\end{tikzpicture}
}
\caption{$\overline{A}\Delta$2 \kwd{SEQ} Generalised Sub Patterns}\label{f:ad2_alt_yxsubpatterns}
\end{figure}
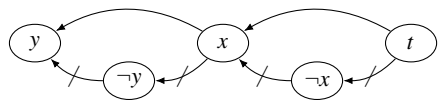

When classifying overnodes along the stem of a process tree, if one of the nodes is classified as $\mathbf{\overset{+}{x}\overset{+}{z}t}$, then the evaluation procedure can be terminated, as this is sufficient proof to conclude that one of the decomposed version of $\overline{A}\Delta2$ is satisfied. That is, the process model associated to the process tree being analysed violates the obligation associated to the $\Delta$-constraint.

Noticing the similarity between the classes used to classify undernodes and overnodes while verifying a generalised sequence pattern, the aggregation tables describing how undernode classes are aggregated to classify their undernode parent can be used to classify an overnode on the stem. More precisely, the column $\kwd{SEQ}(A, B)$ can be used as the aggregation result to classify an overnode of type \kwd{SEQ} on the stem when using the column $A$ as the left-to-right classification aggregation of its undernode children on the left of the stem, and using the column $B$ as the equivalent classification of the overnode's overnode child.

\subsubsection{\kwd{AND} Overnode}

When evaluating an overnode of type \kwd{AND}, its undernodes can be classified using the \emph{Interval Sub-Pattern}. Note that in this case every undernode, either on the left or the right of the stem, must be evaluated.

As illustrated in Figure \ref{f:ad2_ad1s_andaggcomp}, the same pattern used to evaluate $\overline{A}\Delta1S$ can be used to evaluate both $\overline{A}\Delta2.1$ and $\overline{A}\Delta2.2$. In Figure \ref{f:ad2_ad1s_andaggcomp}, the $\overline{A}\Delta2$ decomposed patterns are obtained by adding the trigger task on the right of the interval sub-pattern. 


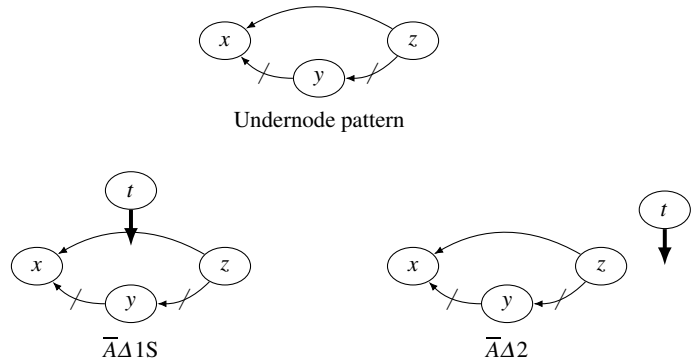
\begin{figure}[h!]
\centering
\scalebox{0.5}{
\begin{tikzpicture}[thick, font=\LARGE]

\node[ellipse,draw,align=center,minimum size=1cm,text width=0.75cm] (a) at (-2.5,0) {$x$};
\node[ellipse,draw,align=center,minimum size=1cm,text width=0.75cm] (b) at (0,-1) {$y$};
\node[ellipse,draw,align=center,minimum size=1cm,text width=0.75cm] (c) at (2.5,0) {$z$};

\draw[->,-{Latex[length=2.5mm]}] (c) to[out=150,in=30] (a);
\draw[->,-{Latex[length=2.5mm]}] (c) to[out=225,in=0] (b); \node[xshift=-1cm,yshift=0.125cm] at (b.west) {$\not$};
\draw[->,-{Latex[length=2.5mm]}] (b) to[out=180,in=315] (a); \node[xshift=0.5cm,yshift=0.125cm] at (b.east) {$\not$};

\node[yshift=-0.25cm,anchor=north,align=center] at (b.south) {Undernode pattern};

\node[ellipse,draw,align=center,minimum size=1cm,text width=0.75cm] (d) at (-7.5,-6) {$x$};
\node[ellipse,draw,align=center,minimum size=1cm,text width=0.75cm] (e) at (-5,-7) {$y$};
\node[ellipse,draw,align=center,minimum size=1cm,text width=0.75cm] (f) at (-2.5,-6) {$z$};

\node[yshift=2.5cm,ellipse,draw,align=center,minimum size=1cm,text width=0.75cm] (t1) at (e.north) {$t$};

\draw[->,-{Latex[length=2.5mm]}] (f) to[out=150,in=30] (d);
\draw[->,-{Latex[length=2.5mm]}] (f) to[out=225,in=0] (e); \node[xshift=-1cm,yshift=0.125cm] at (e.west) {$\not$};
\draw[->,-{Latex[length=2.5mm]}] (e) to[out=180,in=315] (d); \node[xshift=0.5cm,yshift=0.125cm] at (e.east) {$\not$};

\draw[->,-{Latex[length=5mm]},line width=3pt] (t1) -- ([yshift=1cm]e.north);

\node[yshift=-0.25cm,anchor=north,align=center] at (e.south) {$\overline{A}\Delta$1S};

\node[ellipse,draw,align=center,minimum size=1cm,text width=0.75cm] (g) at (2.5,-6) {$x$};
\node[ellipse,draw,align=center,minimum size=1cm,text width=0.75cm] (h) at (5,-7) {$y$};
\node[ellipse,draw,align=center,minimum size=1cm,text width=0.75cm] (i) at (7.5,-6) {$z$};

\node[xshift=1cm,yshift=1.5cm,ellipse,draw,align=center,minimum size=1cm,text width=0.75cm] (t2) at (i.east) {$t$};

\draw[->,-{Latex[length=2.5mm]}] (i) to[out=150,in=30] (g);
\draw[->,-{Latex[length=2.5mm]}] (i) to[out=225,in=0] (h); \node[xshift=-1cm,yshift=0.125cm] at (h.west) {$\not$};
\draw[->,-{Latex[length=2.5mm]}] (h) to[out=180,in=315] (g); \node[xshift=0.5cm,yshift=0.125cm] at (h.east) {$\not$};

\draw[->,-{Latex[length=5mm]},line width=3pt] (t2) -- ([xshift=1cm]i.east);

\node[yshift=-0.25cm,anchor=north,align=center] at (h.south) {$\overline{A}\Delta$2};

\end{tikzpicture}
}
\caption{$\overline{A}\Delta2$ decomposed Patterns \kwd{AND} Overnode Aggregation Comparison With $\overline{A}\Delta1S$}\label{f:ad2_ad1s_andaggcomp}
\end{figure}

While evaluating a generalised sequence pattern: $\kwd{gsp}(x, y, z, k)$, the interval sub-pattern used to classify its undernodes is the following: $\kwd{isp}(x, y, z)$.

\subsubsection{$\overline{A}\Delta2.1$ and $\overline{A}\Delta2.1$:Computational Complexity of the Aggregations}

In general, for both \kwd{SEQ} and \kwd{AND} overnode evaluations, in the worst case scenario at most 3 classifications are used at the same time to classify a single node of the process-tree. Therefore, the overall complexity of evaluating whether a process tree (and, as such, the associated business process model) satisfies one of the decomposed pattern representations of $\overline{A}\Delta2$ (namely $\overline{A}\Delta2.1$ and $\overline{A}\Delta2.2$) is \textbf{O}($3^2n$), where $n$ is the number of nodes in the process tree. Note that independently from the type of overnode on the stem, the complexity of its evaluation is \textbf{O}($3^2n$).

Finally, considering that to actually evaluate $\overline{A}\Delta2$ on a process tree, both its decomposed patterns need to be evaluated. As such, the actual complexity of evaluating $\overline{A}\Delta2$ is \textbf{O}($2\times3^2n$).

\subsection{$\overline{M}\Delta1$}\label{sec:md1}

We hereby show the classifications used by \kwd{Evaluate}($\mathcal{P}'$, $\tau$, $\delta$) when $\delta$ is $\overline{M}\Delta1$. The illustration of the $\Delta$-constraint is recalled in Figure \ref{fig:md1}.

\begin{figure}[h!]
\centering
\scalebox{0.4}{
\begin{tikzpicture}[thick, font=\LARGE]

\node[ellipse,draw,align=center,minimum size=1cm,text width=0.75cm] (t1) at (0,0) {$t$};
\node[xshift=-1.5cm,yshift=-1.5cm,ellipse,draw,align=center,minimum size=1cm,text width=0.75cm] (c1) at (t1.west) {$r$};
\node[xshift=-1.5cm,ellipse,draw,align=center,minimum size=1cm,text width=0.75cm] (notc1) at (c1.west |- t1) {$\neg r$};

\draw[->,-{Latex[length=2.5mm]}] (t1) to[out=240,in=20] (c1) node[xshift=0.5cm,yshift=0.25cm] {$\not$};
\draw[->,-{Latex[length=2.5mm]}] (t1) to[out=160,in=20] (notc1);
\draw[->,-{Latex[length=2.5mm]}] (c1) to[out=160,in=315] (notc1) node[xshift=0.5cm,yshift=-0.625cm] {$\not$};

\end{tikzpicture}
}
\caption{$\overline{M}\Delta1$}\label{fig:md1}
\end{figure}
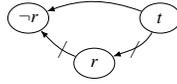

\subsubsection{Evaluate Using LSP}

To evaluate $\overline{M}\Delta1$ we use the left sub-pattern with the following parametrisation:

$$\kwd{lsp}(\neg r, r)$$

The possible classifications of an overnode, derived from the existing classifications of the left sub-pattern, are listed as follows:

\begin{table}
{\normalsize
\begin{tabularx}{\textwidth}{m{0.5cm} X}
$\mathbf{\overset{+}{x}t}$ & The associated process block contains an execution capable of fulfilling $\overline{M}\Delta1$.\\
$\mathbf{\overset{0}{x}t}$ & The associated process block contains an execution that does not contain elements influencing whether $\overline{M}\Delta1$ is satisfied or violated.\\
$\mathbf{\overset{-}{x}t}$ & The associated process block contains only executions hindering the capability of fulfilling $\overline{M}\Delta1$.\\
\end{tabularx}
}
\end{table}

The classes follow the following linear preference order of the classes: 

$$\mathbf{\overset{-}{x}t} < \mathbf{\overset{0}{x}t} < \mathbf{\overset{+}{x}t}$$

The undernodes of an overnode of type \kwd{SEQ} are classified using left sub-pattern classes, and only the undernodes on the left of the stem are used to evaluate the overnode. The aggregation to evaluate the overnode follows exactly the left sub-pattern aggregations.

\subsubsection{\kwd{AND} Overnode Evaluation: Interleaved Generic Pattern}

Let the Interleaved Generic Pattern be the following:

$$igp(k, y)$$

The undernodes of an overnode of type \kwd{AND} are evaluated using to the following Interleaved Generic Pattern classes:

\bigskip

{\normalsize
\noindent\begin{tabularx}{\textwidth}{@{}m{0.375cm} X}
$\mathbf{\overset{+}{k}}$ & The associated process block contains a task having a $k$ annotated and no $y$.\\
$\mathbf{\overset{0}{k}}$ & Otherwise.\\
\end{tabularx}
}

\bigskip

The classification of the undernodes of an \kwd{AND} overnode while evaluating $\overline{M}\Delta1$ is done using the following parametrisation: $igp(\neg r, r)$.
The evaluations of the undernodes can be aggregated with the Left Sub-Pattern classification of the overnode using the \kwd{AND} column of the aggregations while considering one of the columns \kwd{A} of \kwd{B} as the equivalent Interleaved Generic Pattern classification and the other as the classification of the overnode's overnode child.

\subsubsection{$\overline{M}\Delta1$: Computational Complexity of the Aggregations}

As the preferences between the classes is a total order independent from the type of overnode, the overall computational complexity of $\overline{M}\Delta1$ is linear: \textbf{O}$(n)$, where $n$ is the number of nodes in the process tree.

\subsection{$\overline{M}\Delta2$}\label{sec:md2}

We hereby show the classifications used by \kwd{Evaluate}($\mathcal{P}'$, $\tau$, $\delta$) when $\delta$ is $\overline{M}\Delta2$. The evaluation is done using a simplified version of the failure $\Delta$-constraint, defined as follows.

\begin{definition}[$\overline{M}\Delta2$ Simplified]\label{def:mD2S}
$\overline{M}\Delta2$ from Definition \ref{def:mf_delta} can be simplified as follows:

\bigskip
\noindent\begin{tabular}{@{}ll}
$\overline{M}\Delta2S$ & $\exists t_t$ such that: \\
& $\exists t_{\neg r}, \neg \exists t_{d} | t_t \preceq t_{d} \preceq t_{\neg r}$ \\
\end{tabular}



\end{definition}

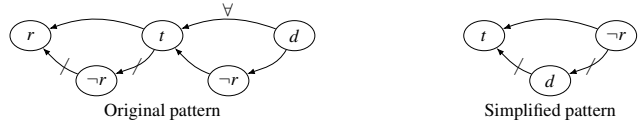
\begin{figure}[h!]
\centering
\scalebox{0.4}{
\begin{tikzpicture}[thick, font=\LARGE]

\node[ellipse,draw,align=center,minimum size=1cm,text width=0.75cm] (t1) at (0,0) {$t$};
\node[xshift=-1.5cm,yshift=-1.5cm,ellipse,draw,align=center,minimum size=1cm,text width=0.75cm] (notc1) at (t1.west) {$\neg r$};
\node[xshift=-1.5cm,ellipse,draw,align=center,minimum size=1cm,text width=0.75cm] (c1) at (notc1.west |- t1) {$r$};
\node[xshift=1.5cm,ellipse,draw,align=center,minimum size=1cm,text width=0.75cm] (notc2) at (t1.east |- notc1) {$\neg r$};
\node[xshift=1.5cm,ellipse,draw,align=center,minimum size=1cm,text width=0.75cm] (d1) at (notc2.east |- t1) {$d$};

\draw[->,-{Latex[length=2.5mm]}] (t1) to[out=240,in=20] (notc1) node[xshift=0.5cm,yshift=0.25cm] {$\not$};
\draw[->,-{Latex[length=2.5mm]}] (t1) to[out=160,in=20] (c1);
\draw[->,-{Latex[length=2.5mm]}] (notc1) to[out=160,in=315] (c1) node[xshift=0.5cm,yshift=-0.625cm] {$\not$};

\draw[->,-{Latex[length=2.5mm]}] (notc2) to[out=160,in=315] (t1); 
\draw[->,-{Latex[length=2.5mm]}] (d1) to[out=240,in=20] (notc2); 
\draw[->,-{Latex[length=2.5mm]}] (d1) to[out=160,in=20] (t1) node[xshift=1.55cm,yshift=0.625cm] {$\forall$};

\node[xshift=10cm,ellipse,draw,align=center,minimum size=1cm,text width=0.75cm] (notc3) at (d1.east) {$\neg r$};
\node[xshift=-1.5cm,yshift=-1.5cm,ellipse,draw,align=center,minimum size=1cm,text width=0.75cm] (d2) at (notc3.west) {$d$};
\node[xshift=-1.5cm,ellipse,draw,align=center,minimum size=1cm,text width=0.75cm] (t2) at (d2.west |- notc3) {$t$};

\draw[->,-{Latex[length=2.5mm]}] (d2) to[out=160,in=315] (t2) node[xshift=0.5cm,yshift=-0.625cm] {$\not$};
\draw[->,-{Latex[length=2.5mm]}] (notc3) to[out=240,in=20] (d2) node[xshift=0.5cm,yshift=0.25cm] {$\not$};
\draw[->,-{Latex[length=2.5mm]}] (notc3) to[out=160,in=20] (t2);

\node[yshift=-0.5cm] at (t1 |- notc1.south) {Original pattern};
\node[yshift=-0.5cm] at (d2.south) {Simplified pattern};

\end{tikzpicture}
}
\caption{$\overline{M}\Delta2$ Simplified}\label{f:mD2S}
\end{figure}

\begin{theorem}
If a violation in business process model is identified by satisfying either $\overline{M}\Delta1$ or $\overline{M}\Delta2$, then the same violation is identified by satisfying either $\overline{M}\Delta1$ or $\overline{M}\Delta2S$.
\end{theorem}

\subsubsection{Evaluate Using RSP}

To evaluate $\overline{M}\Delta2S$, we use the right sub-pattern with the following parametrisation:

$$\kwd{rsp}(r, d)$$

When an overnode is of type \kwd{SEQ}, then the undernodes on the right of the stem can be classified using right sub-pattern classes and used to evaluate the overnode using the same aggregations.

\subsubsection{\kwd{AND} Overnode Evaluation using IGP}

In a similar way for $\overline{M}\Delta1$, the classification of the undernodes of an \kwd{AND} overnode while evaluating $\overline{M}\Delta2S$ is done using the following parametrisation: $igp(\neg r, d)$.

The evaluations of the undernodes can be aggregated with the Right Sub-Pattern classification of the overnode using the \kwd{AND} column of the aggregations while considering one of the columns \kwd{A} of \kwd{B} as the equivalent Interleaved Generic Pattern classification and the other as the classification of the overnode's overnode child.

\subsubsection{$\overline{M}\Delta2S$ Overall Complexity}

As the preferences between the classes is a total order, the overall complexity of verifying whether a process model satisfies $\overline{M}\Delta2S$ is determined by the complexity of aggregating the evaluations of the overnodes belonging to its process tree. As evaluating an overnode of type \kwd{AND} or of type \kwd{SEQ} is linear: \textbf{O}$(n)$, where $n$ is the number of nodes in the process tree.
\section{Conclusion}\label{sec:conclusion}

We have shown that verifying a failure $\Delta$-constraint over a process-tree for one of its trigger leaves is achievable in polynomial time with respect to the size of the tree. Given an obligation, independent of whether such obligation is of type achievement or maintenance, its corresponding failure $\Delta$-constraints can be verified over the process tree for each of its trigger leaves and determine whether the process model corresponding to the tree is \emph{fully compliant} with the obligation as follows:

\begin{description}
\item[Not Fully Compliant] If one of the failure $\Delta$-constraints is positively verified for the process tree and one of its trigger leaves, then it means that there exists an execution of the process violating the obligation associated to the failure $\Delta$-constraint, hence the process model cannot be fully compliant.
\item[Fully Compliant] Otherwise, when no failure $\Delta$-constraint is positively verified for each of the trigger leaves of the process tree, then every execution satisfies the obligation, and the process is fully compliant.
\end{description}

Notice that the number of failure $\Delta$-constraints being verified is 3 when the obligation is of type achievement, and 2 when the obligation is of type maintenance. Moreover, the number of trigger leaves in a process model can in the worst case be the same as the amount of tasks in the process. Although this can be significant, the procedure iterating through each of them for each failure $\Delta$-constraint and aggregating through the nodes of the tree is still polynomial with respect to the size of the process model.

Finally, even considering a set of obligations, the same procedure can be reused by simply reiterating the verification for each obligation in the set composing the regulatory framework, which would increase the number of steps for the procedure to achieve an answer by up to a factor of the cardinality of the set of obligations being checked. However, the verification would be still in polynomial time with respect to the size of the process model and the cardinality of the set of obligations.

Therefore, we have shown that evaluating whether a structured business process is fully compliant with a set of conditional obligations, whose elements are restricted to be propositional literals, belongs to the computational complexity class \textbf{P}.

\section*{Acknowledgments}

\small{This research is supported by the Science and Industry Endowment Fund.}

\bibliographystyle{plain}
\bibliography{bib}

\newpage
\appendix
\section{Failure $\Delta$-constraints}

\subsection{Achievement Obligations to Failure $\Delta$-constraints}

\begin{lemma}\label{l:trigger}
Given a conditional obligation, executing a task having its trigger annotated always results in its following state having the obligation in force, where it can be potentially fulfilled or violated.
\end{lemma}

\begin{proof}
Either the execution changes the previous state where $t$ was not holding to one where it is, or $t$ had been holding already.

In the first case, the execution brings a new in force period for the obligation.

In the second case, either the obligation is in force and not fulfilled in the previous state, or it becomes fulfilled in the previous state. In both cases, the execution of the task brings the obligation in force again and requires to be fulfilled.
\end{proof}

\setcounter{theorem}{0}
\begin{theorem}
Given an execution $\exe$, represented as a sequence of tasks, if $\exe$ satisfies either of the Achievement Failure $\Delta$-Constraints (Definition \ref{def:af_delta}), then $\exe$ satisfies the obligation related to the Achievement Failure $\Delta$-Constraints.
\end{theorem}

\begin{proof}
\textbf{Soundness}:

\textbf{First case}: $\overline{A}\Delta1$
\begin{enumerate}
\item From the hypothesis, and Definition \ref{def:af_delta}, it follows that $\exe$ satisfies the following conditions:
\begin{enumerate}
\item$\exists t_t | \exists t_{\neg r}, t_d | t_{\neg r} \preceq t_t \preceq t_d$
\item$\exists t_t | \neg \exists t_r | t_{\neg r} \preceq t_r \preceq t_d$
\item$\exists t_t |\exists t_{\neg d}, \neg \exists t_d | t_{\neg d} \preceq t_d \preceq t_t$
\end{enumerate}
\item From 1.(a), and Lemma \ref{l:trigger}, it follows that: $t$ holds and $r$ does not, in the state holding after the execution of the task $t_t$.
\item From 2., and 1.(b), it follows that: there is no state included between the state after executing the task $t_t$ and one where $d$ starts holding, where $r$ holds.
\item From 3. and Definition \ref{def:laof_fail}, it follows that $\exe$ would violate the obligation related to the Achievement Failure $\Delta$-constraints.
\end{enumerate}

\textbf{Second case}: $\overline{A}\Delta2$
\begin{enumerate}
\item From the hypothesis, and Definition \ref{def:af_delta}, it follows that $\exe$ satisfies the following conditions:
\begin{enumerate}
\item $\exists t_t | \exists t_{\neg r}, \neg \exists t_r  | t_{\neg r} \preceq t_r \preceq t_t$
\item$\exists t_t | \exists t_d, \neg \exists t_{\neg d} | t_d \preceq t_{\neg d} \preceq t_t$
\end{enumerate}
\item From 1.(a), and Lemma \ref{l:trigger}, it follows that $t$ holds and $r$ does not, in the state holding after the execution of the task $t_t$.
\item From 1.(b), it follows that $t$ holds and $d$ holds, in the state holding after the execution of the task $t_t$.
\item From 2. and 3., it follows that after executing the task $t_t$, $t$ and $d$ hold and $r$ does not hold.
\item From 4. and Definition \ref{def:laof_fail}, it follows that $\exe$ would violate the obligation related to the Achievement Failure $\Delta$-constraints.
\end{enumerate}

\textbf{Completeness}:
\begin{enumerate}
\item Given a trace violating a given achievement obligation.
\item From 1. and Definition \ref{def:laof_fail} it follows that $\exists \sigma_t,\sigma_d | \kwd{contain}(t, \sigma_t)$ and $\sigma_t \preceq \sigma_d$ and $\neg \exists \sigma_r | \sigma_t \preceq \sigma_r \preceq \sigma_d$.
\item From 2., it follows that a task with $t$ annotated is executed.
\item From 2., it follows that a task with $d$ annotated is executed.
\item From 2., it follows that $\neg r$ holds and no task with $r$ annotated is executed between the one with $t$ and the one with $d$.
\item Following from 3., 4., and 5. two cases are possible:
\begin{description}
\item[$t_d \preceq t_t$] this case is covered by the $\overline{A}\Delta2$ conditions in Definition \ref{def:af_delta}.
\item[$t_t \preceq t_d$] this case covered by the $\overline{A}\Delta1$ conditions in Definition \ref{def:af_delta}.
\end{description}
\item Thus, all cases are covered and a violating trace is always identified by the Achievement Failure $\Delta$-constraints.
\end{enumerate}
\end{proof}

\subsubsection{$\overline{A}\Delta1$ Simplified}

\setcounter{theorem}{3}
\begin{theorem}
If a violation in business process model is identified by satisfying either $\overline{A}\Delta1$ or $\overline{A}\Delta2$, then the same violation is identified by satisfying either $\overline{A}\Delta1S$ or $\overline{A}\Delta2$.
\end{theorem}

\begin{proof}

As both procedures identifying whether a regulation violates a business process model contain the constraint $\overline{A}\Delta2$, we focus on proving soundness and completeness when the violation is identified by fulfilling the other constraint, respectively $\overline{A}\Delta1$ and $\overline{A}\Delta1S$.

\textbf{Soundness}:

\begin{enumerate}
\item Given a violated regulation and satisfying $\overline{A}\Delta1$.
\item From 1. and Definition \ref{def:af_delta} for $\overline{A}\Delta1$, if follows that:
\begin{enumerate}
\item $\exists t_{\neg r}, t_d | t_{\neg r} \preceq t_t \preceq t_d$
\item $\neg \exists t_r | t_{\neg r} \preceq t_r \preceq t_d$
\end{enumerate}
\item From 2.(a), 2.(b), and Definition \ref{def:ad1s} if follows that the violation is identified by $\overline{A}\Delta1S$.
\end{enumerate}

\textbf{Completeness}:

\begin{enumerate}
\item Given a violated regulation and satisfying $\overline{A}\Delta1S$.
\item From 1. and Definition \ref{def:ad1s} for $\overline{A}\Delta1S$, if follows that:
\begin{enumerate}
\item $\exists t_t$ such that: $\exists t_{\neg r}, t_d | t_{\neg r} \preceq t_t \preceq t_d$
\item $\exists t_t$ such that: $\neg \exists t_r | t_{\neg r} \preceq t_r \preceq t_d$
\end{enumerate}
\item We consider the two following mutually exclusive cases:
\begin{description}
\item[case 1]: $\exists t_t$ such that: $\exists t_{\neg d}, \neg \exists t_d | t_{\neg d} \preceq t_d \preceq t_t$
\begin{enumerate}
\item[3.1.1] From 2.(a), 2.(b), \textbf{case 1}, and Definition \ref{def:af_delta} for $\overline{A}\Delta1$ if follows that the violation is identified by $\overline{A}\Delta1$.
\end{enumerate}
\item[case 2]:  $\exists t_t$ such that: $\exists t_d, \neg \exists t_{\neg d} | t_d \preceq t_{\neg d} \preceq t_t$
\begin{enumerate}
\item[3.2.1] From 2.(a), and 2.(b), it follows that $\exists t_t$ such that: $\exists t_{\neg r}, \neg \exists t_r  | t_{\neg r} \preceq t_r \preceq t_t$.
\item[3.2.2] From 3.2.1, \textbf{case 2}, and Definition \ref{def:af_delta} for $\overline{A}\Delta2$ if follows that the violation is identified by $\overline{A}\Delta2$.
\end{enumerate}
\end{description}
\item Therefore in each case, it follows that a violation identified by $\overline{A}\Delta1S$ would be identified by either $\overline{A}\Delta1$ or $\overline{A}\Delta2$.
\end{enumerate}
\end{proof}

\subsubsection{$\overline{A}\Delta2$ Decomposed}

\setcounter{theorem}{4}
\begin{theorem}
If a violation in business process model is identified by satisfying $\overline{A}\Delta2$, then the same violation is identified by satisfying either $\overline{A}\Delta2.1$ or $\overline{A}\Delta2.2$.
\end{theorem}

\begin{proof}

\textbf{Soundness}:

\begin{enumerate}
\item Given a violated regulation and satisfying $\overline{A}\Delta2$.
\item From 1. and Definition \ref{def:af_delta} for $\overline{A}\Delta2$, if follows that:
$\exists t_t$ such that:
\begin{description}
\item $\exists t_{\neg r}, \neg \exists t_{r}  | t_{\neg r} \preceq t_{r} \preceq t_t$ and
\item$\exists t_d, \neg \exists t_{\neg d} | t_d \preceq t_{\neg d} \preceq t_t$
\end{description}
\item We distinguish 3 cases:
\begin{description}
\item[$t_{\neg r} \preceq t_d$]
\begin{enumerate}
\item From the constraint introduced by the case, from 2., and from Definition \ref{def:ad2d} it follows that $\overline{A}\Delta2.1$ is satisfied.
\end{enumerate}
\item[$t_{d} \preceq t_{\neg r}$] 
\begin{enumerate}
\item From the constraint introduced by the case, from 2., and from Definition \ref{def:ad2d} it follows that $\overline{A}\Delta2.2$ is satisfied.
\end{enumerate}
\item[$t_{\neg r}$ and $t_d$ are the same task]
\begin{enumerate}
\item From the constraint introduced by the case, from 2., and from Definition \ref{def:ad2d} it follows that both $\overline{A}\Delta2.1$ and $\overline{A}\Delta2.2$ are satisfied.
\end{enumerate}
\end{description}
\item From the cases analysed in 3., it follows that the violation is captured by either $\overline{A}\Delta2.1$ or $\overline{A}\Delta2.2$.
\end{enumerate}

\textbf{Completeness}:
We analyse the two cases separately.

\begin{enumerate}
\item Given a violated regulation and satisfying $\overline{A}\Delta2.1$.
\item From 1. and Definition \ref{def:ad2d} for $\overline{A}\Delta2.1$, if follows that:
$\exists t_t$ such that:
\begin{description}
\item$\exists t_{\neg r}, t_d | t_{\neg r} \preceq t_d \preceq t_t$ and
\item $\neg \exists t_{\neg d} | t_{d} \preceq t_{\neg d} \preceq t_t$ and
\item $\neg \exists t_r | t_{\neg r} \preceq t_r \preceq t_d$
\end{description}
\item From 2., and Definition \ref{def:af_delta} for $\overline{A}\Delta2$, it follows that the violation is captured by $\overline{A}\Delta2$.
\end{enumerate}

\begin{enumerate}
\item Given a violated regulation and satisfying $\overline{A}\Delta2.2$.
\item From 1. and Definition \ref{def:ad2d} for $\overline{A}\Delta2.2$, if follows that:
$\exists t_t$ such that:
\begin{description}
\item$\exists t_{\neg r}, t_d | t_d \preceq t_{\neg r} \preceq t_t$ and
\item $\neg \exists t_{\neg d} | t_{d} \preceq t_{\neg d} \preceq t_{\neg r}$ and
\item $\neg \exists t_r | t_{\neg r} \preceq t_r \preceq t_t$
\end{description}
\item From 2., and Definition \ref{def:af_delta} for $\overline{A}\Delta2$, it follows that the violation is captured by $\overline{A}\Delta2$.
\end{enumerate}

We have shown that in both cases the violations identified by the decomposed by the decomposed $\Delta$-constraints are violations identified by $\overline{A}\Delta2$.
\end{proof}

\subsection{Maintenance Obligations to Failure $\Delta$-constraints}

\setcounter{theorem}{1}
\begin{theorem}
Given an execution $\exe$, represented as a sequence of tasks, if $\exe$ satisfies either of the Maintenance Failure $\Delta$-constraints (Definition \ref{def:mf_delta}), then $\exe$ violates the obligation related to the Maintenance Failure $\Delta$-constraints.
\end{theorem}

\begin{proof}
\textbf{Soundness}:

\textbf{First case}: $\overline{M}\Delta1$
\begin{enumerate}
\item From the hypothesis, and Definition \ref{def:mf_delta}, it follows that $\exe$ satisfies the following condition:
\begin{enumerate}
\item $\exists t_t | \exists t_{\neg r}, \neg \exists t_r | t_{\neg r} \preceq t_r \preceq t_t$
\end{enumerate}
\item From 1.(a), and Lemma \ref{l:trigger}, it follows that: $t$ holds and $r$ does not, after the execution of the task $t_t$ annotated.
\item From 2. and Definition \ref{def:lmof_fail}, it follows that $\exe$ would violate the obligation related to the Maintenance Failure $\Delta$-constraints.
\end{enumerate}

\textbf{Second case}: $\overline{M}\Delta2$
\begin{enumerate}
\item From the hypothesis, and Definition \ref{def:mf_delta}, it follows that $\exe$ satisfies the following condition:
\begin{enumerate}
\item $\exists t_t | \exists t_r, \neg \exists t_{\neg r} | t_r \preceq t_{\neg r} \preceq t_t$ and
\item $\exists t_t | \forall t_d (\exists t_{\neg r} | t_t \preceq t_{\neg r} \prec t_d)$
\end{enumerate}
\item From 1.(a), and Lemma \ref{l:trigger}, it follows that $t$ holds and $r$ holds, in the state holding after executing $t_t$.
\item From 1.(b), and Lemma \ref{l:trigger}, it follows that after executing the $t_t$, $t$ it is always the case that in the following states $r$ stops holding before $d$ starts holding.
\item From 2. it follows that after the execution of the task $t_t$, $r$ holds due to the constraint preventing a task having $\neg r$ annotated to be executed and cancelling $r$ from the process state, which is already holding due to the execution of a task with $r$ annotated.
\item From 3. it follows before the execution of any task having $d$ in its annotation, a task having $\neg r$ annotated is executed, leading to the removal of $r$ from the process state.
\item From 4. and 5. and Definition \ref{def:lmof_fail}, it follows that $\exe$ would violate the obligation related to the Maintenance Failure $\Delta$-constraints.
\end{enumerate}

\textbf{Completeness}:
\begin{enumerate}
\item Given a trace violating a given maintenance obligation.
\item From 1. and Definition \ref{def:lmof_fail} it follows that $\exists \sigma_t$ $\forall \sigma_d \; |$ $\kwd{contain}(t, \sigma)$ and $\sigma_t \preceq \sigma_d$ and $\exists \sigma_{\neg r} | \sigma_t \preceq \sigma_{\neg r} \preceq \sigma_d$.
\item From 2., it follows that a task with $t$ annotated is executed.
\item From 2., it follows that a task with $\neg r$ annotated is executed.
\item Following from 3., and 4. two cases are possible:
\begin{description}
\item[$t_{\neg r} \preceq t_t$] case covered by the $\overline{M}\Delta1$ conditions in Definition \ref{def:mf_delta}.
\item[$t_t \preceq t_{\neg r}$] case covered by the $\overline{M}\Delta2$ conditions in Definition \ref{def:mf_delta}.
\end{description}
\item Thus, all cases are covered and a violating trace is always identified by the Maintenance Failure $\Delta$-constraints.
\end{enumerate}
\end{proof}

\subsubsection{$\overline{M}\Delta2$ Simplified}

\setcounter{theorem}{6}
\begin{theorem}
If a violation in business process model is identified by satisfying either $\overline{M}\Delta1$ or $\overline{M}\Delta2$, then the same violation is identified by satisfying either $\overline{M}\Delta1$ or $\overline{M}\Delta2S$.
\end{theorem}

\begin{proof}
As both procedures identifying whether a regulation violates a business process model contain the constraint $\overline{M}\Delta1$, we focus on proving soundness and completeness when the violation is identified by fulfilling the other constraint, respectively $\overline{M}\Delta2$ and $\overline{M}\Delta2S$.

\textbf{Soundness}:

\begin{enumerate}
\item Given a violated obligation and satisfying $\overline{M}\Delta2$.
\item From 1. and Definition \ref{def:mf_delta} for $\overline{M}\Delta2$, if follows that $\exists t_t$ such that: $\forall t_d (\exists t_{\neg r} | t_t \preceq t_{\neg r} \prec t_d)$
\item Logically, $\exists t_t$ such that: $\forall t_d (\exists t_{\neg r} | t_t \preceq t_{\neg r} \prec t_d)$ $\equiv$ $\exists t_t$ such that: $\exists t_{\neg r}, \neg \exists t_{d} | t_t \preceq t_{d} \preceq t_{\neg r}$
\item From 3., and Definition \ref{def:mD2S}, it follows that the violation is identified by $\overline{M}\Delta2S$.
\end{enumerate}

\textbf{Completeness}:

\begin{enumerate}
\item Given a violated obligation and satisfying $\overline{M}\Delta2S$.
\item From 1., and Definition \ref{def:mD2S}, it follows that $\exists t_t$ such that: $\exists t_{\neg r}, \neg \exists t_{d} | t_t \preceq t_{d} \preceq t_{\neg r}$
\item We consider the two following mutually exclusive cases:
\begin{description}
\item[case 1]: $\exists t_t$ such that: $\exists t_r, \neg \exists t_{\neg r} | t_r \preceq t_{\neg r} \preceq t_t$
\begin{enumerate}
\item[3.1.1] Logically, $\exists t_t$ such that: $\forall t_d (\exists t_{\neg r} | t_t \preceq t_{\neg r} \prec t_d)$ $\equiv$ $\exists t_t$ such that: $\exists t_{\neg r}, \neg \exists t_{d} | t_t \preceq t_{d} \preceq t_{\neg r}$
\item[3.1.2] From \textbf{case 1}, 3.1.1, and Definition \ref{def:mf_delta} for $\overline{M}\Delta2$, it follows that the violation is identified by $\overline{M}\Delta2$.
\end{enumerate}
\item[case 2]: $\exists t_t$ such that: $\exists t_{\neg r}, \neg \exists t_r | t_{\neg r} \preceq t_r \preceq t_t$
\begin{enumerate}
\item[3.2.1] From \textbf{case 2}, and Definition \ref{def:mf_delta} for $\overline{M}\Delta1$, it follows that the violation is identified by $\overline{M}\Delta1$.
\end{enumerate}
\end{description}
\item Therefore in each case, it follows that a violation identified by $\overline{M}\Delta2S$ would be identified by either $\overline{M}\Delta1$ or $\overline{M}\Delta2$.
\end{enumerate}
\end{proof}
\newpage
\section{Stem Evaluation}

\begin{lemma}\label{lem:sub_to_super}
Given a process model $P$ and a process block $B$ contained in $P$. If $B$ satisfies a failure $\Delta$-constraint, then $P$ satisfies the same $\Delta$-constraint.
\end{lemma}

\begin{proof}
\begin{enumerate}
\item Assume that an execution $\exe$ in a process block $B$ satisfies the constraints specified by a $\Delta$-constraint. 
\item Considering a process block $B'$ containing $B$, we analyse three different cases depending on the type of $B'$:
\begin{description}
\item[\kwd{XOR}] The super block $B'$ contains the execution $\exe$, which directly follows from Definition \ref{def:ser}. 
\item[\kwd{SEQ}] Following from Definition \ref{def:ser}, the super block $B'$ contains an execution $\exe'$ which contains $\exe$ as part of its execution. Thus, as $\exe$ is part of $\exe'$, $\exe'$ would also satisfy the same $\Delta$-constraint as $\exe$.
\item[\kwd{AND}] The fact that an execution belonging to the super block of this type, satisfies the same property satisfied by an execution of one of its sub blocks, follows from Lemma \ref{lem:inclusivity} and the case \kwd{SEQ}.
\end{description}
\item Independently on the case shown in 2., the $\Delta$-constraint satisfied in an execution of $B$ is also satisfied by an execution of $B'$.
\item Each case shown in 2. allows us to recursively apply them to each super blocks, which would lead to having the process model block, $P$, still satisfying the same $\Delta$-constraint as one of its sub blocks $B$.
\end{enumerate}
\end{proof}

\begin{lemma}[Stem Pruning]\label{l:stem_prune}
Given a tree representation $\mathcal{P}$ of a business process $P$, and a task $\tau$ in $P$. Let $\mathcal{P}'$ be $\mathcal{P}$ where each \kwd{XOR} block on the path from the root to $\tau$ is substituted with its child leading to $\tau$. The business process model $P'$, corresponding to $\mathcal{P}'$, contains $\tau$ in each of its executions, and there are no executions in $P$ that contain $\tau$ that are not appear in $P'$.
\end{lemma}

\begin{proof}
We first prove that each execution in $P'$ contains $\tau$. We prove it by contradiction.
\begin{enumerate}
\item Assume there exists an execution in $P'$ that does not contain $\tau$.
\item From construction, $P'$ contains $\tau$.
\item From 2. and Definition \ref{def:ser}, it follows that exists an execution of $P'$ containing $\tau$.
\item From 1. and 3., it follows that $P'$ contains at least an execution containing $\tau$, and contains at least an execution not containing $\tau$.
\item From 4., and Definition \ref{def:ser}, it follows that $P'$ contains a \kwd{XOR} block, $X$, having a sub-block containing $\tau$, and a different sub-block not containing $\tau$.
\item From 5., and Definition \ref{def:process_tree}, it follows that $X$ is on the path from the root to $\tau$ in $\mathcal{P}'$.
\item However, 6. contradicts the premises, as $X$ would have been removed by construction.
\item Therefore, there cannot be an execution in $P'$ not containing $\tau$.
\end{enumerate}

We prove now that every execution containing $\tau$ in $P$, is also an execution of $P'$. We prove it by contradiction.

\begin{enumerate}
\item Assume there is an execution in $P$ containing $\tau$, and such execution does not belong to the executions of $P'$.
\item We consider two possible cases:
\begin{description}
\item[There is at least an \kwd{XOR} on the path from the root to $\tau$ in $\mathcal{P}$]:
\begin{enumerate}
\item By construction, the difference between $\mathcal{P}$ and $\mathcal{P'}$ is only on the missing children of \kwd{XOR} nodes on the path from the root to $\tau$.
\item From the assumption 1., (a), and Definition \ref{def:ser}, it follows that a children of one of the \kwd{XOR} on the path from the root to $\tau$, has been removed and it leads to $\tau$.
\item In this case, (b) is in contradiction with both the construction, and both the uniqueness of a task $\tau$.
\end{enumerate}
\item[There is not an \kwd{XOR} on the path from the root to $\tau$ in $\mathcal{P}$]:
\begin{enumerate}
\item By construction, then $\mathcal{P} \equiv \mathcal{P}'$.
\item From (a), and Definition \ref{def:process_tree}, it follows that $P \equiv P'$.
\item From (b), and Definition \ref{def:ser}, it follows that $\Exe{P} \equiv \Exe{P'}$.
\item In this case, (c) is in contradiction with the assumption in 1.
\end{enumerate}
\end{description}
\item We have shown in both cases a contradiction, hence every execution containing $\tau$ in $P$, is also an execution of $P'$.
\end{enumerate}
\end{proof}

\setcounter{theorem}{2}
\begin{theorem}[Algorithm \ref{a:stem} Correctness]
Given a business process model $P$ and a failure $\Delta$-constrain $\delta$, if there exists an execution of $P$ such that the execution satisfies $\delta$, then Algorithm \ref{a:stem} returns \kwd{true}, otherwise it returns \kwd{false}.
\end{theorem}

\begin{proof}
\textbf{Soundness}
\begin{enumerate}
\item From line 1 and line 3, each task containing the obligation's trigger is considered.
\item From 1. lines 4 and 5, and Lemma \ref{l:stem_prune}, it follows that the pruning procedure does not remove executions containing the task considered as the \emph{trigger leaf}.
\item Let the \kwd{Evaluate} function be \kwd{true} when the tree representation contains an execution satisfying $\delta$, and \kwd{false} otherwise.\footnote{Let \kwd{Evaluate} be a function calculating the bottom up aggregation of a node in a tree according to the classification of its children as shown by the aggregations shown in the tables in Appendix \ref{app:agg}.}
\item From 2., and line 6, the function \kwd{Evaluate} is given a tree representation of a process containing each execution having the \emph{trigger leaf}.
\item From 1., and 4., it follows that the \kwd{Evaluate} function considers every executions in $P$ having a task with the obligation's trigger annotated.
\item From 3., and 5., it follows that if $P$ contains an execution fulfilling $\delta$, then \kwd{Evaluate} is evaluated \kwd{true}.
\item From 6., and lines 6 and 7, it follows that Algorithm \ref{a:stem} returns \kwd{true} when $P$ contains an execution fulfilling $\delta$.
\item From 6., and line 8, it follows that if $P$ does not contain an execution fulfilling $\delta$, then Algorithm \ref{a:stem} returns \kwd{false}.
\item From 7., and 8., it follows that Algorithm \ref{a:stem} returns the correct value depending whether $P$ contains an execution satisfying $\delta$.
\end{enumerate}

\textbf{Completeness}
\begin{enumerate}
\item Depending on the output of Algorithm \ref{a:stem}, we distinguish two cases:
\begin{description}
\item[true]
\begin{enumerate}
\item Line 7 of Algorithm \ref{a:stem} is executed.
\item From (a), it follows that \kwd{Evaluate} at line 6 is \kwd{true} at least once.
\item From (b), and lines 1 and 2, it follows that $\mathcal{P}$ contains an execution satisfying $\delta$.
\item From (c), and Definition \ref{def:process_tree}, it follows that $P$ contains an execution satisfying $\delta$.
\end{enumerate}
\item[false]
\begin{enumerate}
\item Line 8 of Algorithm \ref{a:stem} is executed.
\item From (a), it follows that \kwd{Evaluate} at line 6 is always \kwd{false}.
\item From line 1 and line 3, each task containing the obligation's trigger is considered.
\item From (c), lines 4 and 5, and Lemma \ref{l:stem_prune}, it follows that the pruning procedure does not remove executions containing the task considered as the \emph{trigger leaf}.
\item From (d), and line 6, the function \kwd{Evaluate} is given a tree representation of a process containing each execution having the \emph{trigger leaf}.
\item From (c), and (e), it follows that the \kwd{Evaluate} function considers every executions in $P$ having a task with the obligation's trigger annotated.
\item From (b), and (f), it follows that non of the executions of $P$ satisfy $\delta$.
\end{enumerate}
\end{description}
\item We have proven that in both cases, Algorithm \ref{a:stem} correctly recognises whether a business process $P$ contains an execution satisfying $\delta$.
\end{enumerate}
\end{proof}
\newpage
\section{Aggregating the Evaluations}\label{app:agg}
\subsection{Left Sub-Pattern}

\begin{definition}[Left Sub-Pattern $\mathbf{\overset{+}{x}}$]\label{def:lsp_x+}
Left side pattern fulfilling state: there exists an execution belonging to the process block containing $x$, and no $y$ on its right.

\noindent\textbf{Formally}: 

\noindent Given a process block $B$, it belongs to this class if and only if:
\begin{itemize}
\item $\exists \exe \in \Exe{B}$ such that:
\begin{itemize}
\item $\exists x \in \exe$ such that:
\begin{itemize}
\item $\not \exists y \in \exe$ such that $x \preceq y$.
\end{itemize}
\end{itemize}
\end{itemize}
\end{definition}

\begin{definition}[Left Sub-Pattern $\mathbf{\overset{0}{x}}$]\label{def:lsp_x0}
neutral state: there exists an execution belonging to the process block, such that it does not contain $x$ or $y$.

\noindent\textbf{Formally}: 

\noindent Given a process block $B$, it belongs to this class if and only if:
\begin{itemize}
\item $\exists \exe \in \Exe{B}$ such that:
\begin{itemize}
\item $\not\exists x, \in \exe$, and
\item $\not\exists y, \in \exe$.
\end{itemize}
\end{itemize}
\end{definition}

\begin{definition}[Left Sub-Pattern $\mathbf{\overset{-}{x}}$]\label{def:lsp_x-}
blocking state: every execution belonging to the process block contains at least a $y$, and no $x$ on its right.

\noindent\textbf{Formally}: 

\noindent Given a process block $B$, it belongs to this class if and only if:
\begin{itemize}
\item $\forall \exe \in \Exe{B}$ such that:
\begin{itemize}
\item $\exists y \in \exe$ such that:
\begin{itemize}
\item $\not \exists x \in \exe$ such that $y \preceq x$.
\end{itemize}
\end{itemize}
\end{itemize}
\end{definition}

\setcounter{theorem}{7}
\begin{theorem}[Classification Completeness for Left Sub-Pattern]
The set of possible evaluations of Left Sub-Pattern is completely covered by the provided set of classifications.
\end{theorem}

\begin{proof}
\begin{enumerate}
\item The pattern ($x$) in Left Sub-Pattern allows 3 possible evaluations: whether the partial requirement is satisfied, failing, or in a neutral state.
\item The 3 possible evaluations are completely covered by the 3 possible classifications.
\end{enumerate}
\end{proof}

\subsubsection{Lemmas}

\begin{lemma}[Aggregation Neutral Class]\label{l:neutral}
Given a process block $A$, assigned to the evaluation class $\mathbf{\overset{0}{x}}$, and another process block $B$, assigned to any of the available evaluation classes. Let $C$ be a process block having $A$ and $B$ as its sub-blocks, then the evaluation class of $C$ is the same class as the process block $B$.
\end{lemma}

\begin{proof}
We prove Lemma \ref{l:neutral} by contradiction.

\begin{enumerate}
\item Let $A$ be a process block classified as $\mathbf{\overset{0}{x}}$.
\item Let $B$ be a process block classified as one of the possible classifications.
\item Let $C$ be a process block having $A$ and $B$ as its sub-blocks.
\item Assume that $C$ is classified as a different class than the one that classifies $B$.
\item From 4., and Definition \ref{def:lsp_x0}, it follows that the executions in $C$ do not have the same properties as the executions in $B$.
\item From 5., it follows that $A$ contains at least a serialisation capable of changing the properties of the serialisations in $B$ when considered in the super-block $C$.
\item From 6., it follows that $A$ must contain some tasks with some relevant literals in their annotation, in other words $x$ or $y$ must appear in at least one annotation of at least one of the tasks.
\item From 7. and the preference order: $\mathbf{\overset{-}{x}} < \mathbf{\overset{0}{x}} < \mathbf{\overset{+}{x}}$, it follows that either exists an execution in $A$ such that $A$ belongs to one of the most preferred classes than $0$, or each execution of $A$ contains a $x$ annotated in one of the appearing tasks, and it was not possible to assign it to any of the more preferred classes, hence it is assigned to $\mathbf{\overset{-}{x}}$.
\item 1. and 8. are in contradiction, as the former assigns $A$ to $\mathbf{\overset{0}{x}}$, while the latter to any other class but not $\mathbf{\overset{0}{x}}$.
\item From 9., it follows that a process classified as $\mathbf{\overset{0}{x}}$ is a neutral element in any aggregation.
\end{enumerate}

\end{proof}

\begin{lemma}[Aggregation Compositionality]\label{l:compositionality}
Given a process block $B$ having a list of sub-blocks: $A_0, \dots, A_n$. The classification of $B$ is the same as ordered pair aggregation of the classifications of its sub-blocks: \kwd{agg}(\dots( \kwd{agg}(\kwd{agg}($A_0, A_1$), $A_2$) \dots), $A_n$).
\end{lemma}

\begin{proof}
We prove Lemma \ref{l:compositionality} by cases:

For \kwd{SEQ} the order of the aggregation follows the execution order so is natural that compositionally aggregating does not affect the result.

For \kwd{AND} the order is irrelevant. The aggregation tables take care of keeping track of the best possible outcomes of the aggregations, allowing again compositionality.
\end{proof}

\begin{lemma}[Executions Inclusivity]\label{lem:inclusivity}
Given two process blocks $A$ and $B$, and two super-blocks $\kwd{SEQ}(A, B)$ and $\kwd{AND}(A, B)$, then $\Exe{\kwd{SEQ}(A, B)} \subseteq \Exe{\kwd{AND}(A, B)}$.
\end{lemma}

\begin{proof}
Lemma \ref{lem:inclusivity} directly follows from Definition \ref{def:ser}.
\end{proof}

\subsubsection{Aggregations}~\\
\noindent
Finally, for the Left Sub-Pattern, the aggregation of the classifications of the children while evaluating an undernode, is described in Table \ref{tab:atsol}. Notice that the aggregation is performed pair-wise left to right, with the result being aggregated with next brother on the right. The result depends on the classifications of the two children, or the result of a partial aggregation and a child, and the type of the parent undernode.

\begin{table}[ht!]
\centering
\begin{tabular}{|c|c|c|c|c|}
\hline
A & B & \kwd{SEQ}(A, B) & \kwd{AND}(A, B) \\ \hline
$\mathbf{\overset{0}{x}}$ & $\mathbf{\overset{0}{x}}$ & $\mathbf{\overset{0}{x}}$ & $\mathbf{\overset{0}{x}}$ \\ \hline
$\mathbf{\overset{0}{x}}$ & $\mathbf{\overset{+}{x}}$ & $\mathbf{\overset{+}{x}}$ & $\mathbf{\overset{+}{x}}$ \\ \hline
$\mathbf{\overset{0}{x}}$ & $\mathbf{\overset{-}{x}}$ & $\mathbf{\overset{-}{x}}$ & $\mathbf{\overset{-}{x}}$ \\ \hline
$\mathbf{\overset{+}{x}}$ & $\mathbf{\overset{0}{x}}$ & $\mathbf{\overset{+}{x}}$ & $\mathbf{\overset{+}{x}}$ \\ \hline
$\mathbf{\overset{+}{x}}$ & $\mathbf{\overset{+}{x}}$ & $\mathbf{\overset{+}{x}}$ & $\mathbf{\overset{+}{x}}$ \\ \hline
$\mathbf{\overset{+}{x}}$ & $\mathbf{\overset{-}{x}}$ & $\mathbf{\overset{-}{x}}$ & $\mathbf{\overset{+}{x}}$ \\ \hline
$\mathbf{\overset{-}{x}}$ & $\mathbf{\overset{0}{x}}$ & $\mathbf{\overset{-}{x}}$ & $\mathbf{\overset{-}{x}}$ \\ \hline
$\mathbf{\overset{-}{x}}$ & $\mathbf{\overset{+}{x}}$ & $\mathbf{\overset{+}{x}}$ & $\mathbf{\overset{+}{x}}$ \\ \hline
$\mathbf{\overset{-}{x}}$ & $\mathbf{\overset{-}{x}}$ & $\mathbf{\overset{-}{x}}$ & $\mathbf{\overset{-}{x}}$ \\ \hline
\end{tabular}
\caption{Aggregation Table Left Sub-Pattern}\label{tab:atsol}
\end{table}

\begin{proof}[Table \ref{tab:atsol}]

Considering the row in the table where $A$ is $\mathbf{\overset{0}{x}}$ and $B$ is $\mathbf{\overset{0}{x}}$
\begin{description}
\item[SEQ]
\begin{enumerate}
\item \kwd{SEQ}(A, B) is correctly classified $\mathbf{\overset{0}{x}}$ according to Lemma \ref{l:neutral}.
\end{enumerate}
\item[AND]
\begin{enumerate}
\item \kwd{AND}(A, B) is correctly classified $\mathbf{\overset{0}{x}}$ according to Lemma \ref{l:neutral}.
\end{enumerate}
\end{description}

Considering the row in the table where $A$ is $\mathbf{\overset{0}{x}}$ and $B$ is $\mathbf{\overset{+}{x}}$
\begin{description}
\item[SEQ]
\begin{enumerate}
\item \kwd{SEQ}(A, B) is correctly classified $\mathbf{\overset{+}{x}}$ according to Lemma \ref{l:neutral}.
\end{enumerate}
\item[AND]
\begin{enumerate}
\item \kwd{AND}(A, B) is correctly classified $\mathbf{\overset{+}{x}}$ according to Lemma \ref{l:neutral}.
\end{enumerate}
\end{description}

Considering the row in the table where $A$ is $\mathbf{\overset{0}{x}}$ and $B$ is $\mathbf{\overset{-}{x}}$
\begin{description}
\item[SEQ]
\begin{enumerate}
\item \kwd{SEQ}(A, B) is correctly classified $\mathbf{\overset{-}{x}}$ according to Lemma \ref{l:neutral}.
\end{enumerate}
\item[AND]
\begin{enumerate}
\item \kwd{AND}(A, B) is correctly classified $\mathbf{\overset{-}{x}}$ according to Lemma \ref{l:neutral}.
\end{enumerate}
\end{description}

Considering the row in the table where $A$ is $\mathbf{\overset{+}{x}}$ and $B$ is $\mathbf{\overset{0}{x}}$
\begin{description}
\item[SEQ]
\begin{enumerate}
\item \kwd{SEQ}(A, B) is correctly classified $\mathbf{\overset{+}{x}}$ according to Lemma \ref{l:neutral}.
\end{enumerate}
\item[AND]
\begin{enumerate}
\item \kwd{AND}(A, B) is correctly classified $\mathbf{\overset{+}{x}}$ according to Lemma \ref{l:neutral}.
\end{enumerate}
\end{description}

Considering the row in the table where $A$ is $\mathbf{\overset{+}{x}}$ and $B$ is $\mathbf{\overset{+}{x}}$
\begin{description}
\item[SEQ]
\begin{enumerate}
\item From Definition \ref{def:ser}, it follows that the possible executions of a process block $\kwd{SEQ}(A, B)$ are the concatenation of an execution of $A$ and an execution of $B$.
\item From the hypothesis and 1., it follows that $\kwd{SEQ}(A, B)$ contains at least an $x$ and no $y$ on its right.
\item From 2. and Definition \ref{def:lsp_x+}, it follows that $\kwd{SEQ}(A, B)$ is correctly classified as $\mathbf{\overset{+}{x}}$.
\end{enumerate}
\item[AND]
\begin{enumerate}
\item From Lemma \ref{lem:inclusivity} and the result for $\kwd{SEQ}(A, B)$, we can conclude that $\kwd{AND}(A, B)$ is correctly classified as $\mathbf{\overset{+}{x}}$
\end{enumerate}
\end{description}

Considering the row in the table where $A$ is $\mathbf{\overset{+}{x}}$ and $B$ is $\mathbf{\overset{-}{x}}$
\begin{description}
\item[SEQ]
\begin{enumerate}
\item From Definition \ref{def:ser}, it follows that the possible executions of a process block $\kwd{SEQ}(A, B)$ are the concatenation of an execution of $A$ and an execution of $B$.
\item From the hypothesis and 1., it follows that $\kwd{SEQ}(A, B)$ contains at least an $x$ and a $y$ on its right.
\item From 2. and Definition \ref{def:lsp_x-}, it follows that $\kwd{SEQ}(A, B)$ is correctly classified as $\mathbf{\overset{-}{x}}$.
\end{enumerate}
\item[AND]
\begin{enumerate}
\item From Definition \ref{def:ser}, it follows that each execution resulting from $\kwd{AND}(A, B)$ consists of an execution of $A$ interleaved with an execution of $B$.
\item From 1. and Definition \ref{def:ser}, it follows that each execution resulting from $\kwd{SEQ}(B, A)$ is a subset of the executions of $\kwd{AND}(A, B)$.
\item From the hypothesis and 2., it follows that $\kwd{SEQ}(B, A)$ contains at least an $x$ and no $y$ on its right.
\item From 1., 3. and Definition \ref{def:lsp_x+}, it follows that $\kwd{AND}(A, B)$ is correctly classified as $\mathbf{\overset{+}{x}}$.
\end{enumerate}
\end{description}

Considering the row in the table where $A$ is $\mathbf{\overset{-}{x}}$ and $B$ is $\mathbf{\overset{0}{x}}$
\begin{description}
\item[SEQ]
\begin{enumerate}
\item \kwd{SEQ}(A, B) is correctly classified $\mathbf{\overset{-}{x}}$ according to Lemma \ref{l:neutral}.
\end{enumerate}
\item[AND]
\begin{enumerate}
\item \kwd{AND}(A, B) is correctly classified $\mathbf{\overset{-}{x}}$ according to Lemma \ref{l:neutral}.
\end{enumerate}
\end{description}

Considering the row in the table where $A$ is $\mathbf{\overset{-}{x}}$ and $B$ is $\mathbf{\overset{+}{x}}$
\begin{description}
\item[SEQ]
\begin{enumerate}
\item From Definition \ref{def:ser}, it follows that the possible executions of a process block $\kwd{SEQ}(A, B)$ are the concatenation of an execution of $A$ and an execution of $B$.
\item From the hypothesis and 1., it follows that $\kwd{SEQ}(A, B)$ contains at least an $x$ and no $y$ on its right.
\item From 2. and Definition \ref{def:lsp_x+}, it follows that $\kwd{SEQ}(A, B)$ is correctly classified as $\mathbf{\overset{+}{x}}$.
\end{enumerate}
\item[AND]
\begin{enumerate}
\item From Lemma \ref{lem:inclusivity} and the result for $\kwd{SEQ}(A, B)$, we can conclude that $\kwd{AND}(A, B)$ is correctly classified as $\mathbf{\overset{+}{x}}$
\end{enumerate}
\end{description}

Considering the row in the table where $A$ is $\mathbf{\overset{-}{x}}$ and $B$ is $\mathbf{\overset{-}{x}}$
\begin{description}
\item[SEQ]
\begin{enumerate}
\item From Definition \ref{def:ser}, it follows that the possible executions of a process block $\kwd{SEQ}(A, B)$ are the concatenation of an execution of $A$ and an execution of $B$.
\item From the hypothesis and 1., it follows that $\kwd{SEQ}(A, B)$ contains at least a $y$ and no $x$ on its right.
\item From 2. and Definition \ref{def:lsp_x-}, it follows that $\kwd{SEQ}(A, B)$ is correctly classified as $\mathbf{\overset{-}{x}}$.
\end{enumerate}
\item[AND]
\begin{enumerate}
\item From Definition \ref{def:ser}, it follows that each execution resulting from $\kwd{AND}(A, B)$ consists of an execution of $A$ interleaved with an execution of $B$.
\item From 1. and the hypothesis, it follows that independently on the execution order, a $x$ on the right of a not $y$ is not possible, while a $y$ with no $x$ on its right is guaranteed.
\item From 2. and Definition \ref{def:lsp_x-}, it follows that $\kwd{SEQ}(A, B)$ is correctly classified as $\mathbf{\overset{-}{x}}$.
\end{enumerate}
\end{description}

\end{proof}
\newpage
\subsection{Right Sub-Pattern}

\begin{definition}[Right Sub-Pattern $\mathbf{\overset{+}{z}}$]\label{def:lsp_z+}
left side pattern fulfilling state: there exists an execution belonging to the process block containing $z$, and no $y$ on its left.

\noindent\textbf{Formally}: 

\noindent Given a process block $B$, it belongs to this class if and only if:
\begin{itemize}
\item $\exists \exe \in \Exe{B}$ such that:
\begin{itemize}
\item $\exists $z$ \in \exe$ such that:
\begin{itemize}
\item $\not \exists $y$ \in \exe$ such that $z \succeq y$.
\end{itemize}
\end{itemize}
\end{itemize}
\end{definition}

\begin{definition}[Right Sub-Pattern $\mathbf{\overset{0}{z}}$]\label{def:lsp_z0}
neutral state: there exists an execution belonging to the process block, such that it does not contain $z$ or $y$.

\noindent\textbf{Formally}: 

\noindent Given a process block $B$, it belongs to this class if and only if:
\begin{itemize}
\item $\exists \exe \in \Exe{B}$ such that:
\begin{itemize}
\item $\not\exists z \in \exe$, and
\item $\not\exists y \in \exe$.
\end{itemize}
\end{itemize}
\end{definition}

\begin{definition}[Right Sub-Pattern $\mathbf{\overset{-}{z}}$]\label{def:lsp_z-}
blocking state: every execution belonging to the process block contains at least a $y$, and no $z$ on its right.

\noindent\textbf{Formally}: 

\noindent Given a process block $B$, it belongs to this class if and only if:
\begin{itemize}
\item $\forall \exe \in \Exe{B}$ such that:
\begin{itemize}
\item $\exists y \in \exe$ such that:
\begin{itemize}
\item $\not \exists z \in \exe$ such that $y \succeq z$.
\end{itemize}
\end{itemize}
\end{itemize}
\end{definition}

\begin{theorem}[Classification Completeness for Right Sub-Pattern]
The set of possible evaluations of Right Sub-Pattern is completely covered by the provided set of classifications.
\end{theorem}

\begin{proof}
\begin{enumerate}
\item The pattern ($z$) in Right Sub-Pattern allows 3 possible evaluations: whether the partial requirement is satisfied, failing, or in a neutral state.
\item The 3 possible evaluations are completely covered by the 3 possible classifications.
\end{enumerate}
\end{proof}

\subsubsection{Lemmas}

\begin{lemma}[Aggregation Neutral Class]\label{l:neutral_rsp}
Given a process block $A$, assigned to the evaluation class $\mathbf{\overset{0}{z}}$, and another process block $B$, assigned to any of the available evaluation classes. Let $C$ be a process block having $A$ and $B$ as its sub-blocks, then the evaluation class of $C$ is the same class as the process block $B$.
\end{lemma}

\begin{proof}
This proof follows closely the proof for Lemma \ref{l:neutral}
\end{proof}

\subsubsection{Aggregations}~\\
\noindent
Similarly for the left sub-pattern, we illustrate how an undernode is classified on the right side of the stem using the right sub-pattern in a table. The aggregation is shown in Table \ref{tab:atsor}.

\begin{table}[ht!]
\centering
\begin{tabular}{|c|c|c|c|c|}
\hline
A & B & \kwd{SEQ}(A, B) & \kwd{AND}(A, B) \\ \hline
$\mathbf{\overset{0}{z}}$ & $\mathbf{\overset{0}{z}}$ & $\mathbf{\overset{0}{z}}$ & $\mathbf{\overset{0}{z}}$ \\ \hline
$\mathbf{\overset{0}{z}}$ & $\mathbf{\overset{+}{z}}$ & $\mathbf{\overset{+}{z}}$ & $\mathbf{\overset{+}{z}}$ \\ \hline
$\mathbf{\overset{0}{z}}$ & $\mathbf{\overset{-}{z}}$ & $\mathbf{\overset{-}{z}}$ & $\mathbf{\overset{-}{z}}$ \\ \hline
$\mathbf{\overset{+}{z}}$ & $\mathbf{\overset{0}{z}}$ & $\mathbf{\overset{+}{z}}$ & $\mathbf{\overset{+}{z}}$ \\ \hline
$\mathbf{\overset{+}{z}}$ & $\mathbf{\overset{+}{z}}$ & $\mathbf{\overset{+}{z}}$ & $\mathbf{\overset{+}{z}}$ \\ \hline
$\mathbf{\overset{+}{z}}$ & $\mathbf{\overset{-}{z}}$ & $\mathbf{\overset{+}{z}}$ & $\mathbf{\overset{+}{z}}$ \\ \hline
$\mathbf{\overset{-}{z}}$ & $\mathbf{\overset{0}{z}}$ & $\mathbf{\overset{-}{z}}$ & $\mathbf{\overset{-}{z}}$ \\ \hline
$\mathbf{\overset{-}{z}}$ & $\mathbf{\overset{+}{z}}$ & $\mathbf{\overset{-}{z}}$ & $\mathbf{\overset{+}{z}}$ \\ \hline
$\mathbf{\overset{-}{z}}$ & $\mathbf{\overset{-}{z}}$ & $\mathbf{\overset{-}{z}}$ & $\mathbf{\overset{-}{z}}$ \\ \hline
\end{tabular}
\caption{Aggregation Table Right Sub-Pattern}\label{tab:atsor}
\end{table}

\begin{proof}[Table \ref{tab:atsor}]

Considering the row in the table where $A$ is $\mathbf{\overset{0}{z}}$ and $B$ is $\mathbf{\overset{0}{z}}$
\begin{description}
\item[SEQ]
\begin{enumerate}
\item \kwd{SEQ}(A, B) is correctly classified $\mathbf{\overset{0}{z}}$ according to Lemma \ref{l:neutral_rsp}.
\end{enumerate}
\item[AND]
\begin{enumerate}
\item \kwd{AND}(A, B) is correctly classified $\mathbf{\overset{0}{z}}$ according to Lemma \ref{l:neutral_rsp}.
\end{enumerate}
\end{description}

Considering the row in the table where $A$ is $\mathbf{\overset{0}{z}}$ and $B$ is $\mathbf{\overset{+}{z}}$
\begin{description}
\item[SEQ]
\begin{enumerate}
\item \kwd{SEQ}(A, B) is correctly classified $\mathbf{\overset{+}{z}}$ according to Lemma \ref{l:neutral_rsp}.
\end{enumerate}
\item[AND]
\begin{enumerate}
\item \kwd{AND}(A, B) is correctly classified $\mathbf{\overset{+}{z}}$ according to Lemma \ref{l:neutral_rsp}.
\end{enumerate}
\end{description}

Considering the row in the table where $A$ is $\mathbf{\overset{0}{z}}$ and $B$ is $\mathbf{\overset{-}{z}}$
\begin{description}
\item[SEQ]
\begin{enumerate}
\item \kwd{SEQ}(A, B) is correctly classified $\mathbf{\overset{-}{z}}$ according to Lemma \ref{l:neutral_rsp}.
\end{enumerate}
\item[AND]
\begin{enumerate}
\item \kwd{AND}(A, B) is correctly classified $\mathbf{\overset{-}{z}}$ according to Lemma \ref{l:neutral_rsp}.
\end{enumerate}
\end{description}

Considering the row in the table where $A$ is $\mathbf{\overset{+}{z}}$ and $B$ is $\mathbf{\overset{0}{z}}$
\begin{description}
\item[SEQ]
\begin{enumerate}
\item \kwd{SEQ}(A, B) is correctly classified $\mathbf{\overset{+}{z}}$ according to Lemma \ref{l:neutral_rsp}.
\end{enumerate}
\item[AND]
\begin{enumerate}
\item \kwd{AND}(A, B) is correctly classified $\mathbf{\overset{+}{z}}$ according to Lemma \ref{l:neutral_rsp}.
\end{enumerate}
\end{description}

Considering the row in the table where $A$ is $\mathbf{\overset{+}{z}}$ and $B$ is $\mathbf{\overset{+}{z}}$
\begin{description}
\item[SEQ]
\begin{enumerate}
\item From Definition \ref{def:ser}, it follows that the possible executions of a process block $\kwd{SEQ}(A, B)$ are the concatenation of an execution of $A$ and an execution of $B$.
\item From the hypothesis and 1., it follows that $\kwd{SEQ}(A, B)$ contains at least an $z$ and no $y$ on its left.
\item From 2. and Definition \ref{def:lsp_z+}, it follows that $\kwd{SEQ}(A, B)$ is correctly classified as $\mathbf{\overset{+}{z}}$.
\end{enumerate}
\item[AND]
\begin{enumerate}
\item From Lemma \ref{lem:inclusivity} and the result for $\kwd{SEQ}(A, B)$, we can conclude that $\kwd{AND}(A, B)$ is correctly classified as $\mathbf{\overset{+}{z}}$
\end{enumerate}
\end{description}

Considering the row in the table where $A$ is $\mathbf{\overset{+}{z}}$ and $B$ is $\mathbf{\overset{-}{z}}$
\begin{description}
\item[SEQ]
\begin{enumerate}
\item From Definition \ref{def:ser}, it follows that the possible executions of a process block $\kwd{SEQ}(A, B)$ are the concatenation of an execution of $A$ and an execution of $B$.
\item From the hypothesis and 1., it follows that $\kwd{SEQ}(A, B)$ contains at least an $z$ and no $y$ on its left.
\item From 2. and Definition \ref{def:lsp_z+}, it follows that $\kwd{SEQ}(A, B)$ is correctly classified as $\mathbf{\overset{+}{z}}$.
\end{enumerate}
\item[AND]
\begin{enumerate}
\item From Lemma \ref{lem:inclusivity} and the result for $\kwd{SEQ}(A, B)$, we can conclude that $\kwd{AND}(A, B)$ is correctly classified as $\mathbf{\overset{+}{z}}$
\end{enumerate}
\end{description}

Considering the row in the table where $A$ is $\mathbf{\overset{-}{z}}$ and $B$ is $\mathbf{\overset{0}{z}}$
\begin{description}
\item[SEQ]
\begin{enumerate}
\item \kwd{SEQ}(A, B) is correctly classified $\mathbf{\overset{-}{z}}$ according to Lemma \ref{l:neutral_rsp}.
\end{enumerate}
\item[AND]
\begin{enumerate}
\item \kwd{AND}(A, B) is correctly classified $\mathbf{\overset{-}{z}}$ according to Lemma \ref{l:neutral_rsp}.
\end{enumerate}
\end{description}

Considering the row in the table where $A$ is $\mathbf{\overset{-}{z}}$ and $B$ is $\mathbf{\overset{+}{z}}$
\begin{description}
\item[SEQ]
\begin{enumerate}
\item From Definition \ref{def:ser}, it follows that the possible executions of a process block $\kwd{SEQ}(A, B)$ are the concatenation of an execution of $A$ and an execution of $B$.
\item From the hypothesis and 1., it follows that $\kwd{SEQ}(A, B)$ contains at least an $z$ and a $y$ on its left.
\item From 2. and Definition \ref{def:lsp_z-}, it follows that $\kwd{SEQ}(A, B)$ is correctly classified as $\mathbf{\overset{-}{z}}$.
\end{enumerate}
\item[AND]
\begin{enumerate}
\item From Definition \ref{def:ser}, it follows that each execution resulting from $\kwd{AND}(A, B)$ consists of an execution of $A$ interleaved with an execution of $B$.
\item From 1. and Definition \ref{def:ser}, it follows that each execution resulting from $\kwd{SEQ}(B, A)$ is a subset of the executions of $\kwd{AND}(A, B)$.
\item From the hypothesis and 2., it follows that $\kwd{SEQ}(B, A)$ contains at least an $z$ and no $y$ on its left.
\item From 1., 3. and Definition \ref{def:lsp_z+}, it follows that $\kwd{AND}(A, B)$ is correctly classified as $\mathbf{\overset{+}{z}}$.
\end{enumerate}
\end{description}

Considering the row in the table where $A$ is $\mathbf{\overset{-}{z}}$ and $B$ is $\mathbf{\overset{-}{z}}$
\begin{description}
\item[SEQ]
\begin{enumerate}
\item From Definition \ref{def:ser}, it follows that the possible executions of a process block $\kwd{SEQ}(A, B)$ are the concatenation of an execution of $A$ and an execution of $B$.
\item From the hypothesis and 1., it follows that $\kwd{SEQ}(A, B)$ contains at least a $y$ and no $z$ on its left.
\item From 2. and Definition \ref{def:lsp_z-}, it follows that $\kwd{SEQ}(A, B)$ is correctly classified as $\mathbf{\overset{-}{z}}$.
\end{enumerate}
\item[AND]
\begin{enumerate}
\item From Definition \ref{def:ser}, it follows that each execution resulting from $\kwd{AND}(A, B)$ consists of an execution of $A$ interleaved with an execution of $B$.
\item From 1. and the hypothesis, it follows that independently on the execution order, a $z$ on the right of a not $y$ is not possible, while a $y$ with no $z$ on its right is guaranteed.
\item From 2. and Definition \ref{def:lsp_z-}, it follows that $\kwd{SEQ}(A, B)$ is correctly classified as $\mathbf{\overset{-}{z}}$.
\end{enumerate}
\end{description}

\end{proof}
\newpage
\subsection{Interval Sub-Pattern}

\begin{definition}[$\mathbf{\underline{\overset{+}{x}\overset{+}{z}}}$]\label{def:isp_3}
Closed interval class: there exists an execution belonging to the process block containing $x$ and $z$ on its right, while having no $y$ in-between.

\noindent\textbf{Formally}: 

\noindent Given a process block $B$ belongs to this class if and only if:
\begin{itemize}
\item $\exists \exe \in \Exe{B}$ such that:
\begin{itemize}
\item $\exists x, z \in \exe$ such that $x \prec z$, and
\item $\not \exists y \in \exe$ such that $x \preceq y \preceq z$.
\end{itemize}
\end{itemize}
\end{definition}

\begin{definition}[$\mathbf{\underline{\overset{+}{x}\overset{0}{z}}}$]\label{def:isp_1}
Left side pattern fulfilling state: there exists an execution belonging to the process block containing $x$, and no $y$ on its right. 

\noindent\textbf{Formally}: 

\noindent Given a process block $B$, it belongs to this class if and only if:
\begin{itemize}
\item $\exists \exe \in \Exe{B}$ such that:
\begin{itemize}
\item $\exists x \in \exe$ such that:
\begin{itemize}
\item $\not \exists z \in \exe$ such that $x \prec z$, and
\item $\not \exists y \in \exe$ such that $x \preceq y$.
\end{itemize}
\end{itemize}
\end{itemize}
\end{definition}

\begin{definition}[$\mathbf{\underline{\overset{0}{x}\overset{+}{z}}}$]\label{def:isp_1'}
Right side pattern fulfilling state: there exists an execution belonging to the process block containing $z$, and no $y$ on its left. 

\noindent\textbf{Formally}: 

\noindent Given a process block $B$, it belongs to this class if and only if:
\begin{itemize}
\item $\exists \exe \in \Exe{B}$ such that:
\begin{itemize}
\item $\exists z \in \exe$ such that:
\begin{itemize}
\item $\not \exists x \in \exe$ such that $x \prec z$, and
\item $\not \exists y \in \exe$ such that $y \preceq z$.
\end{itemize}
\end{itemize}
\end{itemize}
\end{definition}

\begin{definition}[$\mathbf{\underline{\overset{+}{x}\overset{-}{z}}}$]\label{def:isp_1r}
Left side pattern right blocked state: there exists an execution belonging to the process block containing $x$, and $y$ on its right. 

\noindent\textbf{Formally}: 

\noindent Given a process block $B$, it belongs to this class if and only if:
\begin{itemize}
\item $\exists \exe \in \Exe{B}$ such that:
\begin{itemize}
\item $\exists x \in \exe$ such that:
\begin{itemize}
\item $\not \exists z \in \exe$ such that $x \prec z$, and
\item $\exists y \in \exe$ such that $x \preceq y$.
\end{itemize}
\end{itemize}
\end{itemize}
\end{definition}

\begin{definition}[$\mathbf{\underline{\overset{0}{x}\overset{0}{z}}}$]\label{def:isp_0}
Neutral state: there exists an execution belonging to the process block, such that it does not contain $x$, $y$, or $z$.

\noindent\textbf{Formally}: 

\noindent Given a process block $B$, it belongs to this class if and only if:
\begin{itemize}
\item $\exists \exe \in \Exe{B}$ such that:
\begin{itemize}
\item $\not\exists y, x, z \in \exe$.
\end{itemize}
\end{itemize}
\end{definition}

\begin{definition}[$\mathbf{\underline{\overset{-}{x}\overset{+}{z}}}$]\label{def:isp_1'l}
Right side pattern left blocked state: there exists an execution belonging to the process block containing $z$, and $y$ on its left. 

\noindent\textbf{Formally}: 

\noindent Given a process block $B$, it belongs to this class if and only if:
\begin{itemize}
\item $\exists \exe \in \Exe{B}$ such that:
\begin{itemize}
\item $\exists z \in \exe$ such that:
\begin{itemize}
\item $\not \exists x \in \exe$ such that $x \prec z$, and
\item $\exists y \in \exe$ such that $y \preceq z$.
\end{itemize}
\end{itemize}
\end{itemize}
\end{definition}

\begin{definition}[$\mathbf{\underline{\overset{-}{xz}}}$]\label{def:isp_0-}
Blocking state: every execution belonging to the process block contains $y$, and no $x$ or $z$.

\noindent\textbf{Formally}: 

\noindent Given a process block $B$, it belongs to this class if and only if:
\begin{itemize}
\item $\forall \exe \in \Exe{B}$ such that:
\begin{itemize}
\item $\exists y \in \exe$.
\end{itemize}
\end{itemize}
\end{definition}

\begin{theorem}[Classification Completeness for Interval Sub-Pattern]
The set of possible evaluations of Interval Sub-Pattern is completely covered by the provided set of classifications.
\end{theorem}

\begin{proof}
\begin{enumerate}
\item The two sides of the pattern, the left ($x$) and on the right ($z$), in Interval Sub-Pattern allow 3 possible evaluations each: whether the partial requirement is satisfied, failing, or in a neutral state.
\item While the sides of the pattern share the same undesired element ($y$), the two sides are virtually independent as their evaluation is not influenced directly by the evaluation of the other side of the pattern, but only by the presence of the undesired element and their own required element.
\item From 1. and the independence from 2., it follows that the amount of possible combinations in Interval Sub-Pattern is $3^2$, hence 9 possible combinations. Moreover the case from 2. is also covered within the 9 combinations.
\item Therefore, every possible evaluation of Interval Sub-Pattern is captured by one of the 9 provided classifications.
\end{enumerate}
\end{proof}

\subsubsection{Lemmas}

\begin{lemma}[Aggregation Neutral Class]\label{l:neutral_isp}
Given a process block $A$, assigned to the evaluation class $\mathbf{\underline{\overset{0}{x}\overset{0}{z}}}$, and another process block $B$, assigned to any of the available evaluation classes. Let $C$ be a process block having $A$ and $B$ as its sub-blocks, then the evaluation class of $C$ is the same class as the process block $B$.
\end{lemma}

\begin{proof}
This proof follows closely the proof for Lemma \ref{l:neutral}
\end{proof}

\begin{lemma}[\kwd{AND} Commutativity]\label{l:commutativity}
Given two process blocks $A$ and $B$, the executions of $\kwd{AND}(A,$ $B)$ are the same as $\kwd{AND}(B, A)$.
\end{lemma}

\begin{proof}
The correctness of Lemma \ref{l:commutativity} follows directly from Definition \ref{def:ser}.
\end{proof}

\begin{lemma}[Class $\mathbf{\underline{\overset{+}{x}\overset{+}{z}}}$ Aggregation]\label{l:c3a}
Given a process block $A$ classified as $\mathbf{\underline{\overset{+}{x}\overset{+}{z}}}$ and another process block $B$ with any classification, aggregating $A$ and $B$, independently than the order and the type of their super-block always results in being classified as $\mathbf{\underline{\overset{+}{x}\overset{+}{z}}}$.
\end{lemma}

\begin{proof}
Considering the case where the aggregation is made through a super-block of type \kwd{SEQ}, according to Definition \ref{def:ser}, the aggregation of the executions corresponds to the concatenation of the executions in the sub blocks. From the hypothesis we know that there is an execution satisfying the requirements to be classified as $\mathbf{\underline{\overset{+}{x}\overset{+}{z}}}$.

Independently than the order of aggregation, from Definition \ref{def:isp_3}, we can see that the concatenation does not influence the properties allowing to classify the aggregation as $\mathbf{\underline{\overset{+}{x}\overset{+}{z}}}$, and as shown in Figure \ref{fig:ad1s_c_latt}, given that $\mathbf{\underline{\overset{+}{x}\overset{+}{z}}}$ is the top ranked classification, then the aggregation can be correctly classified as $\mathbf{\underline{\overset{+}{x}\overset{+}{z}}}$.

The correctness of the case concerning a super block of type \kwd{AND} follows directly from the previous case and Lemma \ref{lem:inclusivity}.
\end{proof}

\subsubsection{Aggregations}~\\
\noindent
The aggregation of the classes used to evaluate undernodes while evaluating an overnode of type \kwd{AND} are contained in Tables \ref{tab:0}, \dots, \ref{tab:3}.

\begin{table}[ht!]
\centering
\begin{tabular}{|c|c|c|c|}
\hline
A & B & \kwd{SEQ}(A, B) & \kwd{AND}(A, B) \\ \hline
$\mathbf{\underline{\overset{0}{x}\overset{0}{z}}}$ & $\mathbf{\underline{\overset{0}{x}\overset{0}{z}}}$ & $\mathbf{\underline{\overset{0}{x}\overset{0}{z}}}$ & $\mathbf{\underline{\overset{0}{x}\overset{0}{z}}}$ \\ \hline
$\mathbf{\underline{\overset{0}{x}\overset{0}{z}}}$ & $\mathbf{\underline{\overset{-}{xz}}}$ & $\mathbf{\underline{\overset{-}{xz}}}$ & $\mathbf{\underline{\overset{-}{xz}}}$ \\ \hline
$\mathbf{\underline{\overset{0}{x}\overset{0}{z}}}$ & $\mathbf{\underline{\overset{+}{x}\overset{0}{z}}}$ & $\mathbf{\underline{\overset{+}{x}\overset{0}{z}}}$ & $\mathbf{\underline{\overset{+}{x}\overset{0}{z}}}$ \\ \hline
$\mathbf{\underline{\overset{0}{x}\overset{0}{z}}}$ & $\mathbf{\underline{\overset{+}{x}\overset{-}{z}}}$ & $\mathbf{\underline{\overset{+}{x}\overset{-}{z}}}$ & $\mathbf{\underline{\overset{+}{x}\overset{-}{z}}}$ \\ \hline
$\mathbf{\underline{\overset{0}{x}\overset{0}{z}}}$ & $\mathbf{\underline{\overset{0}{x}\overset{+}{z}}}$ & $\mathbf{\underline{\overset{0}{x}\overset{+}{z}}}$ & $\mathbf{\underline{\overset{0}{x}\overset{+}{z}}}$  \\ \hline
$\mathbf{\underline{\overset{0}{x}\overset{0}{z}}}$ & $\mathbf{\underline{\overset{-}{x}\overset{+}{z}}}$ & $\mathbf{\underline{\overset{-}{x}\overset{+}{z}}}$ & $\mathbf{\underline{\overset{-}{x}\overset{+}{z}}}$ \\ \hline
$\mathbf{\underline{\overset{0}{x}\overset{0}{z}}}$ & $\mathbf{\underline{\overset{+}{x}\overset{+}{z}}}$ & $\mathbf{\underline{\overset{+}{x}\overset{+}{z}}}$ & $\mathbf{\underline{\overset{+}{x}\overset{+}{z}}}$ \\ \hline
\end{tabular}
\caption{Aggregation Table from Interval Sub-Pattern Evaluation Classes}\label{tab:0}
\end{table}

\begin{proof}[Table \ref{tab:0}]
As each of the aggregations shown in Table \ref{tab:0} contains at least an element belonging to the aggregation class $\mathbf{\underline{\overset{0}{x}\overset{0}{z}}}$, the correctness of the columns \kwd{SEQ} and \kwd{AND} follows directly from Lemma \ref{l:neutral_isp}.
\end{proof}

\begin{table}[ht!]
\centering
\begin{tabular}{|c|c|c|c|}
\hline
A & B & \kwd{SEQ}(A, B) & \kwd{AND}(A, B) \\ \hline
$\mathbf{\underline{\overset{-}{xz}}}$ & $\mathbf{\underline{\overset{0}{x}\overset{0}{z}}}$ & $\mathbf{\underline{\overset{-}{xz}}}$ & $\mathbf{\underline{\overset{-}{xz}}}$ \\ \hline
$\mathbf{\underline{\overset{-}{xz}}}$ & $\mathbf{\underline{\overset{-}{xz}}}$ & $\mathbf{\underline{\overset{-}{xz}}}$ & $\mathbf{\underline{\overset{-}{xz}}}$ \\ \hline
$\mathbf{\underline{\overset{-}{xz}}}$ & $\mathbf{\underline{\overset{+}{x}\overset{0}{z}}}$ & $\mathbf{\underline{\overset{+}{x}\overset{0}{z}}}$ & $\mathbf{\underline{\overset{+}{x}\overset{0}{z}}}$ \\ \hline
$\mathbf{\underline{\overset{-}{xz}}}$ & $\mathbf{\underline{\overset{+}{x}\overset{-}{z}}}$ & $\mathbf{\underline{\overset{+}{x}\overset{-}{z}}}$ & $\mathbf{\underline{\overset{+}{x}\overset{-}{z}}}$ \\ \hline
$\mathbf{\underline{\overset{-}{xz}}}$ & $\mathbf{\underline{\overset{0}{x}\overset{+}{z}}}$ & $\mathbf{\underline{\overset{-}{x}\overset{+}{z}}}$ & $\mathbf{\underline{\overset{0}{x}\overset{+}{z}}}$ \\ \hline
$\mathbf{\underline{\overset{-}{xz}}}$ & $\mathbf{\underline{\overset{-}{x}\overset{+}{z}}}$ & $\mathbf{\underline{\overset{-}{x}\overset{+}{z}}}$ & $\mathbf{\underline{\overset{-}{x}\overset{+}{z}}}$ \\ \hline
$\mathbf{\underline{\overset{-}{xz}}}$ & $\mathbf{\underline{\overset{+}{x}\overset{+}{z}}}$ & $\mathbf{\underline{\overset{+}{x}\overset{+}{z}}}$ & $\mathbf{\underline{\overset{+}{x}\overset{+}{z}}}$ \\ \hline
\end{tabular}
\caption{Aggregation Table from Interval Sub-Pattern Evaluation Classes}\label{tab:0-}
\end{table}

\begin{proof}[Table \ref{tab:0-}]
\begin{enumerate}
\item From Definition \ref{def:isp_0-}, and column \textbf{A} of the table, it follows that each execution of the process block $A$ contains a task having $y$ annotated, and no other literals beneficial towards the pattern fulfilment.
\item We prove the correctness by case, depending on the class assigned to the process block in column \textbf{B}, each row of the table and each of the two different types of aggregation.

\begin{description}
\item[$\mathbf{\underline{\overset{0}{x}\overset{0}{z}}}$] Independently on the type of the process block, the correctness of the row follows directly from Lemma \ref{l:neutral_isp}.

\item[$\mathbf{\underline{\overset{-}{xz}}}$] 
\begin{enumerate}
\item From Definition \ref{def:isp_0-}, and the classification of the process block in the column \textbf{B} of the table, it follows that each execution of \textbf{B} contains at least a task having $y$ annotated, and no other literals beneficial towards the pattern fulfilment.
\begin{description}
\item[SEQ]
\begin{enumerate}
\item From Definition \ref{def:ser}, it follows that the possible executions of a process block $\kwd{SEQ}(A, B)$ are the concatenation of an execution of $A$ and an execution of $B$.
\item From i., 1., and (a), it follows that each possible execution of $\kwd{SEQ}(A, B)$ contains at least a task having $y$ annotated, and no other literals beneficial towards the pattern fulfilment.
\item From ii., and Definition \ref{def:isp_0-}, it follows that the process block resulting from $\kwd{SEQ}(A, B)$ is correctly classified as $\mathbf{\underline{\overset{-}{xz}}}$.
\end{enumerate}
\item[AND]
\begin{enumerate}
\item From Definition \ref{def:ser}, it follows that each execution resulting from $\kwd{AND}(A,$ $B)$ consists of an execution of $A$ interleaved with an execution of $B$.
\item From i., 1., and (a), it follows that each possible execution of $\kwd{AND}(A, B)$ contains at least a task having $y$ annotated, and no other literals beneficial towards the pattern fulfilment.
\item From ii., and Definition \ref{def:isp_0-}, it follows that the process block resulting from $\kwd{AND}(A, B)$ is correctly classified as $\mathbf{\underline{\overset{-}{xz}}}$.
\end{enumerate}
\end{description}
\end{enumerate}

\item[$\mathbf{\underline{\overset{+}{x}\overset{0}{z}}}$]
\begin{enumerate}
\item From Definition \ref{def:isp_1}, and the classification of the process block in the column \textbf{B} of the table, it follows that there exists an execution belonging to the process block containing a task having $x$ annotated, and no task on its right having $y$ annotated.
\begin{description}
\item[SEQ]
\begin{enumerate}
\item From Definition \ref{def:ser}, it follows that the possible executions of a process block $\kwd{SEQ}(A, B)$ are the concatenation of an execution of $A$ and an execution of $B$.
\item From i., 1., and (a), it follows that there exists an execution belonging to $\kwd{SEQ}(A, B)$ containing a task having $x$ annotated, and no task on its right having $y$ annotated. This because the $y$ appearing in each execution of $A$ always appears on the left of each whole execution of $B$.
\item From ii., and Definition \ref{def:isp_1}, it follows that the process block resulting from $\kwd{SEQ}(A, B)$ is correctly classified as $\mathbf{\underline{\overset{+}{x}\overset{0}{z}}}$.
\end{enumerate}
\item[AND]
\begin{enumerate}
\item From the \textbf{SEQ} case result, and Lemma \ref{lem:inclusivity}, it follows that $\kwd{AND}(A, B)$ can be classified as $\mathbf{\underline{\overset{+}{x}\overset{0}{z}}}$ or better.
\item From the preference lattice in Figure \ref{fig:ad1s_c_latt} and Definition \ref{def:isp_3}, it follows that for $\kwd{AND}(A, B)$ to achieve a better class than $\mathbf{\underline{\overset{+}{x}\overset{0}{z}}}$, it requires a $z$ annotated in one of its tasks.
\item From (a), and Definition \ref{def:isp_1}, it follows that $B$ does not contain $z$ in its tasks annotations, otherwise it would be assigned to a better class. 
\item From 1., and iii., it follows that $\kwd{AND}(A, B)$ does not contain tasks having $z$ annotated.
\item From i., ii., and iv., it follows that $\kwd{AND}(A, B)$ is correctly classified as $\mathbf{\underline{\overset{+}{x}\overset{0}{z}}}$.
\end{enumerate}
\end{description}
\end{enumerate}

\item[$\mathbf{\underline{\overset{+}{x}\overset{-}{z}}}$]
\begin{enumerate}
\item From Definition \ref{def:isp_1r}, and the classification of the process block in the column \textbf{B} of the table, it follows that there exists an execution belonging to the process block containing a task having $x$ annotated, and a task on its right having $y$ annotated.
\begin{description}
\item[SEQ]
\begin{enumerate}
\item From Definition \ref{def:ser}, it follows that the possible executions of a process block $\kwd{SEQ}(A, B)$ are the concatenation of an execution of $A$ and an execution of $B$.
\item From i., 1., and (a), it follows that there exists an execution belonging to the process block containing a task having $x$ annotated, and a task on its right having $y$ annotated. This because the $y$ appearing in each execution of $A$ always appears on the left of each whole execution of $B$, hence not influencing the execution order of the tasks on the right side of the task having $x$ annotated.
\item From ii., and Definition \ref{def:isp_1r}, it follows that the process block resulting from $\kwd{SEQ}(A, B)$ is correctly classified as $\mathbf{\underline{\overset{+}{x}\overset{-}{z}}}$.
\end{enumerate}
\item[AND]
\begin{enumerate}
\item From the \textbf{SEQ} case result, and Lemma \ref{lem:inclusivity}, it follows that $\kwd{AND}(A, B)$ can be classified as $\mathbf{\underline{\overset{+}{x}\overset{-}{z}}}$ or better.
\item From the preference lattice in Figure \ref{fig:ad1s_c_latt} and Definition \ref{def:isp_3}, it follows that for $\kwd{AND}(A, B)$ to achieve a better class than $\mathbf{\underline{\overset{+}{x}\overset{-}{z}}}$, it requires either a $z$ annotated in one of its tasks, or an execution not having a $y$ annotated on the right of a task having $x$ annotated.
\item From (a), and Definition \ref{def:isp_1r}, it follows that $B$ does not contain a $z$ not preceded by a $y$ after the $x$ in its tasks' annotations, otherwise it would be assigned to a better class. 
\item From 1., and iii., it follows that $\kwd{AND}(A, B)$ does not contain tasks having $z$ annotated.
\item From (a), and Definition \ref{def:isp_1r}, it follows that $B$ does not contain an execution having a task with $x$ annotated and no task on its right with $y$ annotated, otherwise $B$ would be classified as $\mathbf{\underline{\overset{+}{x}\overset{0}{z}}}$.
\item From v., and Definition \ref{def:ser}, it follows that $\kwd{AND}(A, B)$ does not contain an execution not having a $y$ annotated on the right of a task having $x$ annotated.
\item From i., ii., iv., and vi., it follows that $\kwd{AND}(A, B)$ is correctly classified as $\mathbf{\underline{\overset{+}{x}\overset{-}{z}}}$.
\end{enumerate}
\end{description}
\end{enumerate}
\item[$\mathbf{\underline{\overset{0}{x}\overset{+}{z}}}$]
\begin{enumerate}
\item From Definition \ref{def:isp_1'}, and the classification of the process block in the column \textbf{B} of the table, it follows that there exists an execution belonging to the process block containing a task having $z$ annotated, and no tasks on its left having $y$ annotated.
\begin{description}
\item[SEQ]
\begin{enumerate}
\item From Definition \ref{def:ser}, it follows that the possible executions of a process block $\kwd{SEQ}(A, B)$ are the concatenation of an execution of $A$ and an execution of $B$.
\item From i., 1., and (a), it follows that there exists an execution belonging to the process block containing a task having $z$ annotated, and a task on its left having $y$ annotated. This because the $y$ appearing in each execution of $A$ always appears on the left of each whole execution of $B$, hence influencing the execution on the left side of the task having $z$ annotated.
\item From ii., and Definition \ref{def:isp_1'l}, it follows that the process block resulting from $\kwd{SEQ}(A, B)$ is correctly classified as $\mathbf{\underline{\overset{-}{x}\overset{+}{z}}}$.
\end{enumerate}
\item[AND]
\begin{enumerate}
\item From Definition \ref{def:ser}, it follows that each execution resulting from $\kwd{AND}(A,$ $B)$ consists of an execution of $A$ interleaved with an execution of $B$.
\item From i., it follows that an execution of $B$ followed by an execution of $A$ is a valid execution of $\kwd{AND}(A, B)$.
\item From ii., 1., and (a), it follows that there exists an execution belonging to the process block containing a task having $z$ annotated, and no task on its left having $y$ annotated. This because the $y$ appearing in each execution of $A$ can always be put on the right of an execution of $A$, hence not influencing the execution on the left side of the task having $z$ annotated.
\item From iii., and Definition \ref{def:isp_1'}, it follows that the process block resulting from $\kwd{AND}(A, B)$ is correctly classified as $\mathbf{\underline{\overset{0}{x}\overset{+}{z}}}$.
\end{enumerate}
\end{description}
\end{enumerate}
\item[$\mathbf{\underline{\overset{-}{x}\overset{+}{z}}}$]
\begin{enumerate}
\item From Definition \ref{def:isp_1'l}, and the classification of the process block in the column \textbf{B} of the table, it follows that there exists an execution belonging to the process block containing a task having $z$ annotated, and having another task on its left having $y$ annotated.
\begin{description}
\item[SEQ]
\begin{enumerate}
\item From Definition \ref{def:ser}, it follows that the possible executions of a process block $\kwd{SEQ}(A, B)$ are the concatenation of an execution of $A$ and an execution of $B$.
\item From i., 1., and (a), it follows that there exists an execution belonging to the process block containing a task having $z$ annotated, and a task on its left having $y$ annotated. This corresponds to the original class of $B$, as the $y$ from $A$ does not alter the initial properties of $B$.
\item From ii., and Definition \ref{def:isp_1'l}, it follows that the process block resulting from $\kwd{SEQ}(A, B)$ is correctly classified as $\mathbf{\underline{\overset{-}{x}\overset{+}{z}}}$.
\end{enumerate}
\item[AND]
\begin{enumerate}
\item From Definition \ref{def:ser}, it follows that each execution resulting from $\kwd{AND}(A,$ $B)$ consists of an execution of $A$ interleaved with an execution of $B$.
\item From i., 1., and (a), it follows that there exists an execution belonging to the process block containing a task having $z$ annotated, and a task on its left having $y$ annotated. This corresponds to the original class of $B$, as no matter the interleaving chosen, in particular where the $y$ from the execution of $A$ would be positioned within the execution of $B$, it would not change its properties.
\item From ii., and Definition \ref{def:isp_1'l}, it follows that the process block resulting from $\kwd{AND}(A, B)$ is correctly classified as $\mathbf{\underline{\overset{-}{x}\overset{+}{z}}}$.
\end{enumerate}
\end{description}
\end{enumerate}
\item[$\mathbf{\underline{\overset{+}{x}\overset{+}{z}}}$]
\begin{enumerate}
\item Independently on the type of the process block, the correctness of the row follows directly from Lemma \ref{l:c3a}.
\end{enumerate}
\end{description}

\item We have shown that, for each row the \kwd{SEQ} and \kwd{AND} aggregation between the states correctly aggregates in the described result
\item The correctness of column \kwd{XOR} follows directly from the preference order defined in the lattice shown in Figure \ref{fig:ad1s_c_latt}.
\item From 3., and 4., it follows that the Aggregations in Table \ref{tab:0-} are correct.
\end{enumerate}
\end{proof}

\begin{table}[ht!]
\centering
\begin{tabular}{|c|c|c|c|}
\hline
A & B & \kwd{SEQ}(A, B) & \kwd{AND}(A, B) \\ \hline
$\mathbf{\underline{\overset{+}{x}\overset{0}{z}}}$ & $\mathbf{\underline{\overset{0}{x}\overset{0}{z}}}$ & $\mathbf{\underline{\overset{+}{x}\overset{0}{z}}}$ & $\mathbf{\underline{\overset{+}{x}\overset{0}{z}}}$ \\ \hline
$\mathbf{\underline{\overset{+}{x}\overset{0}{z}}}$ & $\mathbf{\underline{\overset{-}{xz}}}$ & $\mathbf{\underline{\overset{+}{x}\overset{-}{z}}}$ & $\mathbf{\underline{\overset{+}{x}\overset{0}{z}}}$ \\ \hline
$\mathbf{\underline{\overset{+}{x}\overset{0}{z}}}$ & $\mathbf{\underline{\overset{+}{x}\overset{0}{z}}}$ & $\mathbf{\underline{\overset{+}{x}\overset{0}{z}}}$ & $\mathbf{\underline{\overset{+}{x}\overset{0}{z}}}$ \\ \hline
$\mathbf{\underline{\overset{+}{x}\overset{0}{z}}}$ & $\mathbf{\underline{\overset{+}{x}\overset{-}{z}}}$ & $\mathbf{\underline{\overset{+}{x}\overset{-}{z}}}$ & $\mathbf{\underline{\overset{+}{x}\overset{0}{z}}}$ \\ \hline
$\mathbf{\underline{\overset{+}{x}\overset{0}{z}}}$ & $\mathbf{\underline{\overset{0}{x}\overset{+}{z}}}$ & $\mathbf{\underline{\overset{+}{x}\overset{+}{z}}}$ & $\mathbf{\underline{\overset{+}{x}\overset{+}{z}}}$ \\ \hline
$\mathbf{\underline{\overset{+}{x}\overset{0}{z}}}$ & $\mathbf{\underline{\overset{-}{x}\overset{+}{z}}}$ & $\mathbf{\underline{\overset{+}{x}\overset{-}{z}}}$ / $\mathbf{\underline{\overset{-}{x}\overset{+}{z}}}$ & $\mathbf{\underline{\overset{+}{x}\overset{+}{z}}}$ \\ \hline
$\mathbf{\underline{\overset{+}{x}\overset{0}{z}}}$ & $\mathbf{\underline{\overset{+}{x}\overset{+}{z}}}$ & $\mathbf{\underline{\overset{+}{x}\overset{+}{z}}}$ & $\mathbf{\underline{\overset{+}{x}\overset{+}{z}}}$ \\ \hline
\end{tabular}
\caption{Aggregation Table from Interval Sub-Pattern Evaluation Classes}\label{tab:1}
\end{table}

\begin{proof}[Table \ref{tab:1}]
\begin{enumerate}
\item From Definition \ref{def:isp_1}, and column \textbf{A} of the table, it follows that there exists an execution belonging to the process block containing a task having $x$ annotated, and no task on its right having $y$ annotated.
\item We prove by case each row of the table and each of the two different types of aggregation.

\begin{description}
\item[$\mathbf{\underline{\overset{0}{x}\overset{0}{z}}}$] Independently on the type of the process block, the correctness of the row follows directly from Lemma \ref{l:neutral_isp}.
\item[$\mathbf{\underline{\overset{-}{xz}}}$]
\begin{enumerate}
\item From Definition \ref{def:isp_0-}, and the classification of the process block in the column \textbf{B} of the table, it follows that each execution of \textbf{B} contains at least a task having $y$ annotated, and no other literals beneficial towards the pattern fulfilment.
\begin{description}
\item[SEQ]
\begin{enumerate}
\item From Definition \ref{def:ser}, it follows that the possible executions of a process block $\kwd{SEQ}(A, B)$ are the concatenation of an execution of $A$ and an execution of $B$.
\item From i., 1., and (a), it follows that there exists an execution belonging to the process block containing a task having $x$ annotated, and a task on its right having $y$ annotated. This because the $y$ appearing in each execution of $B$ always appears on the right of each whole execution of $A$.
\item From ii., and Definition \ref{def:isp_1r}, it follows that the process block resulting from $\kwd{SEQ}(A, B)$ is correctly classified as $\mathbf{\underline{\overset{+}{x}\overset{-}{z}}}$.
\end{enumerate}
\item[AND]
\begin{enumerate}
\item From Lemma \ref{l:commutativity}, and row $\mathbf{\underline{\overset{-}{xz}}}$ $\mathbf{\underline{\overset{+}{x}\overset{0}{z}}}$ of Table \ref{tab:0-}, it follows that the process block resulting from $\kwd{AND}(A, B)$ is correctly classified as $\mathbf{\underline{\overset{+}{x}\overset{0}{z}}}$. 
\end{enumerate}
\end{description}
\end{enumerate}
\item[$\mathbf{\underline{\overset{+}{x}\overset{0}{z}}}$]
\begin{enumerate}
\item From Definition \ref{def:isp_1}, and the classification of the process block in the column \textbf{B} of the table, it follows that there exists an execution belonging to the process block containing a task having $x$ annotated, and no task on its right having $y$ annotated.
\begin{description}
\item[SEQ]
\begin{enumerate}
\item From Definition \ref{def:ser}, it follows that the possible executions of a process block $\kwd{SEQ}(A, B)$ are the concatenation of an execution of $A$ and an execution of $B$.
\item From i., 1., and (a), it follows that there exists an execution belonging to the process block containing a task having $x$ annotated, and no task on its right having $y$ annotated. This because there is an execution of $B$ that without $y$ that can follow each execution of $A$.
\item From ii., and Definition \ref{def:isp_1}, it follows that the process block resulting from $\kwd{SEQ}(A, B)$ is correctly classified as $\mathbf{\underline{\overset{+}{x}\overset{0}{z}}}$.
\end{enumerate}
\item[AND]
\begin{enumerate}
\item From the \textbf{SEQ} case result, and Lemma \ref{lem:inclusivity}, it follows that $\kwd{AND}(A, B)$ can be classified as $\mathbf{\underline{\overset{+}{x}\overset{0}{z}}}$ or better.
\item From the preference lattice in Figure \ref{fig:ad1s_c_latt}, it follows that for $\kwd{AND}(A, B)$ to achieve a better class than $\mathbf{\underline{\overset{+}{x}\overset{0}{z}}}$, it requires a $z$ annotated in one of its tasks.
\item From (a), and Definition \ref{def:isp_1}, it follows that both $A$ and $B$ do not contain $z$ in its tasks annotations, otherwise they would be assigned to a better class. 
\item From 1., and iii., it follows that $\kwd{AND}(A, B)$ does not contain tasks having $z$ annotated.
\item From i., ii., and iv., it follows that $\kwd{AND}(A, B)$ is correctly classified as $\mathbf{\underline{\overset{+}{x}\overset{0}{z}}}$.
\end{enumerate}
\end{description}
\end{enumerate}
\item[$\mathbf{\underline{\overset{+}{x}\overset{-}{z}}}$]
\begin{enumerate}
\item From Definition \ref{def:isp_1r}, and the classification of the process block in the column \textbf{B} of the table, it follows that there exists an execution belonging to the process block containing a task having $x$ annotated, and a task on its right having $y$ annotated.
\begin{description}
\item[SEQ]
\begin{enumerate}
\item From Definition \ref{def:ser}, it follows that the possible executions of a process block $\kwd{SEQ}(A, B)$ are the concatenation of an execution of $A$ and an execution of $B$.
\item From i., 1., and (a), it follows that there exists an execution belonging to the process block containing a task having $x$ annotated, and a task on its right having $y$ annotated. This because the executions of $B$ are not influenced by the executions of $A$ as they appear on their left side.
\item From ii., and Definition \ref{def:isp_1r}, it follows that the process block resulting from $\kwd{SEQ}(A, B)$ is correctly classified as $\mathbf{\underline{\overset{+}{x}\overset{-}{z}}}$.
\end{enumerate}
\item[AND]
\begin{enumerate}
\item From Definition \ref{def:ser}, it follows that each execution resulting from $\kwd{AND}(A,$ $B)$ consists of an execution of $A$ interleaved with an execution of $B$.
\item From i., it follows that an execution of $B$ followed by an execution of $A$ is a valid execution of $\kwd{AND}(A, B)$.
\item From ii., 1., and (a), it follows that there exists an execution belonging to the process block containing a task having $x$ annotated, and no task on its right having $y$ annotated. This because the executions of $A$ are not influenced by the executions of $B$ as they appear on their left side.
\item From iii., and Definition \ref{def:isp_1}, it follows that the process block resulting from $\kwd{AND}(A, B)$ is correctly classified as $\mathbf{\underline{\overset{+}{x}\overset{0}{z}}}$.
\end{enumerate}
\end{description}
\end{enumerate}
\item[$\mathbf{\underline{\overset{0}{x}\overset{+}{z}}}$]
\begin{enumerate}
\item From Definition \ref{def:isp_1'}, and the classification of the process block in the column \textbf{B} of the table, it follows that there exists an execution belonging to the process block containing a task having $z$ annotated, and no tasks on its left having $y$ annotated.
\begin{description}
\item[SEQ]
\begin{enumerate}
\item From Definition \ref{def:ser}, it follows that the possible executions of a process block $\kwd{SEQ}(A, B)$ are the concatenation of an execution of $A$ and an execution of $B$.
\item From i., 1., and (a), it follows that there exists an execution belonging to the process block containing a task having $x$ annotated and another task on its right having $z$ annotated, while having no task in-between having $y$ annotated. As appending the execution of $B$ to an execution of $A$ completes the interval without including any $y$ in-between.
\item From ii., and Definition \ref{def:isp_3}, it follows that the process block resulting from $\kwd{SEQ}(A, B)$ is correctly classified as $\mathbf{\underline{\overset{+}{x}\overset{+}{z}}}$.
\end{enumerate}
\item[AND]
\begin{enumerate}
\item From the \textbf{SEQ} case result, and Lemma \ref{lem:inclusivity}, it follows that $\kwd{AND}(A, B)$ can be correctly classified as $\mathbf{\underline{\overset{+}{x}\overset{+}{z}}}$.
\end{enumerate}
\end{description}
\end{enumerate}
\item[$\mathbf{\underline{\overset{-}{x}\overset{+}{z}}}$]
\begin{enumerate}
\item From Definition \ref{def:isp_1'l}, and the classification of the process block in the column \textbf{B} of the table, it follows that there exists an execution belonging to the process block containing a task having $z$ annotated, and having another task on its left having $y$ annotated.
\begin{description}
\item[SEQ]
\begin{enumerate}
\item From Definition \ref{def:ser}, it follows that the possible executions of a process block $\kwd{SEQ}(A, B)$ are the concatenation of an execution of $A$ and an execution of $B$.
\item From i., 1., and (a), it follows that there exists an execution belonging to the process block containing a task having $x$ annotated and another task on its right having $z$ annotated, while having at least task in-between having $y$ annotated. As appending the execution of $B$ to an execution of $A$ completes the interval and including any $y$ from the execution of $B$ in-between.
\item From ii., Definition \ref{def:isp_1r} and Definition \ref{def:isp_1'l}, it follows that the process block resulting from $\kwd{SEQ}(A, B)$ is correctly classified as both $\mathbf{\underline{\overset{+}{x}\overset{-}{z}}}$ and $\mathbf{\underline{\overset{-}{x}\overset{+}{z}}}$.
\end{enumerate}
\item[AND]
\begin{enumerate}
\item From Definition \ref{def:ser}, it follows that each execution resulting from $\kwd{AND}(A,$ $B)$ consists of an execution of $A$ interleaved with an execution of $B$.
\item From i., 1., and (a), it follows that an execution of $A$ can be preceded by the part of the execution of $B$ containing $y$, and followed by the remainder of that execution containing $z$.
\item From ii., it follows that there exist an execution of $\kwd{AND}(A, B)$ having a task with $x$ annotated and another task on its right having $z$ annotated, without having a task in-between having $y$ annotated.
\item From iii., and Definition \ref{def:isp_3}, it follows that the process block resulting from $\kwd{AND}(A, B)$ is correctly classified as $\mathbf{\underline{\overset{+}{x}\overset{+}{z}}}$.
\end{enumerate}
\end{description}
\end{enumerate}
\item[$\mathbf{\underline{\overset{+}{x}\overset{+}{z}}}$]
\begin{enumerate}
\item Independently on the type of the process block, the correctness of the row follows directly from Lemma \ref{l:c3a}.
\end{enumerate}
\end{description}

\item We have shown that, for each row the \kwd{SEQ} and \kwd{AND} aggregation between the states correctly aggregates in the described result
\item The correctness of column \kwd{XOR} follows directly from the preference order defined in the lattice shown in Figure \ref{fig:ad1s_c_latt}.
\item From 3., and 4., it follows that the Aggregations in Table \ref{tab:1} are correct.
\end{enumerate}

\end{proof}

\begin{table}[ht!]
\centering
\begin{tabular}{|c|c|c|c|}
\hline
A & B & \kwd{SEQ}(A, B) & \kwd{AND}(A, B) \\ \hline
$\mathbf{\underline{\overset{+}{x}\overset{-}{z}}}$ & $\mathbf{\underline{\overset{0}{x}\overset{0}{z}}}$ & $\mathbf{\underline{\overset{+}{x}\overset{-}{z}}}$ & $\mathbf{\underline{\overset{+}{x}\overset{-}{z}}}$ \\ \hline
$\mathbf{\underline{\overset{+}{x}\overset{-}{z}}}$ & $\mathbf{\underline{\overset{-}{xz}}}$ & $\mathbf{\underline{\overset{+}{x}\overset{-}{z}}}$ & $\mathbf{\underline{\overset{+}{x}\overset{-}{z}}}$ \\ \hline
$\mathbf{\underline{\overset{+}{x}\overset{-}{z}}}$ & $\mathbf{\underline{\overset{+}{x}\overset{0}{z}}}$ & $\mathbf{\underline{\overset{+}{x}\overset{0}{z}}}$ & $\mathbf{\underline{\overset{+}{x}\overset{0}{z}}}$ \\ \hline
$\mathbf{\underline{\overset{+}{x}\overset{-}{z}}}$ & $\mathbf{\underline{\overset{+}{x}\overset{-}{z}}}$ & $\mathbf{\underline{\overset{+}{x}\overset{-}{z}}}$ & $\mathbf{\underline{\overset{+}{x}\overset{-}{z}}}$ \\ \hline
$\mathbf{\underline{\overset{+}{x}\overset{-}{z}}}$ & $\mathbf{\underline{\overset{0}{x}\overset{+}{z}}}$ & $\mathbf{\underline{\overset{+}{x}\overset{-}{z}}}$ / $\mathbf{\underline{\overset{-}{x}\overset{+}{z}}}$ & $\mathbf{\underline{\overset{+}{x}\overset{+}{z}}}$ \\ \hline
$\mathbf{\underline{\overset{+}{x}\overset{-}{z}}}$ & $\mathbf{\underline{\overset{-}{x}\overset{+}{z}}}$ & $\mathbf{\underline{\overset{+}{x}\overset{-}{z}}}$ / $\mathbf{\underline{\overset{-}{x}\overset{+}{z}}}$ & $\mathbf{\underline{\overset{+}{x}\overset{+}{z}}}$ \\ \hline
$\mathbf{\underline{\overset{+}{x}\overset{-}{z}}}$ & $\mathbf{\underline{\overset{+}{x}\overset{+}{z}}}$ & $\mathbf{\underline{\overset{+}{x}\overset{+}{z}}}$ & $\mathbf{\underline{\overset{+}{x}\overset{+}{z}}}$ \\ \hline
\end{tabular}
\caption{Aggregation Table from Interval Sub-Pattern Evaluation Classes}\label{tab:1r}
\end{table}

\begin{proof}[Table \ref{tab:1r}]
\begin{enumerate}
\item From Definition \ref{def:isp_1r}, and column \textbf{A} of the table, it follows that there exists an execution belonging to the process block containing a task having $x$ annotated, and a task on its right having $y$ annotated.
\item We prove by case each row of the table and each of the two different types of aggregation.

\begin{description}
\item[$\mathbf{\underline{\overset{0}{x}\overset{0}{z}}}$] Independently on the type of the process block, the correctness of the row follows directly from Lemma \ref{l:neutral_isp}.
\item[$\mathbf{\underline{\overset{-}{xz}}}$]
\begin{enumerate}
\item From Definition \ref{def:isp_0-}, and the classification of the process block in the column \textbf{B} of the table, it follows that each execution of \textbf{B} contains at least a task having $y$ annotated, and no other literals beneficial towards the pattern fulfilment.
\begin{description}
\item[SEQ]
\begin{enumerate}
\item From Definition \ref{def:ser}, it follows that the possible executions of a process block $\kwd{SEQ}(A, B)$ are the concatenation of an execution of $A$ and an execution of $B$.
\item From i., 1., and (a), it follows that there exists an execution in $\kwd{SEQ}(A, B)$ containing $x$ and followed by at least a $y$.
\item From ii. and Definition \ref{def:isp_1r}, it follows that $\kwd{SEQ}(A, B)$ can be correctly classified as $\mathbf{\underline{\overset{+}{x}\overset{-}{z}}}$.
\end{enumerate}
\item[AND]
\begin{enumerate}
\item From Lemma \ref{l:commutativity}, it follows that $\kwd{AND}(A, B)$ is classified as the same class as the row $\mathbf{\underline{\overset{-}{xz}}}$ $\mathbf{\underline{\overset{+}{x}\overset{-}{z}}}$ in Table \ref{tab:0-}, which is correctly $\mathbf{\underline{\overset{+}{x}\overset{-}{z}}}$.
\end{enumerate}
\end{description}
\end{enumerate}
\item[$\mathbf{\underline{\overset{+}{x}\overset{0}{z}}}$]
\begin{enumerate}
\item From Definition \ref{def:isp_1}, and the classification of the process block in the column \textbf{B} of the table, it follows that there exists an execution belonging to the process block containing a task having $x$ annotated, and no task on its right having $y$ annotated.
\begin{description}
\item[SEQ]
\begin{enumerate}
\item From Definition \ref{def:ser}, it follows that the possible executions of a process block $\kwd{SEQ}(A, B)$ are the concatenation of an execution of $A$ and an execution of $B$.
\item From 1., (a) and i., it follows that there is an execution in $\kwd{SEQ}(A, B)$ having the execution of $B$, where there is a $x$ and no $y$ on its right, as its tail.
\item From ii. and Definition \ref{def:isp_1}, it follows that $\kwd{SEQ}(A, B)$ is correctly classified as $\mathbf{\underline{\overset{+}{x}\overset{0}{z}}}$.
\end{enumerate}
\item[AND]
\begin{enumerate}
\item From Lemma \ref{l:commutativity}, it follows that $\kwd{AND}(A, B)$ is classified as the same class as the row $\mathbf{\underline{\overset{+}{x}\overset{0}{z}}}$ $\mathbf{\underline{\overset{+}{x}\overset{-}{z}}}$ in Table \ref{tab:1}, which is correctly $\mathbf{\underline{\overset{+}{x}\overset{0}{z}}}$.
\end{enumerate}
\end{description}
\end{enumerate}
\item[$\mathbf{\underline{\overset{+}{x}\overset{-}{z}}}$]
\begin{enumerate}
\item From Definition \ref{def:isp_1r}, and the classification of the process block in the column \textbf{B} of the table, it follows that there exists an execution belonging to the process block containing a task having $x$ annotated, and a task on its right having $y$ annotated.
\begin{description}
\item[SEQ]
\begin{enumerate}
\item From Definition \ref{def:ser}, it follows that the possible executions of a process block $\kwd{SEQ}(A, B)$ are the concatenation of an execution of $A$ and an execution of $B$.
\item From 1., (a) and i., it follows that there is an execution in $\kwd{SEQ}(A, B)$ having the execution of $B$, where there is a $x$ and a $y$ on its right, as its tail.
\item From ii. and Definition \ref{def:isp_1r}, it follows that $\kwd{SEQ}(A, B)$ is correctly classified as $\mathbf{\underline{\overset{+}{x}\overset{-}{z}}}$.
\end{enumerate}
\item[AND]
\begin{enumerate}
\item From Definition \ref{def:ser}, it follows that each execution resulting from $\kwd{AND}(A,$ $B)$ consists of an execution of $A$ interleaved with an execution of $B$.
\item From Lemma \ref{lem:inclusivity} and the result for the \kwd{SEQ} case, it follows that the classification of $\kwd{AND}(A, B)$ cannot be worse than $\mathbf{\underline{\overset{+}{x}\overset{-}{z}}}$.
\item From 1. and (a), it follows that none of the executions identified by the classifications contain a $z$.
\item From 1. and (a), it follows that none of the executions identified by the classifications contain a $x$ without a $y$ on its right.
\item From i., iii., iv. and the preference order in Figure \ref{fig:ad1s_c_latt}, it follows that there cannot be a better classification than $\mathbf{\underline{\overset{+}{x}\overset{-}{z}}}$.
\item From ii. and v., it follows that $\kwd{AND}(A, B)$ is correctly classified as $\mathbf{\underline{\overset{+}{x}\overset{-}{z}}}$.
\end{enumerate}
\end{description}
\end{enumerate}
\item[$\mathbf{\underline{\overset{0}{x}\overset{+}{z}}}$]
\begin{enumerate}
\item From Definition \ref{def:isp_1'}, and the classification of the process block in the column \textbf{B} of the table, it follows that there exists an execution belonging to the process block containing a task having $z$ annotated, and no tasks on its left having $y$ annotated.
\begin{description}
\item[SEQ]
\begin{enumerate}
\item From Definition \ref{def:ser}, it follows that the possible executions of a process block $\kwd{SEQ}(A, B)$ are the concatenation of an execution of $A$ and an execution of $B$.
\item From 1., (a) and i., it follows that there is an execution composed by a classified execution in $A$ followed by a classification in $B$ where there is a $x$, with a $y$ on its right, and a $z$ on the right of the $y$.
\item From ii. and Definition \ref{def:isp_1r}, it follows that $\kwd{SEQ}(A, B)$ is correctly classified as $\mathbf{\underline{\overset{+}{x}\overset{-}{z}}}$.
\item From ii. and Definition \ref{def:isp_1'l}, it follows that $\kwd{SEQ}(A, B)$ is correctly classified as $\mathbf{\underline{\overset{-}{x}\overset{+}{z}}}$.
\end{enumerate}
\item[AND]
\begin{enumerate}
\item From Definition \ref{def:ser}, it follows that each execution resulting from $\kwd{AND}(A,$ $B)$ consists of an execution of $A$ interleaved with an execution of $B$.
\item From 1., (a) and i., it follows that the execution of $A$, having $x$ and a $y$ on its right, can be interleaved with the execution from $B$ containing a $z$, and the $z$ from $B$ can be placed in the resulting execution between the $x$ and $y$ from $A$.
\item From ii. and Definition \ref{def:isp_3}, it follows that $\kwd{AND}(A, B)$ is correctly classified as $\mathbf{\underline{\overset{+}{x}\overset{+}{z}}}$.
\end{enumerate}
\end{description}
\end{enumerate}
\item[$\mathbf{\underline{\overset{-}{x}\overset{+}{z}}}$]
\begin{enumerate}
\item From Definition \ref{def:isp_1'l}, and the classification of the process block in the column \textbf{B} of the table, it follows that there exists an execution belonging to the process block containing a task having $z$ annotated, and having another task on its left having $y$ annotated.
\begin{description}
\item[SEQ]
\begin{enumerate}
\item From Definition \ref{def:ser}, it follows that the possible executions of a process block $\kwd{SEQ}(A, B)$ are the concatenation of an execution of $A$ and an execution of $B$.
\item From 1., (a) and i., it follows that there is an execution composed by a classified execution in $A$ followed by a classification in $B$ where there is a $x$, with a $y$ on its right, and a $z$ on the right of the $y$.
\item From ii. and Definition \ref{def:isp_1r}, it follows that $\kwd{SEQ}(A, B)$ is correctly classified as $\mathbf{\underline{\overset{+}{x}\overset{-}{z}}}$.
\item From ii. and Definition \ref{def:isp_1'l}, it follows that $\kwd{SEQ}(A, B)$ is correctly classified as $\mathbf{\underline{\overset{-}{x}\overset{+}{z}}}$.
\end{enumerate}
\item[AND]
\begin{enumerate}
\item From Definition \ref{def:ser}, it follows that each execution resulting from $\kwd{AND}(A,$ $B)$ consists of an execution of $A$ interleaved with an execution of $B$.
\item From 1., (a) and i., it follows that the execution of $A$, having $x$ and a $y$ on its right, can be interleaved with the execution from $B$ containing a $z$, and the $z$ from $B$ can be placed in the resulting execution between the $x$ and $y$ from $A$.
\item From ii. and Definition \ref{def:isp_3}, it follows that $\kwd{AND}(A, B)$ is correctly classified as $\mathbf{\underline{\overset{+}{x}\overset{+}{z}}}$.
\end{enumerate}
\end{description}
\end{enumerate}
\item[$\mathbf{\underline{\overset{+}{x}\overset{+}{z}}}$]
\begin{enumerate}
\item Independently on the type of the process block, the correctness of the row follows directly from Lemma \ref{l:c3a}.
\end{enumerate}
\end{description}
\item We have shown that, for each row the \kwd{SEQ} and \kwd{AND} aggregation between the states correctly aggregates in the described result
\item The correctness of column \kwd{XOR} follows directly from the preference order defined in the lattice shown in Figure \ref{fig:ad1s_c_latt}.
\item From 3., and 4., it follows that the Aggregations in Table \ref{tab:1r} are correct.
\end{enumerate}
\end{proof}

\begin{table}[ht!]
\centering
\begin{tabular}{|c|c|c|c|}
\hline
A & B & \kwd{SEQ}(A, B) & \kwd{AND}(A, B) \\ \hline
$\mathbf{\underline{\overset{0}{x}\overset{+}{z}}}$ & $\mathbf{\underline{\overset{0}{x}\overset{0}{z}}}$ & $\mathbf{\underline{\overset{0}{x}\overset{+}{z}}}$ & $\mathbf{\underline{\overset{0}{x}\overset{+}{z}}}$ \\ \hline
$\mathbf{\underline{\overset{0}{x}\overset{+}{z}}}$ & $\mathbf{\underline{\overset{-}{xz}}}$ & $\mathbf{\underline{\overset{0}{x}\overset{+}{z}}}$ & $\mathbf{\underline{\overset{0}{x}\overset{+}{z}}}$ \\ \hline
$\mathbf{\underline{\overset{0}{x}\overset{+}{z}}}$ & $\mathbf{\underline{\overset{+}{x}\overset{0}{z}}}$ & $\mathbf{\underline{\overset{0}{x}\overset{+}{z}}}$ / $\mathbf{\underline{\overset{+}{x}\overset{0}{z}}}$ & $\mathbf{\underline{\overset{+}{x}\overset{+}{z}}}$ \\ \hline
$\mathbf{\underline{\overset{0}{x}\overset{+}{z}}}$ & $\mathbf{\underline{\overset{+}{x}\overset{-}{z}}}$ & $\mathbf{\underline{\overset{0}{x}\overset{+}{z}}}$ / $\mathbf{\underline{\overset{+}{x}\overset{-}{z}}}$ & $\mathbf{\underline{\overset{+}{x}\overset{+}{z}}}$ \\ \hline
$\mathbf{\underline{\overset{0}{x}\overset{+}{z}}}$ & $\mathbf{\underline{\overset{0}{x}\overset{+}{z}}}$ & $\mathbf{\underline{\overset{0}{x}\overset{+}{z}}}$ & $\mathbf{\underline{\overset{0}{x}\overset{+}{z}}}$ \\ \hline
$\mathbf{\underline{\overset{0}{x}\overset{+}{z}}}$ & $\mathbf{\underline{\overset{-}{x}\overset{+}{z}}}$ & $\mathbf{\underline{\overset{0}{x}\overset{+}{z}}}$ & $\mathbf{\underline{\overset{0}{x}\overset{+}{z}}}$ \\ \hline
$\mathbf{\underline{\overset{0}{x}\overset{+}{z}}}$ & $\mathbf{\underline{\overset{+}{x}\overset{+}{z}}}$ & $\mathbf{\underline{\overset{+}{x}\overset{+}{z}}}$ & $\mathbf{\underline{\overset{+}{x}\overset{+}{z}}}$ \\ \hline
\end{tabular}
\caption{Aggregation Table from Interval Sub-Pattern Evaluation State}\label{tab:1'}
\end{table}

\begin{proof}[Table \ref{tab:1'}]
\begin{enumerate}
\item From Definition \ref{def:isp_1'}, and column \textbf{A} of the table, it follows that there exists an execution belonging to the process block containing a task having $z$ annotated, and no task on its left having $y$ annotated.
\item We prove by case each row of the table and each of the two different types of aggregation.

\begin{description}
\item[$\mathbf{\underline{\overset{0}{x}\overset{0}{z}}}$] Independently on the type of the process block, the correctness of the row follows directly from Lemma \ref{l:neutral_isp}.
\item[$\mathbf{\underline{\overset{-}{xz}}}$]
\begin{enumerate}
\item From Definition \ref{def:isp_0-}, and the classification of the process block in the column \textbf{B} of the table, it follows that each execution of \textbf{B} contains at least a task having $y$ annotated, and no other literals beneficial towards the pattern fulfilment.
\begin{description}
\item[SEQ]
\begin{enumerate}
\item From Definition \ref{def:ser}, it follows that the possible executions of a process block $\kwd{SEQ}(A, B)$ are the concatenation of an execution of $A$ and an execution of $B$.
\item From i., 1., and (a), it follows that there exists an execution in $\kwd{SEQ}(A, B)$ containing $z$ and no $y$ on its left, as each execution from $B$ is appended on the right.
\item From ii. and Definition \ref{def:isp_1'}, it follows that $\kwd{SEQ}(A, B)$ can be correctly classified as $\mathbf{\underline{\overset{0}{x}\overset{+}{z}}}$.
\end{enumerate}
\item[AND]
\begin{enumerate}
\item From Lemma \ref{l:commutativity}, it follows that $\kwd{AND}(A, B)$ is classified as the same class as the row $\mathbf{\underline{\overset{-}{xz}}}$ $\mathbf{\underline{\overset{0}{x}\overset{+}{z}}}$ in Table \ref{tab:0-}, which is correctly $\mathbf{\underline{\overset{0}{x}\overset{+}{z}}}$.
\end{enumerate}
\end{description}
\end{enumerate}
\item[$\mathbf{\underline{\overset{+}{x}\overset{0}{z}}}$]
\begin{enumerate}
\item From Definition \ref{def:isp_1}, and the classification of the process block in the column \textbf{B} of the table, it follows that there exists an execution belonging to the process block containing a task having $x$ annotated, and no task on its right having $y$ annotated.
\begin{description}
\item[SEQ]
\begin{enumerate}
\item From Definition \ref{def:ser}, it follows that the possible executions of a process block $\kwd{SEQ}(A, B)$ are the concatenation of an execution of $A$ and an execution of $B$.
\item From i., 1., and (a), it follows that there exists an execution in $\kwd{SEQ}(A, B)$ containing $z$ and a $x$ on its right, as each execution from $B$ is appended on the right. While the $x$ has no $y$ on its right, and $z$ has no $y$ on its left.
\item From ii. and Definition \ref{def:isp_1'}, it follows that $\kwd{SEQ}(A, B)$ can be correctly classified as $\mathbf{\underline{\overset{0}{x}\overset{+}{z}}}$.
\item From ii. and Definition \ref{def:isp_1}, it follows that $\kwd{SEQ}(A, B)$ can be correctly classified as $\mathbf{\underline{\overset{+}{x}\overset{0}{z}}}$.
\item From iii., iv. and the priority ordering in Figure \ref{fig:ad1s_c_latt}, it follows that $\kwd{SEQ}(A, B)$ is correctly classified as both $\mathbf{\underline{\overset{+}{x}\overset{0}{z}}}$ and $\mathbf{\underline{\overset{0}{x}\overset{+}{z}}}$.
\end{enumerate}
\item[AND]
\begin{enumerate}
\item From Lemma \ref{l:commutativity}, it follows that $\kwd{AND}(A, B)$ is classified as the same class as the row $\mathbf{\underline{\overset{+}{x}\overset{0}{z}}}$ $\mathbf{\underline{\overset{0}{x}\overset{+}{z}}}$ in Table \ref{tab:1}, which is correctly $\mathbf{\underline{\overset{+}{x}\overset{+}{z}}}$.
\end{enumerate}
\end{description}
\end{enumerate}
\item[$\mathbf{\underline{\overset{+}{x}\overset{-}{z}}}$]
\begin{enumerate}
\item From Definition \ref{def:isp_1r}, and the classification of the process block in the column \textbf{B} of the table, it follows that there exists an execution belonging to the process block containing a task having $x$ annotated, and a task on its right having $y$ annotated.
\begin{description}
\item[SEQ]
\begin{enumerate}
\item From Definition \ref{def:ser}, it follows that the possible executions of a process block $\kwd{SEQ}(A, B)$ are the concatenation of an execution of $A$ and an execution of $B$.
\item From i., 1., and (a), it follows that there exists an execution in $\kwd{SEQ}(A, B)$ containing $z$ and a $x$ on its right, which in turn has a $y$ on its right, as each execution from $B$ is appended on the right.
\item From ii. and Definition \ref{def:isp_1'}, it follows that $\kwd{SEQ}(A, B)$ can be correctly classified as $\mathbf{\underline{\overset{0}{x}\overset{+}{z}}}$.
\item From ii. and Definition \ref{def:isp_1r}, it follows that $\kwd{SEQ}(A, B)$ can be correctly classified as $\mathbf{\underline{\overset{+}{x}\overset{-}{z}}}$.
\item From iii., iv. and the priority ordering in Figure \ref{fig:ad1s_c_latt}, it follows that $\kwd{SEQ}(A, B)$ is correctly classified as both $\mathbf{\underline{\overset{+}{x}\overset{-}{z}}}$ and $\mathbf{\underline{\overset{0}{x}\overset{+}{z}}}$.
\end{enumerate}
\item[AND]
\begin{enumerate}
\item From Lemma \ref{l:commutativity}, it follows that $\kwd{AND}(A, B)$ is classified as the same class as the row $\mathbf{\underline{\overset{+}{x}\overset{-}{z}}}$ $\mathbf{\underline{\overset{0}{x}\overset{+}{z}}}$ in Table \ref{tab:1r}, which is correctly $\mathbf{\underline{\overset{+}{x}\overset{+}{z}}}$.
\end{enumerate}
\end{description}
\end{enumerate}
\item[$\mathbf{\underline{\overset{0}{x}\overset{+}{z}}}$]
\begin{enumerate}
\item From Definition \ref{def:isp_1'}, and the classification of the process block in the column \textbf{B} of the table, it follows that there exists an execution belonging to the process block containing a task having $z$ annotated, and no tasks on its left having $y$ annotated.
\begin{description}
\item[SEQ]
\begin{enumerate}
\item From Definition \ref{def:ser}, it follows that the possible executions of a process block $\kwd{SEQ}(A, B)$ are the concatenation of an execution of $A$ and an execution of $B$.
\item From i., 1., and (a), it follows that there exists an execution in $\kwd{SEQ}(A, B)$ containing $z$ and no $y$ on its left.
\item From ii. and Definition \ref{def:isp_1'}, it follows that $\kwd{SEQ}(A, B)$ can be correctly classified as $\mathbf{\underline{\overset{0}{x}\overset{+}{z}}}$.
\end{enumerate}
\item[AND]
\begin{enumerate}
\item From Definition \ref{def:ser}, it follows that each execution resulting from $\kwd{AND}(A, B)$ consists of an execution of $A$ interleaved with an execution of $B$.
\item From i., 1., and (a), it follows that the aggregated execution does not contain a $x$.
\item From Lemma \ref{lem:inclusivity} and the result from the \kwd{SEQ} column, it follows that $\kwd{AND}(A,$ $B)$ is at least classified as $\mathbf{\underline{\overset{0}{x}\overset{+}{z}}}$.
\item From ii., Definition \ref{def:isp_1'}, Definition \ref{def:isp_3} and the preference ordering of the classes in Figure \ref{fig:ad1s_c_latt}, it follows that $\kwd{AND}(A, B)$ is classified as $\mathbf{\underline{\overset{0}{x}\overset{+}{z}}}$ at best as no $x$ is contained.
\item From iii. and iv., it follows that $\kwd{SEQ}(A, B)$ can be correctly classified as $\mathbf{\underline{\overset{0}{x}\overset{+}{z}}}$.
\end{enumerate}
\end{description}
\end{enumerate}
\item[$\mathbf{\underline{\overset{-}{x}\overset{+}{z}}}$]
\begin{enumerate}
\item From Definition \ref{def:isp_1'l}, and the classification of the process block in the column \textbf{B} of the table, it follows that there exists an execution belonging to the process block containing a task having $z$ annotated, and having another task on its left having $y$ annotated.
\begin{description}
\item[SEQ]
\begin{enumerate}
\item From Definition \ref{def:ser}, it follows that the possible executions of a process block $\kwd{SEQ}(A, B)$ are the concatenation of an execution of $A$ and an execution of $B$.
\item From i., 1., and (a), it follows that there exists an execution in $\kwd{SEQ}(A, B)$ containing $z$ and no $y$ on its left, since the execution containing $y$ is appended on the right.
\item From ii. and Definition \ref{def:isp_1'}, it follows that $\kwd{SEQ}(A, B)$ can be correctly classified as $\mathbf{\underline{\overset{0}{x}\overset{+}{z}}}$.
\end{enumerate}
\item[AND]
\begin{enumerate}
\item From Definition \ref{def:ser}, it follows that each execution resulting from $\kwd{AND}(A, B)$ consists of an execution of $A$ interleaved with an execution of $B$.
\item From i., 1., and (a), it follows that the aggregated execution does not contain a $x$.
\item From Lemma \ref{lem:inclusivity} and the result from the \kwd{SEQ} column, it follows that $\kwd{AND}(A, B)$ is at least classified as $\mathbf{\underline{\overset{0}{x}\overset{+}{z}}}$.
\item From ii., Definition \ref{def:isp_1'}, Definition \ref{def:isp_3} and the preference ordering of the classes in Figure \ref{fig:ad1s_c_latt}, it follows that $\kwd{AND}(A, B)$ is classified as $\mathbf{\underline{\overset{0}{x}\overset{+}{z}}}$ at best as no $x$ is contained.
\item From iii. and iv., it follows that $\kwd{SEQ}(A, B)$ can be correctly classified as $\mathbf{\underline{\overset{0}{x}\overset{+}{z}}}$.
\end{enumerate}
\end{description}
\end{enumerate}
\item[$\mathbf{\underline{\overset{+}{x}\overset{+}{z}}}$]
\begin{enumerate}
\item Independently on the type of the process block, the correctness of the row follows directly from Lemma \ref{l:c3a}.
\end{enumerate}
\end{description}
\item We have shown that, for each row the \kwd{SEQ} and \kwd{AND} aggregation between the states correctly aggregates in the described result
\item The correctness of column \kwd{XOR} follows directly from the preference order defined in the lattice shown in Figure \ref{fig:ad1s_c_latt}.
\item From 3., and 4., it follows that the Aggregations in Table \ref{tab:1'} are correct.
\end{enumerate}
\end{proof}

\begin{table}[ht!]
\centering
\begin{tabular}{|c|c|c|c|}
\hline
A & B & \kwd{SEQ}(A, B) & \kwd{AND}(A, B) \\ \hline
$\mathbf{\underline{\overset{-}{x}\overset{+}{z}}}$ & $\mathbf{\underline{\overset{0}{x}\overset{0}{z}}}$ & $\mathbf{\underline{\overset{-}{x}\overset{+}{z}}}$ & $\mathbf{\underline{\overset{-}{x}\overset{+}{z}}}$ \\ \hline
$\mathbf{\underline{\overset{-}{x}\overset{+}{z}}}$ & $\mathbf{\underline{\overset{-}{xz}}}$ & $\mathbf{\underline{\overset{-}{x}\overset{+}{z}}}$ & $\mathbf{\underline{\overset{-}{x}\overset{+}{z}}}$ \\ \hline
$\mathbf{\underline{\overset{-}{x}\overset{+}{z}}}$ & $\mathbf{\underline{\overset{+}{x}\overset{0}{z}}}$ & $\mathbf{\underline{\overset{-}{x}\overset{+}{z}}}$ / $\mathbf{\underline{\overset{+}{x}\overset{0}{z}}}$ & $\mathbf{\underline{\overset{+}{x}\overset{+}{z}}}$ \\ \hline
$\mathbf{\underline{\overset{-}{x}\overset{+}{z}}}$ & $\mathbf{\underline{\overset{+}{x}\overset{-}{z}}}$ & $\mathbf{\underline{\overset{-}{x}\overset{+}{z}}}$ / $\mathbf{\underline{\overset{+}{x}\overset{-}{z}}}$ & $\mathbf{\underline{\overset{+}{x}\overset{+}{z}}}$ \\ \hline
$\mathbf{\underline{\overset{-}{x}\overset{+}{z}}}$ & $\mathbf{\underline{\overset{0}{x}\overset{+}{z}}}$ & $\mathbf{\underline{\overset{-}{x}\overset{+}{z}}}$ & $\mathbf{\underline{\overset{0}{x}\overset{+}{z}}}$ \\ \hline
$\mathbf{\underline{\overset{-}{x}\overset{+}{z}}}$ & $\mathbf{\underline{\overset{-}{x}\overset{+}{z}}}$ & $\mathbf{\underline{\overset{-}{x}\overset{+}{z}}}$ & $\mathbf{\underline{\overset{-}{x}\overset{+}{z}}}$ \\ \hline
$\mathbf{\underline{\overset{-}{x}\overset{+}{z}}}$ & $\mathbf{\underline{\overset{+}{x}\overset{+}{z}}}$ & $\mathbf{\underline{\overset{+}{x}\overset{+}{z}}}$ & $\mathbf{\underline{\overset{+}{x}\overset{+}{z}}}$ \\ \hline
\end{tabular}
\caption{Aggregation Table from Interval Sub-Pattern Evaluation State}\label{tab:1'l}
\end{table}

\begin{proof}[Table \ref{tab:1'l}]
\begin{enumerate}
\item From Definition \ref{def:isp_1'l}, and column \textbf{A} of the table, it follows that there exists an execution belonging to the process block containing a task having $z$ annotated, and a task on its left having $y$ annotated.
\item We prove by case each row of the table and each of the two different types of aggregation.

\begin{description}
\item[$\mathbf{\underline{\overset{0}{x}\overset{0}{z}}}$] Independently on the type of the process block, the correctness of the row follows directly from Lemma \ref{l:neutral_isp}.
\item[$\mathbf{\underline{\overset{-}{xz}}}$]
\begin{enumerate}
\item From Definition \ref{def:isp_0-}, and the classification of the process block in the column \textbf{B} of the table, it follows that each execution of \textbf{B} contains at least a task having $y$ annotated, and no other literals beneficial towards the pattern fulfilment.
\begin{description}
\item[SEQ]
\begin{enumerate}
\item From Definition \ref{def:ser}, it follows that the possible executions of a process block $\kwd{SEQ}(A, B)$ are the concatenation of an execution of $A$ and an execution of $B$.
\item From i., 1., and (a), it follows that there exists an execution in $\kwd{SEQ}(A, B)$ containing $z$ and a $y$ on its left.
\item From ii. and Definition \ref{def:isp_1'l}, it follows that $\kwd{SEQ}(A, B)$ can be correctly classified as $\mathbf{\underline{\overset{-}{x}\overset{+}{z}}}$.
\end{enumerate}
\item[AND]
\begin{enumerate}
\item From Lemma \ref{l:commutativity}, it follows that $\kwd{AND}(A, B)$ is classified as the same class as the row $\mathbf{\underline{\overset{-}{xz}}}$ $\mathbf{\underline{\overset{-}{x}\overset{+}{z}}}$ in Table \ref{tab:0-}, which is correctly $\mathbf{\underline{\overset{-}{x}\overset{+}{z}}}$.
\end{enumerate}
\end{description}
\end{enumerate}
\item[$\mathbf{\underline{\overset{+}{x}\overset{0}{z}}}$]
\begin{enumerate}
\item From Definition \ref{def:isp_1}, and the classification of the process block in the column \textbf{B} of the table, it follows that there exists an execution belonging to the process block containing a task having $x$ annotated, and no task on its right having $y$ annotated.
\begin{description}
\item[SEQ]
\begin{enumerate}
\item From Definition \ref{def:ser}, it follows that the possible executions of a process block $\kwd{SEQ}(A, B)$ are the concatenation of an execution of $A$ and an execution of $B$.
\item From i., 1., and (a), it follows that there exists an execution in $\kwd{SEQ}(A, B)$ containing $z$ and a $x$ on its right, as each execution from $B$ is appended on the right. While the $x$ has no $y$ on its right, and $z$ has a $y$ on its left.
\item From ii. and Definition \ref{def:isp_1'l}, it follows that $\kwd{SEQ}(A, B)$ can be correctly classified as $\mathbf{\underline{\overset{-}{x}\overset{+}{z}}}$.
\item From ii. and Definition \ref{def:isp_1}, it follows that $\kwd{SEQ}(A, B)$ can be correctly classified as $\mathbf{\underline{\overset{+}{x}\overset{0}{z}}}$.
\item From iii., iv. and the priority ordering in Figure \ref{fig:ad1s_c_latt}, it follows that $\kwd{SEQ}(A, B)$ is correctly classified as both $\mathbf{\underline{\overset{+}{x}\overset{0}{z}}}$ and $\mathbf{\underline{\overset{-}{x}\overset{+}{z}}}$.
\end{enumerate}
\item[AND]
\begin{enumerate}
\item From Lemma \ref{l:commutativity}, it follows that $\kwd{AND}(A, B)$ is classified as the same class as the row $\mathbf{\underline{\overset{+}{x}\overset{0}{z}}}$ $\mathbf{\underline{\overset{-}{x}\overset{+}{z}}}$ in Table \ref{tab:1}, which is correctly $\mathbf{\underline{\overset{+}{x}\overset{+}{z}}}$.
\end{enumerate}
\end{description}
\end{enumerate}
\item[$\mathbf{\underline{\overset{+}{x}\overset{-}{z}}}$]
\begin{enumerate}
\item From Definition \ref{def:isp_1r}, and the classification of the process block in the column \textbf{B} of the table, it follows that there exists an execution belonging to the process block containing a task having $x$ annotated, and a task on its right having $y$ annotated.
\begin{description}
\item[SEQ]
\begin{enumerate}
\item From Definition \ref{def:ser}, it follows that the possible executions of a process block $\kwd{SEQ}(A, B)$ are the concatenation of an execution of $A$ and an execution of $B$.
\item From i., 1., and (a), it follows that there exists an execution in $\kwd{SEQ}(A, B)$ containing $z$ and a $x$ on its right, as each execution from $B$ is appended on the right. While the $x$ has a $y$ on its right, and $z$ has a $y$ on its left.
\item From ii. and Definition \ref{def:isp_1'l}, it follows that $\kwd{SEQ}(A, B)$ can be correctly classified as $\mathbf{\underline{\overset{-}{x}\overset{+}{z}}}$.
\item From ii. and Definition \ref{def:isp_1r}, it follows that $\kwd{SEQ}(A, B)$ can be correctly classified as $\mathbf{\underline{\overset{+}{x}\overset{-}{z}}}$.
\item From iii., iv. and the priority ordering in Figure \ref{fig:ad1s_c_latt}, it follows that $\kwd{SEQ}(A, B)$ is correctly classified as both $\mathbf{\underline{\overset{+}{x}\overset{-}{z}}}$ and $\mathbf{\underline{\overset{-}{x}\overset{+}{z}}}$.
\end{enumerate}
\item[AND]
\begin{enumerate}
\item From Lemma \ref{l:commutativity}, it follows that $\kwd{AND}(A, B)$ is classified as the same class as the row $\mathbf{\underline{\overset{+}{x}\overset{-}{z}}}$ $\mathbf{\underline{\overset{-}{x}\overset{+}{z}}}$ in Table \ref{tab:1r}, which is correctly $\mathbf{\underline{\overset{+}{x}\overset{+}{z}}}$.
\end{enumerate}
\end{description}
\end{enumerate}
\item[$\mathbf{\underline{\overset{0}{x}\overset{+}{z}}}$]
\begin{enumerate}
\item From Definition \ref{def:isp_1'}, and the classification of the process block in the column \textbf{B} of the table, it follows that there exists an execution belonging to the process block containing a task having $z$ annotated, and no tasks on its left having $y$ annotated.
\begin{description}
\item[SEQ]
\begin{enumerate}
\item From Definition \ref{def:ser}, it follows that the possible executions of a process block $\kwd{SEQ}(A, B)$ are the concatenation of an execution of $A$ and an execution of $B$.
\item From i., 1., and (a), it follows that there exists an execution in $\kwd{SEQ}(A, B)$ containing $z$ and a $y$ on its left.
\item From ii. and Definition \ref{def:isp_1'l}, it follows that $\kwd{SEQ}(A, B)$ can be correctly classified as $\mathbf{\underline{\overset{-}{x}\overset{+}{z}}}$.
\end{enumerate}
\item[AND]
\begin{enumerate}
\item From Lemma \ref{l:commutativity}, it follows that $\kwd{AND}(A, B)$ is classified as the same class as the row $\mathbf{\underline{\overset{0}{x}\overset{+}{z}}}$ $\mathbf{\underline{\overset{-}{x}\overset{+}{z}}}$ in Table \ref{tab:1'}, which is correctly $\mathbf{\underline{\overset{0}{x}\overset{+}{z}}}$.
\end{enumerate}
\end{description}
\end{enumerate}
\item[$\mathbf{\underline{\overset{-}{x}\overset{+}{z}}}$]
\begin{enumerate}
\item From Definition \ref{def:isp_1'l}, and the classification of the process block in the column \textbf{B} of the table, it follows that there exists an execution belonging to the process block containing a task having $z$ annotated, and having another task on its left having $y$ annotated.
\begin{description}
\item[SEQ]
\begin{enumerate}
\item From Definition \ref{def:ser}, it follows that the possible executions of a process block $\kwd{SEQ}(A, B)$ are the concatenation of an execution of $A$ and an execution of $B$.
\item From i., 1., and (a), it follows that there exists an execution in $\kwd{SEQ}(A, B)$ containing $z$ and a $y$ on its left.
\item From ii. and Definition \ref{def:isp_1'l}, it follows that $\kwd{SEQ}(A, B)$ can be correctly classified as $\mathbf{\underline{\overset{-}{x}\overset{+}{z}}}$.
\end{enumerate}
\item[AND]
\begin{enumerate}
\item From Definition \ref{def:ser}, it follows that each execution resulting from $\kwd{AND}(A, B)$ consists of an execution of $A$ interleaved with an execution of $B$.
\item From i., 1., and (a), it follows that the aggregated execution does not contain a $x$.
\item From i., 1., and (a), it follows that the aggregated execution always contains a $z$ with a $y$ on its left.
\item From Lemma \ref{lem:inclusivity} and the result from the \kwd{SEQ} column, it follows that $\kwd{AND}(A, B)$ is at least classified as $\mathbf{\underline{\overset{-}{x}\overset{+}{z}}}$.
\item From ii., iii., Definition \ref{def:isp_1'l}, Definition \ref{def:isp_1'}, Definition \ref{def:isp_3} and the preference ordering of the classes in Figure \ref{fig:ad1s_c_latt}, it follows that $\kwd{AND}(A, B)$ is classified as $\mathbf{\underline{\overset{0}{x}\overset{+}{z}}}$ at best as no $x$ is contained in the aggregated execution and $y$ cannot be removed from the left of $z$.
\item From iv. and v., it follows that $\kwd{SEQ}(A, B)$ can be correctly classified as $\mathbf{\underline{\overset{-}{x}\overset{+}{z}}}$.
\end{enumerate}
\end{description}
\end{enumerate}
\item[$\mathbf{\underline{\overset{+}{x}\overset{+}{z}}}$]
\begin{enumerate}
\item Independently on the type of the process block, the correctness of the row follows directly from Lemma \ref{l:c3a}.
\end{enumerate}
\end{description}
\item We have shown that, for each row the \kwd{SEQ} and \kwd{AND} aggregation between the states correctly aggregates in the described result
\item The correctness of column \kwd{XOR} follows directly from the preference order defined in the lattice shown in Figure \ref{fig:ad1s_c_latt}.
\item From 3., and 4., it follows that the Aggregations in Table \ref{tab:1'l} are correct.
\end{enumerate}
\end{proof}

\begin{table}[ht!]
\centering
\begin{tabular}{|c|c|c|c|}
\hline
A & B & \kwd{SEQ}(A, B) & \kwd{AND}(A, B) \\ \hline
$\mathbf{\underline{\overset{+}{x}\overset{+}{z}}}$ & $\mathbf{\underline{\overset{0}{x}\overset{0}{z}}}$ & $\mathbf{\underline{\overset{+}{x}\overset{+}{z}}}$ & $\mathbf{\underline{\overset{+}{x}\overset{+}{z}}}$ \\ \hline
$\mathbf{\underline{\overset{+}{x}\overset{+}{z}}}$ & $\mathbf{\underline{\overset{-}{xz}}}$ & $\mathbf{\underline{\overset{+}{x}\overset{+}{z}}}$ & $\mathbf{\underline{\overset{+}{x}\overset{+}{z}}}$ \\ \hline
$\mathbf{\underline{\overset{+}{x}\overset{+}{z}}}$ & $\mathbf{\underline{\overset{+}{x}\overset{0}{z}}}$ & $\mathbf{\underline{\overset{+}{x}\overset{+}{z}}}$ & $\mathbf{\underline{\overset{+}{x}\overset{+}{z}}}$ \\ \hline
$\mathbf{\underline{\overset{+}{x}\overset{+}{z}}}$ & $\mathbf{\underline{\overset{+}{x}\overset{-}{z}}}$ & $\mathbf{\underline{\overset{+}{x}\overset{+}{z}}}$ & $\mathbf{\underline{\overset{+}{x}\overset{+}{z}}}$ \\ \hline
$\mathbf{\underline{\overset{+}{x}\overset{+}{z}}}$ & $\mathbf{\underline{\overset{0}{x}\overset{+}{z}}}$ & $\mathbf{\underline{\overset{+}{x}\overset{+}{z}}}$ & $\mathbf{\underline{\overset{+}{x}\overset{+}{z}}}$ \\ \hline
$\mathbf{\underline{\overset{+}{x}\overset{+}{z}}}$ & $\mathbf{\underline{\overset{-}{x}\overset{+}{z}}}$ & $\mathbf{\underline{\overset{+}{x}\overset{+}{z}}}$ & $\mathbf{\underline{\overset{+}{x}\overset{+}{z}}}$ \\ \hline
$\mathbf{\underline{\overset{+}{x}\overset{+}{z}}}$ & $\mathbf{\underline{\overset{+}{x}\overset{+}{z}}}$ & $\mathbf{\underline{\overset{+}{x}\overset{+}{z}}}$ & $\mathbf{\underline{\overset{+}{x}\overset{+}{z}}}$ \\ \hline
\end{tabular}
\caption{Aggregation Table from Interval Sub-Pattern Evaluation State}\label{tab:3}
\end{table}

\begin{proof}[Table \ref{tab:3}]
As each of the aggregations shown in Table \ref{tab:3} contains at least an element belonging to the aggregation class $\mathbf{\underline{\overset{+}{x}\overset{+}{z}}}$, the correctness of the columns \kwd{SEQ} and \kwd{AND} follows directly from Lemma \ref{l:c3a}.
\end{proof}
\newpage
\subsection{Interval Overnode Pattern}

\begin{definition}[Interval Overnode Pattern $\mathbf{\overset{+}{x}t\overset{+}{z}}$]\label{def:iop_+t+}
Fulfilment class: there exists an execution belonging to the process block containing the \emph{trigger leaf}, a $x$ on its left, a $z$ on the \emph{trigger leaf} right, and no $y$ between $x$ and $z$.

\noindent\textbf{Formally}: 

\noindent Given a process block $B$ belongs to this class if and only if:
\begin{itemize}
\item $\exists \exe \in \Exe{B}$ such that:
\begin{itemize}
\item $\exists x, t_{t}, z \in \exe$ such that $x \preceq t_{t} \preceq z$, and
\begin{itemize}
\item $\not \exists y \in \exe$ such that $x \preceq y \preceq z$.
\end{itemize}
\end{itemize}
\end{itemize}
\end{definition}

\begin{definition}[Interval Overnode Pattern $\mathbf{\overset{+}{x}t\overset{0}{z}}$]\label{def:iop_+t0}
Left class: there exists an execution belonging to the process block containing the \emph{trigger leaf}, a $x$ on its left, and no $y$ between $x$ and the \emph{trigger leaf}. Moreover there are no $z$ or $y$ on the right of the \emph{trigger leaf}.

\noindent\textbf{Formally}: 

\noindent Given a process block $B$ belongs to this class if and only if:
\begin{itemize}
\item $\exists \exe \in \Exe{B}$ such that:
\begin{itemize}
\item $\exists x, t_{t}\in \exe$ such that $x \preceq t_{t}$, and
\begin{itemize}
\item $\not \exists y \in \exe$ such that $x \preceq y \preceq t_t$, and
\item $\not \exists z \in \exe$ such that $t_{t} \preceq z$, and
\item $\not \exists y \in \exe$ such that $t_{t} \preceq y$.
\end{itemize}
\end{itemize}
\end{itemize}
\end{definition}

\begin{definition}[Interval Overnode Pattern $\mathbf{\overset{0}{x}t\overset{+}{z}}$]\label{def:iop_0t+}
Right class: there exists an execution belonging to the process block containing the \emph{trigger leaf}, a $z$ on its right, and no $y$ between the \emph{trigger leaf} and $z$. Moreover there are no $x$ or $y$ on the left of the \emph{trigger leaf}.

\noindent\textbf{Formally}: 

\noindent Given a process block $B$ belongs to this class if and only if:
\begin{itemize}
\item $\exists \exe \in \Exe{B}$ such that:
\begin{itemize}
\item $\exists z, t_{t}\in \exe$ such that $t_{t} \preceq z$, and
\begin{itemize}
\item $\not \exists y \in \exe$ such that $t_{t} \preceq y \preceq z$, and
\item $\not \exists y \in \exe$ such that $y \preceq t_{t}$.
\end{itemize}
\end{itemize}
\end{itemize}
\end{definition}

\begin{definition}[Interval Overnode Pattern $\mathbf{\overset{+}{x}t\overset{-}{z}}$]\label{def:iop_+t-}
Left right blocked class: there exists an execution belonging to the process block containing the \emph{trigger leaf}, a $x$ on its left, and no $y$ between $x$ and the \emph{trigger leaf}. Moreover there is $y$ on the right of the \emph{trigger leaf}.

\noindent\textbf{Formally}: 

\noindent Given a process block $B$ belongs to this class if and only if:
\begin{itemize}
\item $\exists \exe \in \Exe{B}$ such that:
\begin{itemize}
\item $\exists x, t_{t}\in \exe$ such that $x \preceq t_{t}$, and
\begin{itemize}
\item $\not \exists y \in \exe$ such that $x \preceq y \preceq t_t$, and
\item $\exists y \in \exe$ such that $t_{t} \preceq y$.
\end{itemize}
\end{itemize}
\end{itemize}
\end{definition}

\begin{definition}[Interval Overnode Pattern $\mathbf{\overset{0}{x}t\overset{0}{z}}$]\label{def:iop_0t0}
Neutral class: there exists an execution belonging to the process block containing the \emph{trigger leaf}, there is no $x$ on the left of the \emph{trigger leaf}, and there is no $z$ on the right of the \emph{trigger leaf}. Moreover there is $y$ in the execution.

\noindent\textbf{Formally}: 

\noindent Given a process block $B$ belongs to this class if and only if:
\begin{itemize}
\item $\exists \exe \in \Exe{B}$ such that:
\begin{itemize}
\item $\exists t_{t} \in \exe$ such that
\begin{itemize}
\item $\not \exists x \in \exe$ such that $x \preceq t_t$,
\item $\not \exists z \in \exe$ such that $t_{t} \preceq z$, and
\item $\not \exists y \in \exe$.
\end{itemize}
\end{itemize}
\end{itemize}
\end{definition}

\begin{definition}[Interval Overnode Pattern $\mathbf{\overset{-}{x}t\overset{+}{z}}$]\label{def:iop_-t+}
Right left blocked class: there exists an execution belonging to the process block containing the \emph{trigger leaf}, a $z$ on its right, and no $y$ between the \emph{trigger leaf} and $z$. Moreover there is a $y$ on the left of the \emph{trigger leaf}.

\noindent\textbf{Formally}: 

\noindent Given a process block $B$ belongs to this class if and only if:
\begin{itemize}
\item $\exists \exe \in \Exe{B}$ such that:
\begin{itemize}
\item $\exists z, t_{t}\in \exe$ such that $t_{t} \preceq z$, and
\begin{itemize}
\item $\not \exists y \in \exe$ such that $t_{t} \preceq y \preceq z$, and
\item $\exists y \in \exe$ such that $y \preceq t_{t}$.
\end{itemize}
\end{itemize}
\end{itemize}
\end{definition}

\begin{definition}[Interval Overnode Pattern $\mathbf{\overset{0}{x}t\overset{-}{z}}$]\label{def:iop_0t-}
Right blocked class: there exists an execution belonging to the process block containing the \emph{trigger leaf}, and there is a $y$ on its right.

\noindent\textbf{Formally}: 

\noindent Given a process block $B$ belongs to this class if and only if:
\begin{itemize}
\item $\exists \exe \in \Exe{B}$ such that:
\begin{itemize}
\item $\exists y \in \exe$ such that $t_{t} \preceq y$.
\end{itemize}
\end{itemize}
\end{definition}

\begin{definition}[Interval Overnode Pattern $\mathbf{\overset{-}{x}t\overset{0}{z}}$]\label{def:iop_-t0}
Left blocked class: there exists an execution belonging to the process block containing the \emph{trigger leaf}, and there is a $y$ on its left.

\noindent\textbf{Formally}: 

\noindent Given a process block $B$ belongs to this class if and only if:
\begin{itemize}
\item $\exists \exe \in \Exe{B}$ such that:
\begin{itemize}
\item $\exists y \in \exe$ such that $y \preceq t_{t}$.
\end{itemize}
\end{itemize}
\end{definition}

\begin{definition}[Interval Overnode Pattern $\mathbf{\overset{-}{x}t\overset{-}{z}}$]\label{def:iop_-t-}
Both blocked class: for each execution belonging to the process block, it contains the \emph{trigger leaf}, and there are $y$s on both of its right and left side.

\noindent\textbf{Formally}: 

\noindent Given a process block $B$ belongs to this class if and only if:
\begin{itemize}
\item $\exists \exe \in \Exe{B}$ such that:
\begin{itemize}
\item $\exists y \in \exe$ such that $y \preceq t_{t}$, and
\item $\exists y \in \exe$ such that $t_{t} \preceq y$.
\end{itemize}
\end{itemize}
\end{definition}

\begin{theorem}[Classification Completeness for  Interval Overnode Pattern]
The set of possible evaluations of  Interval Overnode Pattern is completely covered by the provided set of classifications.
\end{theorem}

\begin{proof}
\begin{enumerate}
\item The two sides on the left ($x$) and on the right ($z$) of $t$ in  Interval Overnode Pattern allow 3 possible evaluations each: whether the partial requirement is satisfied, failing, or in a neutral state.
\item The two sides are independent. Considering that the only shared element between the two sides is $t$, even if the undesired element $y$ is ordinally at the same position as $t$, meaning that $t \leq y$ and $t \geq y$, then the consequence is that both sides are falsified.
\item From 1. and the independence from 2., it follows that the amount of possible combinations in  Interval Overnode Pattern is $3^2$, hence 9 possible combinations. Moreover the case from 2. is also covered within the 9 combinations.
\item Therefore, every possible evaluation of  Interval Overnode Pattern is captured by one of the 9 provided classifications.
\end{enumerate}
\end{proof}

\subsubsection{Lemmas}

\begin{lemma}[Aggregation Neutral Class]\label{l:neutral_iop}
Given a process block $A$, assigned to the evaluation class $\mathbf{\overset{0}{x}t\overset{0}{z}}$, and another process block $B$, assigned to any of the available evaluation classes. Let $C$ be a process block having $A$ and $B$ as its sub-blocks, then the evaluation class of $C$ is the same class as the process block $B$.
\end{lemma}

\begin{proof}
This proof follows closely the proof for Lemma \ref{l:neutral}
\end{proof}

\begin{lemma}[$\mathbf{\underline{\overset{+}{x}\overset{+}{z}}}$ super class]\label{l:3super}
Independently from the classification of the overnode's overnode child, when aggregated with the overnode's undernode class $\mathbf{\underline{\overset{+}{x}\overset{+}{z}}}$, it always results in $\mathbf{\overset{+}{x}t\overset{+}{z}}$.
\end{lemma}

\begin{proof}
The correctness of Lemma \ref{l:3super} follows from Definition \ref{def:ser}, where the execution of a process block classified as $\mathbf{\underline{\overset{+}{x}\overset{+}{z}}}$, can be interleaved with another execution from whatever classification\footnote{Notice that the trigger task must not contain a $y$, otherwise such interleaving is not possible anymore (this is ensure by one of the premature exit condition in the stem algorithm).} and result in $\mathbf{\overset{+}{x}t\overset{+}{z}}$. According to Definition \ref{def:isp_3}, $\mathbf{\underline{\overset{+}{x}\overset{+}{z}}}$ contains an $x$ with a $z$ on its right and no $y$ inbetween. Thus independently than what appears around a trigger task $t$ and its classification, such task can be plugged within the interval identified in the undernode classified execution, while leaving outside the interval the remainder of the elements in the overnode, making them virtually useless for the classification, as independently from them, the resulting block clearly allows an execution satisfying the criteria described in Definition \ref{def:iop_+t+}, allowing the result to be always classified as $\mathbf{\overset{+}{x}t\overset{+}{z}}$, since is the most preferred class as show in Figure \ref{fig:ad1s_stem_ev}.
\end{proof}

\begin{lemma}[Top Classification]\label{l:topclass}
Independently from the classification of the overnode's undernode children, when aggregated with the overnode's overnode child classification $\mathbf{\overset{+}{x}t\overset{+}{z}}$, it always results in $\mathbf{\overset{+}{x}t\overset{+}{z}}$.
\end{lemma}

\begin{proof}
Considering Definition \ref{def:ser} and that the aggregations considered refer to \kwd{AND} blocks as the way executions from different sub blocks are merged together, considering that one of the sub-blocks is classified as $\mathbf{\overset{+}{x}t\overset{+}{z}}$, the executions relative to the other sub-block(s) can be appended before and / or after the execution in the block classified $\mathbf{\overset{+}{x}t\overset{+}{z}}$, hence without disrupting the properties allowing the aggregation to be classified as $\mathbf{\overset{+}{x}t\overset{+}{z}}$ in accordance to Definition \ref{def:iop_+t+}. Finally, as the preference order shown in Figure \ref{fig:ad1s_stem_ev} illustrates that $\mathbf{\overset{+}{x}t\overset{+}{z}}$ is the best classification achievable, then we can state that Lemma \ref{l:topclass} is correct.
\end{proof}

\begin{lemma}[Hide The Bubu]\label{l:htb}
Given a overnode's overnode child \emph{not} classified as $\mathbf{\overset{0}{x}t\overset{0}{z}}$ and a overnode's undernode children classified as $\mathbf{\underline{\overset{-}{xz}}}$, their aggregation always corresponds to the classification of the overnode's overnode child. In the case Given a overnode's overnode child \emph{is} classified as $\mathbf{\overset{0}{x}t\overset{0}{z}}$, then the aggregation results in both: $\mathbf{\overset{-}{x}t\overset{0}{z}}$ and $\mathbf{\overset{0}{x}t\overset{-}{z}}$.
\end{lemma}

\begin{proof}
Considering the case where the overnode's overnode child is classified as $\mathbf{\overset{0}{x}t\overset{0}{z}}$, the $y$ from the execution of the block classified as $\mathbf{\underline{\overset{-}{xz}}}$ can be put on the left or the right of the execution allowing to classify overnode's overnode child as $\mathbf{\overset{0}{x}t\overset{0}{z}}$, which would then allow to classify it either as $\mathbf{\overset{-}{x}t\overset{0}{z}}$ or $\mathbf{\overset{0}{x}t\overset{+}{z}}$. Notice that the resulting classification are strictly worse than the starting one.

In the opposite case, where the overnode's overnode child is not classified as $\mathbf{\overset{0}{x}t\overset{0}{z}}$, considering that the $y$ from the other block's executions cannot make any classification better, to avoid making the current classification worse it is sufficient to place the $y$ on the left side of an $x$ on the left of a $t$, on the right of a $z$ on the right of a $t$ or on the left / right of a $y$ left / right of a $t$. In such a way, the classification of the side where the newly set $y$ would be dominated by the pre-existing elements, which would not make that side of the classification to change. Which overall means that the whole classification remains the same as the overnode's overnode child.
\end{proof}

\begin{lemma}[Not Worse]\label{l:nw}
Given a overnode's overnode child \emph{not} classified as $\mathbf{\overset{0}{x}t\overset{0}{z}}$ and a overnode's undernode children classified either as $\mathbf{\underline{\overset{+}{x}\overset{-}{z}}}$ or $\mathbf{\underline{\overset{-}{x}\overset{+}{z}}}$, their aggregation cannot be worse than the classification of the overnode's overnode child.
\end{lemma}

\begin{proof}
Considering Definition \ref{def:isp_1r} and Definition \ref{def:isp_1'l}, we know that both $\mathbf{\underline{\overset{+}{x}\overset{-}{z}}}$ are $\mathbf{\underline{\overset{-}{x}\overset{+}{z}}}$ contain an element $y$, and another element that can potentially be useful to bring either the left or right side of an overnode's overnode child classification to a $\mathbf{+}$.

We can apply Lemma \ref{l:htb} the the $y$ element, while the remainder element can be either used to improve the classification, or not if that is not possible. Leading to the conclusion that the classification of the overnode's overnode child cannot get worse.
\end{proof}

\subsubsection{\kwd{SEQ} Aggregations}~\\
\noindent
The aggregations leading to the classification of the overnode of type \kwd{SEQ} are illustrated in Table \ref{t:ad1s_oa_l} and Table \ref{t:ad1s_oa_r}. The two table show the aggregation between the current classification of the overnode's overnode child with the left and right undernodes classifications. As the two sub-patterns are independent, the order in which the tables are applied is irrelevant.

\begin{table}[ht!]
\centering
\begin{tabular}{|c|c|c|}
\hline
Overnode's Overnode Child & Left Sub-Pattern Classification & Result \\ \hline
$\mathbf{\overset{+}{x}t\overset{+}{z}}$ & $\mathbf{\overset{+}{x}}$ & $\mathbf{\overset{+}{x}t\overset{+}{z}}$ \\ \hline
$\mathbf{\overset{+}{x}t\overset{+}{z}}$ & $\mathbf{\overset{0}{x}}$ & $\mathbf{\overset{+}{x}t\overset{+}{z}}$ \\ \hline
$\mathbf{\overset{+}{x}t\overset{+}{z}}$ & $\mathbf{\overset{-}{x}}$ & $\mathbf{\overset{+}{x}t\overset{+}{z}}$ \\ \hline
$\mathbf{\overset{+}{x}t\overset{0}{z}}$ & $\mathbf{\overset{+}{x}}$ & $\mathbf{\overset{+}{x}t\overset{0}{z}}$ \\ \hline
$\mathbf{\overset{+}{x}t\overset{0}{z}}$ & $\mathbf{\overset{0}{x}}$ & $\mathbf{\overset{+}{x}t\overset{0}{z}}$ \\ \hline
$\mathbf{\overset{+}{x}t\overset{0}{z}}$ & $\mathbf{\overset{-}{x}}$ & $\mathbf{\overset{+}{x}t\overset{0}{z}}$ \\ \hline
$\mathbf{\overset{0}{x}t\overset{+}{z}}$ & $\mathbf{\overset{+}{x}}$ & $\mathbf{\overset{+}{x}t\overset{+}{z}}$ \\ \hline
$\mathbf{\overset{0}{x}t\overset{+}{z}}$ & $\mathbf{\overset{0}{x}}$ & $\mathbf{\overset{0}{x}t\overset{+}{z}}$ \\ \hline
$\mathbf{\overset{0}{x}t\overset{+}{z}}$ & $\mathbf{\overset{-}{x}}$ & $\mathbf{\overset{-}{x}t\overset{+}{z}}$ \\ \hline
$\mathbf{\overset{+}{x}t\overset{-}{z}}$ & $\mathbf{\overset{+}{x}}$ & $\mathbf{\overset{+}{x}t\overset{-}{z}}$ \\ \hline
$\mathbf{\overset{+}{x}t\overset{-}{z}}$ & $\mathbf{\overset{0}{x}}$ & $\mathbf{\overset{+}{x}t\overset{-}{z}}$ \\ \hline
$\mathbf{\overset{+}{x}t\overset{-}{z}}$ & $\mathbf{\overset{-}{x}}$ & $\mathbf{\overset{+}{x}t\overset{-}{z}}$ \\ \hline
$\mathbf{\overset{0}{x}t\overset{0}{z}}$ & $\mathbf{\overset{+}{x}}$ & $\mathbf{\overset{+}{x}t\overset{0}{z}}$ \\ \hline
$\mathbf{\overset{0}{x}t\overset{0}{z}}$ & $\mathbf{\overset{0}{x}}$ & $\mathbf{\overset{0}{x}t\overset{0}{z}}$ \\ \hline
$\mathbf{\overset{0}{x}t\overset{0}{z}}$ & $\mathbf{\overset{-}{x}}$ & $\mathbf{\overset{-}{x}t\overset{0}{z}}$ \\ \hline
$\mathbf{\overset{-}{x}t\overset{+}{z}}$ & $\mathbf{\overset{+}{x}}$ & $\mathbf{\overset{-}{x}t\overset{+}{z}}$ \\ \hline
$\mathbf{\overset{-}{x}t\overset{+}{z}}$ & $\mathbf{\overset{0}{x}}$ & $\mathbf{\overset{-}{x}t\overset{+}{z}}$ \\ \hline
$\mathbf{\overset{-}{x}t\overset{+}{z}}$ & $\mathbf{\overset{-}{x}}$ & $\mathbf{\overset{-}{x}t\overset{+}{z}}$ \\ \hline
$\mathbf{\overset{0}{x}t\overset{-}{z}}$ & $\mathbf{\overset{+}{x}}$ & $\mathbf{\overset{+}{x}t\overset{-}{z}}$ \\ \hline
$\mathbf{\overset{0}{x}t\overset{-}{z}}$ & $\mathbf{\overset{0}{x}}$ & $\mathbf{\overset{0}{x}t\overset{-}{z}}$ \\ \hline
$\mathbf{\overset{0}{x}t\overset{-}{z}}$ & $\mathbf{\overset{-}{x}}$ & $\mathbf{\overset{-}{x}t\overset{-}{z}}$ \\ \hline
$\mathbf{\overset{-}{x}t\overset{0}{z}}$ & $\mathbf{\overset{+}{x}}$ & $\mathbf{\overset{-}{x}t\overset{0}{z}}$ \\ \hline
$\mathbf{\overset{-}{x}t\overset{0}{z}}$ & $\mathbf{\overset{0}{x}}$ & $\mathbf{\overset{-}{x}t\overset{0}{z}}$ \\ \hline
$\mathbf{\overset{-}{x}t\overset{0}{z}}$ & $\mathbf{\overset{-}{x}}$ & $\mathbf{\overset{-}{x}t\overset{0}{z}}$ \\ \hline
$\mathbf{\overset{-}{x}t\overset{-}{z}}$ & $\mathbf{\overset{+}{x}}$ & $\mathbf{\overset{-}{x}t\overset{-}{z}}$ \\ \hline
$\mathbf{\overset{-}{x}t\overset{-}{z}}$ & $\mathbf{\overset{0}{x}}$ & $\mathbf{\overset{-}{x}t\overset{-}{z}}$ \\ \hline
$\mathbf{\overset{-}{x}t\overset{-}{z}}$ & $\mathbf{\overset{-}{x}}$ & $\mathbf{\overset{-}{x}t\overset{-}{z}}$ \\ \hline
\end{tabular}
\caption{Aggregation Table \kwd{SEQ} Overnode Left}\label{t:ad1s_oa_l}
\end{table}

\begin{proof}[Table \ref{t:ad1s_oa_l}]
Let the left part of the overnode classification be $B$, while the left sub-pattern be $A$ and, as the aggregation in Table \ref{t:ad1s_oa_l} refers to aggregations of a \kwd{SEQ} block, we can use the same aggregations from Table \ref{tab:atsol} using the result column \kwd{SEQ}(A, B), as it computes exactly the same aggregation required in Table \ref{t:ad1s_oa_l} for the right side of the overnode classification, and the aggregation does not influence the classification on the right side of the current overnode classification, which remains unchanged in the result.
\end{proof}

\begin{table}[ht!]
\centering
\begin{tabular}{|c|c|c|}
\hline
Overnode's Overnode Child & Right Sub-Pattern Classification & Result \\ \hline
$\mathbf{\overset{+}{x}t\overset{+}{z}}$ & $\mathbf{\overset{+}{z}}$ & $\mathbf{\overset{+}{x}t\overset{+}{z}}$ \\ \hline
$\mathbf{\overset{+}{x}t\overset{+}{z}}$ & $\mathbf{\overset{0}{z}}$ & $\mathbf{\overset{+}{x}t\overset{+}{z}}$ \\ \hline
$\mathbf{\overset{+}{x}t\overset{+}{z}}$ & $\mathbf{\overset{-}{z}}$ & $\mathbf{\overset{+}{x}t\overset{+}{z}}$ \\ \hline
$\mathbf{\overset{+}{x}t\overset{0}{z}}$ & $\mathbf{\overset{+}{z}}$ & $\mathbf{\overset{+}{x}t\overset{+}{z}}$ \\ \hline
$\mathbf{\overset{+}{x}t\overset{0}{z}}$ & $\mathbf{\overset{0}{z}}$ & $\mathbf{\overset{+}{x}t\overset{0}{z}}$ \\ \hline
$\mathbf{\overset{+}{x}t\overset{0}{z}}$ & $\mathbf{\overset{-}{z}}$ & $\mathbf{\overset{+}{x}t\overset{-}{z}}$ \\ \hline
$\mathbf{\overset{0}{x}t\overset{+}{z}}$ & $\mathbf{\overset{+}{z}}$ & $\mathbf{\overset{0}{x}t\overset{+}{z}}$ \\ \hline
$\mathbf{\overset{0}{x}t\overset{+}{z}}$ & $\mathbf{\overset{0}{z}}$ & $\mathbf{\overset{0}{x}t\overset{+}{z}}$ \\ \hline
$\mathbf{\overset{0}{x}t\overset{+}{z}}$ & $\mathbf{\overset{-}{z}}$ & $\mathbf{\overset{0}{x}t\overset{+}{z}}$ \\ \hline
$\mathbf{\overset{+}{x}t\overset{-}{z}}$ & $\mathbf{\overset{+}{z}}$ & $\mathbf{\overset{+}{x}t\overset{-}{z}}$ \\ \hline
$\mathbf{\overset{+}{x}t\overset{-}{z}}$ & $\mathbf{\overset{0}{z}}$ & $\mathbf{\overset{+}{x}t\overset{-}{z}}$ \\ \hline
$\mathbf{\overset{+}{x}t\overset{-}{z}}$ & $\mathbf{\overset{-}{z}}$ & $\mathbf{\overset{+}{x}t\overset{-}{z}}$ \\ \hline
$\mathbf{\overset{0}{x}t\overset{0}{z}}$ & $\mathbf{\overset{+}{z}}$ & $\mathbf{\overset{0}{x}t\overset{+}{z}}$ \\ \hline
$\mathbf{\overset{0}{x}t\overset{0}{z}}$ & $\mathbf{\overset{0}{z}}$ & $\mathbf{\overset{0}{x}t\overset{0}{z}}$ \\ \hline
$\mathbf{\overset{0}{x}t\overset{0}{z}}$ & $\mathbf{\overset{-}{z}}$ & $\mathbf{\overset{0}{x}t\overset{-}{z}}$ \\ \hline
$\mathbf{\overset{-}{x}t\overset{+}{z}}$ & $\mathbf{\overset{+}{z}}$ & $\mathbf{\overset{-}{x}t\overset{+}{z}}$ \\ \hline
$\mathbf{\overset{-}{x}t\overset{+}{z}}$ & $\mathbf{\overset{0}{z}}$ & $\mathbf{\overset{-}{x}t\overset{+}{z}}$ \\ \hline
$\mathbf{\overset{-}{x}t\overset{+}{z}}$ & $\mathbf{\overset{-}{z}}$ & $\mathbf{\overset{-}{x}t\overset{+}{z}}$ \\ \hline
$\mathbf{\overset{0}{x}t\overset{-}{z}}$ & $\mathbf{\overset{+}{z}}$ & $\mathbf{\overset{0}{x}t\overset{-}{z}}$ \\ \hline
$\mathbf{\overset{0}{x}t\overset{-}{z}}$ & $\mathbf{\overset{0}{z}}$ & $\mathbf{\overset{0}{x}t\overset{-}{z}}$ \\ \hline
$\mathbf{\overset{0}{x}t\overset{-}{z}}$ & $\mathbf{\overset{-}{z}}$ & $\mathbf{\overset{0}{x}t\overset{-}{z}}$ \\ \hline
$\mathbf{\overset{-}{x}t\overset{0}{z}}$ & $\mathbf{\overset{+}{z}}$ & $\mathbf{\overset{-}{x}t\overset{+}{z}}$ \\ \hline
$\mathbf{\overset{-}{x}t\overset{0}{z}}$ & $\mathbf{\overset{0}{z}}$ & $\mathbf{\overset{-}{x}t\overset{0}{z}}$ \\ \hline
$\mathbf{\overset{-}{x}t\overset{0}{z}}$ & $\mathbf{\overset{-}{z}}$ & $\mathbf{\overset{-}{x}t\overset{-}{z}}$ \\ \hline
$\mathbf{\overset{-}{x}t\overset{-}{z}}$ & $\mathbf{\overset{+}{z}}$ & $\mathbf{\overset{-}{x}t\overset{-}{z}}$ \\ \hline
$\mathbf{\overset{-}{x}t\overset{-}{z}}$ & $\mathbf{\overset{0}{z}}$ & $\mathbf{\overset{-}{x}t\overset{-}{z}}$ \\ \hline
$\mathbf{\overset{-}{x}t\overset{-}{z}}$ & $\mathbf{\overset{-}{z}}$ & $\mathbf{\overset{-}{x}t\overset{-}{z}}$ \\ \hline
\end{tabular}
\caption{Aggregation Table \kwd{SEQ} Overnode Right}\label{t:ad1s_oa_r}
\end{table}

\begin{proof}[Table \ref{t:ad1s_oa_r}]
Let the right part of the overnode classification be $A$, while the right sub-pattern be $B$ and, as the aggregation in Table \ref{t:ad1s_oa_r} refers to aggregations of a \kwd{SEQ} block, we can use the same aggregations from Table \ref{tab:atsor} using the result column \kwd{SEQ}(A, B), as it computes exactly the same aggregation required in Table \ref{t:ad1s_oa_r} for the right side of the overnode classification, and the aggregation does not influence the classification on the left side of the current overnode classification, which remains unchanged in the result.
\end{proof}

\subsubsection{\kwd{AND} Aggregations}~\\
\noindent
The classification of an overnode of type \kwd{AND} is obtained by the aggregating the aggregation result of the classification of its undernode children with the classification of its overnode child. The aggregations are illustrated in Table \ref{tab:ad1s_aoa1}, Table \ref{tab:ad1s_aoa2} and Table \ref{tab:ad1s_aoa3}.

\begin{table}[ht!]
\centering
\begin{tabular}{|c|c|c|}
\hline
Overnode's Overnode Child & Overnode's Undernode Children & Result \\ \hline
$\mathbf{\overset{+}{x}t\overset{+}{z}}$ & $\mathbf{\underline{\overset{+}{x}\overset{+}{z}}}$ & $\mathbf{\overset{+}{x}t\overset{+}{z}}$ \\ \hline
$\mathbf{\overset{+}{x}t\overset{+}{z}}$ & $\mathbf{\underline{\overset{+}{x}\overset{0}{z}}}$ & $\mathbf{\overset{+}{x}t\overset{+}{z}}$ \\ \hline
$\mathbf{\overset{+}{x}t\overset{+}{z}}$ & $\mathbf{\underline{\overset{0}{x}\overset{+}{z}}}$ & $\mathbf{\overset{+}{x}t\overset{+}{z}}$ \\ \hline
$\mathbf{\overset{+}{x}t\overset{+}{z}}$ & $\mathbf{\underline{\overset{+}{x}\overset{-}{z}}}$ & $\mathbf{\overset{+}{x}t\overset{+}{z}}$ \\ \hline
$\mathbf{\overset{+}{x}t\overset{+}{z}}$ & $\mathbf{\underline{\overset{0}{x}\overset{0}{z}}}$ & $\mathbf{\overset{+}{x}t\overset{+}{z}}$ \\ \hline
$\mathbf{\overset{+}{x}t\overset{+}{z}}$ & $\mathbf{\underline{\overset{-}{x}\overset{+}{z}}}$ & $\mathbf{\overset{+}{x}t\overset{+}{z}}$ \\ \hline
$\mathbf{\overset{+}{x}t\overset{+}{z}}$ & $\mathbf{\underline{\overset{-}{xz}}}$ & $\mathbf{\overset{+}{x}t\overset{+}{z}}$ \\ \hline
$\mathbf{\overset{+}{x}t\overset{0}{z}}$ & $\mathbf{\underline{\overset{+}{x}\overset{+}{z}}}$ & $\mathbf{\overset{+}{x}t\overset{+}{z}}$ \\ \hline
$\mathbf{\overset{+}{x}t\overset{0}{z}}$ & $\mathbf{\underline{\overset{+}{x}\overset{0}{z}}}$ & $\mathbf{\overset{+}{x}t\overset{0}{z}}$ \\ \hline
$\mathbf{\overset{+}{x}t\overset{0}{z}}$ & $\mathbf{\underline{\overset{0}{x}\overset{+}{z}}}$ & $\mathbf{\overset{+}{x}t\overset{+}{z}}$ \\ \hline
$\mathbf{\overset{+}{x}t\overset{0}{z}}$ & $\mathbf{\underline{\overset{+}{x}\overset{-}{z}}}$ & $\mathbf{\overset{+}{x}t\overset{0}{z}}$ \\ \hline
$\mathbf{\overset{+}{x}t\overset{0}{z}}$ & $\mathbf{\underline{\overset{0}{x}\overset{0}{z}}}$ & $\mathbf{\overset{+}{x}t\overset{0}{z}}$ \\ \hline
$\mathbf{\overset{+}{x}t\overset{0}{z}}$ & $\mathbf{\underline{\overset{-}{x}\overset{+}{z}}}$ & $\mathbf{\overset{+}{x}t\overset{+}{z}}$ \\ \hline
$\mathbf{\overset{+}{x}t\overset{0}{z}}$ & $\mathbf{\underline{\overset{-}{xz}}}$ & $\mathbf{\overset{+}{x}t\overset{0}{z}}$ \\ \hline
$\mathbf{\overset{0}{x}t\overset{+}{z}}$ & $\mathbf{\underline{\overset{+}{x}\overset{+}{z}}}$ & $\mathbf{\overset{+}{x}t\overset{+}{z}}$ \\ \hline
$\mathbf{\overset{0}{x}t\overset{+}{z}}$ & $\mathbf{\underline{\overset{+}{x}\overset{0}{z}}}$ & $\mathbf{\overset{+}{x}t\overset{+}{z}}$ \\ \hline
$\mathbf{\overset{0}{x}t\overset{+}{z}}$ & $\mathbf{\underline{\overset{0}{x}\overset{+}{z}}}$ & $\mathbf{\overset{0}{x}t\overset{+}{z}}$ \\ \hline
$\mathbf{\overset{0}{x}t\overset{+}{z}}$ & $\mathbf{\underline{\overset{+}{x}\overset{-}{z}}}$ & $\mathbf{\overset{+}{x}t\overset{+}{z}}$ \\ \hline
$\mathbf{\overset{0}{x}t\overset{+}{z}}$ & $\mathbf{\underline{\overset{0}{x}\overset{0}{z}}}$ & $\mathbf{\overset{0}{x}t\overset{+}{z}}$ \\ \hline
$\mathbf{\overset{0}{x}t\overset{+}{z}}$ & $\mathbf{\underline{\overset{-}{x}\overset{+}{z}}}$ & $\mathbf{\overset{0}{x}t\overset{+}{z}}$ \\ \hline
$\mathbf{\overset{0}{x}t\overset{+}{z}}$ & $\mathbf{\underline{\overset{-}{xz}}}$ & $\mathbf{\overset{0}{x}t\overset{+}{z}}$ \\ \hline
\end{tabular}
\caption{Aggregation Table from Interval Sub-Pattern to Interval Overnode Pattern}\label{tab:ad1s_aoa1}
\end{table}

\begin{proof}[Table \ref{tab:ad1s_aoa1}]
\begin{enumerate}
\item As the aggregation involves \kwd{AND} blocks, from Definition \ref{def:ser}, it follows that each execution resulting from the aggregation consists of an execution of the \emph{Overnode's Overnode Child} interleaved with an execution of the \emph{Overnode's Undernode Children}.
\item We prove the correctness of the table by proving the correctness of the aggregation in each of its line. Considering by case the \emph{Overnode's Overnode Child} classifications:
\begin{description}
\item[$\mathbf{\overset{+}{x}t\overset{+}{z}}$] Independently on the \emph{Overnode's Undernode Children} classification, the correctness of the row follows directly from Lemma \ref{l:topclass}.

\item[$\mathbf{\overset{+}{x}t\overset{0}{z}}$] 
\begin{enumerate}
\item From Definition \ref{def:iop_+t0}, it follows that there exists an execution of the \emph{Overnode's Overnode Child} where there is a $x$ on the left of $t$ with no $y$ between them, and no $y$ or $z$ on the right of the $t$.
\item Considering by case the \emph{Overnode's Undernode Children} classification:
\begin{description}
\item[$\mathbf{\underline{\overset{+}{x}\overset{+}{z}}}$] The correctness of the row follows directly from Lemma \ref{l:3super}.
\item[$\mathbf{\underline{\overset{+}{x}\overset{0}{z}}}$] 
\begin{enumerate}
\item From Definition \ref{def:isp_1}, it follows that \emph{Overnode's Undernode Children} contains an execution containing a $x$ and no $y$ or $z$.
\item From i., it follows that it does not contain any element capable of improving the \emph{Overnode's Overnode Child} classification, or making it worse.
\item From 1., (a) and ii., it follows that the aggregation is then correctly classified as $\mathbf{\overset{+}{x}t\overset{0}{z}}$.
\end{enumerate}
\item[$\mathbf{\underline{\overset{0}{x}\overset{+}{z}}}$] 
\begin{enumerate}
\item From Definition \ref{def:isp_1'}, it follows that \emph{Overnode's Undernode Children} contains an execution containing a $z$ and no $y$ or $x$.
\item From 1., (a) and i., it follows that a possible interleaving would allow to place the $z$ on the right of the $t$.
\item From ii and Definition \ref{def:iop_+t+}, it follows that the aggregation is then correctly classified as $\mathbf{\overset{+}{x}t\overset{+}{z}}$.
\end{enumerate}
\item[$\mathbf{\underline{\overset{+}{x}\overset{-}{z}}}$] 
\begin{enumerate}
\item From Definition \ref{def:isp_1r}, it follows that \emph{Overnode's Undernode Children} contains an execution containing a $x$ and a $y$ on its right, while not containing $z$.
\item From i., it follows that it does not contain any element capable of improving the \emph{Overnode's Overnode Child} classification. 
\item From ii., (a) and Lemma \ref{l:nw}, it follows that the aggregation is then correctly classified as $\mathbf{\overset{+}{x}t\overset{0}{z}}$.
\end{enumerate}
\item[$\mathbf{\underline{\overset{0}{x}\overset{0}{z}}}$] The correctness of the row follows directly from Lemma \ref{l:neutral_iop}.
\item[$\mathbf{\underline{\overset{-}{x}\overset{+}{z}}}$] 
\begin{enumerate}
\item From Definition \ref{def:isp_1'l}, it follows that \emph{Overnode's Undernode Children} contains an execution containing a $z$ and a $y$ on its left, while not containing $x$.
\item From 1., (a) and i., it follows that the $z$ from $\mathbf{\underline{\overset{-}{x}\overset{+}{z}}}$ can be set on the right of $t$, while the $y$, thanks to the interleaving, can be set on the left of the $x$ from $\mathbf{\overset{+}{x}t\overset{0}{z}}$.
\item From ii and Definition \ref{def:iop_+t+}, it follows that the aggregation is then correctly classified as $\mathbf{\overset{+}{x}t\overset{+}{z}}$.
\end{enumerate}
\item[$\mathbf{\underline{\overset{-}{xz}}}$] The correctness of the row follows directly from Lemma \ref{l:htb}.
\end{description}
\end{enumerate}
\item[$\mathbf{\overset{0}{x}t\overset{+}{z}}$]
\begin{enumerate}
\item From Definition \ref{def:iop_0t+}, it follows that there exists an execution of the \emph{Overnode's Overnode Child} where there is a $z$ on the right of $t$ with no $y$ between them, and no $y$ or $x$ on the left of the $t$.
\item Considering by case the \emph{Overnode's Undernode Children} classification:
\begin{description}
\item[$\mathbf{\underline{\overset{+}{x}\overset{+}{z}}}$] The correctness of the row follows directly from Lemma \ref{l:3super}.
\item[$\mathbf{\underline{\overset{+}{x}\overset{0}{z}}}$] 
\begin{enumerate}
\item From Definition \ref{def:isp_1}, it follows that \emph{Overnode's Undernode Children} contains an execution containing a $x$ and no $y$ or $z$.
\item From 1., (a) and i., it follows that a possible interleaving would allow to place the $x$ on the left of the $t$.
\item From ii and Definition \ref{def:iop_+t+}, it follows that the aggregation is then correctly classified as $\mathbf{\overset{+}{x}t\overset{+}{z}}$.
\end{enumerate}
\item[$\mathbf{\underline{\overset{0}{x}\overset{+}{z}}}$] 
\begin{enumerate}
\item From Definition \ref{def:isp_1'}, it follows that \emph{Overnode's Undernode Children} contains an execution containing a $z$ and no $y$ or $x$.
\item From i., it follows that it does not contain any element capable of improving the \emph{Overnode's Overnode Child} classification, or making it worse.
\item From 1., (a) and ii., it follows that the aggregation is then correctly classified as $\mathbf{\overset{0}{x}t\overset{+}{z}}$.
\end{enumerate}
\item[$\mathbf{\underline{\overset{+}{x}\overset{-}{z}}}$] 
\begin{enumerate}
\item From Definition \ref{def:isp_1r}, it follows that \emph{Overnode's Undernode Children} contains an execution containing a $x$ and a $y$ on its right, while not containing $z$.
\item From 1., (a) and i., it follows that the $x$ from $\mathbf{1'_{r}}$ can be set on the left of $t$, while the $y$, thanks to the interleaving, can be set on the right of the $z$ from $\mathbf{\overset{0}{x}t\overset{+}{z}}$.
\item From ii and Definition \ref{def:iop_+t+}, it follows that the aggregation is then correctly classified as $\mathbf{\overset{+}{x}t\overset{+}{z}}$.
\end{enumerate}
\item[$\mathbf{\underline{\overset{0}{x}\overset{0}{z}}}$] The correctness of the row follows directly from Lemma \ref{l:neutral_iop}.
\item[$\mathbf{\underline{\overset{-}{x}\overset{+}{z}}}$] 
\begin{enumerate}
\item From Definition \ref{def:isp_1'l}, it follows that \emph{Overnode's Undernode Children} contains an execution containing a $z$ and a $y$ on its left, while not containing $x$.
\item From i., it follows that it does not contain any element capable of improving the \emph{Overnode's Overnode Child} classification. 
\item From ii., (a) and Lemma \ref{l:nw}, it follows that the aggregation is then correctly classified as $\mathbf{\overset{0}{x}t\overset{+}{z}}$.
\end{enumerate}
\item[$\mathbf{\underline{\overset{-}{xz}}}$] The correctness of the row follows directly from Lemma \ref{l:htb}.
\end{description}
\end{enumerate}
\end{description}
\end{enumerate}
\end{proof}

\begin{table}[ht!]
\centering
\begin{tabular}{|c|c|c|}
\hline
Overnode's Overnode Child & Overnode's Undernode Children & Result \\ \hline
$\mathbf{\overset{+}{x}t\overset{-}{z}}$ & $\mathbf{\underline{\overset{+}{x}\overset{+}{z}}}$ & $\mathbf{\overset{+}{x}t\overset{+}{z}}$ \\ \hline
$\mathbf{\overset{+}{x}t\overset{-}{z}}$ & $\mathbf{\underline{\overset{+}{x}\overset{0}{z}}}$ & $\mathbf{\overset{+}{x}t\overset{-}{z}}$ \\ \hline
$\mathbf{\overset{+}{x}t\overset{-}{z}}$ & $\mathbf{\underline{\overset{0}{x}\overset{+}{z}}}$ & $\mathbf{\overset{+}{x}t\overset{+}{z}}$ \\ \hline
$\mathbf{\overset{+}{x}t\overset{-}{z}}$ & $\mathbf{\underline{\overset{+}{x}\overset{-}{z}}}$ & $\mathbf{\overset{+}{x}t\overset{-}{z}}$ \\ \hline
$\mathbf{\overset{+}{x}t\overset{-}{z}}$ & $\mathbf{\underline{\overset{0}{x}\overset{0}{z}}}$ & $\mathbf{\overset{+}{x}t\overset{-}{z}}$ \\ \hline
$\mathbf{\overset{+}{x}t\overset{-}{z}}$ & $\mathbf{\underline{\overset{-}{x}\overset{+}{z}}}$ & $\mathbf{\overset{+}{x}t\overset{+}{z}}$ \\ \hline
$\mathbf{\overset{+}{x}t\overset{-}{z}}$ & $\mathbf{\underline{\overset{-}{xz}}}$ & $\mathbf{\overset{+}{x}t\overset{-}{z}}$ \\ \hline
$\mathbf{\overset{0}{x}t\overset{0}{z}}$ & $\mathbf{\underline{\overset{+}{x}\overset{+}{z}}}$ & $\mathbf{\overset{+}{x}t\overset{+}{z}}$ \\ \hline
$\mathbf{\overset{0}{x}t\overset{0}{z}}$ & $\mathbf{\underline{\overset{+}{x}\overset{0}{z}}}$ & $\mathbf{\overset{+}{x}t\overset{0}{z}}$ \\ \hline
$\mathbf{\overset{0}{x}t\overset{0}{z}}$ & $\mathbf{\underline{\overset{0}{x}\overset{+}{z}}}$ & $\mathbf{\overset{0}{x}t\overset{+}{z}}$ \\ \hline
$\mathbf{\overset{0}{x}t\overset{0}{z}}$ & $\mathbf{\underline{\overset{+}{x}\overset{-}{z}}}$ & $\mathbf{\overset{+}{x}t\overset{-}{z}}$ / $\mathbf{\overset{-}{x}t\overset{0}{z}}$ \\ \hline
$\mathbf{\overset{0}{x}t\overset{0}{z}}$ & $\mathbf{\underline{\overset{0}{x}\overset{0}{z}}}$ & $\mathbf{\overset{0}{x}t\overset{0}{z}}$ \\ \hline
$\mathbf{\overset{0}{x}t\overset{0}{z}}$ & $\mathbf{\underline{\overset{-}{x}\overset{+}{z}}}$ & $\mathbf{\overset{-}{x}t\overset{+}{z}}$ / $\mathbf{\overset{0}{x}t\overset{-}{z}}$ \\ \hline
$\mathbf{\overset{0}{x}t\overset{0}{z}}$ & $\mathbf{\underline{\overset{-}{xz}}}$ & $\mathbf{\overset{0}{x}t\overset{-}{z}}$ / $\mathbf{\overset{-}{x}t\overset{0}{z}}$ \\ \hline
$\mathbf{\overset{-}{x}t\overset{+}{z}}$ & $\mathbf{\underline{\overset{+}{x}\overset{+}{z}}}$ & $\mathbf{\overset{+}{x}t\overset{+}{z}}$ \\ \hline
$\mathbf{\overset{-}{x}t\overset{+}{z}}$ & $\mathbf{\underline{\overset{+}{x}\overset{0}{z}}}$ & $\mathbf{\overset{+}{x}t\overset{+}{z}}$ \\ \hline
$\mathbf{\overset{-}{x}t\overset{+}{z}}$ & $\mathbf{\underline{\overset{0}{x}\overset{+}{z}}}$ & $\mathbf{\overset{-}{x}t\overset{+}{z}}$ \\ \hline
$\mathbf{\overset{-}{x}t\overset{+}{z}}$ & $\mathbf{\underline{\overset{+}{x}\overset{-}{z}}}$ & $\mathbf{\overset{+}{x}t\overset{+}{z}}$ \\ \hline
$\mathbf{\overset{-}{x}t\overset{+}{z}}$ & $\mathbf{\underline{\overset{0}{x}\overset{0}{z}}}$ & $\mathbf{\overset{-}{x}t\overset{+}{z}}$ \\ \hline
$\mathbf{\overset{-}{x}t\overset{+}{z}}$ & $\mathbf{\underline{\overset{-}{x}\overset{+}{z}}}$ & $\mathbf{\overset{-}{x}t\overset{+}{z}}$ \\ \hline
$\mathbf{\overset{-}{x}t\overset{+}{z}}$ & $\mathbf{\underline{\overset{-}{xz}}}$ & $\mathbf{\overset{-}{x}t\overset{+}{z}}$ \\ \hline
\end{tabular}
\caption{Aggregation Table from Interval Sub-Pattern to Interval Overnode Pattern}\label{tab:ad1s_aoa2}
\end{table}

\begin{proof}[Table \ref{tab:ad1s_aoa2}]
\begin{enumerate}
\item As the aggregation involves \kwd{AND} blocks, from Definition \ref{def:ser}, it follows that each execution resulting from the aggregation consists of an execution of the \emph{Overnode's Overnode Child} interleaved with an execution of the \emph{Overnode's Undernode Children}.
\item We prove the correctness of the table by proving the correctness of the aggregation in each of its line. Considering by case the \emph{Overnode's Overnode Child} classifications:
\begin{description}
\item[$\mathbf{\overset{+}{x}t\overset{-}{z}}$]
\begin{enumerate}
\item From Definition \ref{def:iop_+t-}, it follows that there exists an execution of the \emph{Overnode's Overnode Child} where there is a $z$ on the right of $t$ with no $y$ between them, and no $y$ or $x$ on the left of the $t$.
\item Considering by case the \emph{Overnode's Undernode Children} classification:
\begin{description}
\item[$\mathbf{\underline{\overset{+}{x}\overset{+}{z}}}$] The correctness of the row follows directly from Lemma \ref{l:3super}.
\item[$\mathbf{\underline{\overset{+}{x}\overset{0}{z}}}$] 
\begin{enumerate}
\item From Definition \ref{def:isp_1}, it follows that \emph{Overnode's Undernode Children} contains an execution containing a $x$ and no $y$ or $z$.
\item From i., it follows that it does not contain any element capable of improving the \emph{Overnode's Overnode Child} classification, or making it worse.
\item From 1., (a) and ii., it follows that the aggregation is then correctly classified as $\mathbf{\overset{+}{x}t\overset{-}{z}}$.
\end{enumerate}
\item[$\mathbf{\underline{\overset{0}{x}\overset{+}{z}}}$] 
\begin{enumerate}
\item From Definition \ref{def:isp_1'}, it follows that \emph{Overnode's Undernode Children} contains an execution containing a $z$ and no $y$ or $x$.
\item From 1., (a) and i., it follows that a possible interleaving would allow to place the $z$ on the right of the $t$, before the existing $y$.
\item From ii and Definition \ref{def:iop_+t+}, it follows that the aggregation is then correctly classified as $\mathbf{\overset{+}{x}t\overset{+}{z}}$.
\end{enumerate}
\item[$\mathbf{\underline{\overset{+}{x}\overset{-}{z}}}$] 
\begin{enumerate}
\item From Definition \ref{def:isp_1r}, it follows that \emph{Overnode's Undernode Children} contains an execution containing a $x$ and a $y$ on its right, while not containing $z$.
\item From i., it follows that it does not contain any element capable of improving the \emph{Overnode's Overnode Child} classification. 
\item From ii., (a) and Lemma \ref{l:nw}, it follows that the aggregation is then correctly classified as $\mathbf{\overset{+}{x}t\overset{-}{z}}$.
\end{enumerate}
\item[$\mathbf{\underline{\overset{0}{x}\overset{0}{z}}}$] The correctness of the row follows directly from Lemma \ref{l:neutral_iop}.
\item[$\mathbf{\underline{\overset{-}{x}\overset{+}{z}}}$] 
\begin{enumerate}
\item From Definition \ref{def:isp_1'l}, it follows that \emph{Overnode's Undernode Children} contains an execution containing a $z$ and a $y$ on its left, while not containing $x$.
\item From 1., (a) and i., it follows that the $z$ from $\mathbf{\underline{\overset{-}{x}\overset{+}{z}}}$ can be set on the right of $t$, before the existing $y$ in $\mathbf{\overset{+}{x}t\overset{-}{z}}$, while the $y$, thanks to the interleaving, can be set on the left of the $x$ from $\mathbf{\overset{+}{x}t\overset{0}{z}}$.
\item From ii and Definition \ref{def:iop_+t+}, it follows that the aggregation is then correctly classified as $\mathbf{\overset{+}{x}t\overset{+}{z}}$.
\end{enumerate}
\item[$\mathbf{\underline{\overset{-}{xz}}}$] The correctness of the row follows directly from Lemma \ref{l:htb}.
\end{description}
\end{enumerate}
\item[$\mathbf{\overset{0}{x}t\overset{0}{z}}$]
\begin{enumerate}
\item From Definition \ref{def:iop_0t0}, it follows that there exists an execution of the \emph{Overnode's Overnode Child} where there is a $z$ on the right of $t$ with no $y$ between them, and no $y$ or $x$ on the left of the $t$.
\item Considering by case the \emph{Overnode's Undernode Children} classification:
\begin{description}
\item[$\mathbf{\underline{\overset{+}{x}\overset{+}{z}}}$] The correctness of the row follows directly from Lemma \ref{l:3super}.
\item[$\mathbf{\underline{\overset{+}{x}\overset{0}{z}}}$] 
\begin{enumerate}
\item From Definition \ref{def:isp_1}, it follows that \emph{Overnode's Undernode Children} contains an execution containing a $x$ and no $y$ or $z$.
\item From 1., (a) and i., it follows that the $x$ from $\mathbf{\underline{\overset{+}{x}\overset{0}{z}}}$ can be placed on the left of the $t$ from $\mathbf{\overset{0}{x}t\overset{0}{z}}$.
\item From ii. and Definition \ref{def:iop_+t0}, it follows that the aggregation is then correctly classified as $\mathbf{\overset{+}{x}t\overset{0}{z}}$.
\end{enumerate}
\item[$\mathbf{\underline{\overset{0}{x}\overset{+}{z}}}$] 
\begin{enumerate}
\item From Definition \ref{def:isp_1'}, it follows that \emph{Overnode's Undernode Children} contains an execution containing a $z$ and no $y$ or $x$.
\item From 1., (a) and i., it follows that the $y$ from $\mathbf{\underline{\overset{0}{x}\overset{+}{z}}}$ can be placed on the right of the $t$ from $\mathbf{\overset{0}{x}t\overset{0}{z}}$.
\item From ii. and Definition \ref{def:iop_+t0}, it follows that the aggregation is then correctly classified as $\mathbf{\overset{0}{x}t\overset{+}{z}}$.
\end{enumerate}
\item[$\mathbf{\underline{\overset{+}{x}\overset{-}{z}}}$] 
\begin{enumerate}
\item From Definition \ref{def:isp_1r}, it follows that \emph{Overnode's Undernode Children} contains an execution containing a $x$ and a $y$ on its right, while not containing $z$.
\item From 1., (a) and i., it follows that the $x$ from $\mathbf{\underline{\overset{+}{x}\overset{-}{z}}}$ can be placed on the left of the $t$ from $\mathbf{\overset{0}{x}t\overset{0}{z}}$, while the $y$ can be placed on the right.
\item From 1., (a) and i., it follows that the $x$ from $\mathbf{\underline{\overset{+}{x}\overset{-}{z}}}$ can be placed on the left of the $t$ from $\mathbf{\overset{0}{x}t\overset{0}{z}}$, while the $y$ can be placed between the $x$ and the $t$.
\item From ii. and Definition \ref{def:iop_+t-}, it follows that the aggregation is then correctly classified as $\mathbf{\overset{+}{x}t\overset{-}{z}}$.
\item From iii. and Definition \ref{def:iop_-t0}, it follows that the aggregation is then correctly classified as $\mathbf{\overset{-}{x}t\overset{0}{z}}$.
\item From iv. and v., it follows that the classification of the row is correct.
\end{enumerate}
\item[$\mathbf{\underline{\overset{0}{x}\overset{0}{z}}}$] The correctness of the row follows directly from Lemma \ref{l:neutral_iop}.
\item[$\mathbf{\underline{\overset{-}{x}\overset{+}{z}}}$] 
\begin{enumerate}
\item From Definition \ref{def:isp_1'l}, it follows that \emph{Overnode's Undernode Children} contains an execution containing a $z$ and a $y$ on its left, while not containing $x$.
\item From 1., (a) and i., it follows that the $z$ from $\mathbf{\underline{\overset{-}{x}\overset{+}{z}}}$ can be placed on the right of the $t$ from $\mathbf{\overset{0}{x}t\overset{0}{z}}$, while the $y$ can be placed on the left.
\item From 1., (a) and i., it follows that the $z$ from $\mathbf{\underline{\overset{-}{x}\overset{+}{z}}}$ can be placed on the right of the $t$ from $\mathbf{\overset{0}{x}t\overset{0}{z}}$, while the $y$ can be placed between the $t$ and the $z$.
\item From ii. and Definition \ref{def:iop_-t+}, it follows that the aggregation is then correctly classified as $\mathbf{\overset{-}{x}t\overset{+}{z}}$.
\item From iii. and Definition \ref{def:iop_0t-}, it follows that the aggregation is then correctly classified as $\mathbf{\overset{0}{x}t\overset{-}{z}}$.
\item From iv. and v., it follows that the classification of the row is correct.
\end{enumerate}
\item[$\mathbf{\underline{\overset{-}{xz}}}$] The correctness of the row follows directly from Lemma \ref{l:htb}.
\end{description}
\end{enumerate}
\item[$\mathbf{\overset{-}{x}t\overset{+}{z}}$]
\begin{enumerate}
\item From Definition \ref{def:iop_-t+}, it follows that there exists an execution of the \emph{Overnode's Overnode Child} where there is a $z$ on the right of $t$ with no $y$ between them, and no $y$ or $x$ on the left of the $t$.
\item Considering by case the \emph{Overnode's Undernode Children} classification:
\begin{description}
\item[$\mathbf{\underline{\overset{+}{x}\overset{+}{z}}}$] The correctness of the row follows directly from Lemma \ref{l:3super}.
\item[$\mathbf{\underline{\overset{+}{x}\overset{0}{z}}}$] 
\begin{enumerate}
\item From Definition \ref{def:isp_1}, it follows that \emph{Overnode's Undernode Children} contains an execution containing a $x$ and no $y$ or $z$.
\item From 1., (a) and i., it follows that a possible interleaving would allow to place the $x$ on the left of the $t$.
\item From ii and Definition \ref{def:iop_+t+}, it follows that the aggregation is then correctly classified as $\mathbf{\overset{+}{x}t\overset{+}{z}}$.
\end{enumerate}
\item[$\mathbf{\underline{\overset{0}{x}\overset{+}{z}}}$] 
\begin{enumerate}
\item From Definition \ref{def:isp_1'}, it follows that \emph{Overnode's Undernode Children} contains an execution containing a $z$ and no $y$ or $x$.
\item From i., it follows that it does not contain any element capable of improving the \emph{Overnode's Overnode Child} classification, or making it worse.
\item From 1., (a) and ii., it follows that the aggregation is then correctly classified as $\mathbf{\overset{-}{x}t\overset{+}{z}}$.
\end{enumerate}
\item[$\mathbf{\underline{\overset{+}{x}\overset{-}{z}}}$] 
\begin{enumerate}
\item From Definition \ref{def:isp_1r}, it follows that \emph{Overnode's Undernode Children} contains an execution containing a $x$ and a $y$ on its right, while not containing $z$.
\item From 1., (a) and i., it follows that the $x$ from $\mathbf{1'_{r}}$ can be set on the left of $t$, while the $y$, thanks to the interleaving, can be set on the right of the $z$ from $\mathbf{\overset{0}{x}t\overset{+}{z}}$.
\item From ii and Definition \ref{def:iop_+t+}, it follows that the aggregation is then correctly classified as $\mathbf{\overset{+}{x}t\overset{+}{z}}$.
\end{enumerate}
\item[$\mathbf{\underline{\overset{0}{x}\overset{0}{z}}}$] The correctness of the row follows directly from Lemma \ref{l:neutral_iop}.
\item[$\mathbf{\underline{\overset{-}{x}\overset{+}{z}}}$] 
\begin{enumerate}
\item From Definition \ref{def:isp_1'l}, it follows that \emph{Overnode's Undernode Children} contains an execution containing a $z$ and a $y$ on its left, while not containing $x$.
\item From i., it follows that it does not contain any element capable of improving the \emph{Overnode's Overnode Child} classification. 
\item From ii., (a) and Lemma \ref{l:nw}, it follows that the aggregation is then correctly classified as $\mathbf{\overset{-}{x}t\overset{+}{z}}$.
\end{enumerate}
\item[$\mathbf{\underline{\overset{-}{xz}}}$] The correctness of the row follows directly from Lemma \ref{l:htb}.
\end{description}
\end{enumerate}
\end{description}
\end{enumerate}
\end{proof}

\begin{table}[ht!]
\centering
\begin{tabular}{|c|c|c|}
\hline
Overnode's Overnode Child & Overnode's Undernode Children & Result \\ \hline
$\mathbf{\overset{0}{x}t\overset{-}{z}}$ & $\mathbf{\underline{\overset{+}{x}\overset{+}{z}}}$ & $\mathbf{\overset{+}{x}t\overset{+}{z}}$ \\ \hline
$\mathbf{\overset{0}{x}t\overset{-}{z}}$ & $\mathbf{\underline{\overset{+}{x}\overset{0}{z}}}$ & $\mathbf{\overset{+}{x}t\overset{-}{z}}$ \\ \hline
$\mathbf{\overset{0}{x}t\overset{-}{z}}$ & $\mathbf{\underline{\overset{0}{x}\overset{+}{z}}}$ & $\mathbf{\overset{0}{x}t\overset{+}{z}}$ \\ \hline
$\mathbf{\overset{0}{x}t\overset{-}{z}}$ & $\mathbf{\underline{\overset{+}{x}\overset{-}{z}}}$ & $\mathbf{\overset{+}{x}t\overset{-}{z}}$ \\ \hline
$\mathbf{\overset{0}{x}t\overset{-}{z}}$ & $\mathbf{\underline{\overset{0}{x}\overset{0}{z}}}$ & $\mathbf{\overset{0}{x}t\overset{-}{z}}$ \\ \hline
$\mathbf{\overset{0}{x}t\overset{-}{z}}$ & $\mathbf{\underline{\overset{-}{x}\overset{+}{z}}}$ & $\mathbf{\overset{-}{x}t\overset{+}{z}}$ / $\mathbf{\overset{0}{x}t\overset{-}{z}}$ \\ \hline
$\mathbf{\overset{0}{x}t\overset{-}{z}}$ & $\mathbf{\underline{\overset{-}{xz}}}$ & $\mathbf{\overset{0}{x}t\overset{-}{z}}$ \\ \hline
$\mathbf{\overset{-}{x}t\overset{0}{z}}$ & $\mathbf{\underline{\overset{+}{x}\overset{+}{z}}}$ & $\mathbf{\overset{+}{x}t\overset{+}{z}}$ \\ \hline
$\mathbf{\overset{-}{x}t\overset{0}{z}}$ & $\mathbf{\underline{\overset{+}{x}\overset{0}{z}}}$ & $\mathbf{\overset{+}{x}t\overset{0}{z}}$ \\ \hline
$\mathbf{\overset{-}{x}t\overset{0}{z}}$ & $\mathbf{\underline{\overset{0}{x}\overset{+}{z}}}$ & $\mathbf{\overset{-}{x}t\overset{+}{z}}$ \\ \hline
$\mathbf{\overset{-}{x}t\overset{0}{z}}$ & $\mathbf{\underline{\overset{+}{x}\overset{-}{z}}}$ & $\mathbf{\overset{+}{x}t\overset{-}{z}}$ / $\mathbf{\overset{-}{x}t\overset{0}{z}}$ \\ \hline
$\mathbf{\overset{-}{x}t\overset{0}{z}}$ & $\mathbf{\underline{\overset{0}{x}\overset{0}{z}}}$ & $\mathbf{\overset{-}{x}t\overset{0}{z}}$ \\ \hline
$\mathbf{\overset{-}{x}t\overset{0}{z}}$ & $\mathbf{\underline{\overset{-}{x}\overset{+}{z}}}$ & $\mathbf{\overset{-}{x}t\overset{+}{z}}$ \\ \hline
$\mathbf{\overset{-}{x}t\overset{0}{z}}$ & $\mathbf{\underline{\overset{-}{xz}}}$ & $\mathbf{\overset{-}{x}t\overset{0}{z}}$ \\ \hline
$\mathbf{\overset{-}{x}t\overset{-}{z}}$ & $\mathbf{\underline{\overset{+}{x}\overset{+}{z}}}$ & $\mathbf{\overset{+}{x}t\overset{+}{z}}$ \\ \hline
$\mathbf{\overset{-}{x}t\overset{-}{z}}$ & $\mathbf{\underline{\overset{+}{x}\overset{0}{z}}}$ & $\mathbf{\overset{+}{x}t\overset{-}{z}}$ \\ \hline
$\mathbf{\overset{-}{x}t\overset{-}{z}}$ & $\mathbf{\underline{\overset{0}{x}\overset{+}{z}}}$ & $\mathbf{\overset{-}{x}t\overset{+}{z}}$ \\ \hline
$\mathbf{\overset{-}{x}t\overset{-}{z}}$ & $\mathbf{\underline{\overset{+}{x}\overset{-}{z}}}$ & $\mathbf{\overset{+}{x}t\overset{-}{z}}$ \\ \hline
$\mathbf{\overset{-}{x}t\overset{-}{z}}$ & $\mathbf{\underline{\overset{0}{x}\overset{0}{z}}}$ & $\mathbf{\overset{-}{x}t\overset{-}{z}}$ \\ \hline
$\mathbf{\overset{-}{x}t\overset{-}{z}}$ & $\mathbf{\underline{\overset{-}{x}\overset{+}{z}}}$ & $\mathbf{\overset{-}{x}t\overset{+}{z}}$ \\ \hline
$\mathbf{\overset{-}{x}t\overset{-}{z}}$ & $\mathbf{\underline{\overset{-}{xz}}}$ & $\mathbf{\overset{-}{x}t\overset{-}{z}}$ \\ \hline
\end{tabular}
\caption{Aggregation Table from Interval Sub-Pattern to Interval Overnode Pattern}\label{tab:ad1s_aoa3}
\end{table}

\begin{proof}[Table \ref{tab:ad1s_aoa3}]
\begin{enumerate}
\item As the aggregation involves \kwd{AND} blocks, from Definition \ref{def:ser}, it follows that each execution resulting from the aggregation consists of an execution of the \emph{Overnode's Overnode Child} interleaved with an execution of the \emph{Overnode's Undernode Children}.
\item We prove the correctness of the table by proving the correctness of the aggregation in each of its line. Considering by case the \emph{Overnode's Overnode Child} classifications:
\begin{description}
\item[$\mathbf{\overset{0}{x}t\overset{-}{z}}$]
\begin{enumerate}
\item From Definition \ref{def:iop_0t-}, it follows that there exists an execution of the \emph{Overnode's Overnode Child} where there is a $z$ on the right of $t$ with no $y$ between them, and no $y$ or $x$ on the left of the $t$.
\item Considering by case the \emph{Overnode's Undernode Children} classification:
\begin{description}
\item[$\mathbf{\underline{\overset{+}{x}\overset{+}{z}}}$] The correctness of the row follows directly from Lemma \ref{l:3super}.
\item[$\mathbf{\underline{\overset{+}{x}\overset{0}{z}}}$] 
\begin{enumerate}
\item From Definition \ref{def:isp_1}, it follows that \emph{Overnode's Undernode Children} contains an execution containing a $x$ and no $y$ or $z$.
\item From 1., (a) and i., it follows that the $x$ from $\mathbf{\underline{\overset{+}{x}\overset{0}{z}}}$ can be placed on the left of the $t$ from $\mathbf{\overset{0}{x}t\overset{-}{z}}$, while the $y$ remains on the right of the $t$.
\item From ii. and Definition \ref{def:iop_+t-}, it follows that the aggregation is then correctly classified as $\mathbf{\overset{+}{x}t\overset{-}{z}}$.
\end{enumerate}
\item[$\mathbf{\underline{\overset{0}{x}\overset{+}{z}}}$] 
\begin{enumerate}
\item From Definition \ref{def:isp_1'}, it follows that \emph{Overnode's Undernode Children} contains an execution containing a $z$ and no $y$ or $x$.
\item From 1., (a) and i., it follows that the $z$ from $\mathbf{\underline{\overset{0}{x}\overset{+}{z}}}$ can be placed between the $t$ and the $y$ from $\mathbf{\overset{0}{x}t\overset{-}{z}}$.
\item From ii. and Definition \ref{def:iop_0t+}, it follows that the aggregation is then correctly classified as $\mathbf{\overset{0}{x}t\overset{+}{z}}$.
\end{enumerate}
\item[$\mathbf{\underline{\overset{+}{x}\overset{-}{z}}}$] 
\begin{enumerate}
\item From Definition \ref{def:isp_1r}, it follows that \emph{Overnode's Undernode Children} contains an execution containing a $x$ and a $y$ on its right, while not containing $z$.
\item From 1., (a) and i., it follows that the $x$ from $\mathbf{\underline{\overset{+}{x}\overset{-}{z}}}$ can be placed on the left of the $t$ from $\mathbf{\overset{0}{x}t\overset{-}{z}}$, while the $y$ from $\mathbf{\underline{\overset{+}{x}\overset{-}{z}}}$ can be placed on the right of $t$.
\item From ii. and Definition \ref{def:iop_+t-}, it follows that the aggregation is then correctly classified as $\mathbf{\overset{+}{x}t\overset{-}{z}}$.
\end{enumerate}
\item[$\mathbf{\underline{\overset{0}{x}\overset{0}{z}}}$] The correctness of the row follows directly from Lemma \ref{l:neutral_iop}.
\item[$\mathbf{\underline{\overset{-}{x}\overset{+}{z}}}$] 
\begin{enumerate}
\item From Definition \ref{def:isp_1'l}, it follows that \emph{Overnode's Undernode Children} contains an execution containing a $z$ and a $y$ on its left, while not containing $x$.
\item From 1., (a) and i., it follows that the $z$ from $\mathbf{\underline{\overset{-}{x}\overset{+}{z}}}$ can be placed between the $t$ and the $y$ from $\mathbf{\overset{0}{x}t\overset{-}{z}}$, while the $y$ from $\mathbf{\underline{\overset{-}{x}\overset{+}{z}}}$ remains on the right of $z$.
\item From 1., (a) and i., it follows that the $z$ and the $y$ from $\mathbf{\underline{\overset{-}{x}\overset{+}{z}}}$ can be placed on the right the $t$ and the $y$ from $\mathbf{\overset{0}{x}t\overset{-}{z}}$.
\item From ii. and Definition \ref{def:iop_-t+}, it follows that the aggregation is then correctly classified as $\mathbf{\overset{-}{x}t\overset{+}{z}}$.
\item From iii. and Definition \ref{def:iop_0t-}, it follows that the aggregation is then correctly classified as $\mathbf{\overset{0}{x}t\overset{-}{z}}$.
\item From iv. and v., it follows that the classification of the row is correct.
\end{enumerate}
\item[$\mathbf{\underline{\overset{-}{xz}}}$] The correctness of the row follows directly from Lemma \ref{l:htb}.
\end{description}
\end{enumerate}
\item[$\mathbf{\overset{-}{x}t\overset{0}{z}}$]
\begin{enumerate}
\item From Definition \ref{def:iop_-t0}, it follows that there exists an execution of the \emph{Overnode's Overnode Child} where there is a $z$ on the right of $t$ with no $y$ between them, and no $y$ or $x$ on the left of the $t$.
\item Considering by case the \emph{Overnode's Undernode Children} classification:
\begin{description}
\item[$\mathbf{\underline{\overset{+}{x}\overset{+}{z}}}$] The correctness of the row follows directly from Lemma \ref{l:3super}.
\item[$\mathbf{\underline{\overset{+}{x}\overset{0}{z}}}$] 
\begin{enumerate}
\item From Definition \ref{def:isp_1}, it follows that \emph{Overnode's Undernode Children} contains an execution containing a $x$ and no $y$ or $z$.
\item From 1., (a) and i., it follows that the $x$ from $\mathbf{\underline{\overset{+}{x}\overset{0}{z}}}$ can be placed between the $y$ and $t$ from $\mathbf{\overset{-}{x}t\overset{0}{z}}$.
\item From ii. and Definition \ref{def:iop_+t0}, it follows that the aggregation is then correctly classified as $\mathbf{\overset{+}{x}t\overset{0}{z}}$.
\end{enumerate}
\item[$\mathbf{\underline{\overset{0}{x}\overset{+}{z}}}$] 
\begin{enumerate}
\item From Definition \ref{def:isp_1'}, it follows that \emph{Overnode's Undernode Children} contains an execution containing a $z$ and no $y$ or $x$.
\item From 1., (a) and i., it follows the $z$ from $\mathbf{\underline{\overset{0}{x}\overset{+}{z}}}$ can be placed on the right of the $y$ and the $t$ from $\mathbf{\overset{-}{x}t\overset{0}{z}}$.
\item From ii. and Definition \ref{def:iop_-t+}, it follows that the aggregation is then correctly classified as $\mathbf{\overset{-}{x}t\overset{+}{z}}$.
\end{enumerate}
\item[$\mathbf{\underline{\overset{+}{x}\overset{-}{z}}}$] 
\begin{enumerate}
\item From Definition \ref{def:isp_1r}, it follows that \emph{Overnode's Undernode Children} contains an execution containing a $x$ and a $y$ on its right, while not containing $z$.
\item From 1., (a) and i., it follows that the $x$ from $\mathbf{\underline{\overset{+}{x}\overset{-}{z}}}$ can be placed between the $y$ and $t$ from $\mathbf{\overset{-}{x}t\overset{0}{z}}$, while the $y$ from $\mathbf{\underline{\overset{+}{x}\overset{-}{z}}}$ can be placed on the left of $t$.
\item From 1., (a) and i., it follows that the $x$ and the $y$ from $\mathbf{\underline{\overset{+}{x}\overset{-}{z}}}$ can be placed on the left of the $y$ and $t$ from $\mathbf{\overset{-}{x}t\overset{0}{z}}$.
\item From ii. and Definition \ref{def:iop_+t-}, it follows that the aggregation is then correctly classified as $\mathbf{\overset{+}{x}t\overset{-}{z}}$.
\item From iii. and Definition \ref{def:iop_-t0}, it follows that the aggregation is then correctly classified as $\mathbf{\overset{-}{x}t\overset{0}{z}}$.
\item From iv. and v., it follows that the classification of the row is correct.
\end{enumerate}
\item[$\mathbf{\underline{\overset{0}{x}\overset{0}{z}}}$] The correctness of the row follows directly from Lemma \ref{l:neutral_iop}.
\item[$\mathbf{\underline{\overset{-}{x}\overset{+}{z}}}$] 
\begin{enumerate}
\item From Definition \ref{def:isp_1'l}, it follows that \emph{Overnode's Undernode Children} contains an execution containing a $z$ and a $y$ on its left, while not containing $x$.
\item From 1., (a) and i., it follows that the $z$ from $\mathbf{\underline{\overset{-}{x}\overset{+}{z}}}$ can be placed on the right of the $t$ from $\mathbf{\overset{-}{x}t\overset{0}{z}}$, and the $y$ from $\mathbf{\underline{\overset{-}{x}\overset{+}{z}}}$ can be placed on the left of $t$.
\item From ii. and Definition \ref{def:iop_-t+}, it follows that the aggregation is then correctly classified as $\mathbf{\overset{-}{x}t\overset{+}{z}}$.
\end{enumerate}
\item[$\mathbf{\underline{\overset{-}{xz}}}$] The correctness of the row follows directly from Lemma \ref{l:htb}.
\end{description}
\end{enumerate}
\item[$\mathbf{\overset{-}{x}t\overset{-}{z}}$]
\begin{enumerate}
\item From Definition \ref{def:iop_-t-}, it follows that there exists an execution of the \emph{Overnode's Overnode Child} where there is a $z$ on the right of $t$ with no $y$ between them, and no $y$ or $x$ on the left of the $t$.
\item Considering by case the \emph{Overnode's Undernode Children} classification:
\begin{description}
\item[$\mathbf{\underline{\overset{+}{x}\overset{+}{z}}}$] The correctness of the row follows directly from Lemma \ref{l:3super}.
\item[$\mathbf{\underline{\overset{+}{x}\overset{0}{z}}}$] 
\begin{enumerate}
\item From Definition \ref{def:isp_1}, it follows that \emph{Overnode's Undernode Children} contains an execution containing a $x$ and no $y$ or $z$.
\item From 1., (a) and i., it follows that the $x$ from $\mathbf{\underline{\overset{+}{x}\overset{0}{z}}}$ can be placed between the $t$ and its left $y$ from $\mathbf{\overset{-}{x}t\overset{-}{z}}$.
\item From ii. and Definition \ref{def:iop_+t-}, it follows that the aggregation is then correctly classified as $\mathbf{\overset{+}{x}t\overset{-}{z}}$.
\end{enumerate}
\item[$\mathbf{\underline{\overset{0}{x}\overset{+}{z}}}$] 
\begin{enumerate}
\item From Definition \ref{def:isp_1'}, it follows that \emph{Overnode's Undernode Children} contains an execution containing a $z$ and no $y$ or $x$.
\item From 1., (a) and i., it follows that the $x$ from $\mathbf{\underline{\overset{0}{x}\overset{+}{z}}}$ can be placed between the $t$ and its right $y$ from $\mathbf{\overset{-}{x}t\overset{-}{z}}$.
\item From ii. and Definition \ref{def:iop_-t+}, it follows that the aggregation is then correctly classified as $\mathbf{\overset{-}{x}t\overset{+}{z}}$.
\end{enumerate}
\item[$\mathbf{\underline{\overset{+}{x}\overset{-}{z}}}$] 
\begin{enumerate}
\item From Definition \ref{def:isp_1r}, it follows that \emph{Overnode's Undernode Children} contains an execution containing a $x$ and a $y$ on its right, while not containing $z$.
\item From 1., (a) and i., it follows that the $x$ from $\mathbf{\underline{\overset{+}{x}\overset{-}{z}}}$ can be placed between the $t$ and its left $y$ from $\mathbf{\overset{-}{x}t\overset{-}{z}}$, while the $y$ from $\mathbf{\underline{\overset{+}{x}\overset{-}{z}}}$ can be placed on the right of the $t$.
\item From ii. and Definition \ref{def:iop_+t-}, it follows that the aggregation is then correctly classified as $\mathbf{\overset{+}{x}t\overset{-}{z}}$.
\end{enumerate}
\item[$\mathbf{\underline{\overset{0}{x}\overset{0}{z}}}$] The correctness of the row follows directly from Lemma \ref{l:neutral_iop}.
\item[$\mathbf{\underline{\overset{-}{x}\overset{+}{z}}}$] 
\begin{enumerate}
\item From Definition \ref{def:isp_1'l}, it follows that \emph{Overnode's Undernode Children} contains an execution containing a $z$ and a $y$ on its left, while not containing $x$.
\item From 1., (a) and i., it follows that the $x$ from $\mathbf{\underline{\overset{-}{x}\overset{+}{z}}}$ can be placed between the $t$ and its right $y$ from $\mathbf{\overset{-}{x}t\overset{-}{z}}$, while the $y$ from $\mathbf{\underline{\overset{-}{x}\overset{+}{z}}}$ can be placed on the left of the $t$.
\item From ii. and Definition \ref{def:iop_-t+}, it follows that the aggregation is then correctly classified as $\mathbf{\overset{-}{x}t\overset{+}{z}}$.
\end{enumerate}
\item[$\mathbf{\underline{\overset{-}{xz}}}$] The correctness of the row follows directly from Lemma \ref{l:htb}.
\end{description}
\end{enumerate}
\end{description}
\end{enumerate}
\end{proof}
\newpage
\subsection{Generalised Sequence Pattern}

\begin{definition}[Generalised Sequence Pattern $\mathbf{\overset{+}{x}\overset{+}{z}}$]\label{def:gsc_x+z+}
Exists an execution such that: both $x$ and $z$ sub-patterns are fulfilled and in the correct order.

\noindent\textbf{Formally}: 

\noindent Given a process block $B$, it belongs to this class if and only if:
\begin{itemize}
\item $\exists \exe \in \Exe{B}$ such that:
\begin{itemize}
\item $\exists $x$ \in \exe$ such that:
\begin{itemize}
\item $\not \exists $y$ \in \exe$ such that $x \preceq y$, and
\item $\exists$ z $\in \exe$ such that:
\begin{itemize}
\item $x \preceq z$, and
\item $\not \exists $k$ \in \exe$ such that $z \preceq k$
\end{itemize}
\end{itemize}
\end{itemize}
\end{itemize}
\end{definition}

\begin{definition}[Generalised Sequence Pattern $\mathbf{\overset{+}{x}\overset{0}{z}}$]\label{def:gsc_x+z0}
Exists an execution such that: $x$ sub-pattern satisfied, but $z$ sub-pattern not-satisfied on the right of the former, neither invalidated.

\noindent\textbf{Formally}: 

\noindent Given a process block $B$, it belongs to this class if and only if:
\begin{itemize}
\item $\exists \exe \in \Exe{B}$ such that:
\begin{itemize}
\item $\exists $x$ \in \exe$ such that:
\begin{itemize}
\item $\not \exists $y$ \in \exe$ such that $x \preceq y$, and
\item $\not \exists$ z $\in \exe$ such that $x \preceq z$, and
\item $\not \exists $k$ \in \exe$ such that $x \preceq k$
\end{itemize}
\end{itemize}
\end{itemize}
\end{definition}

\begin{definition}[Generalised Sequence Pattern $\mathbf{\overset{0}{x}\overset{+}{z}}$]\label{def:gsc_x0z+}
Exists an execution such that: $z$ sub-pattern satisfied and $x$ is neither satisfied or invalidated on the left of the former.

\noindent\textbf{Formally}: 

\noindent Given a process block $B$, it belongs to this class if and only if:
\begin{itemize}
\item $\exists \exe \in \Exe{B}$ such that:
\begin{itemize}
\item $\exists $z$ \in \exe$ such that:
\begin{itemize}
\item $\not \exists $y$ \in \exe$ , and
\item $\not \exists$ x $\in \exe$ such that $x \preceq z$, and
\item $\not \exists $k$ \in \exe$ such that $z \preceq k$
\end{itemize}
\end{itemize}
\end{itemize}
\end{definition}

\begin{definition}[Generalised Sequence Pattern $\mathbf{\overset{0}{x}\overset{0}{z}}$]\label{def:gsc_x0z0}
Exists an execution such that: neither sub-pattern invalidated or satisfied.

\noindent\textbf{Formally}: 

\noindent Given a process block $B$, it belongs to this class if and only if:
\begin{itemize}
\item $\exists \exe \in \Exe{B}$ such that:
\begin{itemize}
\item $\not \exists $x$ \in \exe$,
\item $\not \exists $y$ \in \exe$,
\item $\not \exists$ z $\in \exe$ , and
\item $\not \exists $k$ \in \exe$
\end{itemize}
\end{itemize}
\end{definition}

\begin{definition}[Generalised Sequence Pattern $\mathbf{\overset{-}{x}\overset{+}{z}}$]\label{def:gsc_x-z+}
Exists an execution such that: $z$ sub-pattern satisfied and $x$ is invalidated on the left of the former.

\noindent\textbf{Formally}: 

\noindent Given a process block $B$, it belongs to this class if and only if:
\begin{itemize}
\item $\exists \exe \in \Exe{B}$ such that:
\begin{itemize}
\item $\exists $z$ \in \exe$ such that:
\begin{itemize}
\item $\exists $y$ \in \exe$ such that $y \preceq z$, 
\item $\not \exists$ x $\in \exe$ such that $x \preceq z$, and
\item $\not \exists $k$ \in \exe$ such that $z \preceq k$
\end{itemize}
\end{itemize}
\end{itemize}
\end{definition}

\begin{definition}[$\mathbf{\overset{+}{x}\overset{-}{z}}$]\label{def:gsc_x+z-}
Exists an execution such that: $x$ sub-pattern satisfied and $z$ is invalidated on the right of the former.

\noindent\textbf{Formally}: 

\noindent Given a process block $B$, it belongs to this class if and only if:
\begin{itemize}
\item $\exists \exe \in \Exe{B}$ such that:
\begin{itemize}
\item $\exists $x$ \in \exe$ such that:
\begin{itemize}
\item $\exists $k$ \in \exe$ such that $x \preceq k$, 
\item $\not \exists$ y $\in \exe$ such that $k \preceq y$, and
\item $\not \exists $z$ \in \exe$ such that $k \preceq z$
\end{itemize}
\end{itemize}
\end{itemize}
\end{definition}

\begin{definition}[Generalised Sequence Pattern $\mathbf{\overset{0}{x}\overset{-}{z}}$]\label{def:gsc_x0z-}
Exists an execution: $z$ is falsified and $x$ is neither satisfied or falsified on the left of $z$.

\noindent\textbf{Formally}: 

\noindent Given a process block $B$, it belongs to this class if and only if:
\begin{itemize}
\item $\exists \exe \in \Exe{B}$ such that:
\begin{itemize}
\item $\not \exists $y$ \in \exe$, and
\item $\exists $k$ \in \exe$ such that:
\begin{itemize}
\item $\not \exists$ x $\in \exe$ such that $x \preceq z$, and
\item $\not \exists $z$ \in \exe$ such that $k \preceq z$
\end{itemize}
\end{itemize}
\end{itemize}
\end{definition}

\begin{definition}[Generalised Sequence Pattern $\mathbf{\overset{-}{xz}}$]\label{def:gsc_xz-}
Exists an execution: $z$ is falsified and $x$ is neither satisfied or falsified on the left of $z$.

\noindent\textbf{Formally}: 

\noindent Given a process block $B$, it belongs to this class if and only if:
\begin{itemize}
\item $\forall \exe \in \Exe{B}$ such that:
\begin{itemize}
\item $\exists $y$ \in \exe$, such that:
\begin{itemize}
\item $\not \exists$ x $\in \exe$, such that $y \preceq x$, and
\item $\forall$ z $\in \exe$, such that $y \preceq z$:
\begin{itemize}
\item $\exists$ k $\in \exe$, such that $z \preceq k$
\end{itemize}
\end{itemize}
\item OR $\exists $y$ \in \exe$, such that:
\begin{itemize}
\item $\not \exists$ x $\in \exe$, such that $y \preceq x$, and
\item $\exists$ z $\in \exe$, such that $y \preceq z$:
\begin{itemize}
\item $\not\exists$ k $\in \exe$, such that $z \preceq k$
\end{itemize}
\end{itemize}
\end{itemize}
\end{itemize}
\end{definition}

\begin{theorem}[Classification Completeness for Generalised Sequence Pattern]
The set of possible evaluations of Generalised Sequence Pattern is completely covered by the provided set of classifications. 
\end{theorem}

\begin{proof}

\begin{enumerate}
\item The Generalised Sequence Pattern is composed of two sub patterns on the left of $t$, the far ($x$) and the near ($z$) one, and each allows 3 possible evaluations: whether the partial requirement is satisfied, failing, or in a neutral state.
\item The near and far sub patterns are independent, and the only shared element between the two sub patterns is $t$.
\item From 1. and 2., it follows that the amount of possible combinations in Generalised Sequence Pattern is $3^2$, hence 9 possible combinations.
\item 7 of the 9 possible combinations are covered by the classifications described from Definition \ref{def:gsc_x+z+} to Definition \ref{def:gsc_x0z-}.
\item The remainder 2 combinations are covered by the classification in Definition \ref{def:gsc_xz-}.
\item From 3., 4., and 5., it follows that every possible combination is covered by the classifications allowed while evaluating the Generalised Sequence Pattern.
\end{enumerate}
\end{proof}

\subsubsection{Lemmas}

\begin{lemma}[Aggregation Neutral Class]\label{l:gsc_neutral}
Given a process block $A$, assigned to the evaluation class $\mathbf{\overset{0}{x}\overset{0}{z}}$, and another process block $B$, assigned to any of the available evaluation classes. Let $C$ be a process block having $A$ and $B$ as its sub-blocks, then the evaluation class of $C$ is the same class as the process block $B$.
\end{lemma}

\begin{proof}
This proof follows closely the proof for Lemma \ref{l:neutral}
\end{proof}

\begin{lemma}[\kwd{AND} Super Element]\label{l:gsc_super}
Given a process block $A$, assigned to the evaluation class $\mathbf{\overset{+}{x}\overset{+}{z}}$, and another process block $B$, assigned to any of the available evaluation classes. Let $B$ be a process block of type \kwd{AND} having $A$ as one of its sub-blocks, then the evaluation class of $B$ is always $\mathbf{\overset{+}{x}\overset{+}{z}}$.
\end{lemma}

\begin{proof}
\begin{enumerate}
\item From Definition \ref{def:gsc_x+z+}, it follows the classification $\mathbf{\overset{+}{x}\overset{+}{z}}$ contains the required properties to satisfy the Generalised Sequence Pattern.
\item From the premise, we have that a process block $A$ is classified as $\mathbf{\overset{+}{x}\overset{+}{z}}$ and is one of the sub-blocks of a process block $C$ of type \kwd{AND}.
\item From 1., 2., and Definition \ref{def:ser}, it follows that as the requirements for the Generalised Sequence Pattern are satisfied already by $A$, then one of the possible executions of $B$ is the result of append the partial executions of from the other sub-blocks of $B$ either before of after the partial execution of $A$ that allows it to be classified as $\mathbf{\overset{+}{x}\overset{+}{z}}$.
\item From 3., it follows that there exists an execution in $B$ that maintains the same properties as the one allowing to classify $A$ as  $\mathbf{\overset{+}{x}\overset{+}{z}}$.
\item From 4., and Definition \ref{def:gsc_x+z+}, it follows that $B$ can be indeed classified as  $\mathbf{\overset{+}{x}\overset{+}{z}}$.
\end{enumerate}
\end{proof}

\begin{lemma}[\kwd{SEQ} Right Dominant]\label{l:srd}
Given an aggregation $\kwd{SEQ}(A,B)$, where $B$ is classified either as $\mathbf{\overset{+}{x}\overset{+}{z}}$ , $\mathbf{\overset{-}{x}\overset{+}{z}}$ , $\mathbf{\overset{0}{x}\overset{-}{z}}$, $\mathbf{\overset{+}{x}\overset{-}{z}}$ or $\mathbf{\overset{-}{xz}}$, then the result of the aggregation is the same class as $B$ independently from the classification of $A$.
\end{lemma}

\begin{proof}
If one of such classes is on the right, it does not matter the left part, because it becomes irrelevant for various reasons, for instance because $B$ is a super, or because $B$ be contains a failure invalidating whatever $A$ can bring.
\end{proof}

\begin{lemma}[Both by Unblocking]\label{l:bbu}
Given an aggregation $\kwd{SEQ}(A,B)$, where $B$ is classified as $\mathbf{\overset{+}{x}\overset{0}{z}}$, the result is the same as $\kwd{XOR}(A,B)$.
\end{lemma}

\begin{proof}
$\mathbf{\overset{+}{x}\overset{0}{z}}$ is non blocking as it does not contain a $y$ or a $k$. In accordance, whatever is on the left can be potentially be used to improve the class through a an AND aggregation for both classifications. As well, this type of aggregation cannot improve $\mathbf{\overset{+}{x}\overset{0}{z}}$ as it is already on the right, unless the left element classification is already strictly better, but this is captured by the XOR way of aggregating plus the priority given by the lattice for pruning strictly worse classifications.
\end{proof}

\begin{lemma}[Mirror]\label{l:mirror}
Given an aggregation $\kwd{AND}(A,B)$, where $A$ and $B$ have the same classification, then $\kwd{AND}(A,B)$ has the same classification.
\end{lemma}

\begin{proof}
Provable by contradiction. Assuming that $\kwd{AND}(A,B)$ can achieve a better classification than $A$ or $B$, which entails that either $A$ or $B$ contain some elements capable of improving the classification of the other. This is verifiably false as $A$ and $B$ are constrained to have the same classification, and by going through the possible classifications and their properties.
\end{proof}

\subsubsection{Aggregations}~\\
\noindent
The classification of the undernode children according to the Generalised Sequence Pattern is described in the following tables. Mind that the aggregation is performed pair-wise left to right, with the result being aggregated with next brother on the right. While an ordered aggregation is not necessarily required when dealing with parent nodes of type \kwd{AND}, however this allows to deal with both types of parent nodes in the same way.

\begin{table}[ht!]
\centering
\begin{tabular}{|c|c|c|c|}
\hline
A & B & \kwd{SEQ}(A, B) & \kwd{AND}(A, B) \\ \hline
$\mathbf{\overset{+}{x}\overset{+}{z}}$ & $\mathbf{\overset{+}{x}\overset{+}{z}}$ & $\mathbf{\overset{+}{x}\overset{+}{z}}$ & $\mathbf{\overset{+}{x}\overset{+}{z}}$ \\ \hline
$\mathbf{\overset{+}{x}\overset{+}{z}}$ & $\mathbf{\overset{+}{x}\overset{0}{z}}$ & $\mathbf{\overset{+}{x}\overset{+}{z}}$ & $\mathbf{\overset{+}{x}\overset{+}{z}}$ \\ \hline
$\mathbf{\overset{+}{x}\overset{+}{z}}$ & $\mathbf{\overset{0}{x}\overset{+}{z}}$ & $\mathbf{\overset{+}{x}\overset{+}{z}}$ & $\mathbf{\overset{+}{x}\overset{+}{z}}$ \\ \hline
$\mathbf{\overset{+}{x}\overset{+}{z}}$ & $\mathbf{\overset{0}{x}\overset{0}{z}}$& $\mathbf{\overset{+}{x}\overset{+}{z}}$ & $\mathbf{\overset{+}{x}\overset{+}{z}}$ \\ \hline
$\mathbf{\overset{+}{x}\overset{+}{z}}$ & $\mathbf{\overset{-}{x}\overset{+}{z}}$& $\mathbf{\overset{-}{x}\overset{+}{z}}$ & $\mathbf{\overset{+}{x}\overset{+}{z}}$ \\ \hline
$\mathbf{\overset{+}{x}\overset{+}{z}}$ & $\mathbf{\overset{+}{x}\overset{-}{z}}$& $\mathbf{\overset{+}{x}\overset{-}{z}}$ & $\mathbf{\overset{+}{x}\overset{+}{z}}$ \\ \hline
$\mathbf{\overset{+}{x}\overset{+}{z}}$ & $\mathbf{\overset{0}{x}\overset{-}{z}}$& $\mathbf{\overset{0}{x}\overset{-}{z}}$ & $\mathbf{\overset{+}{x}\overset{+}{z}}$ \\ \hline
$\mathbf{\overset{+}{x}\overset{+}{z}}$ & $\mathbf{\overset{-}{xz}}$& $\mathbf{\overset{-}{xz}}$ & $\mathbf{\overset{+}{x}\overset{+}{z}}$ \\ \hline
\end{tabular}
\caption{Generalised Sequence Aggregations}\label{t:so_ad2_alt_agg_++}
\end{table}

\begin{proof}[Table \ref{t:so_ad2_alt_agg_++}]
\begin{enumerate}
\item From Definition \ref{def:gsc_x+z+}, it follows that $A$ contains an execution where a $x$ is followed by or on the same task as a $z$, and there are no $y$ or $k$ on their respective right.
\item For $B$:
\begin{description}
\item[$\mathbf{\overset{+}{x}\overset{+}{z}}$]
\begin{enumerate}
\item From Definition \ref{def:gsc_x+z+}, it follows that $B$ contains an execution where a $x$ is followed by or on the same task as a $z$, and there are no $y$ or $k$ on their respective right.
\item For the aggregation types:
\begin{description}
\item[$\kwd{SEQ}(A,B)$]
\begin{enumerate}
\item The correctness of this row follows directly from Lemma \ref{l:srd}.
\end{enumerate}
\item[$\kwd{AND}(A,B)$]
\begin{enumerate}
\item The correctness of this row follows directly from Lemma \ref{l:gsc_super}.
\end{enumerate}
\end{description}
\end{enumerate}

\item[$\mathbf{\overset{+}{x}\overset{0}{z}}$]
\begin{enumerate}
\item From Definition \ref{def:gsc_x+z0}, it follows that $B$ contains an execution where a $x$, and there are no $y$ or $k$ on its right.
\item For the aggregation types:
\begin{description}
\item[$\kwd{SEQ}(A,B)$]
\begin{enumerate}
\item The correctness of this row follows directly from Lemma \ref{l:bbu} and the preference lattice in Figure \ref{f:ad2_alt_sequndernodeclasses}.
\end{enumerate}
\item[$\kwd{AND}(A,B)$]
\begin{enumerate}
\item The correctness of this row follows directly from Lemma \ref{l:gsc_super}.
\end{enumerate}
\end{description}
\end{enumerate}

\item[$\mathbf{\overset{0}{x}\overset{+}{z}}$]
\begin{enumerate}
\item From Definition \ref{def:gsc_x0z+}, it follows that $B$ contains an execution where a there are no $x$ and $y$, contains a $z$, and there is no $k$ on its right.
\item For the aggregation types:
\begin{description}
\item[$\kwd{SEQ}(A,B)$]
\begin{enumerate}
\item From Definition \ref{def:ser}, it follows that the possible executions of a process block $\kwd{SEQ}(A, B)$ are the concatenation of an execution of $A$ and an execution of $B$.
\item From the hypothesis and 1., it follows that $\kwd{SEQ}(A, B)$ contains an execution where a $x$ is followed by or on the same task as a $z$, and there are no $y$ or $k$ on their respective right.
\item From ii. and Definition \ref{def:gsc_x+z+}, it follows that the aggregation is correctly classified as $\mathbf{\overset{+}{x}\overset{+}{z}}$.
\end{enumerate}
\item[$\kwd{AND}(A,B)$]
\begin{enumerate}
\item The correctness of this row follows directly from Lemma \ref{l:gsc_super}.
\end{enumerate}
\end{description}
\end{enumerate}

\item[$\mathbf{\overset{0}{x}\overset{0}{z}}$]
\begin{enumerate}
\item From Definition \ref{def:gsc_x0z0}, it follows that $B$ contains an execution where a there are no $x$, $y$, $z$ and $k$. 
\item For the aggregation types:
\begin{description}
\item[$\kwd{SEQ}(A,B)$]
\begin{enumerate}
\item The correctness of this row follows directly from Lemma \ref{l:gsc_neutral}.
\end{enumerate}
\item[$\kwd{AND}(A,B)$]
\begin{enumerate}
\item The correctness of this row follows directly from Lemma \ref{l:gsc_super}.
\end{enumerate}
\end{description}
\end{enumerate}

\item[$\mathbf{\overset{-}{x}\overset{+}{z}}$]
\begin{enumerate}
\item From Definition \ref{def:gsc_x-z+}, it follows that $B$ contains an execution containing a $z$, there is no $k$ on its right, and there is a $y$ on the $z$ left.
\item For the aggregation types:
\begin{description}
\item[$\kwd{SEQ}(A,B)$]
\begin{enumerate}
\item The correctness of this row follows directly from Lemma \ref{l:srd}.
\end{enumerate}
\item[$\kwd{AND}(A,B)$]
\begin{enumerate}
\item The correctness of this row follows directly from Lemma \ref{l:gsc_super}.
\end{enumerate}
\end{description}
\end{enumerate}

\item[$\mathbf{\overset{+}{x}\overset{-}{z}}$]
\begin{enumerate}
\item From Definition \ref{def:gsc_x+z-}, it follows that $B$ contains an execution containing a $x$, there is no $y$ on its right, and there is a $k$ on the right, with no $z$ on the right of the $k$.
\item For the aggregation types:
\begin{description}
\item[$\kwd{SEQ}(A,B)$]
\begin{enumerate}
\item The correctness of this row follows directly from Lemma \ref{l:srd}.
\end{enumerate}
\item[$\kwd{AND}(A,B)$]
\begin{enumerate}
\item The correctness of this row follows directly from Lemma \ref{l:gsc_super}.
\end{enumerate}
\end{description}
\end{enumerate}

\item[$\mathbf{\overset{0}{x}\overset{-}{z}}$]
\begin{enumerate}
\item From Definition \ref{def:gsc_x0z-}, it follows that $B$ contains an execution containing a $k$ with no $z$ on its right, and not containing any $x$ or $y$.
\item For the aggregation types:
\begin{description}
\item[$\kwd{SEQ}(A,B)$]
\begin{enumerate}
\item The correctness of this row follows directly from Lemma \ref{l:srd}.
\end{enumerate}
\item[$\kwd{AND}(A,B)$]
\begin{enumerate}
\item The correctness of this row follows directly from Lemma \ref{l:gsc_super}.
\end{enumerate}
\end{description}
\end{enumerate}

\item[$\mathbf{\overset{-}{xz}}$]
\begin{enumerate}
\item From Definition \ref{def:gsc_xz-}, it follows that every execution in $B$ contains a $y$ and a $k$ with no $x$ and $z$ on their respective rights, OR contains a $y$ with no $x$ on its right, and no $z$ or $k$.
\item For the aggregation types:
\begin{description}
\item[$\kwd{SEQ}(A,B)$]
\begin{enumerate}
\item The correctness of this row follows directly from Lemma \ref{l:srd}.
\end{enumerate}
\item[$\kwd{AND}(A,B)$]
\begin{enumerate}
\item The correctness of this row follows directly from Lemma \ref{l:gsc_super}.
\end{enumerate}
\end{description}
\end{enumerate}
\end{description}
\end{enumerate}
\end{proof}

\begin{table}[ht!]
\centering
\begin{tabular}{|c|c|c|c|}
\hline
A & B & \kwd{SEQ}(A, B) & \kwd{AND}(A, B) \\ \hline
$\mathbf{\overset{+}{x}\overset{0}{z}}$ & $\mathbf{\overset{+}{x}\overset{+}{z}}$ & $\mathbf{\overset{+}{x}\overset{+}{z}}$ & $\mathbf{\overset{+}{x}\overset{+}{z}}$ \\ \hline
$\mathbf{\overset{+}{x}\overset{0}{z}}$ & $\mathbf{\overset{+}{x}\overset{0}{z}}$ & $\mathbf{\overset{+}{x}\overset{0}{z}}$ & $\mathbf{\overset{+}{x}\overset{0}{z}}$ \\ \hline
$\mathbf{\overset{+}{x}\overset{0}{z}}$ & $\mathbf{\overset{0}{x}\overset{+}{z}}$ & $\mathbf{\overset{+}{x}\overset{+}{z}}$ & $\mathbf{\overset{+}{x}\overset{+}{z}}$ \\ \hline
$\mathbf{\overset{+}{x}\overset{0}{z}}$ & $\mathbf{\overset{0}{x}\overset{0}{z}}$& $\mathbf{\overset{+}{x}\overset{0}{z}}$ & $\mathbf{\overset{+}{x}\overset{0}{z}}$ \\ \hline
$\mathbf{\overset{+}{x}\overset{0}{z}}$ & $\mathbf{\overset{-}{x}\overset{+}{z}}$& $\mathbf{\overset{-}{x}\overset{+}{z}}$ & $\mathbf{\overset{+}{x}\overset{+}{z}}$ \\ \hline
$\mathbf{\overset{+}{x}\overset{0}{z}}$ & $\mathbf{\overset{+}{x}\overset{-}{z}}$& $\mathbf{\overset{+}{x}\overset{-}{z}}$ & $\mathbf{\overset{+}{x}\overset{0}{z}}$ \\ \hline
$\mathbf{\overset{+}{x}\overset{0}{z}}$ & $\mathbf{\overset{0}{x}\overset{-}{z}}$& $\mathbf{\overset{0}{x}\overset{-}{z}}$ & $\mathbf{\overset{+}{x}\overset{0}{z}}$ \\ \hline
$\mathbf{\overset{+}{x}\overset{0}{z}}$ & $\mathbf{\overset{-}{xz}}$& $\mathbf{\overset{-}{xz}}$ & $\mathbf{\overset{+}{x}\overset{0}{z}}$ \\ \hline
\end{tabular}
\caption{Generalised Sequence Aggregations}\label{t:so_ad2_alt_agg_+0}
\end{table}

\begin{proof}[Table \ref{t:so_ad2_alt_agg_+0}]
\begin{enumerate}
\item From Definition \ref{def:gsc_x+z0}, it follows that $A$ contains an execution where a $x$, and there are no $y$ or $k$ on its right.
\item For $B$:
\begin{description}
\item[$\mathbf{\overset{+}{x}\overset{+}{z}}$]
\begin{enumerate}
\item From Definition \ref{def:gsc_x+z+}, it follows that $B$ contains an execution where a $x$ is followed by or on the same task as a $z$, and there are no $y$ or $k$ on their respective right.
\item For the aggregation types:
\begin{description}
\item[$\kwd{SEQ}(A,B)$]
\begin{enumerate}
\item The correctness of this row follows directly from Lemma \ref{l:srd}.
\end{enumerate}
\item[$\kwd{AND}(A,B)$]
\begin{enumerate}
\item The correctness of this row follows directly from Lemma \ref{l:gsc_super}.
\end{enumerate}
\end{description}
\end{enumerate}

\item[$\mathbf{\overset{+}{x}\overset{0}{z}}$]
\begin{enumerate}
\item From Definition \ref{def:gsc_x+z0}, it follows that $B$ contains an execution where a $x$, and there are no $y$ or $k$ on its right.
\item For the aggregation types:
\begin{description}
\item[$\kwd{SEQ}(A,B)$]
\begin{enumerate}
\item The correctness of this row follows directly from Lemma \ref{l:bbu} and the preference lattice in Figure \ref{f:ad2_alt_sequndernodeclasses}.
\end{enumerate}
\item[$\kwd{AND}(A,B)$]
\begin{enumerate}
\item From 1. and (a), it follows that neither $A$ or $B$ contain elements capable of improving the classification over $\mathbf{\overset{+}{x}\overset{0}{z}}$, following the preference order in Figure \ref{f:ad2_alt_sequndernodeclasses}.
\item From i., the result of column \kwd{SEQ} and Lemma \ref{lem:inclusivity}, it follows that the aggregation is correctly classified.
\end{enumerate}
\end{description}
\end{enumerate}

\item[$\mathbf{\overset{0}{x}\overset{+}{z}}$]
\begin{enumerate}
\item From Definition \ref{def:gsc_x0z+}, it follows that $B$ contains an execution where a there are no $x$ and $y$, contains a $z$, and there is no $k$ on its right.
\item For the aggregation types:
\begin{description}
\item[$\kwd{SEQ}(A,B)$]
\begin{enumerate}
\item From Definition \ref{def:ser}, it follows that the possible executions of a process block $\kwd{SEQ}(A, B)$ are the concatenation of an execution of $A$ and an execution of $B$.
\item From 1., (a), and i., it follows that $\kwd{SEQ}(A, B)$ contains a $x$ and a $z$ on its right, with no $y$ on the right of the $x$, and no $k$ on the right of the $z$.
\item From ii. and Definition \ref{def:gsc_x+z+}, it follows that the aggregation is correctly classified as $\mathbf{\overset{+}{x}\overset{+}{z}}$.
\end{enumerate}
\item[$\kwd{AND}(A,B)$]
\begin{enumerate}
\item From the result of column \kwd{SEQ} and Lemma \ref{lem:inclusivity}, it follows that the aggregation is correctly classified.
\end{enumerate}
\end{description}
\end{enumerate}

\item[$\mathbf{\overset{0}{x}\overset{0}{z}}$]
\begin{enumerate}
\item From Definition \ref{def:gsc_x0z0}, it follows that $B$ contains an execution where a there are no $x$, $y$, $z$ and $k$. 
\item For the aggregation types:
\begin{description}
\item[$\kwd{SEQ}(A,B)$]
\begin{enumerate}
\item The correctness of this row follows directly from Lemma \ref{l:gsc_neutral}.
\end{enumerate}
\item[$\kwd{AND}(A,B)$]
\begin{enumerate}
\item The correctness of this row follows directly from Lemma \ref{l:gsc_neutral}.
\end{enumerate}
\end{description}
\end{enumerate}

\item[$\mathbf{\overset{-}{x}\overset{+}{z}}$]
\begin{enumerate}
\item From Definition \ref{def:gsc_x-z+}, it follows that $B$ contains an execution containing a $z$, there is no $k$ on its right, and there is a $y$ on the $z$ left.
\item For the aggregation types:
\begin{description}
\item[$\kwd{SEQ}(A,B)$]
\begin{enumerate}
\item The correctness of this row follows directly from Lemma \ref{l:srd}.
\end{enumerate}
\item[$\kwd{AND}(A,B)$]
\begin{enumerate}
\item From Definition \ref{def:ser}, it follows that each execution resulting from $\kwd{AND}(A, B)$ consists of an execution of $A$ interleaved with an execution of $B$.
\item From 1., (a) and i., it follows that a valid execution of $\kwd{AND}(A, B)$ contains the $z$ from $B$ and has the $x$ from $A$ on the left of $z$ and on the right of $y$ of $B$.
\item From ii. and Definition \ref{def:gsc_x+z+}, it follows that the aggregation is correctly classified as $\mathbf{\overset{+}{x}\overset{+}{z}}$. 
\end{enumerate}
\end{description}
\end{enumerate}

\item[$\mathbf{\overset{+}{x}\overset{-}{z}}$]
\begin{enumerate}
\item From Definition \ref{def:gsc_x+z-}, it follows that $B$ contains an execution containing a $x$, there is no $y$ on its right, and there is a $k$ on the right, with no $z$ on the right of the $k$.
\item For the aggregation types:
\begin{description}
\item[$\kwd{SEQ}(A,B)$]
\begin{enumerate}
\item The correctness of this row follows directly from Lemma \ref{l:srd}.
\end{enumerate}
\item[$\kwd{AND}(A,B)$]
\begin{enumerate}
\item From Definition \ref{def:ser}, it follows that each execution resulting from $\kwd{AND}(A, B)$ consists of an execution of $A$ interleaved with an execution of $B$.
\item From i. and Lemma \ref{lem:inclusivity}, it follows that the executions of $\kwd{SEQ}(B,A)$ are included in $\kwd{AND}(A,B)$.
\item From ii., Lemma \ref{l:bbu} and the preference lattice in Figure \ref{f:ad2_alt_sequndernodeclasses}, it follows that a classification of $\kwd{SEQ}(A,B)$ is $\mathbf{\overset{+}{x}\overset{0}{z}}$.
\item From 1. and (a), it follows that $B$ does not contain elements capable of improving the classification over $\mathbf{\overset{+}{x}\overset{0}{z}}$.
\item From iii., iv. and Lemma \ref{lem:inclusivity} it follows that the aggregation is correctly classified as $\mathbf{\overset{+}{x}\overset{0}{z}}$.
\end{enumerate}
\end{description}
\end{enumerate}

\item[$\mathbf{\overset{0}{x}\overset{-}{z}}$]
\begin{enumerate}
\item From Definition \ref{def:gsc_x0z-}, it follows that $B$ contains an execution containing a $k$ with no $z$ on its right, and not containing any $x$ or $y$.
\item For the aggregation types:
\begin{description}
\item[$\kwd{SEQ}(A,B)$]
\begin{enumerate}
\item The correctness of this row follows directly from Lemma \ref{l:srd}.
\end{enumerate}
\item[$\kwd{AND}(A,B)$]
\begin{enumerate}
\item From Definition \ref{def:ser}, it follows that each execution resulting from $\kwd{AND}(A, B)$ consists of an execution of $A$ interleaved with an execution of $B$.
\item From i. and Lemma \ref{lem:inclusivity}, it follows that the executions of $\kwd{SEQ}(B,A)$ are included in $\kwd{AND}(A,B)$.
\item From ii., Lemma \ref{l:bbu} and the preference lattice in Figure \ref{f:ad2_alt_sequndernodeclasses}, it follows that a classification of $\kwd{SEQ}(A,B)$ is $\mathbf{\overset{+}{x}\overset{0}{z}}$.
\item From 1. and (a), it follows that $B$ does not contain elements capable of improving the classification over $\mathbf{\overset{+}{x}\overset{0}{z}}$.
\item From iii., iv. and Lemma \ref{lem:inclusivity} it follows that the aggregation is correctly classified as $\mathbf{\overset{+}{x}\overset{0}{z}}$.
\end{enumerate}
\end{description}
\end{enumerate}

\item[$\mathbf{\overset{-}{xz}}$]
\begin{enumerate}
\item From Definition \ref{def:gsc_xz-}, it follows that every execution in $B$ contains a $y$ and a $k$ with no $x$ and $z$ on their respective rights, OR contains a $y$ with no $x$ on its right, and no $z$ or $k$.
\item For the aggregation types:
\begin{description}
\item[$\kwd{SEQ}(A,B)$]
\begin{enumerate}
\item The correctness of this row follows directly from Lemma \ref{l:srd}.
\end{enumerate}
\item[$\kwd{AND}(A,B)$]
\begin{enumerate}
\item From Definition \ref{def:ser}, it follows that each execution resulting from $\kwd{AND}(A, B)$ consists of an execution of $A$ interleaved with an execution of $B$.
\item From i. and Lemma \ref{lem:inclusivity}, it follows that the executions of $\kwd{SEQ}(B,A)$ are included in $\kwd{AND}(A,B)$.
\item From ii., Lemma \ref{l:bbu} and the preference lattice in Figure \ref{f:ad2_alt_sequndernodeclasses}, it follows that a classification of $\kwd{SEQ}(A,B)$ is $\mathbf{\overset{+}{x}\overset{0}{z}}$.
\item From 1. and (a), it follows that $B$ does not contain elements capable of improving the classification over $\mathbf{\overset{+}{x}\overset{0}{z}}$.
\item From iii., iv. and Lemma \ref{lem:inclusivity} it follows that the aggregation is correctly classified as $\mathbf{\overset{+}{x}\overset{0}{z}}$.
\end{enumerate}
\end{description}
\end{enumerate}
\end{description}
\end{enumerate}
\end{proof}

\begin{table}[ht!]
\centering
\begin{tabular}{|c|c|c|c|}
\hline
A & B & \kwd{SEQ}(A, B) & \kwd{AND}(A, B) \\ \hline
$\mathbf{\overset{0}{x}\overset{+}{z}}$ & $\mathbf{\overset{+}{x}\overset{+}{z}}$ & $\mathbf{\overset{+}{x}\overset{+}{z}}$ & $\mathbf{\overset{+}{x}\overset{+}{z}}$ \\ \hline
$\mathbf{\overset{0}{x}\overset{+}{z}}$ & $\mathbf{\overset{+}{x}\overset{0}{z}}$ & $\mathbf{\overset{0}{x}\overset{+}{z}}$ / $\mathbf{\overset{+}{x}\overset{0}{z}}$ & $\mathbf{\overset{+}{x}\overset{+}{z}}$ \\ \hline
$\mathbf{\overset{0}{x}\overset{+}{z}}$ & $\mathbf{\overset{0}{x}\overset{+}{z}}$ & $\mathbf{\overset{0}{x}\overset{+}{z}}$ & $\mathbf{\overset{0}{x}\overset{+}{z}}$ \\ \hline
$\mathbf{\overset{0}{x}\overset{+}{z}}$ & $\mathbf{\overset{0}{x}\overset{0}{z}}$& $\mathbf{\overset{0}{x}\overset{+}{z}}$ & $\mathbf{\overset{0}{x}\overset{+}{z}}$ \\ \hline
$\mathbf{\overset{0}{x}\overset{+}{z}}$ & $\mathbf{\overset{-}{x}\overset{+}{z}}$& $\mathbf{\overset{-}{x}\overset{+}{z}}$ & $\mathbf{\overset{-}{x}\overset{+}{z}}$ \\ \hline
$\mathbf{\overset{0}{x}\overset{+}{z}}$ & $\mathbf{\overset{+}{x}\overset{-}{z}}$& $\mathbf{\overset{+}{x}\overset{-}{z}}$ & $\mathbf{\overset{+}{x}\overset{+}{z}}$ \\ \hline
$\mathbf{\overset{0}{x}\overset{+}{z}}$ & $\mathbf{\overset{0}{x}\overset{-}{z}}$& $\mathbf{\overset{0}{x}\overset{-}{z}}$ & $\mathbf{\overset{0}{x}\overset{+}{z}}$\\ \hline
$\mathbf{\overset{0}{x}\overset{+}{z}}$ & $\mathbf{\overset{-}{xz}}$& $\mathbf{\overset{-}{xz}}$ & $\mathbf{\overset{-}{x}\overset{+}{z}}$ \\ \hline
\end{tabular}
\caption{Generalised Sequence Aggregations}\label{t:so_ad2_alt_agg_0+}
\end{table}

\begin{proof}[Table \ref{t:so_ad2_alt_agg_0+}]
\begin{enumerate}
\item From Definition \ref{def:gsc_x0z+}, it follows that $A$ contains an execution where a there are no $x$ and $y$, contains a $z$, and there is no $k$ on its right.
\item For $B$:
\begin{description}
\item[$\mathbf{\overset{+}{x}\overset{+}{z}}$]
\begin{enumerate}
\item From Definition \ref{def:gsc_x+z+}, it follows that $B$ contains an execution where a $x$ is followed by or on the same task as a $z$, and there are no $y$ or $k$ on their respective right.
\item For the aggregation types:
\begin{description}
\item[$\kwd{SEQ}(A,B)$]
\begin{enumerate}
\item The correctness of this row follows directly from Lemma \ref{l:srd}.
\end{enumerate}
\item[$\kwd{AND}(A,B)$]
\begin{enumerate}
\item The correctness of this row follows directly from Lemma \ref{l:gsc_super}.
\end{enumerate}
\end{description}
\end{enumerate}

\item[$\mathbf{\overset{+}{x}\overset{0}{z}}$]
\begin{enumerate}
\item From Definition \ref{def:gsc_x+z0}, it follows that $B$ contains an execution where a $x$, and there are no $y$ or $k$ on its right.
\item For the aggregation types:
\begin{description}
\item[$\kwd{SEQ}(A,B)$]
\begin{enumerate}
\item The correctness of this row follows directly from Lemma \ref{l:bbu} and the preference lattice in Figure \ref{f:ad2_alt_sequndernodeclasses}.
\end{enumerate}
\item[$\kwd{AND}(A,B)$]
\begin{enumerate}
\item The correctness of this row follows directly from the result of row $\mathbf{\overset{+}{x}\overset{0}{z}}$ and $\mathbf{\overset{0}{x}\overset{+}{z}}$ in Table \ref{t:so_ad2_alt_agg_+0}, and Lemma \ref{l:commutativity}
\end{enumerate}
\end{description}
\end{enumerate}

\item[$\mathbf{\overset{0}{x}\overset{+}{z}}$]
\begin{enumerate}
\item From Definition \ref{def:gsc_x0z+}, it follows that $B$ contains an execution where a there are no $x$ and $y$, contains a $z$, and there is no $k$ on its right.
\item For the aggregation types:
\begin{description}
\item[$\kwd{SEQ}(A,B)$]
\begin{enumerate}
\item From Definition \ref{def:ser}, it follows that the possible executions of a process block $\kwd{SEQ}(A, B)$ are the concatenation of an execution of $A$ and an execution of $B$.
\item From 1., (a), and i., it follows that the execution of $A$ appended on the left of the execution of $B$ does not contain any element capable of manipulating the value of $\mathbf{x}$ is the classification.
\item From ii. it follows that the classification it correct.
\end{enumerate}
\item[$\kwd{AND}(A,B)$]
\begin{enumerate}
\item From 1. and (a), it follows that neither $A$ or $B$ contain elements capable of improving the classification over $\mathbf{\overset{+}{x}\overset{0}{z}}$, following the preference order in Figure \ref{f:ad2_alt_sequndernodeclasses}.
\item From i., the result of column \kwd{SEQ} and Lemma \ref{lem:inclusivity}, it follows that the aggregation is correctly classified.
\end{enumerate}
\end{description}
\end{enumerate}

\item[$\mathbf{\overset{0}{x}\overset{0}{z}}$]
\begin{enumerate}
\item From Definition \ref{def:gsc_x0z0}, it follows that $B$ contains an execution where a there are no $x$, $y$, $z$ and $k$. 
\item For the aggregation types:
\begin{description}
\item[$\kwd{SEQ}(A,B)$]
\begin{enumerate}
\item The correctness of this row follows directly from Lemma \ref{l:gsc_neutral}.
\end{enumerate}
\item[$\kwd{AND}(A,B)$]
\begin{enumerate}
\item The correctness of this row follows directly from Lemma \ref{l:gsc_neutral}.
\end{enumerate}
\end{description}
\end{enumerate}

\item[$\mathbf{\overset{-}{x}\overset{+}{z}}$]
\begin{enumerate}
\item From Definition \ref{def:gsc_x-z+}, it follows that $B$ contains an execution containing a $z$, there is no $k$ on its right, and there is a $y$ on the $z$ left.
\item For the aggregation types:
\begin{description}
\item[$\kwd{SEQ}(A,B)$]
\begin{enumerate}
\item The correctness of this row follows directly from Lemma \ref{l:srd}.
\end{enumerate}
\item[$\kwd{AND}(A,B)$]
\begin{enumerate}
\item From Definition \ref{def:ser}, it follows that each execution resulting from $\kwd{AND}(A, B)$ consists of an execution of $A$ interleaved with an execution of $B$.
\item From 1., (a), and i., it follows that $A$ does not contain elements capable of improving the classification of $\mathbf{x}$.
\item From i. Lemma \ref{lem:inclusivity} and the $\kwd{SEQ}(A, B)$ result, it follows that the aggregation is correctly classified as $\mathbf{\overset{-}{x}\overset{+}{z}}$.
\end{enumerate}
\end{description}
\end{enumerate}

\item[$\mathbf{\overset{+}{x}\overset{-}{z}}$]
\begin{enumerate}
\item From Definition \ref{def:gsc_x+z-}, it follows that $B$ contains an execution containing a $x$, there is no $y$ on its right, and there is a $k$ on the right, with no $z$ on the right of the $k$.
\item For the aggregation types:
\begin{description}
\item[$\kwd{SEQ}(A,B)$]
\begin{enumerate}
\item The correctness of this row follows directly from Lemma \ref{l:srd}.
\end{enumerate}
\item[$\kwd{AND}(A,B)$]
\begin{enumerate}
\item From Definition \ref{def:ser}, it follows that each execution resulting from $\kwd{AND}(A, B)$ consists of an execution of $A$ interleaved with an execution of $B$.
\item From 1., (a) and i., it follows that a valid execution of $\kwd{AND}(A, B)$ contains the $z$ from $A$ and has the $x$ from $B$ on the left of $z$ and on the right of $y$ of $B$.
\item From ii. and Definition \ref{def:gsc_x+z+}, it follows that the aggregation is correctly classified as $\mathbf{\overset{+}{x}\overset{+}{z}}$. 
\end{enumerate}
\end{description}
\end{enumerate}

\item[$\mathbf{\overset{0}{x}\overset{-}{z}}$]
\begin{enumerate}
\item From Definition \ref{def:gsc_x0z-}, it follows that $B$ contains an execution containing a $k$ with no $z$ on its right, and not containing any $x$ or $y$.
\item For the aggregation types:
\begin{description}
\item[$\kwd{SEQ}(A,B)$]
\begin{enumerate}
\item The correctness of this row follows directly from Lemma \ref{l:srd}.
\end{enumerate}
\item[$\kwd{AND}(A,B)$]
\begin{enumerate}
\item From Definition \ref{def:ser}, it follows that each execution resulting from $\kwd{AND}(A, B)$ consists of an execution of $A$ interleaved with an execution of $B$.
\item From 1., (a) and i., it follows that a valid execution of $\kwd{AND}(A, B)$ contains the $z$ from $A$ while the $k$ from $B$ is attached on the left, making it irrelevant.
\item From ii. and Definition  \ref{def:gsc_x0z+} it follows that the aggregation is correctly classified as $\mathbf{\overset{0}{x}\overset{+}{z}}$.
\end{enumerate}
\end{description}
\end{enumerate}

\item[$\mathbf{\overset{-}{xz}}$]
\begin{enumerate}
\item From Definition \ref{def:gsc_xz-}, it follows that every execution in $B$ contains a $y$ and a $k$ with no $x$ and $z$ on their respective rights, OR contains a $y$ with no $x$ on its right, and no $z$ or $k$.
\item For the aggregation types:
\begin{description}
\item[$\kwd{SEQ}(A,B)$]
\begin{enumerate}
\item The correctness of this row follows directly from Lemma \ref{l:srd}.
\end{enumerate}
\item[$\kwd{AND}(A,B)$]
\begin{enumerate}
\item From Definition \ref{def:ser}, it follows that each execution resulting from $\kwd{AND}(A, B)$ consists of an execution of $A$ interleaved with an execution of $B$.
\item From 1., (a) and i., it follows that a valid execution of $\kwd{AND}(A, B)$ contains $z$ from $A$ while every other element from $B$ is on the left, where the only relevant one is $y$.
\item From ii. and Definition  \ref{def:gsc_x-z+} it follows that the aggregation is correctly classified as $\mathbf{\overset{-}{x}\overset{+}{z}}$.
\end{enumerate}
\end{description}
\end{enumerate}
\end{description}
\end{enumerate}
\end{proof}

\begin{table}[ht!]
\centering
\begin{tabular}{|c|c|c|c|}
\hline
A & B & \kwd{SEQ}(A, B) & \kwd{AND}(A, B) \\ \hline
$\mathbf{\overset{0}{x}\overset{0}{z}}$ & $\mathbf{\overset{+}{x}\overset{+}{z}}$ & $\mathbf{\overset{+}{x}\overset{+}{z}}$ & $\mathbf{\overset{+}{x}\overset{+}{z}}$ \\ \hline
$\mathbf{\overset{0}{x}\overset{0}{z}}$ & $\mathbf{\overset{+}{x}\overset{0}{z}}$ & $\mathbf{\overset{+}{x}\overset{0}{z}}$ & $\mathbf{\overset{+}{x}\overset{0}{z}}$ \\ \hline
$\mathbf{\overset{0}{x}\overset{0}{z}}$ & $\mathbf{\overset{0}{x}\overset{+}{z}}$ & $\mathbf{\overset{0}{x}\overset{+}{z}}$ & $\mathbf{\overset{0}{x}\overset{+}{z}}$ \\ \hline
$\mathbf{\overset{0}{x}\overset{0}{z}}$ & $\mathbf{\overset{0}{x}\overset{0}{z}}$& $\mathbf{\overset{0}{x}\overset{0}{z}}$ & $\mathbf{\overset{0}{x}\overset{0}{z}}$ \\ \hline
$\mathbf{\overset{0}{x}\overset{0}{z}}$ & $\mathbf{\overset{-}{x}\overset{+}{z}}$& $\mathbf{\overset{-}{x}\overset{+}{z}}$ & $\mathbf{\overset{-}{x}\overset{+}{z}}$ \\ \hline
$\mathbf{\overset{0}{x}\overset{0}{z}}$ & $\mathbf{\overset{+}{x}\overset{-}{z}}$& $\mathbf{\overset{+}{x}\overset{-}{z}}$ & $\mathbf{\overset{+}{x}\overset{-}{z}}$ \\ \hline
$\mathbf{\overset{0}{x}\overset{0}{z}}$ & $\mathbf{\overset{0}{x}\overset{-}{z}}$& $\mathbf{\overset{0}{x}\overset{-}{z}}$ & $\mathbf{\overset{0}{x}\overset{-}{z}}$ \\ \hline
$\mathbf{\overset{0}{x}\overset{0}{z}}$ & $\mathbf{\overset{-}{xz}}$& $\mathbf{\overset{-}{xz}}$ & $\mathbf{\overset{-}{xz}}$ \\ \hline
\end{tabular}
\caption{Generalised Sequence Aggregations}\label{t:so_ad2_alt_agg_0}
\end{table}

\begin{proof}[Table \ref{t:so_ad2_alt_agg_0}]
\begin{enumerate}
\item From Definition \ref{def:gsc_x0z0}, it follows that $A$ contains an execution where a there are no $x$, $y$, $z$ and $k$. 
\item For $B$:
\begin{description}
\item[$\mathbf{\overset{+}{x}\overset{+}{z}}$]
\begin{enumerate}
\item From Definition \ref{def:gsc_x+z+}, it follows that $B$ contains an execution where a $x$ is followed by or on the same task as a $z$, and there are no $y$ or $k$ on their respective right.
\item For the aggregation types:
\begin{description}
\item[$\kwd{SEQ}(A,B)$]
\begin{enumerate}
\item The correctness of this row follows directly from Lemma \ref{l:srd}.
\end{enumerate}
\item[$\kwd{AND}(A,B)$]
\begin{enumerate}
\item The correctness of this row follows directly from Lemma \ref{l:gsc_super}.
\end{enumerate}
\end{description}
\end{enumerate}

\item[$\mathbf{\overset{+}{x}\overset{0}{z}}$]
\begin{enumerate}
\item From Definition \ref{def:gsc_x+z0}, it follows that $B$ contains an execution where a $x$, and there are no $y$ or $k$ on its right.
\item For the aggregation types:
\begin{description}
\item[$\kwd{SEQ}(A,B)$]
\begin{enumerate}
\item The correctness of this row follows directly from Lemma \ref{l:bbu} and the preference lattice in Figure \ref{f:ad2_alt_sequndernodeclasses}.
\end{enumerate}
\item[$\kwd{AND}(A,B)$]
\begin{enumerate}
\item The correctness of this row follows directly from Lemma \ref{l:gsc_neutral}.
\end{enumerate}
\end{description}
\end{enumerate}

\item[$\mathbf{\overset{0}{x}\overset{+}{z}}$]
\begin{enumerate}
\item From Definition \ref{def:gsc_x0z+}, it follows that $B$ contains an execution where a there are no $x$ and $y$, contains a $z$, and there is no $k$ on its right.
\item For the aggregation types:
\begin{description}
\item[$\kwd{SEQ}(A,B)$]
\begin{enumerate}
\item The correctness of this row follows directly from Lemma \ref{l:gsc_neutral}.
\end{enumerate}
\item[$\kwd{AND}(A,B)$]
\begin{enumerate}
\item The correctness of this row follows directly from Lemma \ref{l:gsc_neutral}.
\end{enumerate}
\end{description}
\end{enumerate}

\item[$\mathbf{\overset{0}{x}\overset{0}{z}}$]
\begin{enumerate}
\item From Definition \ref{def:gsc_x0z0}, it follows that $B$ contains an execution where a there are no $x$, $y$, $z$ and $k$. 
\item For the aggregation types:
\begin{description}
\item[$\kwd{SEQ}(A,B)$]
\begin{enumerate}
\item The correctness of this row follows directly from Lemma \ref{l:gsc_neutral}.
\end{enumerate}
\item[$\kwd{AND}(A,B)$]
\begin{enumerate}
\item The correctness of this row follows directly from Lemma \ref{l:gsc_neutral}.
\end{enumerate}
\end{description}
\end{enumerate}

\item[$\mathbf{\overset{-}{x}\overset{+}{z}}$]
\begin{enumerate}
\item From Definition \ref{def:gsc_x-z+}, it follows that $B$ contains an execution containing a $z$, there is no $k$ on its right, and there is a $y$ on the $z$ left.
\item For the aggregation types:
\begin{description}
\item[$\kwd{SEQ}(A,B)$]
\begin{enumerate}
\item The correctness of this row follows directly from Lemma \ref{l:srd}.
\end{enumerate}
\item[$\kwd{AND}(A,B)$]
\begin{enumerate}
\item The correctness of this row follows directly from Lemma \ref{l:gsc_neutral}.
\end{enumerate}
\end{description}
\end{enumerate}

\item[$\mathbf{\overset{+}{x}\overset{-}{z}}$]
\begin{enumerate}
\item From Definition \ref{def:gsc_x+z-}, it follows that $B$ contains an execution containing a $x$, there is no $y$ on its right, and there is a $k$ on the right, with no $z$ on the right of the $k$.
\item For the aggregation types:
\begin{description}
\item[$\kwd{SEQ}(A,B)$]
\begin{enumerate}
\item The correctness of this row follows directly from Lemma \ref{l:srd}.
\end{enumerate}
\item[$\kwd{AND}(A,B)$]
\begin{enumerate}
\item The correctness of this row follows directly from Lemma \ref{l:gsc_neutral}.
\end{enumerate}
\end{description}
\end{enumerate}

\item[$\mathbf{\overset{0}{x}\overset{-}{z}}$]
\begin{enumerate}
\item From Definition \ref{def:gsc_x0z-}, it follows that $B$ contains an execution containing a $k$ with no $z$ on its right, and not containing any $x$ or $y$.
\item For the aggregation types:
\begin{description}
\item[$\kwd{SEQ}(A,B)$]
\begin{enumerate}
\item The correctness of this row follows directly from Lemma \ref{l:srd}.
\end{enumerate}
\item[$\kwd{AND}(A,B)$]
\begin{enumerate}
\item The correctness of this row follows directly from Lemma \ref{l:gsc_neutral}.
\end{enumerate}
\end{description}
\end{enumerate}

\item[$\mathbf{\overset{-}{xz}}$]
\begin{enumerate}
\item From Definition \ref{def:gsc_xz-}, it follows that every execution in $B$ contains a $y$ and a $k$ with no $x$ and $z$ on their respective rights, OR contains a $y$ with no $x$ on its right, and no $z$ or $k$.
\item For the aggregation types:
\begin{description}
\item[$\kwd{SEQ}(A,B)$]
\begin{enumerate}
\item The correctness of this row follows directly from Lemma \ref{l:srd}.
\end{enumerate}
\item[$\kwd{AND}(A,B)$]
\begin{enumerate}
\item The correctness of this row follows directly from Lemma \ref{l:gsc_neutral}.
\end{enumerate}
\end{description}
\end{enumerate}
\end{description}
\end{enumerate}
\end{proof}

\begin{table}[ht!]
\centering
\begin{tabular}{|c|c|c|c|}
\hline
A & B & \kwd{SEQ}(A, B) & \kwd{AND}(A, B) \\ \hline
$\mathbf{\overset{-}{x}\overset{+}{z}}$ & $\mathbf{\overset{+}{x}\overset{+}{z}}$ & $\mathbf{\overset{+}{x}\overset{+}{z}}$ & $\mathbf{\overset{+}{x}\overset{+}{z}}$ \\ \hline
$\mathbf{\overset{-}{x}\overset{+}{z}}$ & $\mathbf{\overset{+}{x}\overset{0}{z}}$ & $\mathbf{\overset{-}{x}\overset{+}{z}}$ / $\mathbf{\overset{+}{x}\overset{0}{z}}$ & $\mathbf{\overset{+}{x}\overset{+}{z}}$ \\ \hline
$\mathbf{\overset{-}{x}\overset{+}{z}}$ & $\mathbf{\overset{0}{x}\overset{+}{z}}$ & $\mathbf{\overset{-}{x}\overset{+}{z}}$ & $\mathbf{\overset{-}{x}\overset{+}{z}}$ \\ \hline
$\mathbf{\overset{-}{x}\overset{+}{z}}$ & $\mathbf{\overset{0}{x}\overset{0}{z}}$& $\mathbf{\overset{-}{x}\overset{+}{z}}$ & $\mathbf{\overset{-}{x}\overset{+}{z}}$ \\ \hline
$\mathbf{\overset{-}{x}\overset{+}{z}}$ & $\mathbf{\overset{-}{x}\overset{+}{z}}$& $\mathbf{\overset{-}{x}\overset{+}{z}}$ & $\mathbf{\overset{-}{x}\overset{+}{z}}$ \\ \hline
$\mathbf{\overset{-}{x}\overset{+}{z}}$ & $\mathbf{\overset{+}{x}\overset{-}{z}}$& $\mathbf{\overset{+}{x}\overset{-}{z}}$ & $\mathbf{\overset{+}{x}\overset{+}{z}}$ \\ \hline
$\mathbf{\overset{-}{x}\overset{+}{z}}$ & $\mathbf{\overset{0}{x}\overset{-}{z}}$& $\mathbf{\overset{0}{x}\overset{-}{z}}$ & $\mathbf{\overset{-}{x}\overset{+}{z}}$ \\ \hline
$\mathbf{\overset{-}{x}\overset{+}{z}}$ & $\mathbf{\overset{-}{xz}}$& $\mathbf{\overset{-}{xz}}$ & $\mathbf{\overset{-}{x}\overset{+}{z}}$ \\ \hline
\end{tabular}
\caption{Generalised Sequence Aggregations}\label{t:so_ad2_alt_agg_-+}
\end{table}

\begin{proof}[Table \ref{t:so_ad2_alt_agg_-+}]
\begin{enumerate}
\item From Definition \ref{def:gsc_x-z+}, it follows that $A$ contains an execution containing a $z$, there is no $k$ on its right, and there is a $y$ on the $z$ left.
\item For $B$:
\begin{description}
\item[$\mathbf{\overset{+}{x}\overset{+}{z}}$]
\begin{enumerate}
\item From Definition \ref{def:gsc_x+z+}, it follows that $B$ contains an execution where a $x$ is followed by or on the same task as a $z$, and there are no $y$ or $k$ on their respective right.
\item For the aggregation types:
\begin{description}
\item[$\kwd{SEQ}(A,B)$]
\begin{enumerate}
\item The correctness of this row follows directly from Lemma \ref{l:srd}.
\end{enumerate}
\item[$\kwd{AND}(A,B)$]
\begin{enumerate}
\item The correctness of this row follows directly from Lemma \ref{l:gsc_super}.
\end{enumerate}
\end{description}
\end{enumerate}

\item[$\mathbf{\overset{+}{x}\overset{0}{z}}$]
\begin{enumerate}
\item From Definition \ref{def:gsc_x+z0}, it follows that $B$ contains an execution where a $x$, and there are no $y$ or $k$ on its right.
\item For the aggregation types:
\begin{description}
\item[$\kwd{SEQ}(A,B)$]
\begin{enumerate}
\item The correctness of this row follows directly from Lemma \ref{l:bbu} and the preference lattice in Figure \ref{f:ad2_alt_sequndernodeclasses}.
\end{enumerate}
\item[$\kwd{AND}(A,B)$]
\begin{enumerate}
\item The correctness of the aggregation follows directly from Lemma \ref{l:commutativity} and row $\mathbf{\overset{+}{x}\overset{0}{z}}$ and $\mathbf{\overset{-}{x}\overset{+}{z}}$ of Table \ref{t:so_ad2_alt_agg_+0}.
\end{enumerate}
\end{description}
\end{enumerate}

\item[$\mathbf{\overset{0}{x}\overset{+}{z}}$]
\begin{enumerate}
\item From Definition \ref{def:gsc_x0z+}, it follows that $B$ contains an execution where a there are no $x$ and $y$, contains a $z$, and there is no $k$ on its right.
\item For the aggregation types:
\begin{description}
\item[$\kwd{SEQ}(A,B)$]
\begin{enumerate}
\item From Definition \ref{def:ser}, it follows that the possible executions of a process block $\kwd{SEQ}(A, B)$ are the concatenation of an execution of $A$ and an execution of $B$.
\item From 1., (a), and i., it follows that the evaluation of $\mathbf{x}$ from $A$ cannot be influenced by $B$ appended on the left, while the evaluation of $\mathbf{z}$ is not influenced by the classification of $B$.
\item From ii. it follows that the classification of the aggregation is correct.
\end{enumerate}
\item[$\kwd{AND}(A,B)$]
\begin{enumerate}
\item The correctness of the aggregation follows directly from Lemma \ref{l:commutativity} and row $\mathbf{\overset{0}{x}\overset{+}{z}}$ and $\mathbf{\overset{-}{x}\overset{+}{z}}$ of Table \ref{t:so_ad2_alt_agg_0+}.
\end{enumerate}
\end{description}
\end{enumerate}

\item[$\mathbf{\overset{0}{x}\overset{0}{z}}$]
\begin{enumerate}
\item From Definition \ref{def:gsc_x0z0}, it follows that $B$ contains an execution where a there are no $x$, $y$, $z$ and $k$. 
\item For the aggregation types:
\begin{description}
\item[$\kwd{SEQ}(A,B)$]
\begin{enumerate}
\item The correctness of this row follows directly from Lemma \ref{l:gsc_neutral}.
\end{enumerate}
\item[$\kwd{AND}(A,B)$]
\begin{enumerate}
\item The correctness of this row follows directly from Lemma \ref{l:gsc_neutral}.
\end{enumerate}
\end{description}
\end{enumerate}

\item[$\mathbf{\overset{-}{x}\overset{+}{z}}$]
\begin{enumerate}
\item From Definition \ref{def:gsc_x-z+}, it follows that $B$ contains an execution containing a $z$, there is no $k$ on its right, and there is a $y$ on the $z$ left.
\item For the aggregation types:
\begin{description}
\item[$\kwd{SEQ}(A,B)$]
\begin{enumerate}
\item The correctness of this row follows directly from Lemma \ref{l:srd}.
\end{enumerate}
\item[$\kwd{AND}(A,B)$]
\begin{enumerate}
\item From 1. and (a), it follows that neither $A$ or $B$ contain elements capable of improving the classification over $\mathbf{\overset{-}{x}\overset{+}{z}}$, following the preference order in Figure \ref{f:ad2_alt_sequndernodeclasses}.
\item From i., the result of column \kwd{SEQ} and Lemma \ref{lem:inclusivity}, it follows that the aggregation is correctly classified.
\end{enumerate}
\end{description}
\end{enumerate}

\item[$\mathbf{\overset{+}{x}\overset{-}{z}}$]
\begin{enumerate}
\item From Definition \ref{def:gsc_x+z-}, it follows that $B$ contains an execution containing a $x$, there is no $y$ on its right, and there is a $k$ on the right, with no $z$ on the right of the $k$.
\item For the aggregation types:
\begin{description}
\item[$\kwd{SEQ}(A,B)$]
\begin{enumerate}
\item The correctness of this row follows directly from Lemma \ref{l:srd}.
\end{enumerate}
\item[$\kwd{AND}(A,B)$]
\begin{enumerate}
\item From Definition \ref{def:ser}, it follows that each execution resulting from $\kwd{AND}(A, B)$ consists of an execution of $A$ interleaved with an execution of $B$.
\item From 1., (a) and i., it follows that a valid execution of $\kwd{AND}(A, B)$ contains $z$ from $A$, and on its left appears the $x$ from $B$, while the $y$ from $A$ is kept on the left of $x$ and the $k$ from $B$ on the left of the $z$ from $A$.
\item From ii. and Definition \ref{def:gsc_x+z+}, it follows that the aggregation is correct.
\end{enumerate}
\end{description}
\end{enumerate}

\item[$\mathbf{\overset{0}{x}\overset{-}{z}}$]
\begin{enumerate}
\item From Definition \ref{def:gsc_x0z-}, it follows that $B$ contains an execution containing a $k$ with no $z$ on its right, and not containing any $x$ or $y$.
\item For the aggregation types:
\begin{description}
\item[$\kwd{SEQ}(A,B)$]
\begin{enumerate}
\item The correctness of this row follows directly from Lemma \ref{l:srd}.
\end{enumerate}
\item[$\kwd{AND}(A,B)$]
\begin{enumerate}
\item From Definition \ref{def:ser}, it follows that each execution resulting from $\kwd{AND}(A, B)$ consists of an execution of $A$ interleaved with an execution of $B$.
\item From 1., (a) and i., it follows that a valid execution of $\kwd{AND}(A, B)$ contains the $z$ from $B$ and keeps the $k$ from $A$ on the left, making it irrelevant. The execution still contains a $y$ on the left of $z$ and no $x$ on the right of $y$.
\item From ii. and Definition \ref{def:gsc_x-z+}, it follows that the aggregation is correct.
\end{enumerate}
\end{description}
\end{enumerate}

\item[$\mathbf{\overset{-}{xz}}$]
\begin{enumerate}
\item From Definition \ref{def:gsc_xz-}, it follows that every execution in $B$ contains a $y$ and a $k$ with no $x$ and $z$ on their respective rights, OR contains a $y$ with no $x$ on its right, and no $z$ or $k$.
\item For the aggregation types:
\begin{description}
\item[$\kwd{SEQ}(A,B)$]
\begin{enumerate}
\item The correctness of this row follows directly from Lemma \ref{l:srd}.
\end{enumerate}
\item[$\kwd{AND}(A,B)$]
\begin{enumerate}
\item From Definition \ref{def:ser}, it follows that each execution resulting from $\kwd{AND}(A, B)$ consists of an execution of $A$ interleaved with an execution of $B$.
\item From 1., (a) and i., it follows that a valid execution of $\kwd{AND}(A, B)$ contains the $z$ from $B$ and keeps the $k$, and eventual $y$, from $A$ on the left, making it irrelevant. The execution still contains a $y$ on the left of $z$ and no $x$ on the right of $y$.
\item From ii. and Definition \ref{def:gsc_x-z+}, it follows that the aggregation is correct.
\end{enumerate}
\end{description}
\end{enumerate}
\end{description}
\end{enumerate}
\end{proof}

\begin{table}[ht!]
\centering
\begin{tabular}{|c|c|c|c|}
\hline
A & B & \kwd{SEQ}(A, B) & \kwd{AND}(A, B) \\ \hline
$\mathbf{\overset{+}{x}\overset{-}{z}}$ & $\mathbf{\overset{+}{x}\overset{+}{z}}$ & $\mathbf{\overset{+}{x}\overset{+}{z}}$ & $\mathbf{\overset{+}{x}\overset{+}{z}}$ \\ \hline
$\mathbf{\overset{+}{x}\overset{-}{z}}$ & $\mathbf{\overset{+}{x}\overset{0}{z}}$ & $\mathbf{\overset{+}{x}\overset{0}{z}}$ & $\mathbf{\overset{+}{x}\overset{0}{z}}$ \\ \hline
$\mathbf{\overset{+}{x}\overset{-}{z}}$ & $\mathbf{\overset{0}{x}\overset{+}{z}}$ & $\mathbf{\overset{+}{x}\overset{+}{z}}$ & $\mathbf{\overset{+}{x}\overset{+}{z}}$ \\ \hline
$\mathbf{\overset{+}{x}\overset{-}{z}}$ & $\mathbf{\overset{0}{x}\overset{0}{z}}$& $\mathbf{\overset{+}{x}\overset{-}{z}}$ & $\mathbf{\overset{+}{x}\overset{-}{z}}$ \\ \hline
$\mathbf{\overset{+}{x}\overset{-}{z}}$ & $\mathbf{\overset{-}{x}\overset{+}{z}}$& $\mathbf{\overset{-}{x}\overset{+}{z}}$ & $\mathbf{\overset{+}{x}\overset{+}{z}}$ \\ \hline
$\mathbf{\overset{+}{x}\overset{-}{z}}$ & $\mathbf{\overset{+}{x}\overset{-}{z}}$& $\mathbf{\overset{+}{x}\overset{-}{z}}$ & $\mathbf{\overset{+}{x}\overset{-}{z}}$ \\ \hline
$\mathbf{\overset{+}{x}\overset{-}{z}}$ & $\mathbf{\overset{0}{x}\overset{-}{z}}$& $\mathbf{\overset{0}{x}\overset{-}{z}}$ & $\mathbf{\overset{+}{x}\overset{-}{z}}$ \\ \hline
$\mathbf{\overset{+}{x}\overset{-}{z}}$ & $\mathbf{\overset{-}{xz}}$& $\mathbf{\overset{-}{xz}}$ & $\mathbf{\overset{+}{x}\overset{-}{z}}$ \\ \hline
\end{tabular}
\caption{Generalised Sequence Aggregations}\label{t:so_ad2_alt_agg_+-}
\end{table}

\begin{proof}[Table \ref{t:so_ad2_alt_agg_+-}]
\begin{enumerate}
\item From Definition \ref{def:gsc_x+z-}, it follows that $A$ contains an execution containing a $x$, there is no $y$ on its right, and there is a $k$ on the right, with no $z$ on the right of the $k$.
\item For $B$:
\begin{description}
\item[$\mathbf{\overset{+}{x}\overset{+}{z}}$]
\begin{enumerate}
\item From Definition \ref{def:gsc_x+z+}, it follows that $B$ contains an execution where a $x$ is followed by or on the same task as a $z$, and there are no $y$ or $k$ on their respective right.
\item For the aggregation types:
\begin{description}
\item[$\kwd{SEQ}(A,B)$]
\begin{enumerate}
\item The correctness of this row follows directly from Lemma \ref{l:srd}.
\end{enumerate}
\item[$\kwd{AND}(A,B)$]
\begin{enumerate}
\item The correctness of this row follows directly from Lemma \ref{l:gsc_super}.
\end{enumerate}
\end{description}
\end{enumerate}

\item[$\mathbf{\overset{+}{x}\overset{0}{z}}$]
\begin{enumerate}
\item From Definition \ref{def:gsc_x+z0}, it follows that $B$ contains an execution where a $x$, and there are no $y$ or $k$ on its right.
\item For the aggregation types:
\begin{description}
\item[$\kwd{SEQ}(A,B)$]
\begin{enumerate}
\item The correctness of this row follows directly from Lemma \ref{l:bbu} and the preference lattice in Figure \ref{f:ad2_alt_sequndernodeclasses}.
\end{enumerate}
\item[$\kwd{AND}(A,B)$]
\begin{enumerate}
\item The correctness of the aggregation follows directly from Lemma \ref{l:commutativity} and row $\mathbf{\overset{+}{x}\overset{0}{z}}$ and $\mathbf{\overset{+}{x}\overset{-}{z}}$ of Table \ref{t:so_ad2_alt_agg_+0}.
\end{enumerate}
\end{description}
\end{enumerate}

\item[$\mathbf{\overset{0}{x}\overset{+}{z}}$]
\begin{enumerate}
\item From Definition \ref{def:gsc_x0z+}, it follows that $B$ contains an execution where a there are no $x$ and $y$, contains a $z$, and there is no $k$ on its right.
\item For the aggregation types:
\begin{description}
\item[$\kwd{SEQ}(A,B)$]
\begin{enumerate}
\item From Definition \ref{def:ser}, it follows that the possible executions of a process block $\kwd{SEQ}(A, B)$ are the concatenation of an execution of $A$ and an execution of $B$.
\item From 1., (a), and i., it follows that the $z$ in $B$ is on the right of the $k$ from $A$, while the $x$ from $A$ is on the left of $x$, and no $y$ is contained.
\item From ii. and Definition \ref{def:gsc_x+z+}, it follows that the classification of the aggregation is correct.
\end{enumerate}
\item[$\kwd{AND}(A,B)$]
\begin{enumerate}
\item The correctness of this row follows directly from Lemma \ref{lem:inclusivity}, the result of the column $\kwd{SEQ}(A,B)$ and the preference lattice in Figure \ref{f:ad2_alt_sequndernodeclasses}.
\end{enumerate}
\end{description}
\end{enumerate}

\item[$\mathbf{\overset{0}{x}\overset{0}{z}}$]
\begin{enumerate}
\item From Definition \ref{def:gsc_x0z0}, it follows that $B$ contains an execution where a there are no $x$, $y$, $z$ and $k$. 
\item For the aggregation types:
\begin{description}
\item[$\kwd{SEQ}(A,B)$]
\begin{enumerate}
\item The correctness of this row follows directly from Lemma \ref{l:gsc_neutral}.
\end{enumerate}
\item[$\kwd{AND}(A,B)$]
\begin{enumerate}
\item The correctness of this row follows directly from Lemma \ref{l:gsc_neutral}.
\end{enumerate}
\end{description}
\end{enumerate}

\item[$\mathbf{\overset{-}{x}\overset{+}{z}}$]
\begin{enumerate}
\item From Definition \ref{def:gsc_x-z+}, it follows that $B$ contains an execution containing a $z$, there is no $k$ on its right, and there is a $y$ on the $z$ left.
\item For the aggregation types:
\begin{description}
\item[$\kwd{SEQ}(A,B)$]
\begin{enumerate}
\item The correctness of this row follows directly from Lemma \ref{l:srd}.
\end{enumerate}
\item[$\kwd{AND}(A,B)$]
\begin{enumerate}
\item The correctness of the aggregation follows directly from Lemma \ref{l:commutativity} and row $\mathbf{\overset{-}{x}\overset{+}{z}}$ and $\mathbf{\overset{+}{x}\overset{-}{z}}$ of Table \ref{t:so_ad2_alt_agg_+0}.
\end{enumerate}
\end{description}
\end{enumerate}

\item[$\mathbf{\overset{+}{x}\overset{-}{z}}$]
\begin{enumerate}
\item From Definition \ref{def:gsc_x+z-}, it follows that $B$ contains an execution containing a $x$, there is no $y$ on its right, and there is a $k$ on the right, with no $z$ on the right of the $k$.
\item For the aggregation types:
\begin{description}
\item[$\kwd{SEQ}(A,B)$]
\begin{enumerate}
\item The correctness of this row follows directly from Lemma \ref{l:srd}.
\end{enumerate}
\item[$\kwd{AND}(A,B)$]
\begin{enumerate}
\item The correctness of this row follows directly from Lemma \ref{l:mirror}.
\end{enumerate}
\end{description}
\end{enumerate}

\item[$\mathbf{\overset{0}{x}\overset{-}{z}}$]
\begin{enumerate}
\item From Definition \ref{def:gsc_x0z-}, it follows that $B$ contains an execution containing a $k$ with no $z$ on its right, and not containing any $x$ or $y$.
\item For the aggregation types:
\begin{description}
\item[$\kwd{SEQ}(A,B)$]
\begin{enumerate}
\item The correctness of this row follows directly from Lemma \ref{l:srd}.
\end{enumerate}
\item[$\kwd{AND}(A,B)$]
\begin{enumerate}
\item From Definition \ref{def:ser}, it follows that each execution resulting from $\kwd{AND}(A, B)$ consists of an execution of $A$ interleaved with an execution of $B$.
\item From 1., (a) and i., it follows that a valid execution of $\kwd{AND}(A, B)$ contains the $x$ and $k$ from $A$ in the same order, while the $k$ from $B$ is attached on the left, making it irrelevant.
\item From ii. and Definition  \ref{def:gsc_x+z-} it follows that the aggregation is correctly classified as $\mathbf{\overset{+}{x}\overset{-}{z}}$.
\end{enumerate}
\end{description}
\end{enumerate}

\item[$\mathbf{\overset{-}{xz}}$]
\begin{enumerate}
\item From Definition \ref{def:gsc_xz-}, it follows that every execution in $B$ contains a $y$ and a $k$ with no $x$ and $z$ on their respective rights, OR contains a $y$ with no $x$ on its right, and no $z$ or $k$.
\item For the aggregation types:
\begin{description}
\item[$\kwd{SEQ}(A,B)$]
\begin{enumerate}
\item The correctness of this row follows directly from Lemma \ref{l:srd}.
\end{enumerate}
\item[$\kwd{AND}(A,B)$]
\begin{enumerate}
\item From Definition \ref{def:ser}, it follows that each execution resulting from $\kwd{AND}(A, B)$ consists of an execution of $A$ interleaved with an execution of $B$.
\item From 1., (a) and i., it follows that a valid execution of $\kwd{AND}(A, B)$ contains the $x$ and $k$ from $A$ in the same order, while the $k$, and eventual $y$, from $B$ is attached on the left, making it irrelevant.
\item From ii. and Definition  \ref{def:gsc_x+z-} it follows that the aggregation is correctly classified as $\mathbf{\overset{+}{x}\overset{-}{z}}$.
\end{enumerate}
\end{description}
\end{enumerate}
\end{description}
\end{enumerate}
\end{proof}

\begin{table}[ht!]
\centering
\begin{tabular}{|c|c|c|c|}
\hline
A & B & \kwd{SEQ}(A, B) & \kwd{AND}(A, B) \\ \hline
$\mathbf{\overset{0}{x}\overset{-}{z}}$ & $\mathbf{\overset{+}{x}\overset{+}{z}}$ & $\mathbf{\overset{+}{x}\overset{+}{z}}$ & $\mathbf{\overset{+}{x}\overset{+}{z}}$ \\ \hline
$\mathbf{\overset{0}{x}\overset{-}{z}}$ & $\mathbf{\overset{+}{x}\overset{0}{z}}$ & $\mathbf{\overset{+}{x}\overset{0}{z}}$ & $\mathbf{\overset{+}{x}\overset{0}{z}}$ \\ \hline
$\mathbf{\overset{0}{x}\overset{-}{z}}$ & $\mathbf{\overset{0}{x}\overset{+}{z}}$ & $\mathbf{\overset{0}{x}\overset{+}{z}}$ & $\mathbf{\overset{0}{x}\overset{+}{z}}$ \\ \hline
$\mathbf{\overset{0}{x}\overset{-}{z}}$ & $\mathbf{\overset{0}{x}\overset{0}{z}}$& $\mathbf{\overset{0}{x}\overset{-}{z}}$ & $\mathbf{\overset{0}{x}\overset{-}{z}}$ \\ \hline
$\mathbf{\overset{0}{x}\overset{-}{z}}$ & $\mathbf{\overset{-}{x}\overset{+}{z}}$& $\mathbf{\overset{-}{x}\overset{+}{z}}$ & $\mathbf{\overset{-}{x}\overset{+}{z}}$ \\ \hline
$\mathbf{\overset{0}{x}\overset{-}{z}}$ & $\mathbf{\overset{+}{x}\overset{-}{z}}$& $\mathbf{\overset{+}{x}\overset{-}{z}}$ & $\mathbf{\overset{+}{x}\overset{-}{z}}$ \\ \hline
$\mathbf{\overset{0}{x}\overset{-}{z}}$ & $\mathbf{\overset{0}{x}\overset{-}{z}}$& $\mathbf{\overset{0}{x}\overset{-}{z}}$ & $\mathbf{\overset{0}{x}\overset{-}{z}}$ \\ \hline
$\mathbf{\overset{0}{x}\overset{-}{z}}$ & $\mathbf{\overset{-}{xz}}$& $\mathbf{\overset{-}{xz}}$ & $\mathbf{\overset{-}{xz}}$ \\ \hline
\end{tabular}
\caption{Generalised Sequence Aggregations}\label{t:so_ad2_alt_agg_0-}
\end{table}

\begin{proof}[Table \ref{t:so_ad2_alt_agg_0-}]
\begin{enumerate}
\item From Definition \ref{def:gsc_x0z-}, it follows that $A$ contains an execution containing a $k$ with no $z$ on its right, and not containing any $x$ or $y$.
\item For $B$:
\begin{description}
\item[$\mathbf{\overset{+}{x}\overset{+}{z}}$]
\begin{enumerate}
\item From Definition \ref{def:gsc_x+z+}, it follows that $B$ contains an execution where a $x$ is followed by or on the same task as a $z$, and there are no $y$ or $k$ on their respective right.
\item For the aggregation types:
\begin{description}
\item[$\kwd{SEQ}(A,B)$]
\begin{enumerate}
\item The correctness of this row follows directly from Lemma \ref{l:srd}.
\end{enumerate}
\item[$\kwd{AND}(A,B)$]
\begin{enumerate}
\item The correctness of this row follows directly from Lemma \ref{l:gsc_super}.
\end{enumerate}
\end{description}
\end{enumerate}

\item[$\mathbf{\overset{+}{x}\overset{0}{z}}$]
\begin{enumerate}
\item From Definition \ref{def:gsc_x+z0}, it follows that $B$ contains an execution where a $x$, and there are no $y$ or $k$ on its right.
\item For the aggregation types:
\begin{description}
\item[$\kwd{SEQ}(A,B)$]
\begin{enumerate}
\item The correctness of this row follows directly from Lemma \ref{l:bbu} and the preference lattice in Figure \ref{f:ad2_alt_sequndernodeclasses}.
\end{enumerate}
\item[$\kwd{AND}(A,B)$]
\begin{enumerate}
\item The correctness of the aggregation follows directly from Lemma \ref{l:commutativity} and row $\mathbf{\overset{+}{x}\overset{0}{z}}$ and $\mathbf{\overset{0}{x}\overset{-}{z}}$ of Table \ref{t:so_ad2_alt_agg_+0}.
\end{enumerate}
\end{description}
\end{enumerate}

\item[$\mathbf{\overset{0}{x}\overset{+}{z}}$]
\begin{enumerate}
\item From Definition \ref{def:gsc_x0z+}, it follows that $B$ contains an execution where a there are no $x$ and $y$, contains a $z$, and there is no $k$ on its right.
\item For the aggregation types:
\begin{description}
\item[$\kwd{SEQ}(A,B)$]
\begin{enumerate}
\item From Definition \ref{def:ser}, it follows that the possible executions of a process block $\kwd{SEQ}(A, B)$ are the concatenation of an execution of $A$ and an execution of $B$.
\item From 1., (a), and i., it follows that the $z$ in $B$ is on the right of the $k$ from $A$, while no $y$ is contained.
\item From ii. and Definition \ref{def:gsc_x0z+}, it follows that the classification of the aggregation is correct.
\end{enumerate}
\item[$\kwd{AND}(A,B)$]
\begin{enumerate}
\item The correctness of the aggregation follows directly from Lemma \ref{l:commutativity} and row $\mathbf{\overset{0}{x}\overset{+}{z}}$ and $\mathbf{\overset{0}{x}\overset{-}{z}}$ of Table \ref{t:so_ad2_alt_agg_0+}.
\end{enumerate}
\end{description}
\end{enumerate}

\item[$\mathbf{\overset{0}{x}\overset{0}{z}}$]
\begin{enumerate}
\item From Definition \ref{def:gsc_x0z0}, it follows that $B$ contains an execution where a there are no $x$, $y$, $z$ and $k$. 
\item For the aggregation types:
\begin{description}
\item[$\kwd{SEQ}(A,B)$]
\begin{enumerate}
\item The correctness of this row follows directly from Lemma \ref{l:gsc_neutral}.
\end{enumerate}
\item[$\kwd{AND}(A,B)$]
\begin{enumerate}
\item The correctness of this row follows directly from Lemma \ref{l:gsc_neutral}.
\end{enumerate}
\end{description}
\end{enumerate}

\item[$\mathbf{\overset{-}{x}\overset{+}{z}}$]
\begin{enumerate}
\item From Definition \ref{def:gsc_x-z+}, it follows that $B$ contains an execution containing a $z$, there is no $k$ on its right, and there is a $y$ on the $z$ left.
\item For the aggregation types:
\begin{description}
\item[$\kwd{SEQ}(A,B)$]
\begin{enumerate}
\item The correctness of this row follows directly from Lemma \ref{l:srd}.
\end{enumerate}
\item[$\kwd{AND}(A,B)$]
\begin{enumerate}
\item The correctness of the aggregation follows directly from Lemma \ref{l:commutativity} and row $\mathbf{\overset{-}{x}\overset{+}{z}}$ and $\mathbf{\overset{0}{x}\overset{-}{z}}$ of Table \ref{t:so_ad2_alt_agg_-+}.
\end{enumerate}
\end{description}
\end{enumerate}

\item[$\mathbf{\overset{+}{x}\overset{-}{z}}$]
\begin{enumerate}
\item From Definition \ref{def:gsc_x+z-}, it follows that $B$ contains an execution containing a $x$, there is no $y$ on its right, and there is a $k$ on the right, with no $z$ on the right of the $k$.
\item For the aggregation types:
\begin{description}
\item[$\kwd{SEQ}(A,B)$]
\begin{enumerate}
\item The correctness of this row follows directly from Lemma \ref{l:srd}.
\end{enumerate}
\item[$\kwd{AND}(A,B)$]
\begin{enumerate}
\item The correctness of the aggregation follows directly from Lemma \ref{l:commutativity} and row $\mathbf{\overset{+}{x}\overset{-}{z}}$ and $\mathbf{\overset{0}{x}\overset{-}{z}}$ of Table \ref{t:so_ad2_alt_agg_+-}.
\end{enumerate}
\end{description}
\end{enumerate}

\item[$\mathbf{\overset{0}{x}\overset{-}{z}}$]
\begin{enumerate}
\item From Definition \ref{def:gsc_x0z-}, it follows that $B$ contains an execution containing a $k$ with no $z$ on its right, and not containing any $x$ or $y$.
\item For the aggregation types:
\begin{description}
\item[$\kwd{SEQ}(A,B)$]
\begin{enumerate}
\item The correctness of this row follows directly from Lemma \ref{l:srd}.
\end{enumerate}
\item[$\kwd{AND}(A,B)$]
\begin{enumerate}
\item The correctness of this row follows directly from Lemma \ref{l:mirror}.
\end{enumerate}
\end{description}
\end{enumerate}

\item[$\mathbf{\overset{-}{xz}}$]
\begin{enumerate}
\item From Definition \ref{def:gsc_xz-}, it follows that every execution in $B$ contains a $y$ and a $k$ with no $x$ and $z$ on their respective rights, OR contains a $y$ with no $x$ on its right, and no $z$ or $k$.
\item For the aggregation types:
\begin{description}
\item[$\kwd{SEQ}(A,B)$]
\begin{enumerate}
\item The correctness of this row follows directly from Lemma \ref{l:srd}.
\end{enumerate}
\item[$\kwd{AND}(A,B)$]
\begin{enumerate}
\item From Definition \ref{def:ser}, it follows that each execution resulting from $\kwd{AND}(A, B)$ consists of an execution of $A$ interleaved with an execution of $B$.
\item From 1., (a) and i., it follows that a valid execution of $\kwd{AND}(A, B)$ contains the $k$ from $A$ and $y$ from $B$. The order is not relevant.
\item From ii. and Definition  \ref{def:gsc_xz-} it follows that the aggregation is correctly classified as $\mathbf{\overset{-}{xz}}$.
\end{enumerate}
\end{description}
\end{enumerate}
\end{description}
\end{enumerate}
\end{proof}

\begin{table}[ht!]
\centering
\begin{tabular}{|c|c|c|c|}
\hline
A & B & \kwd{SEQ}(A, B) & \kwd{AND}(A, B) \\ \hline
$\mathbf{\overset{-}{xz}}$ & $\mathbf{\overset{+}{x}\overset{+}{z}}$ & $\mathbf{\overset{+}{x}\overset{+}{z}}$ & $\mathbf{\overset{+}{x}\overset{+}{z}}$ \\ \hline
$\mathbf{\overset{-}{xz}}$ & $\mathbf{\overset{+}{x}\overset{0}{z}}$ & $\mathbf{\overset{+}{x}\overset{0}{z}}$ & $\mathbf{\overset{+}{x}\overset{0}{z}}$ \\ \hline
$\mathbf{\overset{-}{xz}}$ & $\mathbf{\overset{0}{x}\overset{+}{z}}$ & $\mathbf{\overset{-}{x}\overset{+}{z}}$ & $\mathbf{\overset{-}{x}\overset{+}{z}}$ \\ \hline
$\mathbf{\overset{-}{xz}}$ & $\mathbf{\overset{0}{x}\overset{0}{z}}$& $\mathbf{\overset{-}{xz}}$ & $\mathbf{\overset{-}{xz}}$ \\ \hline
$\mathbf{\overset{-}{xz}}$ & $\mathbf{\overset{-}{x}\overset{+}{z}}$& $\mathbf{\overset{-}{x}\overset{+}{z}}$ & $\mathbf{\overset{-}{x}\overset{+}{z}}$ \\ \hline
$\mathbf{\overset{-}{xz}}$ & $\mathbf{\overset{+}{x}\overset{-}{z}}$& $\mathbf{\overset{+}{x}\overset{-}{z}}$ & $\mathbf{\overset{+}{x}\overset{-}{z}}$ \\ \hline
$\mathbf{\overset{-}{xz}}$ & $\mathbf{\overset{0}{x}\overset{-}{z}}$& $\mathbf{\overset{-}{xz}}$ & $\mathbf{\overset{-}{xz}}$ \\ \hline
$\mathbf{\overset{-}{xz}}$ & $\mathbf{\overset{-}{xz}}$& $\mathbf{\overset{-}{xz}}$ & $\mathbf{\overset{-}{xz}}$ \\ \hline
\end{tabular}
\caption{Generalised Sequence Aggregations}\label{t:so_ad2_alt_agg_-}
\end{table}

\begin{proof}[Table \ref{t:so_ad2_alt_agg_-}]
\begin{enumerate}
\item From Definition \ref{def:gsc_xz-}, it follows that every execution in $A$ contains a $y$ and a $k$ with no $x$ and $z$ on their respective rights, OR contains a $y$ with no $x$ on its right, and no $z$ or $k$.
\item For $B$:
\begin{description}
\item[$\mathbf{\overset{+}{x}\overset{+}{z}}$]
\begin{enumerate}
\item From Definition \ref{def:gsc_x+z+}, it follows that $B$ contains an execution where a $x$ is followed by or on the same task as a $z$, and there are no $y$ or $k$ on their respective right.
\item For the aggregation types:
\begin{description}
\item[$\kwd{SEQ}(A,B)$]
\begin{enumerate}
\item The correctness of this row follows directly from Lemma \ref{l:srd}.
\end{enumerate}
\item[$\kwd{AND}(A,B)$]
\begin{enumerate}
\item The correctness of this row follows directly from Lemma \ref{l:gsc_super}.
\end{enumerate}
\end{description}
\end{enumerate}

\item[$\mathbf{\overset{+}{x}\overset{0}{z}}$]
\begin{enumerate}
\item From Definition \ref{def:gsc_x+z0}, it follows that $B$ contains an execution where a $x$, and there are no $y$ or $k$ on its right.
\item For the aggregation types:
\begin{description}
\item[$\kwd{SEQ}(A,B)$]
\begin{enumerate}
\item The correctness of this row follows directly from Lemma \ref{l:bbu} and the preference lattice in Figure \ref{f:ad2_alt_sequndernodeclasses}.
\end{enumerate}
\item[$\kwd{AND}(A,B)$]
\begin{enumerate}
\item asdf
\end{enumerate}
\end{description}
\end{enumerate}

\item[$\mathbf{\overset{0}{x}\overset{+}{z}}$]
\begin{enumerate}
\item From Definition \ref{def:gsc_x0z+}, it follows that $B$ contains an execution where a there are no $x$ and $y$, contains a $z$, and there is no $k$ on its right.
\item For the aggregation types:
\begin{description}
\item[$\kwd{SEQ}(A,B)$]
\begin{enumerate}
\item From Definition \ref{def:ser}, it follows that the possible executions of a process block $\kwd{SEQ}(A, B)$ are the concatenation of an execution of $A$ and an execution of $B$.
\item From 1., (a), and i., it follows that the $z$ in $B$ is on the right of the $y$, and eventual $k$ from $A$.
\item From ii. and Definition \ref{def:gsc_x-z+}, it follows that the classification of the aggregation is correct.
\end{enumerate}
\item[$\kwd{AND}(A,B)$]
\begin{enumerate}
\item The correctness of the aggregation follows directly from Lemma \ref{l:commutativity} and row $\mathbf{\overset{0}{x}\overset{+}{z}}$ and $\mathbf{\overset{-}{xz}}$ of Table \ref{t:so_ad2_alt_agg_0+}.
\end{enumerate}
\end{description}
\end{enumerate}

\item[$\mathbf{\overset{0}{x}\overset{0}{z}}$]
\begin{enumerate}
\item From Definition \ref{def:gsc_x0z0}, it follows that $B$ contains an execution where a there are no $x$, $y$, $z$ and $k$. 
\item For the aggregation types:
\begin{description}
\item[$\kwd{SEQ}(A,B)$]
\begin{enumerate}
\item The correctness of this row follows directly from Lemma \ref{l:gsc_neutral}.
\end{enumerate}
\item[$\kwd{AND}(A,B)$]
\begin{enumerate}
\item The correctness of this row follows directly from Lemma \ref{l:gsc_neutral}.
\end{enumerate}
\end{description}
\end{enumerate}

\item[$\mathbf{\overset{-}{x}\overset{+}{z}}$]
\begin{enumerate}
\item From Definition \ref{def:gsc_x-z+}, it follows that $B$ contains an execution containing a $z$, there is no $k$ on its right, and there is a $y$ on the $z$ left.
\item For the aggregation types:
\begin{description}
\item[$\kwd{SEQ}(A,B)$]
\begin{enumerate}
\item The correctness of this row follows directly from Lemma \ref{l:srd}.
\end{enumerate}
\item[$\kwd{AND}(A,B)$]
\begin{enumerate}
\item The correctness of the aggregation follows directly from Lemma \ref{l:commutativity} and row $\mathbf{\overset{-}{x}\overset{+}{z}}$ and $\mathbf{\overset{-}{xz}}$ of Table \ref{t:so_ad2_alt_agg_-+}.
\end{enumerate}
\end{description}
\end{enumerate}

\item[$\mathbf{\overset{+}{x}\overset{-}{z}}$]
\begin{enumerate}
\item From Definition \ref{def:gsc_x+z-}, it follows that $B$ contains an execution containing a $x$, there is no $y$ on its right, and there is a $k$ on the right, with no $z$ on the right of the $k$.
\item For the aggregation types:
\begin{description}
\item[$\kwd{SEQ}(A,B)$]
\begin{enumerate}
\item The correctness of this row follows directly from Lemma \ref{l:srd}.
\end{enumerate}
\item[$\kwd{AND}(A,B)$]
\begin{enumerate}
\item The correctness of the aggregation follows directly from Lemma \ref{l:commutativity} and row $\mathbf{\overset{+}{x}\overset{-}{z}}$ and $\mathbf{\overset{-}{xz}}$ of Table \ref{t:so_ad2_alt_agg_+-}.
\end{enumerate}
\end{description}
\end{enumerate}

\item[$\mathbf{\overset{0}{x}\overset{-}{z}}$]
\begin{enumerate}
\item From Definition \ref{def:gsc_x0z-}, it follows that $B$ contains an execution containing a $k$ with no $z$ on its right, and not containing any $x$ or $y$.
\item For the aggregation types:
\begin{description}
\item[$\kwd{SEQ}(A,B)$]
\begin{enumerate}
\item The correctness of this row follows directly from Lemma \ref{l:srd}.
\end{enumerate}
\item[$\kwd{AND}(A,B)$]
\begin{enumerate}
\item The correctness of the aggregation follows directly from Lemma \ref{l:commutativity} and row $\mathbf{\overset{0}{x}\overset{-}{z}}$ and $\mathbf{\overset{-}{xz}}$ of Table \ref{t:so_ad2_alt_agg_0-}.
\end{enumerate}
\end{description}
\end{enumerate}

\item[$\mathbf{\overset{-}{xz}}$]
\begin{enumerate}
\item From Definition \ref{def:gsc_xz-}, it follows that every execution in $B$ contains a $y$ and a $k$ with no $x$ and $z$ on their respective rights, OR contains a $y$ with no $x$ on its right, and no $z$ or $k$.
\item For the aggregation types:
\begin{description}
\item[$\kwd{SEQ}(A,B)$]
\begin{enumerate}
\item The correctness of this row follows directly from Lemma \ref{l:srd}.
\end{enumerate}
\item[$\kwd{AND}(A,B)$]
\begin{enumerate}
\item The correctness of this row follows directly from Lemma \ref{l:mirror}.
\end{enumerate}
\end{description}
\end{enumerate}
\end{description}
\end{enumerate}
\end{proof}
\newpage
\subsection{Aggregating GSP with ISP for \kwd{AND} Overnode Evaluation}

When evaluating an overnode of type \kwd{A} for $\overline{A}\Delta2.1$ or $\overline{A}\Delta2.2$, its overnode child is classified according to the \emph{Generalised Sequence Pattern}, while its direct undernode children are classified according to the \emph{Interval Sub-Pattern}. The following tables illustrate the classification of an overnode of type \kwd{AND} where its overnode child classification is aggregated with the overall aggregation of its direct children classification according to the interval sub-pattern aggregation tables.

\subsubsection{Lemmas}

\begin{lemma}[Aggregation Neutral Class]\label{l:aoe_neutral}
Given a process block $A$, assigned to the evaluation class $\mathbf{\underline{\overset{0}{x}\overset{0}{z}}}$, and another process block $B$, assigned to any of the available evaluation classes. Let $C$ be a process block having $A$ and $B$ as its sub-blocks, then the evaluation class of $C$ is the same class as the process block $B$.
\end{lemma}

\begin{proof}
This proof follows closely the proof for Lemma \ref{l:neutral}.
\end{proof}

\begin{lemma}[Super Elements]\label{l:aoe_super}
$\mathbf{\overset{+}{x}\overset{+}{z}t}$ and $\mathbf{\underline{\overset{+}{x}\overset{+}{z}}}$
\end{lemma}

\begin{proof}
As the aggregations between a GSP class and an ISP class is always done when the common parent overnode is of type \kwd{AND}, the proof follows closely the proof for Lemma \ref{l:gsc_super}.
\end{proof}

\begin{lemma}[Hide the Bubu]\label{l:aoe_htb}
When evaluating an \kwd{AND} overnode and one of the elements being aggregated is classified as $\mathbf{\underline{\overset{-}{xz}}}$, then the resulting classification is always the overnode classification corresponding to the other element.
\end{lemma}

\begin{proof}
The worst element classification, $\mathbf{\underline{\overset{-}{xz}}}$, as it cannot improve any classification the aggregation as it involves an \kwd{AND} block it always allows to ``hide'' the \emph{problematic} segment on the right of $t$, hence not modifying (i.e. making it worse) the current classification on the stem.
\end{proof}

\subsubsection{Aggregations}

\begin{table}[ht!]
\centering
\begin{tabular}{|c|c|c|}
\hline
Overnode Child & Undernode Children & Result \\ \hline
$\mathbf{\overset{+}{x}\overset{+}{z}t}$ & $\mathbf{\underline{\overset{+}{x}\overset{+}{z}}}$ & $\mathbf{\overset{+}{x}\overset{+}{z}t}$ \\ \hline
$\mathbf{\overset{+}{x}\overset{+}{z}t}$ & $\mathbf{\underline{\overset{+}{x}\overset{0}{z}}}$ & $\mathbf{\overset{+}{x}\overset{+}{z}t}$ \\ \hline
$\mathbf{\overset{+}{x}\overset{+}{z}t}$ & $\mathbf{\underline{\overset{0}{x}\overset{+}{z}}}$ & $\mathbf{\overset{+}{x}\overset{+}{z}t}$ \\ \hline
$\mathbf{\overset{+}{x}\overset{+}{z}t}$ & $\mathbf{\underline{\overset{+}{x}\overset{-}{z}}}$ & $\mathbf{\overset{+}{x}\overset{+}{z}t}$ \\ \hline
$\mathbf{\overset{+}{x}\overset{+}{z}t}$ & $\mathbf{\underline{\overset{0}{x}\overset{0}{z}}}$ & $\mathbf{\overset{+}{x}\overset{+}{z}t}$ \\ \hline
$\mathbf{\overset{+}{x}\overset{+}{z}t}$ & $\mathbf{\underline{\overset{-}{x}\overset{+}{z}}}$ & $\mathbf{\overset{+}{x}\overset{+}{z}t}$ \\ \hline
$\mathbf{\overset{+}{x}\overset{+}{z}t}$ & $\mathbf{\underline{\overset{-}{xz}}}$ & $\mathbf{\overset{+}{x}\overset{+}{z}t}$ \\ \hline
\end{tabular}
\caption{Aggregation Table \kwd{AND} Overnode}\label{t:ad2_alt_aoa_++t}
\end{table}

\begin{proof}[Table \ref{t:ad2_alt_aoa_++t}]
\begin{enumerate}
\item The Overnode Child is classified as $\mathbf{\overset{+}{x}\overset{+}{z}t}$.
\item For Undernode Children:
\begin{description}
\item[$\mathbf{\underline{\overset{+}{x}\overset{+}{z}}}$] The correctness of the row follows directly from Lemma \ref{l:aoe_super}.
\item[$\mathbf{\underline{\overset{+}{x}\overset{0}{z}}}$] The correctness of the row follows directly from Lemma \ref{l:aoe_super}.
\item[$\mathbf{\underline{\overset{0}{x}\overset{+}{z}}}$] The correctness of the row follows directly from Lemma \ref{l:aoe_super}.
\item[$\mathbf{\underline{\overset{+}{x}\overset{-}{z}}}$] The correctness of the row follows directly from Lemma \ref{l:aoe_super}.
\item[$\mathbf{\underline{\overset{0}{x}\overset{0}{z}}}$] The correctness of the row follows directly from Lemma \ref{l:aoe_neutral}.
\item[$\mathbf{\underline{\overset{-}{x}\overset{+}{z}}}$] The correctness of the row follows directly from Lemma \ref{l:aoe_super}.
\item[$\mathbf{\underline{\overset{-}{xz}}}$] The correctness of the row follows directly from Lemma \ref{l:aoe_htb}.
\end{description}
\end{enumerate}
\end{proof}

\begin{table}[ht!]
\centering
\begin{tabular}{|c|c|c|}
\hline
Overnode Child & Undernode Children & Result \\ \hline
$\mathbf{\overset{+}{x}\overset{0}{z}t}$ & $\mathbf{\underline{\overset{+}{x}\overset{+}{z}}}$ & $\mathbf{\overset{+}{x}\overset{+}{z}t}$ \\ \hline
$\mathbf{\overset{+}{x}\overset{0}{z}t}$ & $\mathbf{\underline{\overset{+}{x}\overset{0}{z}}}$ & $\mathbf{\overset{+}{x}\overset{0}{z}t}$ \\ \hline
$\mathbf{\overset{+}{x}\overset{0}{z}t}$ & $\mathbf{\underline{\overset{0}{x}\overset{+}{z}}}$ & $\mathbf{\overset{+}{x}\overset{+}{z}t}$ \\ \hline
$\mathbf{\overset{+}{x}\overset{0}{z}t}$ & $\mathbf{\underline{\overset{+}{x}\overset{-}{z}}}$ & $\mathbf{\overset{+}{x}\overset{0}{z}t}$ \\ \hline
$\mathbf{\overset{+}{x}\overset{0}{z}t}$ & $\mathbf{\underline{\overset{0}{x}\overset{0}{z}}}$ & $\mathbf{\overset{+}{x}\overset{0}{z}t}$ \\ \hline
$\mathbf{\overset{+}{x}\overset{0}{z}t}$ & $\mathbf{\underline{\overset{-}{x}\overset{+}{z}}}$ & $\mathbf{\overset{+}{x}\overset{+}{z}t}$ \\ \hline
$\mathbf{\overset{+}{x}\overset{0}{z}t}$ & $\mathbf{\underline{\overset{-}{xz}}}$ & $\mathbf{\overset{+}{x}\overset{0}{z}t}$ \\ \hline
\end{tabular}
\caption{Aggregation Table \kwd{AND} Overnode}\label{t:ad2_alt_aoa_+0t}
\end{table}

\begin{proof}[Table \ref{t:ad2_alt_aoa_+0t}]
\begin{enumerate}
\item From Definition \ref{def:gsc_x+z0}, it follows that Overnode Child contains an execution where a $x$, and there are no $y$ or $k$ on its right.
\item The aggregation is of type $\kwd{AND}$ and from Definition \ref{def:ser}, it follows that each execution resulting from the aggregation consists of an execution of Overnode Child interleaved with an execution of Undernode Children.
\item For Undernode Children:
\begin{description}
\item[$\mathbf{\underline{\overset{+}{x}\overset{+}{z}}}$] The correctness of the row follows directly from Lemma \ref{l:aoe_super}.
\item[$\mathbf{\underline{\overset{+}{x}\overset{0}{z}}}$] 
\begin{enumerate}
\item From Definition \ref{def:isp_1}, it follows that \emph{Overnode's Undernode Children} contains an execution containing a $x$ and no $y$ or $z$.
\item From 1. and (a), it follows that the Undernode Children do not contain elements to improve the classification for the $\mathbf{z}$ sub-classification of Overnode Child.
\item From 2., it follows that the execution of  Undernode Children can be put on the right of the $t$ from Overnode Child.
\item From (b) and (c), it follows that the classification is correct.
\end{enumerate}
\item[$\mathbf{\underline{\overset{0}{x}\overset{+}{z}}}$] 
\begin{enumerate}
\item From Definition \ref{def:isp_1'}, it follows that \emph{Overnode's Undernode Children} contains an execution containing a $z$ and no $y$ or $x$.
\item From 1., 2. and (a), it follows that the $z$ from Undernode Children can be put between the $x$ and the $t$ from Overnode Child.
\item From (b) and Definition \ref{def:gsc_x+z+}, it follows that the classification of the aggregation is correct.
\end{enumerate}
\item[$\mathbf{\underline{\overset{+}{x}\overset{-}{z}}}$] 
\begin{enumerate}
\item From Definition \ref{def:isp_1r}, it follows that \emph{Overnode's Undernode Children} contains an execution containing a $x$ and a $y$ on its right, while not containing $z$.
\item From 1. and (a), it follows that the Undernode Children do not contain elements to improve the classification for the $\mathbf{z}$ sub-classification of Overnode Child.
\item From 2., it follows that the execution of  Undernode Children can be put on the right of the $t$ from Overnode Child.
\item From (b) and (c), it follows that the classification is correct.
\end{enumerate}
\item[$\mathbf{\underline{\overset{0}{x}\overset{0}{z}}}$] The correctness of the row follows directly from Lemma \ref{l:aoe_neutral}.
\item[$\mathbf{\underline{\overset{-}{x}\overset{+}{z}}}$] 
\begin{enumerate}
\item From Definition \ref{def:isp_1'l}, it follows that \emph{Overnode's Undernode Children} contains an execution containing a $z$ and a $y$ on its left, while not containing $x$.
\item From 1., 2. and (a), it follows that the $z$ from Undernode Children can be put between the $x$ and the $t$ from Overnode Child, while the $y$ from Undernode Children can be put on the left of $x$ from Overnode Child.
\item From (b) and Definition \ref{def:gsc_x+z+}, it follows that the classification of the aggregation is correct.
\end{enumerate}
\item[$\mathbf{\underline{\overset{-}{xz}}}$] The correctness of the row follows directly from Lemma \ref{l:aoe_htb}.
\end{description}
\end{enumerate}
\end{proof}

\begin{table}[ht!]
\centering
\begin{tabular}{|c|c|c|}
\hline
Overnode Child & Undernode Children & Result \\ \hline
$\mathbf{\overset{0}{x}\overset{+}{z}t}$ & $\mathbf{\underline{\overset{+}{x}\overset{+}{z}}}$ & $\mathbf{\overset{+}{x}\overset{+}{z}t}$ \\ \hline
$\mathbf{\overset{0}{x}\overset{+}{z}t}$ & $\mathbf{\underline{\overset{+}{x}\overset{0}{z}}}$ & $\mathbf{\overset{+}{x}\overset{+}{z}t}$ \\ \hline
$\mathbf{\overset{0}{x}\overset{+}{z}t}$ & $\mathbf{\underline{\overset{0}{x}\overset{+}{z}}}$ & $\mathbf{\overset{0}{x}\overset{+}{z}t}$ \\ \hline
$\mathbf{\overset{0}{x}\overset{+}{z}t}$ & $\mathbf{\underline{\overset{+}{x}\overset{-}{z}}}$ & $\mathbf{\overset{+}{x}\overset{+}{z}t}$ \\ \hline
$\mathbf{\overset{0}{x}\overset{+}{z}t}$ & $\mathbf{\underline{\overset{0}{x}\overset{0}{z}}}$ & $\mathbf{\overset{0}{x}\overset{+}{z}t}$ \\ \hline
$\mathbf{\overset{0}{x}\overset{+}{z}t}$ & $\mathbf{\underline{\overset{-}{x}\overset{+}{z}}}$ & $\mathbf{\overset{0}{x}\overset{+}{z}t}$ \\ \hline
$\mathbf{\overset{0}{x}\overset{+}{z}t}$ & $\mathbf{\underline{\overset{-}{xz}}}$ & $\mathbf{\overset{0}{x}\overset{+}{z}t}$ \\ \hline
\end{tabular}
\caption{Aggregation Table \kwd{AND} Overnode}\label{t:ad2_alt_aoa_0+t}
\end{table}

\begin{proof}[Table \ref{t:ad2_alt_aoa_0+t}]
\begin{enumerate}
\item From Definition \ref{def:gsc_x0z+}, it follows that Overnode Child contains an execution where a there are no $x$ and $y$, contains a $z$, and there is no $k$ on its right.
\item The aggregation is of type $\kwd{AND}$ and from Definition \ref{def:ser}, it follows that each execution resulting from the aggregation consists of an execution of Overnode Child interleaved with an execution of Undernode Children.
\item For Undernode Children:
\begin{description}
\item[$\mathbf{\underline{\overset{+}{x}\overset{+}{z}}}$] The correctness of the row follows directly from Lemma \ref{l:aoe_super}.
\item[$\mathbf{\underline{\overset{+}{x}\overset{0}{z}}}$] 
\begin{enumerate}
\item From Definition \ref{def:isp_1}, it follows that \emph{Overnode's Undernode Children} contains an execution containing a $x$ and no $y$ or $z$.
\item From 1., 2. and (a), it follows that the $x$ from Undernode Children can be put on the left of the $z$ and the $t$ from Overnode Child.
\item From (b) and Definition \ref{def:gsc_x+z+}, it follows that the classification of the aggregation is correct.
\end{enumerate}
\item[$\mathbf{\underline{\overset{0}{x}\overset{+}{z}}}$] 
\begin{enumerate}
\item From Definition \ref{def:isp_1'}, it follows that \emph{Overnode's Undernode Children} contains an execution containing a $z$ and no $y$ or $x$.
\item From 1. and (a), it follows that the Undernode Children do not contain elements to improve the classification for the $\mathbf{x}$ sub-classification of Overnode Child.
\item From 2., it follows that the execution of  Undernode Children can be put on the right of the $t$ from Overnode Child.
\item From (b) and (c), it follows that the classification is correct.
\end{enumerate}
\item[$\mathbf{\underline{\overset{+}{x}\overset{-}{z}}}$] 
\begin{enumerate}
\item From Definition \ref{def:isp_1r}, it follows that \emph{Overnode's Undernode Children} contains an execution containing a $x$ and a $y$ on its right, while not containing $z$.
\item From 1., 2. and (a), it follows that the $x$ from Undernode Children can be put on the left of the $z$ and the $t$ from Overnode Child, while the $y$ from Undernode Children can be put on the right of $t$ from Overnode Child.
\item From (b) and Definition \ref{def:gsc_x+z+}, it follows that the classification of the aggregation is correct.
\end{enumerate}
\item[$\mathbf{\underline{\overset{0}{x}\overset{0}{z}}}$] The correctness of the row follows directly from Lemma \ref{l:aoe_neutral}.
\item[$\mathbf{\underline{\overset{-}{x}\overset{+}{z}}}$] 
\begin{enumerate}
\item From Definition \ref{def:isp_1'l}, it follows that \emph{Overnode's Undernode Children} contains an execution containing a $z$ and a $y$ on its left, while not containing $x$.
\item From 1. and (a), it follows that the Undernode Children do not contain elements to improve the classification for the $\mathbf{x}$ sub-classification of Overnode Child.
\item From 2., it follows that the execution of  Undernode Children can be put on the right of the $t$ from Overnode Child.
\item From (b) and (c), it follows that the classification is correct.
\end{enumerate}
\item[$\mathbf{\underline{\overset{-}{xz}}}$] The correctness of the row follows directly from Lemma \ref{l:aoe_htb}.
\end{description}
\end{enumerate}
\end{proof}

\begin{table}[ht!]
\centering
\begin{tabular}{|c|c|c|}
\hline
Overnode Child & Undernode Children & Result \\ \hline
$\mathbf{\overset{0}{x}\overset{0}{z}t}$ & $\mathbf{\underline{\overset{+}{x}\overset{+}{z}}}$ & $\mathbf{\overset{+}{x}\overset{+}{z}t}$ \\ \hline
$\mathbf{\overset{0}{x}\overset{0}{z}t}$ & $\mathbf{\underline{\overset{+}{x}\overset{0}{z}}}$ & $\mathbf{\overset{+}{x}\overset{0}{z}t}$ \\ \hline
$\mathbf{\overset{0}{x}\overset{0}{z}t}$ & $\mathbf{\underline{\overset{0}{x}\overset{+}{z}}}$ & $\mathbf{\overset{0}{x}\overset{+}{z}t}$ \\ \hline
$\mathbf{\overset{0}{x}\overset{0}{z}t}$ & $\mathbf{\underline{\overset{+}{x}\overset{-}{z}}}$ & $\mathbf{\overset{+}{x}\overset{0}{z}t}$ \\ \hline
$\mathbf{\overset{0}{x}\overset{0}{z}t}$ & $\mathbf{\underline{\overset{0}{x}\overset{0}{z}}}$ & $\mathbf{\overset{0}{x}\overset{0}{z}t}$ \\ \hline
$\mathbf{\overset{0}{x}\overset{0}{z}t}$ & $\mathbf{\underline{\overset{-}{x}\overset{+}{z}}}$ & $\mathbf{\overset{0}{x}\overset{0}{z}t}$ / $\mathbf{\overset{-}{x}\overset{+}{z}t}$ \\ \hline
$\mathbf{\overset{0}{x}\overset{0}{z}t}$ & $\mathbf{\underline{\overset{-}{xz}}}$ & $\mathbf{\overset{0}{x}\overset{0}{z}t}$ \\ \hline
\end{tabular}
\caption{Aggregation Table \kwd{AND} Overnode}\label{t:ad2_alt_aoa_00t}
\end{table}

\begin{proof}[Table \ref{t:ad2_alt_aoa_00t}]
\begin{enumerate}
\item From Definition \ref{def:gsc_x0z0}, it follows that Overnode Child contains an execution where a there are no $x$, $y$, $z$ and $k$. 
\item The aggregation is of type $\kwd{AND}$ and from Definition \ref{def:ser}, it follows that each execution resulting from the aggregation consists of an execution of Overnode Child interleaved with an execution of Undernode Children.
\item For Undernode Children:
\begin{description}
\item[$\mathbf{\underline{\overset{+}{x}\overset{+}{z}}}$] The correctness of the row follows directly from Lemma \ref{l:aoe_super}.
\item[$\mathbf{\underline{\overset{+}{x}\overset{0}{z}}}$] 
\begin{enumerate}
\item From Definition \ref{def:isp_1}, it follows that \emph{Overnode's Undernode Children} contains an execution containing a $x$ and no $y$ or $z$.
\item From 1. and (a), it follows that the $x$ from the Undernode Children can be placed on the left of the $t$ from Overnode Child.
\item From (b) and Definition \ref{def:gsc_x+z0}, it follows that the classification of the aggregation is correct.
\end{enumerate}
\item[$\mathbf{\underline{\overset{0}{x}\overset{+}{z}}}$] 
\begin{enumerate}
\item From Definition \ref{def:isp_1'}, it follows that \emph{Overnode's Undernode Children} contains an execution containing a $z$ and no $y$ or $x$.
\item From 1. and (a), it follows that the $z$ from the Undernode Children can be placed on the left of the $t$ from Overnode Child.
\item From (b) and Definition \ref{def:gsc_x0z+}, it follows that the classification of the aggregation is correct.
\end{enumerate}
\item[$\mathbf{\underline{\overset{+}{x}\overset{-}{z}}}$] 
\begin{enumerate}
\item From Definition \ref{def:isp_1r}, it follows that \emph{Overnode's Undernode Children} contains an execution containing a $x$ and a $y$ on its right, while not containing $z$.
\item From 1. and (a), it follows that the $x$ from the Undernode Children can be placed on the left of the $t$ from Overnode Child, while the $y$ from Undernode Children can be put on the right of the $t$.
\item From (b) and Definition \ref{def:gsc_x+z0}, it follows that the classification of the aggregation is correct.
\end{enumerate}
\item[$\mathbf{\underline{\overset{0}{x}\overset{0}{z}}}$] The correctness of the row follows directly from Lemma \ref{l:aoe_neutral}.
\item[$\mathbf{\underline{\overset{-}{x}\overset{+}{z}}}$] 
\begin{enumerate}
\item From Definition \ref{def:isp_1'l}, it follows that \emph{Overnode's Undernode Children} contains an execution containing a $z$ and a $y$ on its left, while not containing $x$.
\item From 1. and (a), it follows that the whole execution from Undernode Children can be put on the right of $t$ from Overnode Child.
\item From 1. and (a), it follows that the $y$ and $z$ from Undernode children can be placed on the left of $t$ from Overnode Child.
\item From (b) and Definition \ref{def:gsc_x0z0}, it follows that  $\mathbf{\overset{0}{x}\overset{0}{z}t}$ is a valid classification.
\item From (c) and Definition \ref{def:gsc_x-z+}, it follows that $\mathbf{\overset{-}{x}\overset{+}{z}t}$ is a valid classification.
\item From (d), (e) and the preference orders in Figure \ref{f:ad2_alt_sequndernodeclasses}, it follows that both classifications are viable.
\end{enumerate}
\item[$\mathbf{\underline{\overset{-}{xz}}}$] The correctness of the row follows directly from Lemma \ref{l:aoe_htb}.
\end{description}
\end{enumerate}
\end{proof}

\begin{table}[ht!]
\centering
\begin{tabular}{|c|c|c|}
\hline
Overnode Child & Undernode Children & Result \\ \hline
$\mathbf{\overset{-}{x}\overset{+}{z}t}$ & $\mathbf{\underline{\overset{+}{x}\overset{+}{z}}}$ & $\mathbf{\overset{+}{x}\overset{+}{z}t}$ \\ \hline
$\mathbf{\overset{-}{x}\overset{+}{z}t}$ & $\mathbf{\underline{\overset{+}{x}\overset{0}{z}}}$ & $\mathbf{\overset{+}{x}\overset{+}{z}t}$ \\ \hline
$\mathbf{\overset{-}{x}\overset{+}{z}t}$ & $\mathbf{\underline{\overset{0}{x}\overset{+}{z}}}$ & $\mathbf{\overset{-}{x}\overset{+}{z}t}$ \\ \hline
$\mathbf{\overset{-}{x}\overset{+}{z}t}$ & $\mathbf{\underline{\overset{+}{x}\overset{-}{z}}}$ & $\mathbf{\overset{+}{x}\overset{+}{z}t}$ \\ \hline
$\mathbf{\overset{-}{x}\overset{+}{z}t}$ & $\mathbf{\underline{\overset{0}{x}\overset{0}{z}}}$ & $\mathbf{\overset{-}{x}\overset{+}{z}t}$ \\ \hline
$\mathbf{\overset{-}{x}\overset{+}{z}t}$ & $\mathbf{\underline{\overset{-}{x}\overset{+}{z}}}$ & $\mathbf{\overset{-}{x}\overset{+}{z}t}$ \\ \hline
$\mathbf{\overset{-}{x}\overset{+}{z}t}$ & $\mathbf{\underline{\overset{-}{xz}}}$ & $\mathbf{\overset{-}{x}\overset{+}{z}t}$ \\ \hline
\end{tabular}
\caption{Aggregation Table \kwd{AND} Overnode}\label{t:ad2_alt_aoa_-+t}
\end{table}

\begin{proof}[Table \ref{t:ad2_alt_aoa_-+t}]
\begin{enumerate}
\item From Definition \ref{def:gsc_x-z+}, it follows that Overnode Child contains an execution containing a $z$, there is no $k$ on its right, and there is a $y$ on the $z$ left.
\item The aggregation is of type $\kwd{AND}$ and from Definition \ref{def:ser}, it follows that each execution resulting from the aggregation consists of an execution of Overnode Child interleaved with an execution of Undernode Children.
\item For Undernode Children:
\begin{description}
\item[$\mathbf{\underline{\overset{+}{x}\overset{+}{z}}}$] The correctness of the row follows directly from Lemma \ref{l:aoe_super}.
\item[$\mathbf{\underline{\overset{+}{x}\overset{0}{z}}}$] 
\begin{enumerate}
\item From Definition \ref{def:isp_1}, it follows that \emph{Overnode's Undernode Children} contains an execution containing a $x$ and no $y$ or $z$.
\item From 1., 2. and (a), it follows that the $x$ from Undernode Children can be put on the left of the $z$ and the $t$ from Overnode Child, while still being on the right of the $y$ from Overnode Child.
\item From (b) and Definition \ref{def:gsc_x+z+}, it follows that the classification of the aggregation is correct.
\end{enumerate}
\item[$\mathbf{\underline{\overset{0}{x}\overset{+}{z}}}$] 
\begin{enumerate}
\item From Definition \ref{def:isp_1'}, it follows that \emph{Overnode's Undernode Children} contains an execution containing a $z$ and no $y$ or $x$.
\item From 1. and (a), it follows that the Undernode Children do not contain elements to improve the classification for the $\mathbf{x}$ sub-classification of Overnode Child.
\item From 2., it follows that the execution of  Undernode Children can be put on the right of the $t$ from Overnode Child.
\item From (b) and (c), it follows that the classification is correct.
\end{enumerate}
\item[$\mathbf{\underline{\overset{+}{x}\overset{-}{z}}}$] 
\begin{enumerate}
\item From Definition \ref{def:isp_1r}, it follows that \emph{Overnode's Undernode Children} contains an execution containing a $x$ and a $y$ on its right, while not containing $z$.
\item From 1., 2. and (a), it follows that the $x$ from Undernode Children can be put on the left of the $z$ and the $t$ from Overnode Child and having the $y$ from Overnode Child on its left, while the $y$ from Undernode Children can be put on the right of $t$ from Overnode Child.
\item From (b) and Definition \ref{def:gsc_x+z+}, it follows that the classification of the aggregation is correct.
\end{enumerate}
\item[$\mathbf{\underline{\overset{0}{x}\overset{0}{z}}}$] The correctness of the row follows directly from Lemma \ref{l:aoe_neutral}.
\item[$\mathbf{\underline{\overset{-}{x}\overset{+}{z}}}$] 
\begin{enumerate}
\item From Definition \ref{def:isp_1'l}, it follows that \emph{Overnode's Undernode Children} contains an execution containing a $z$ and a $y$ on its left, while not containing $x$.
\item From 1. and (a), it follows that the Undernode Children do not contain elements to improve the classification for the $\mathbf{x}$ sub-classification of Overnode Child.
\item From 2., it follows that the execution of  Undernode Children can be put on the right of the $t$ from Overnode Child.
\item From (b) and (c), it follows that the classification is correct.
\end{enumerate}
\item[$\mathbf{\underline{\overset{-}{xz}}}$] The correctness of the row follows directly from Lemma \ref{l:aoe_htb}.
\end{description}
\end{enumerate}
\end{proof}

\begin{table}[ht!]
\centering
\begin{tabular}{|c|c|c|}
\hline
Overnode Child & Undernode Children & Result \\ \hline
$\mathbf{\overset{+}{x}\overset{-}{z}t}$ & $\mathbf{\underline{\overset{+}{x}\overset{+}{z}}}$ & $\mathbf{\overset{+}{x}\overset{+}{z}t}$ \\ \hline
$\mathbf{\overset{+}{x}\overset{-}{z}t}$ & $\mathbf{\underline{\overset{+}{x}\overset{0}{z}}}$ & $\mathbf{\overset{+}{x}\overset{-}{z}t}$ \\ \hline
$\mathbf{\overset{+}{x}\overset{-}{z}t}$ & $\mathbf{\underline{\overset{0}{x}\overset{+}{z}}}$ & $\mathbf{\overset{+}{x}\overset{+}{z}t}$ \\ \hline
$\mathbf{\overset{+}{x}\overset{-}{z}t}$ & $\mathbf{\underline{\overset{+}{x}\overset{-}{z}}}$ & $\mathbf{\overset{+}{x}\overset{-}{z}t}$ \\ \hline
$\mathbf{\overset{+}{x}\overset{-}{z}t}$ & $\mathbf{\underline{\overset{0}{x}\overset{0}{z}}}$ & $\mathbf{\overset{+}{x}\overset{-}{z}t}$ \\ \hline
$\mathbf{\overset{+}{x}\overset{-}{z}t}$ & $\mathbf{\underline{\overset{-}{x}\overset{+}{z}}}$ & $\mathbf{\overset{+}{x}\overset{+}{z}t}$ \\ \hline
$\mathbf{\overset{+}{x}\overset{-}{z}t}$ & $\mathbf{\underline{\overset{-}{xz}}}$ & $\mathbf{\overset{+}{x}\overset{-}{z}t}$  \\ \hline
\end{tabular}
\caption{Aggregation Table \kwd{AND} Overnode}\label{t:ad2_alt_aoa_+-t}
\end{table}

\begin{proof}[Table \ref{t:ad2_alt_aoa_+-t}]
\begin{enumerate}
\item From Definition \ref{def:gsc_x+z-}, it follows that Overnode Child contains an execution containing a $x$, there is no $y$ on its right, and there is a $k$ on the right, with no $z$ on the right of the $k$.
\item The aggregation is of type $\kwd{AND}$ and from Definition \ref{def:ser}, it follows that each execution resulting from the aggregation consists of an execution of Overnode Child interleaved with an execution of Undernode Children.
\item For Undernode Children:
\begin{description}
\item[$\mathbf{\underline{\overset{+}{x}\overset{+}{z}}}$] The correctness of the row follows directly from Lemma \ref{l:aoe_super}.
\item[$\mathbf{\underline{\overset{+}{x}\overset{0}{z}}}$] 
\begin{enumerate}
\item From Definition \ref{def:isp_1}, it follows that \emph{Overnode's Undernode Children} contains an execution containing a $x$ and no $y$ or $z$.
\item From 1. and (a), it follows that the Undernode Children do not contain elements to improve the classification for the $\mathbf{z}$ sub-classification of Overnode Child.
\item From 2., it follows that the execution of  Undernode Children can be put on the right of the $t$ from Overnode Child.
\item From (b) and (c), it follows that the classification is correct.
\end{enumerate}
\item[$\mathbf{\underline{\overset{0}{x}\overset{+}{z}}}$] 
\begin{enumerate}
\item From Definition \ref{def:isp_1'}, it follows that \emph{Overnode's Undernode Children} contains an execution containing a $z$ and no $y$ or $x$.
\item From 1., 2. and (a), it follows that the $z$ from Undernode Children can be put between the $x$ and the $t$ from Overnode Child, while also being on the right of the $k$ from Overnode Child.
\item From (b) and Definition \ref{def:gsc_x+z+}, it follows that the classification of the aggregation is correct.
\end{enumerate}
\item[$\mathbf{\underline{\overset{+}{x}\overset{-}{z}}}$] 
\begin{enumerate}
\item From Definition \ref{def:isp_1r}, it follows that \emph{Overnode's Undernode Children} contains an execution containing a $x$ and a $y$ on its right, while not containing $z$.
\item From 1. and (a), it follows that the Undernode Children do not contain elements to improve the classification for the $\mathbf{z}$ sub-classification of Overnode Child.
\item From 2., it follows that the execution of  Undernode Children can be put on the right of the $t$ from Overnode Child.
\item From (b) and (c), it follows that the classification is correct.
\end{enumerate}
\item[$\mathbf{\underline{\overset{0}{x}\overset{0}{z}}}$] The correctness of the row follows directly from Lemma \ref{l:aoe_neutral}.
\item[$\mathbf{\underline{\overset{-}{x}\overset{+}{z}}}$] 
\begin{enumerate}
\item From Definition \ref{def:isp_1'l}, it follows that \emph{Overnode's Undernode Children} contains an execution containing a $z$ and a $y$ on its left, while not containing $x$.
\item From 1., 2. and (a), it follows that the $z$ from Undernode Children can be put between the $x$ and the $t$ from Overnode Child, while also being on the right of the $k$ from Overnode Child. Moreover the $y$ from Undernode Children can be put on the left of $x$ from Overnode Child.
\item From (b) and Definition \ref{def:gsc_x+z+}, it follows that the classification of the aggregation is correct.
\end{enumerate}
\item[$\mathbf{\underline{\overset{-}{xz}}}$] The correctness of the row follows directly from Lemma \ref{l:aoe_htb}.
\end{description}
\end{enumerate}
\end{proof}

\begin{table}[ht!]
\centering
\begin{tabular}{|c|c|c|}
\hline
Overnode Child & Undernode Children & Result \\ \hline
$\mathbf{\overset{0}{x}\overset{-}{z}t}$ & $\mathbf{\underline{\overset{+}{x}\overset{+}{z}}}$ & $\mathbf{\overset{+}{x}\overset{+}{z}t}$ \\ \hline
$\mathbf{\overset{0}{x}\overset{-}{z}t}$ & $\mathbf{\underline{\overset{+}{x}\overset{0}{z}}}$ & $\mathbf{\overset{+}{x}\overset{0}{z}t}$ \\ \hline
$\mathbf{\overset{0}{x}\overset{-}{z}t}$ & $\mathbf{\underline{\overset{0}{x}\overset{+}{z}}}$ & $\mathbf{\overset{-}{x}\overset{+}{z}t}$ \\ \hline
$\mathbf{\overset{0}{x}\overset{-}{z}t}$ & $\mathbf{\underline{\overset{+}{x}\overset{-}{z}}}$ & $\mathbf{\overset{+}{x}\overset{0}{z}t}$ \\ \hline
$\mathbf{\overset{0}{x}\overset{-}{z}t}$ & $\mathbf{\underline{\overset{0}{x}\overset{0}{z}}}$ & $\mathbf{\overset{0}{x}\overset{-}{z}t}$ \\ \hline
$\mathbf{\overset{0}{x}\overset{-}{z}t}$ & $\mathbf{\underline{\overset{-}{x}\overset{+}{z}}}$ & $\mathbf{\overset{0}{x}\overset{-}{z}t}$ / $\mathbf{\overset{-}{x}\overset{+}{z}t}$ \\ \hline
$\mathbf{\overset{0}{x}\overset{-}{z}t}$ & $\mathbf{\underline{\overset{-}{xz}}}$ & $\mathbf{\overset{0}{x}\overset{-}{z}t}$ \\ \hline
\end{tabular}
\caption{Aggregation Table \kwd{AND} Overnode}\label{t:ad2_alt_aoa_0-t}
\end{table}

\begin{proof}[Table \ref{t:ad2_alt_aoa_0-t}]
\begin{enumerate}
\item From Definition \ref{def:gsc_x0z-}, it follows that Overnode Child contains an execution containing a $k$ with no $z$ on its right, and not containing any $x$ or $y$.
\item The aggregation is of type $\kwd{AND}$ and from Definition \ref{def:ser}, it follows that each execution resulting from the aggregation consists of an execution of Overnode Child interleaved with an execution of Undernode Children.
\item For Undernode Children:
\begin{description}
\item[$\mathbf{\underline{\overset{+}{x}\overset{+}{z}}}$] The correctness of the row follows directly from Lemma \ref{l:aoe_super}.
\item[$\mathbf{\underline{\overset{+}{x}\overset{0}{z}}}$] 
\begin{enumerate}
\item From Definition \ref{def:isp_1}, it follows that \emph{Overnode's Undernode Children} contains an execution containing a $x$ and no $y$ or $z$.
\item From 1., 2. and (a), it follows that the $x$ from Undernode Children can be placed immediately on the left of the $t$ from the Overnode Child.
\item From (b) and Definition \ref{def:gsc_x+z0}, it follows that the classification of the aggregation is correct.
\end{enumerate}
\item[$\mathbf{\underline{\overset{0}{x}\overset{+}{z}}}$] 
\begin{enumerate}
\item From Definition \ref{def:isp_1'}, it follows that \emph{Overnode's Undernode Children} contains an execution containing a $z$ and no $y$ or $x$.
\item From 1., 2. and (a), it follows that the $z$ from Undernode Children can be placed immediately on the left of the $t$ from the Overnode Child.
\item From (b) and Definition \ref{def:gsc_x-z+}, it follows that the classification of the aggregation is correct.
\end{enumerate}
\item[$\mathbf{\underline{\overset{+}{x}\overset{-}{z}}}$] 
\begin{enumerate}
\item From Definition \ref{def:isp_1r}, it follows that \emph{Overnode's Undernode Children} contains an execution containing a $x$ and a $y$ on its right, while not containing $z$.
\item From 1., 2. and (a), it follows that the $x$ from Undernode Children can be placed immediately on the left of the $t$ from the Overnode Child, while the $y$ from the Undernode Children is set on the right of the $t$.
\item From (b) and Definition \ref{def:gsc_x+z0}, it follows that the classification of the aggregation is correct.
\end{enumerate}
\item[$\mathbf{\underline{\overset{0}{x}\overset{0}{z}}}$] The correctness of the row follows directly from Lemma \ref{l:aoe_neutral}.
\item[$\mathbf{\underline{\overset{-}{x}\overset{+}{z}}}$] 
\begin{enumerate}
\item From Definition \ref{def:isp_1'l}, it follows that \emph{Overnode's Undernode Children} contains an execution containing a $z$ and a $y$ on its left, while not containing $x$.
\item From 1., 2. and (a), it follows that the $z$ from Undernode Children can be placed immediately on the left of the $t$ from the Overnode Child.
\item From 1., 2. and (a), it follows that the whole execution from Undernode Children can be placed on the right of the $t$ from the Overnode Child.
\item From (b) and Definition \ref{def:gsc_x-z+}, it follows that $\mathbf{\overset{-}{x}\overset{+}{z}t}$ is a valid classification of the aggregation.
\item From (c) and Definition \ref{def:gsc_x0z-}, it follows that $\mathbf{\overset{0}{x}\overset{-}{z}t}$ is a valid classification of the aggregation.
\item From (d), (e) and the preference orders in Figure \ref{f:ad2_alt_sequndernodeclasses}, it follows that both classifications are viable.
\end{enumerate}
\item[$\mathbf{\underline{\overset{-}{xz}}}$] The correctness of the row follows directly from Lemma \ref{l:aoe_htb}.
\end{description}
\end{enumerate}
\end{proof}

\begin{table}[ht!]
\centering
\begin{tabular}{|c|c|c|}
\hline
Overnode Child & Undernode Children & Result \\ \hline
$\mathbf{\overset{-}{xz}t}$ & $\mathbf{\underline{\overset{+}{x}\overset{+}{z}}}$ & $\mathbf{\overset{+}{x}\overset{+}{z}t}$ \\ \hline
$\mathbf{\overset{-}{xz}t}$ & $\mathbf{\underline{\overset{+}{x}\overset{0}{z}}}$ & $\mathbf{\overset{+}{x}\overset{0}{z}t}$ \\ \hline
$\mathbf{\overset{-}{xz}t}$ & $\mathbf{\underline{\overset{0}{x}\overset{+}{z}}}$ & $\mathbf{\overset{-}{x}\overset{+}{z}t}$\\ \hline
$\mathbf{\overset{-}{xz}t}$ & $\mathbf{\underline{\overset{+}{x}\overset{-}{z}}}$ & $\mathbf{\overset{+}{x}\overset{0}{z}t}$ \\ \hline
$\mathbf{\overset{-}{xz}t}$ & $\mathbf{\underline{\overset{0}{x}\overset{0}{z}}}$ & $\mathbf{\overset{-}{xz}t}$ \\ \hline
$\mathbf{\overset{-}{xz}t}$ & $\mathbf{\underline{\overset{-}{x}\overset{+}{z}}}$ &  $\mathbf{\overset{-}{x}\overset{+}{z}t}$ \\ \hline
$\mathbf{\overset{-}{xz}t}$ & $\mathbf{\underline{\overset{-}{xz}}}$ & $\mathbf{\overset{-}{xz}t}$ \\ \hline
\end{tabular}
\caption{Aggregation Table \kwd{AND} Overnode}\label{t:ad2_alt_aoa_-t}
\end{table}

\begin{proof}[Table \ref{t:ad2_alt_aoa_-t}]
\begin{enumerate}
\item From Definition \ref{def:gsc_x0z0}, it follows that every execution of the Overnode Child contains a $y$ and a $k$ with no $x$ and $z$ on their respective rights, OR contains a $y$ with no $x$ on its right, and no $z$ or $k$.
\item The aggregation is of type $\kwd{AND}$ and from Definition \ref{def:ser}, it follows that each execution resulting from the aggregation consists of an execution of Overnode Child interleaved with an execution of Undernode Children.
\item For Undernode Children:
\begin{description}
\item[$\mathbf{\underline{\overset{+}{x}\overset{+}{z}}}$] The correctness of the row follows directly from Lemma \ref{l:aoe_super}.
\item[$\mathbf{\underline{\overset{+}{x}\overset{0}{z}}}$] 
\begin{enumerate}
\item From Definition \ref{def:isp_1}, it follows that \emph{Overnode's Undernode Children} contains an execution containing a $x$ and no $y$ or $z$.
\item From 1., 2. and (a), it follows that the $x$ from Undernode Children can be placed immediately on the left of the $t$ from the Overnode Child.
\item From (b) and Definition \ref{def:gsc_x+z0}, it follows that the classification of the aggregation is correct.
\end{enumerate}
\item[$\mathbf{\underline{\overset{0}{x}\overset{+}{z}}}$] 
\begin{enumerate}
\item From Definition \ref{def:isp_1'}, it follows that \emph{Overnode's Undernode Children} contains an execution containing a $z$ and no $y$ or $x$.
\item From 1., 2. and (a), it follows that the $z$ from Undernode Children can be placed immediately on the left of the $t$ from the Overnode Child.
\item From (b) and Definition \ref{def:gsc_x-z+}, it follows that the classification of the aggregation is correct.
\end{enumerate}
\item[$\mathbf{\underline{\overset{+}{x}\overset{-}{z}}}$] 
\begin{enumerate}
\item From Definition \ref{def:isp_1r}, it follows that \emph{Overnode's Undernode Children} contains an execution containing a $x$ and a $y$ on its right, while not containing $z$.
\item From 1., 2. and (a), it follows that the $x$ from Undernode Children can be placed immediately on the left of the $t$ from the Overnode Child, while the $y$ from the Undernode Children is set on the right of the $t$.
\item From (b) and Definition \ref{def:gsc_x+z0}, it follows that the classification of the aggregation is correct.
\end{enumerate}
\item[$\mathbf{\underline{\overset{0}{x}\overset{0}{z}}}$] The correctness of the row follows directly from Lemma \ref{l:aoe_neutral}.
\item[$\mathbf{\underline{\overset{-}{x}\overset{+}{z}}}$] 
\begin{enumerate}
\item From Definition \ref{def:isp_1'l}, it follows that \emph{Overnode's Undernode Children} contains an execution containing a $z$ and a $y$ on its left, while not containing $x$.
\item From 1., 2. and (a), it follows that the $z$ from Undernode Children can be placed immediately on the left of the $t$ from the Overnode Child.
\item From (b) and Definition \ref{def:gsc_x-z+}, it follows that the classification of the aggregation is correct.
\end{enumerate}
\item[$\mathbf{\underline{\overset{-}{xz}}}$] The correctness of the row follows directly from Lemma \ref{l:aoe_htb}.
\end{description}
\end{enumerate}
\end{proof}

\newpage
\subsection{Interleaved Generic Pattern}

\begin{definition}[$\mathbf{\overset{+}{k}}$]\label{def:igp_+} generic pattern fulfilling state: there exists an execution belonging to the process block containing a task having $d$ annotated, and not containing the failing annotation.

\noindent\textbf{Formally}: 

\noindent Given a process block $B$, it belongs to this class if and only if:
\begin{itemize}
\item $\exists \exe \in \Exe{B}$ such that:
\begin{itemize}
\item $\exists t_k \in \exe$ and
\item $\not\exists t_y \in \exe$ such that $t_y$ is the same task as $t_k$.
\end{itemize}
\end{itemize}
\end{definition}

\begin{definition}[$\mathbf{\overset{0}{k}}$]\label{def:igp_0} neutral state: the block cannot be classified as $\mathbf{\overset{+}{k}}$.

\noindent\textbf{Formally}: 

\noindent Given a process block $B$, it belongs to this class if and only if:
\begin{itemize}
\item $\forall \exe \in \Exe{B}$ such that either:
\begin{itemize}
\item $\not\exists t_k \in \exe$ or,
\item $\forall t_k \in \exe, \exists t_y$ such that $t_k$ and $t_y$ are the same task.
\end{itemize}
\end{itemize}
\end{definition}

\begin{theorem}[Classification Completeness for Interleaved Generic Pattern]
The set of possible evaluations of Interleaved Generic Pattern is completely covered by the provided set of classifications. 
\end{theorem}

\begin{proof}
As the classification for the Interleaved Generic Pattern is binary, we can see that the completeness is given by the two mutually exclusive and complementary classes in Definition \ref{def:igp_+} and Definition \ref{def:igp_0}.
\end{proof}

\subsubsection{Aggregations}

\begin{table}[ht!]
\centering
\begin{tabular}{|c|c|c|c|}
\hline
A & B & \kwd{SEQ}(A, B) & \kwd{AND}(A, B) \\ \hline
$\mathbf{\overset{0}{k}}$ & $\mathbf{\overset{0}{k}}$ & $\mathbf{\overset{0}{k}}$ & $\mathbf{\overset{0}{k}}$\\ \hline
$\mathbf{\overset{0}{k}}$ & $\mathbf{\overset{+}{k}}$ & $\mathbf{\overset{+}{k}}$ & $\mathbf{\overset{+}{k}}$  \\ \hline
$\mathbf{\overset{+}{k}}$ & $\mathbf{\overset{0}{k}}$ & $\mathbf{\overset{+}{k}}$ & $\mathbf{\overset{+}{k}}$ \\ \hline
$\mathbf{\overset{+}{k}}$ & $\mathbf{\overset{+}{k}}$ & $\mathbf{\overset{+}{k}}$ & $\mathbf{\overset{+}{k}}$  \\ \hline
\end{tabular}
\caption{Interleaved Generic Pattern Aggregation}\label{t:md1_ataol_app}
\end{table}

\begin{proof}[Table \ref{t:md1_ataol_app}]
The table checks whether exists a task with a given property in a possible execution of the block being evaluated. Therefore either the task exists ($\mathbf{\overset{+}{k}}$) and it therefore keeps existing in every further aggregation, or it does not ($\mathbf{\overset{0}{k}}$) but the classification can still be overwritten when it is aggregated with a positive classification.
\end{proof}

\begin{table}[ht!]
\centering
\begin{tabular}{|c|c|c|}
\hline
Overnode's Overnode Child & Overnode's Undernode Children & Result \\ \hline
$\mathbf{\overset{-}{x/z}t}$ & $\mathbf{\overset{-}{x}}$ & $\mathbf{\overset{0}{x}t}$ \\ \hline
$\mathbf{\overset{-}{x/z}t}$ & $\mathbf{\overset{-}{x}}$ & $\mathbf{\overset{+}{x}t}$ \\ \hline
$\mathbf{\overset{0}{x/z}t}$ & $\mathbf{\overset{0}{k}}$ & $\mathbf{\overset{0}{x}t}$ \\ \hline
$\mathbf{\overset{0}{x/z}t}$ & $\mathbf{\overset{+}{k}}$ & $\mathbf{\overset{+}{x}t}$ \\ \hline
$\mathbf{\overset{+}{x/z}t}$ & $\mathbf{\overset{0}{k}}$ & $\mathbf{\overset{+}{x}t}$ \\ \hline
$\mathbf{\overset{+}{x/z}t}$ & $\mathbf{\overset{+}{k}}$ & $\mathbf{\overset{+}{x}t}$ \\ \hline
\end{tabular}
\caption{Aggregating Interleaved Generic Pattern with Left/Right Sub-Pattern Classes for Overnodes}\label{t:md1_ataol_stem_app}
\end{table}

\begin{proof}[Table \ref{t:md1_ataol_stem_app}]
Table \ref{t:md1_ataol_stem_app} shows how the Interleaved Generic Pattern classification of the undernodes is aggregated with the Left/Right Sub-Pattern classification of an overnode. The correctness of these aggregations follows directly from the correctness of the \kwd{AND} columns of Table \ref{tab:atsol} and Table \ref{tab:atsor} when considering one of the columns $A$ or $B$ containing the equivalent classification for the Interleaved Generic Pattern, and the other the respective version of the overnode class for the Left or Right Sub-Pattern respectively.
\end{proof}
\newpage

\end{document}